\documentclass[sigconf, nonacm]{acmart}

\newcommand\vldbdoi{XX.XX/XXX.XX}
\newcommand\vldbpages{XXX-XXX}
\newcommand\vldbvolume{15}
\newcommand\vldbissue{11}
\newcommand\vldbyear{2022}
\newcommand\vldbauthors{Chenghao Lyu, Qi Fan, Fei Song, Arnab Sinha, Yanlei Diao, Wei Chen, Li Ma, Yihui Feng, Yaliang Li, Kai Zeng, Jingren Zhou}
\newcommand\vldbtitle{\shorttitle} 
\newcommand\vldbavailabilityurl{}
\newcommand\vldbpagestyle{empty} 

\UseRawInputEncoding
\usepackage{graphicx}
\usepackage{subfigure,wrapfig}
\usepackage{balance}
\usepackage{amsmath}
\usepackage{color}
\usepackage{tabto}
\usepackage{mathrsfs}
\usepackage{mathtools}
\usepackage{multirow}
\usepackage{enumitem}
\usepackage{booktabs}
\newcommand{\ra}[1]{\renewcommand{\arraystretch}{#1}}
\usepackage{algorithmic}
\usepackage[ruled,vlined]{algorithm2e}
\usepackage{bm}
\usepackage{soul}
\usepackage{amsthm} 

\mathtoolsset{showonlyrefs}
\newcommand{\minip}[1]{\vspace{0.05in} \noindent \textbf{#1}}

\newcommand{\cut}[1]{}
\newcommand{\todo}[1]{{\color{blue}[\textit{#1}]}}
\newcommand{\chenghao}[1]{{\color{black}#1}}

\newcommand{\sketch}[1]{{\color{cyan}[\textit{#1}]}}
\newcommand{\status}[1]{{\color{magenta}#1}}
\newcommand{\rv}[1]{{\color{black}#1}}

\newcommand{\po}{Pareto-optimal }
\newcommand{\bs}{\boldsymbol}

\newboolean{show-tr}
\setboolean{show-tr}{true}
\newcommand{\techreport}[2]{\ifthenelse{\boolean{show-tr}}{{#1}}{{#2}}}

\theoremstyle{definition}
\newtheorem{definition}{Definition}[section] 
\newtheorem{proposition}{Proposition}[section]
\newtheorem{lemma}{Lemma}
\newtheorem{theorem}{Theorem}[section]

\begin{document}
\title{Fine-Grained Modeling and Optimization for Intelligent Resource Management in Big Data Processing}



\author{
Chenghao Lyu$^\dagger$,
Qi Fan$^\ddagger$,
Fei Song$^\ddagger$,
Arnab Sinha$^\ddagger$,
Yanlei Diao$^{\dagger \ddagger}$}
\author{Wei Chen$^*$,
Li Ma$^*$,
Yihui Feng$^*$,
Yaliang Li$^*$,
Kai Zeng$^*$,
Jingren Zhou$^*$} 
\affiliation{
$^\dagger$ University of Massachusetts, Amherst;
$^\ddagger$ Ecole Polytechnique;
$^*$ Alibaba Group} 
\email{chenghao@cs.umass.edu, {qi.fan, fei.song, arnab.sinha, yanlei.diao}@polytechnique.edu}
\email{{wickeychen.cw, mali.mali, yihui.feng, yaliang.li, zengkai.zk, jingren.zhou}@alibaba-inc.com}

%
%
%
%
%

\begin{abstract}

Big data processing at the production scale presents a highly complex environment for resource optimization (RO), a problem crucial for meeting performance goals and budgetary constraints of analytical users. 
The RO problem is challenging because it involves a set of decisions (the partition count, placement of parallel instances on machines, and resource allocation to each instance), requires multi-objective optimization (MOO), and is compounded by the scale and complexity of big data systems while having to meet stringent time constraints for scheduling. 
This paper presents a MaxCompute based integrated system to support multi-objective resource optimization via fine-grained instance-level modeling and optimization. We propose a new architecture that breaks RO into a series of simpler problems, new fine-grained predictive models, and novel optimization methods that exploit these models to make effective instance-level RO decisions well under a second.
Evaluation using production workloads shows that our new RO system could reduce \rv{37-72\%} latency and \rv{43-78\%} cost at the same time, compared to the current optimizer and scheduler, while running in \rv{0.02-0.23s}.

\end{abstract}


\maketitle

\pagestyle{\vldbpagestyle}
\techreport{}{
\begingroup\small\noindent\raggedright\textbf{PVLDB Reference Format:}\\
\vldbauthors. \vldbtitle. PVLDB, \vldbvolume(\vldbissue): \vldbpages, \vldbyear.\\
\href{https://doi.org/\vldbdoi}{doi:\vldbdoi}
\endgroup
\begingroup
\renewcommand\thefootnote{}\footnote{\noindent
This work is licensed under the Creative Commons BY-NC-ND 4.0 International License. Visit \url{https://creativecommons.org/licenses/by-nc-nd/4.0/} to view a copy of this license. For any use beyond those covered by this license, obtain permission by emailing \href{mailto:info@vldb.org}{info@vldb.org}. Copyright is held by the owner/author(s). Publication rights licensed to the VLDB Endowment. \\
\raggedright Proceedings of the VLDB Endowment, Vol. \vldbvolume, No. \vldbissue\ %
ISSN 2150-8097. \\
\href{https://doi.org/\vldbdoi}{doi:\vldbdoi} \\
}\addtocounter{footnote}{-1}\endgroup
}

\ifdefempty{\vldbavailabilityurl}{}{
\vspace{.3cm}
\begingroup\small\noindent\raggedright\textbf{PVLDB Artifact Availability:}\\
The source code, data, and/or other artifacts have been made available at \url{\vldbavailabilityurl}.
\endgroup
}


\section{Introduction}
\label{sec:intro}

While big data query processing has become commonplace in enterprise businesses and many platforms have been developed for this purpose~\cite{BorkarCGOV11,DeanG04,hadoop,fuxi-vldb14,Flink-recovery2015,GatesNCKNORSS09,maxcompute,Murray+13:naiad,scope-vldb12,spark-rdd-nsdi12,ThusooSJSCALWM09,trident-vldb21,Xin+2013:shark}, 
resource optimization in large clusters~\cite{JyothiCMNTYMGKK16,RajanKCK16,cleo-sigmod20,spark-moo-icde21} has received less attention. 
However, we observe from real-world experiences of running large compute clusters in the Alibaba Cloud that resource management plays a vital role in meeting both performance goals and budgetary constraints of internal and external analytical users.

\cut{
\begin{figure*}
\centering
\begin{minipage}{.7\textwidth}
  \centering
  \includegraphics[height=6.4cm,width=\linewidth]{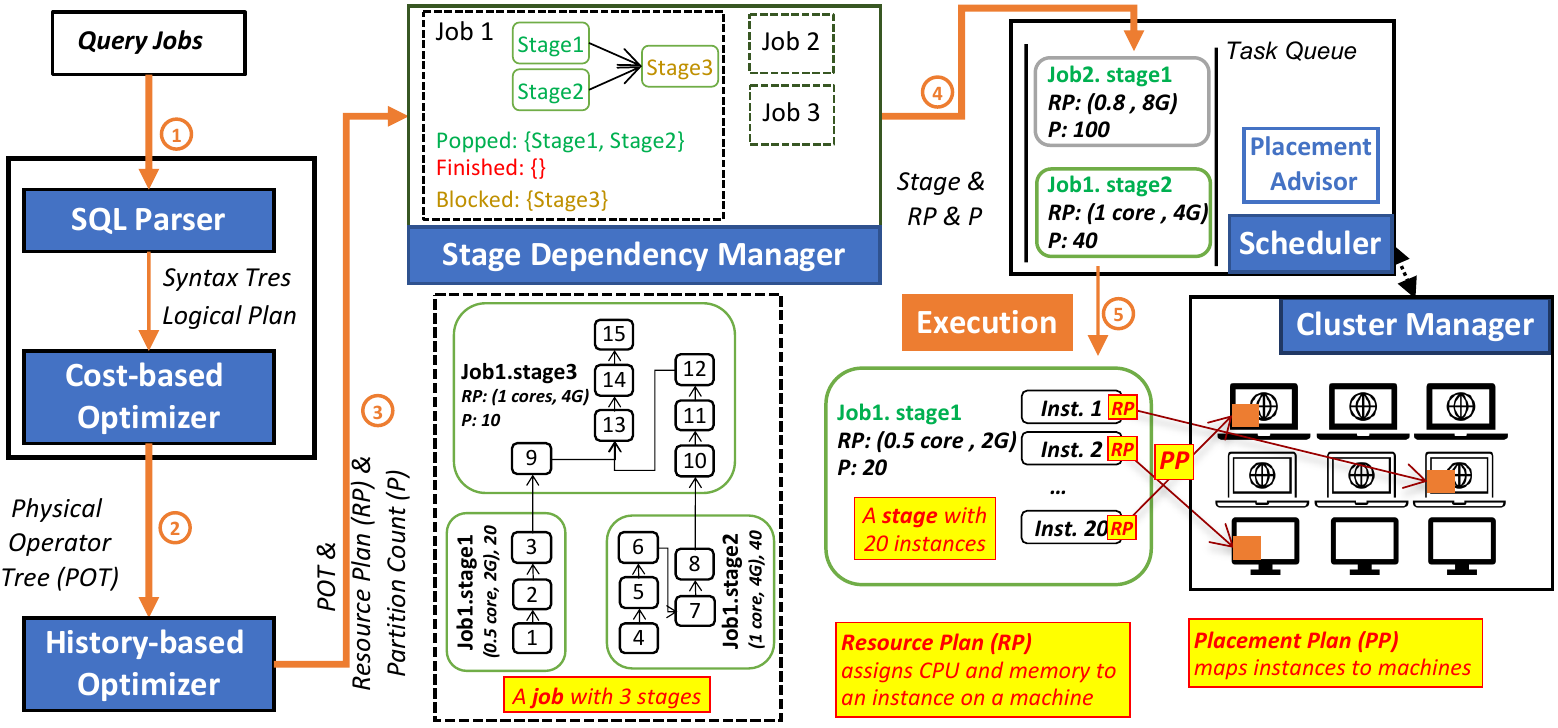}
   \vspace{-0.2in}
  \captionof{figure}{The lifecycle of a query job in MaxCompute}
  \label{fig:mc-job-lifecycle}
\end{minipage}%
\begin{minipage}{.27\textwidth}
  \centering
  \begin{tabular}{ll}
		\subfigure[\small{CDF for \#stages / job}]
		{\label{fig:cdf-job-stage-num-low}\includegraphics[height=1.6cm,width=.98\linewidth]{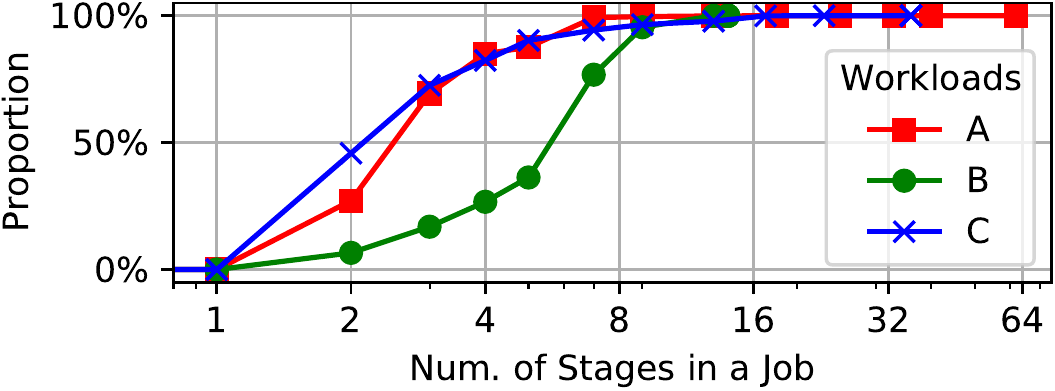}}
		\\
		\subfigure[\small{CDF for \#instances / stage}]
		{\label{fig:cdf-stage-inst-num-low}\includegraphics[height=1.6cm,width=.98\linewidth]{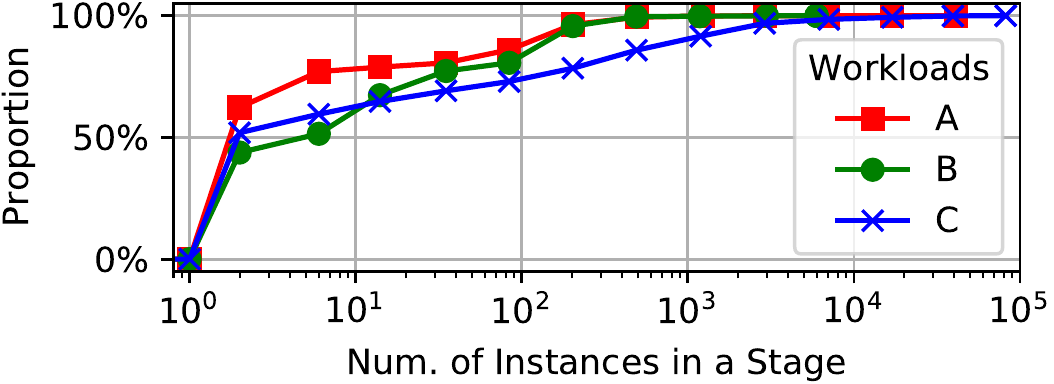}}
		\\
		\subfigure[\small{Timeline of 6716 instances in a stage}]
		{\label{fig:stage-lats}\includegraphics[height=1.6cm,width=.98\linewidth]{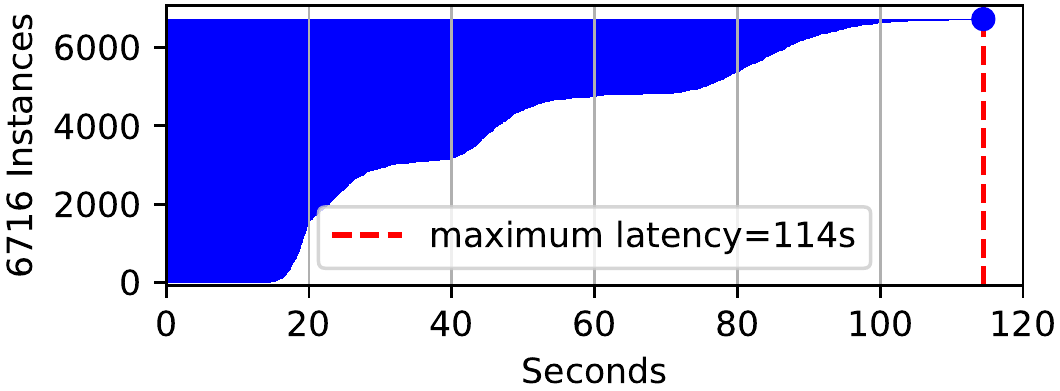}}	  					  
  \end{tabular}
  \vspace{-0.2in}
  \captionof{figure}{Trace overview}
  \label{fig:intro-stats}
\end{minipage}
 \vspace{-0.1in}
\end{figure*}
}

\begin{figure*}
\centering
\begin{minipage}{.7\textwidth}
  \centering
  \includegraphics[height=5.6cm,width=\linewidth]{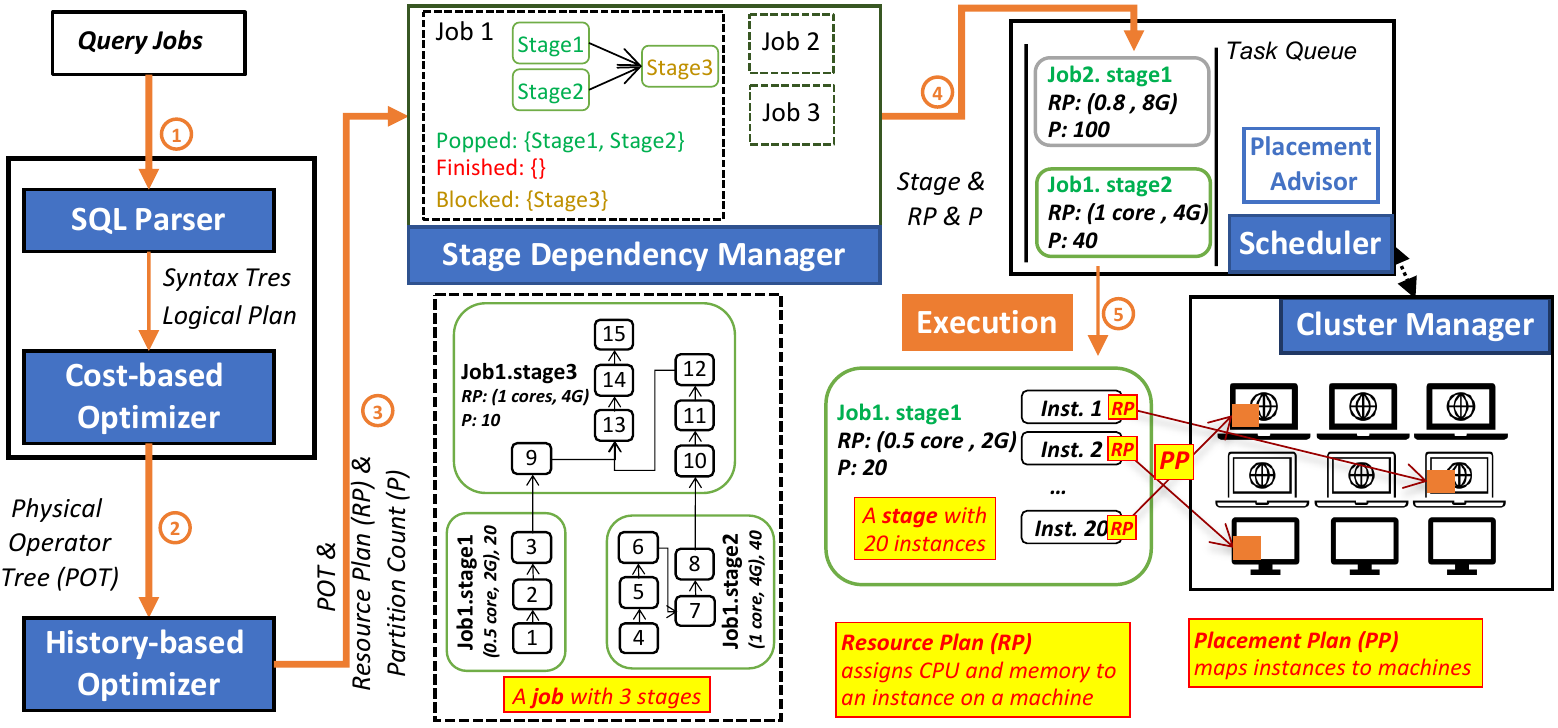}
   \vspace{-0.25in}
  \captionof{figure}{The lifecycle of a query job in MaxCompute}
  \label{fig:mc-job-lifecycle}
\end{minipage}%
\begin{minipage}{.27\textwidth}
  \centering
  \begin{tabular}{ll}
		\subfigure[\small{CDF for \#stages / job}]
		{\label{fig:cdf-job-stage-num-low}\includegraphics[height=1.4cm,width=.98\linewidth]{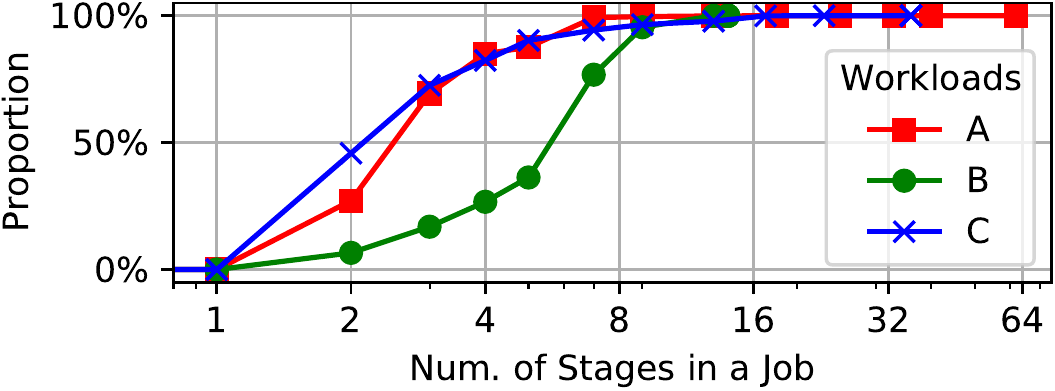}}
		\vspace{-0.02in}
		\\
		\subfigure[\small{CDF for \#instances / stage}]
		{\label{fig:cdf-stage-inst-num-low}\includegraphics[height=1.4cm,width=.98\linewidth]{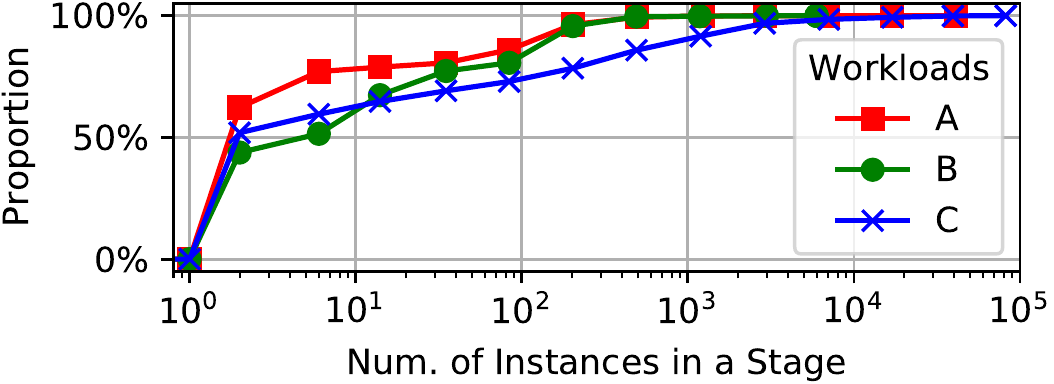}}
		\vspace{-0.02in}
		\\
		\subfigure[\small{Timeline of 6716 instances in a stage}]
		{\label{fig:stage-lats}\includegraphics[height=1.4cm,width=.98\linewidth]{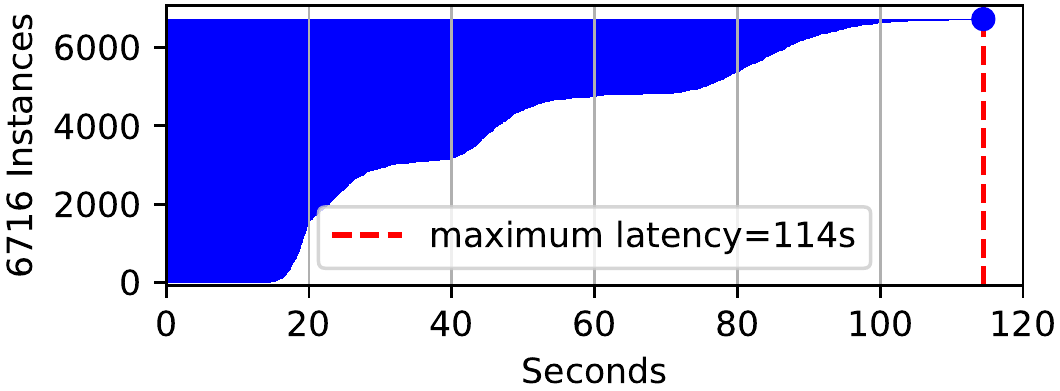}}	  					  
  \end{tabular}
  \vspace{-0.2in}
  \captionof{figure}{Trace overview}
  \label{fig:intro-stats}
\end{minipage}
 \vspace{-0.1in}
\end{figure*}

Production-scale big data processing presents a highly complex environment for resource optimization. 
We show an example through the life cycle of a submitted job in MaxCompute~\cite{maxcompute}, the big data query processing system at Alibaba, as shown in Fig.~\ref{fig:mc-job-lifecycle}. A submitted SQL job is first processed by a  "Cost-Based Optimizer" (CBO) to construct a query plan in the form of a Directed Acyclic Graph (DAG) of operators. 
These operators are further grouped into several stages connected with data shuffle (exchanging) dependencies. 
As shown in Fig.~\ref{fig:mc-job-lifecycle}, job 1 is composed of stage 1 (operators 1-3), stage 2 (operators 4-8), and stage 3 (operators 9-15), and their boundaries are data shuffling operations. 
To explore data parallelism, each \textit{stage} runs over multiple machines and each machine runs an \textit{instance} of a stage over a partition of the input data. To enable such parallel execution, a  "\rv{History}-Based Optimizer" (HBO)  recommends a \textit{partition count} (number of instances)  for each stage and a \textit{resource plan}  (the number of cores and memory needed) for all instances of the stage based on previous experiences. 

During job execution, a Stage Dependency Manager will maintain the stage dependencies of a job and pop out all the unblocked stages to the Fuxi scheduler~\cite{fuxi-vldb14}. The scheduler treats each stage as a task to be scheduled and maintains tasks in queues. For each stage, the Fuxi scheduler uses a heuristic-based approach to recommending a \textit{placement plan} that sends instances to machines. 
After an instance is assigned to a machine,  it will be executed using the resource plan that HBO has created for this instance, which is the same for all instances of the same stage. 

\textbf{Challenges.} MaxCompute's large, complex big data processing environment poses a number of challenges to resource optimization. 

First, resource optimization involves a set of decisions that need to be made in the life cycle of a big data query: 
(1) the \textit{partition count} of a stage; 
(2) the \textit{placement plan} that maps the instances of a stage to the available machines; 
(3) the \textit{resource plan} that determines the resources 
assigned to each instance on a given machine. 
All of these factors will have a significant impact on the performance of the job, e.g., its latency and computing cost.  
Among the three issues, the partition count is best studied in the literature~\cite{LiNN14,JyothiCMNTYMGKK16,RajanKCK16,cleo-sigmod20}. 
But this decision alone is not enough -- both decisions 2 and 3 can affect the performance a lot. 
While most of the existing work neglects the placement problem, due to the use of virtualization or container technology, a service provider like Alibaba Cloud does need to solve the placement problem in the physical clusters. 
As for the resource plan, it is determined by HBO from past experiences, without considering the latencies of the current instances in hand.
However, such solutions are far from ideal: if one underestimates the resource needs, the job can miss its deadline. On the other hand, overestimation leads to wasted resources and higher costs on (both internal and external) users. A few systems addressed the placement problem~\cite{LiNN14,JyothiCMNTYMGKK16} or the resource problem~\cite{spark-moo-icde21,VanAken:2017:ADM}  in isolation and in simpler system environments. 
We are not aware of any solution that can address all three problems in an integrated system. 

Second, the scale and complexity of our production clusters make the above decisions a challenging problem. 
Most notably,  our clusters can easily extend to \textit{tens of thousands} of machines, while all the resource optimization decisions must be made \textit{well under a second}. An algorithm for the placement problem with quadratic complexity~\cite{LiNN14} may work for a small cluster of 10 machines and a small number of jobs, as in the original paper, but not for the scale of 10's thousands of machines and millions of jobs. 
Furthermore, the optimal solutions to both the placement and resource plans depend on the workload characteristics of instances, available machines, and current system states. It is nontrivial to characterize all the data needed for resource optimization, let alone the question of using all the data to derive optimal solutions well under a second. 

Third, our use cases clearly indicate that resource optimization is a multi-objective optimization (MOO) problem. 
A real example we encountered in production is that after the user changed to allocate 10x more resources (hence paying 10x the cost), the latency of a job was reduced only by half, which indicates ineffective use of resources and a poor tradeoff between latency and cost. Over such a complex processing system, the user has no insights into how latency and cost trade-off. Therefore, there is a growing demand that the resource optimizer makes the decisions automatically to achieve the best tradeoff between multiple, often competing, objectives. 
Most existing work on resource optimization focuses on a single objective, i.e., job latency~\cite{cleo-sigmod20,Li:2018:MCD,LiNN14,JyothiCMNTYMGKK16,RajanKCK16,VanAken:2017:ADM,Zhang:2019:EAC}. Only a few recent systems~\cite{spark-moo-icde21,tan2016tempo} offer a multi-objective approach to the resource optimization problem. Still, their solutions do not suit the complex structure of big data queries, which we detail below. 
  

Fourth, to enable multi-objective resource optimization, the system must have a model for each objective to predict its value under any possible solution that the optimizer would consider. 
Existing models,  derived from domain knowledge~\cite{LeisK21,LiNN14,RajanKCK16} or   machine learning methods~\cite{qppnet-vldb19,tlstm-cost-estimator-vldb19,MB2-sigmod21,gpredictor-vldb20}, have been used to improve SQL query plans~\cite{neo-vldb19,VaidyaDNC21-vldb21} or to improve selectivity estimation~\cite{NegiMKMTKA21-vldb21,face-vldb21,Sun0021-sigmod21,nura-vldb19,NeuroCard-vldb20,bayescard,FLAT-vldb21,QiuWYLWZ21-sigmod21,LiuD0Z21-vldb21,LuKKC21-vldb21,cloudcard,WuC21-sigmod21,HasanTAK020-sigmod21,WoltmannOHHL21-vldb21}. 
But a key observation made in this work is that \textit{existing models do not suit the needs of resource optimization of big data queries because they perform only coarse-grained modeling}:  by capturing only end-to-end query latency or operator latency across parallel instances, these models may yield highly variable performance as they often involve large numbers of stages and parallel instances. Running optimization on highly-variable models gives undesirable results while missing opportunities for instance-level recommendations. 


\underline{Example 1.}
Fig.~\ref{fig:cdf-job-stage-num-low} and~\ref{fig:cdf-stage-inst-num-low} show that in a production trace of 0.62 million jobs, there are 1.9 million stages in total, with up to 64 stages in each job, and 121 million instances, with up to 81430 instances in a stage. For a particular stage with 6716 instances, Fig.~\ref{fig:stage-lats} shows that the latencies of different instances vary a lot. 
If a performance model captures only the overall stage latency, i.e., the maximum instance latency,
when the resource optimizer is asked to reduce latency, it will assign more resources uniformly to all the instances (as they are not distinguishable by the model). 
The extra resources do not contribute to the stage latency for those short-running instances, while incurring a higher cost. 
Instead, an optimal solution would be only to assign more resources to long-running instances while reducing resources for short-running ones. Such decisions require fine-grained instance-level models as well as instance-specific resource plans. 


\textbf{Contributions.}
Based on the above discussion, our work aims to support \textit{multi-objective resource optimization via fine-grained instance-level modeling and optimization}, and devises new system architecture and algorithms to enable \textit{fast resource optimization decisions, well under a second, in the face of 10's thousands of machines and 10's thousands of instances per stage}. By way of addressing the above challenges, our work makes the following contributions. 

1. \textit{New architecture of a resource optimizer} (Section 3).
To enable all the resource optimization (RO) decisions of a stage within a second, we propose a new architecture that breaks RO into a series of simpler problems. The solutions to these problems leverage existing cost-based optimization (CBO) and history-based optimization (HBO), while fixing their suboptimal decisions using a new Intelligent Placement Advisor (IPA) and Resource Assignment Advisor (RAA), both of which exploit fine-grained predictive models to enable effective instance-level recommendations.

2. \textit{Fine-grained models} (Section 4).
To suit the complexity of our big data system, our fine-grained instance-level models capture all relevant aspects from the outputs of CBO and HBO, the hardware, and machine states, and embed these heterogenous channels of information in a number of deep neural network architectures. 

3. \textit{Optimizing placement and resource plans} (Section 5). We design a new stage optimizer that employs a new IPA module to derive a placement plan to reduce the stage latency, and a novel RAA model to derive instance-specific resource plans to further reduce latency and cost in a hierarchical MOO framework. Both methods \rv{are proven to achieve optimality with respect to their own set of variables, } and are optimized to run well under a second.

4. \textit{Evaluation} (Section 6). Using production workloads of 0.6M jobs and  1.9M stages and a simulator of the extended MaxCompute environment, our evaluation shows promising results:  
(1) Our best model achieved 7-15\% median error and 9-19\% weighted mean absolute percentage error. 
(2) Compared to the HBO and Fuxi scheduler, IPA  reduced the latency by 10-44\% and the \rv{cloud cost} by 3-12\% while running in 0.04s. 
(3) IPA + RAA further achieved the reduction of \rv{37-72\%} latency and \rv{43-78\%} cost while running in \rv{0.02-0.23s}.






\section{Related Work}
\label{sec:related-work}


{\bf ML-based query performance prediction.} Recent work has developed various Machine Learning (ML) methods to predict query performance. 
QPPNet~\cite{qppnet-vldb19} builds separate neural network models (neural units) for individual query operators and constructs more extensive neural networks based on the query plan structure. Each neural unit learns the latency of the subquery rooted in a given operator. 
TLSTM~\cite{tlstm-cost-estimator-vldb19} constructs a uniform feature representation for each query operator and feeds the operator representations into a TreeLSTM to learn the query latency. 	
Both approaches, however, are tailored only for a {\em single} machine with an isolated runtime environment and fixed resources. 
Hence, they are not directly applicable to our RO problem in large-scale complex clusters.

ModelBot2~\cite{MB2-sigmod21} trains ML models for fine-grained operating units decomposed from a DBMS architecture to enable a self-driving DBMS.
However, it is designed for a local DMBS but not big data systems as in our work.
GPredictor~\cite{gpredictor-vldb20} predicts the latency of concurrent queries in a single machine with a graph-embedding-based model. 
It further improves accuracy by using profiling features such as data-sharing, data-conflict and resource competition from the local DBMS.
In our work, however, concurrency information regarding the operators from different queries is unavailable 
due to the container technology in large clusters~\cite{maxcompute,scope-vldb12,spark-rdd-nsdi12,hadoop}. 


ML-based models have been used for different purposes.  
NEO~\cite{neo-vldb19} learns a DNN-based cost model for (sub)queries and uses it to build a value-based Reinforcement Learning (RL) algorithm for improving query execution plans. 
Vaidya et al. \cite{VaidyaDNC21-vldb21} train ML models from query logs to improve query plans for parametric queries.
Phoebe~\cite{phoebe-vldb21} uses ML models for improving checkpointing decisions. 
Many recent works~\cite{NegiMKMTKA21-vldb21,face-vldb21,Sun0021-sigmod21,nura-vldb19,NeuroCard-vldb20,bayescard,FLAT-vldb21,QiuWYLWZ21-sigmod21,LiuD0Z21-vldb21,LuKKC21-vldb21,cloudcard,WuC21-sigmod21,HasanTAK020-sigmod21,WoltmannOHHL21-vldb21} have applied ML-based approaches to improve cardinality estimation.

A final, yet important, comment on the above systems is that they do not address the RO problem like in our work. 

{\bf Performance tuning in DBMS and big data systems.}
Performance tuning systems require a dedicated, often iterative, tuning session for each workload, which can take long to run (e.g., 15-45 minutes~\cite{VanAken:2017:ADM,Zhang:2019:EAC}). As such, they are not designed for production workloads that need to be executed on demand. In addition, they aim to optimize a single objective, e.g., query latency. 
Among {\em search-based} methods, BestConfig~\cite{zhu2017bestconfig} searches for good configurations by dividing high-dimensional configuration space into subspaces based on samples, but it cold-starts each tuning request. ClassyTune~\cite{zhu2019classytune} solves the optimization problem by classification, which cannot be easily extended to the MOO setting.
Among {\em learning-based} methods, 
Ottertune~\cite{VanAken:2017:ADM}  
 builds a predictive model for each query by leveraging similarities to past queries,  and runs Gaussian Process exploration to try other configurations to reduce latency. 
CDBTune~\cite{Zhang:2019:EAC} uses Deep RL to predict the reward (a weighted sum of latency and throughput) of a given configuration and explores a series of configurations to optimize the reward. 
ResTune~\cite{ResTune-sigmod21} uses a meta-learning model to learn the accumulated knowledge from historical workloads to accelerate the tuning process for optimizing resource utilization without violating SLA constraints.

{\bf Task scheduling in big data systems.}
Fuxi~\cite{fuxi-vldb14} and Yarn~\cite{hadoop} make the scheduling decisions based on locality information, while Trident~\cite{trident-vldb21} improves Yarn by considering the locality in different storage tiers. 
However, these systems treat each task as a blackbox and make scheduling decisions without knowing the task characteristics, such as the query plan structures and performance predictions. Hence, they cannot find optimal solutions to the machine placement and/or resource allocation problems.

{\bf Resource optimization in big data systems.}
In cluster computing, a resource optimizer (RO)  determines the optimal resource configuration {\em on demand} and {\em with low latency} as jobs are submitted. 
RO for parallel databases~\cite{LiNN14} determines the best data partitioning strategy across different machines to minimize a single objective, latency. 
Its time complexity of solving the placement problem is quadratic to the number of machines (9 machines in~\cite{LiNN14}), which is not affordable on today's productive scale (>10K machines).
Morpheus~\cite{JyothiCMNTYMGKK16} 
codifies user expectations as  SLOs and enforces them using scheduling methods. 
\cut{Its online packing algorithm minimizes the maximal total allocation with a log-competitive bound.}
However, its optimization focuses on system utilization and predictability,  but not minimizing the cost and latency of individual jobs as in our work. 
PerfOrator~\cite{RajanKCK16} 
solves a single-objective (latency) optimization problem via an exhaustive search of the solution space while calling its model for predicting the performance of each solution. 
WiseDB~\cite{MarcusP16} manages cloud resources based on  a decision tree trained on  performance and cost features from minimum-cost schedules of sample workloads, 
while such schedules are not available in our case. 
Li et al.~\cite{Li:2018:MCD} minimize end-to-end tuple processing time using deep RL and requires defining scheduling actions and the associated reward, which is not available in our problem.
Recent work~\cite{LeisK21} proposes a heuristic-based model to recommend a cloud instance (e.g., EC2 instance) that achieves cost optimality for OLAP queries, which is different from our resource optimization problem. 

CLEO~\cite{cleo-sigmod20,qrop-icde18} learns the end-to-end latency model of query operators, and based on the model, minimizes the latency of a stage (involving multiple operators) by tuning the partition count. It has a set of limitations: First, its modeling target concerns the latency of multiple instances over different machines,  which can be highly variable (e.g., due to uneven data partitions or scheduling delays) and hard to predict. 
Second, CLEO's latency model does not permit instance-level recommendations for the placement and resource allocation problems. Third, its optimization supports  a single objective, and determines only the partition count in a stage, but not other RO decisions. 
Bag et al.~\cite{plan-aware-resource-opt-hotcloud20} propose a plan-aware resource allocation approach to save resource usage without impacting the job latency but not minimize latency and cost like in our work.

	
UDAO~\cite{spark-moo-icde21,udao-vldb19} tunes Spark configurations to optimize for multiple objectives. However, it works only on the granularity of an entire query and neglects its internal structure. 
Such coarse-grained modeling of latency is unlikely to be accurate for complex big data queries, which leads to poor results of optimization. 
TEMPO~\cite{tan2016tempo} considers multiple Service-Level Objectives (SLOs) of SQL queries and guarantees max-min fairness when they cannot be all met.
But its MOO works for entire queries, but not fine-grained MOO that suits the complex structure of big data systems. 

\cut{
In short, existing work on resource optimization suffers from (one of) the following two issues: 
(1) It does not support MOO. 
(2) Its optimization procedure requires an end-to-end query latency model or operator latency model that involves running multiple instances on different machines. Such coarse-grained latency models can involve unpredictable scheduling delays, a large number of internal stages, and numerous parallel instances. Running optimization on such highly-variable latency models is not effective while missing opportunities for instance-level recommendations.  
Therefore, our work aims to support multi-objective resource optimization via fine-grained instance-level performance modeling and optimization.
}

\textbf{Multi-objective optimization} (MOO).  
MOO for SQL~\cite{Hulgeri:2002:PQO,Kllapi:2011:SOD,Trummer:2014,Trummer:2014:ASM,TrummerK15} finds Pareto-optimal  query plans by efficiently searching through a large  set of them. The problem is fundamentally different from RO, including machine placement and resource tuning problems.
MOO for workflow scheduling~\cite{Kllapi:2011:SOD} assigns operators  to containers to minimize total running time and money cost. But its method is limited to searching through 20 possible numbers of containers and solving a constrained optimization for each option. 

Theoretical MOO solutions suffer from a range of performance issues when used in a RO:
\emph{Weighted Sum}~\cite{marler2004survey} is known to have \textit{poor coverage} of the Pareto frontier~\cite{messac2012from}. 
\emph{Normalized Constraints}~\cite{messac2003nc} lacks in  \textit{efficiency} due to  repeated recomputation to return more solutions. 
\emph{Evolutionary Methods}~\cite{Emmerich:2018:TMO} 
approximately compute a Pareto set but suffer from  \textit{inconsistent solutions}.
\cut{ 
\emph{Multi-objective Bayesian Optimization} extends the Bayesian approach to modeling an unknown function with an acquisition function for deciding to explore the next point(s) that are likely to be Pareto optimal. By taking a long time to run, it lacks the \textit{efficiency} required by RO.
}
\chenghao{
\emph{Multi-objective Bayesian Optimization} 
explores the potential \po points iteratively by extending the Bayesian approach, but lacks the \textit{efficiency} required by RO due to a long-running time.
}
\emph{Progressive frontier}~\cite{spark-moo-icde21} is the first MOO solution suitable for RO, offering good coverage, efficiency, and consistency.
However, none of them consider complex structures of big data queries and require modeling end-to-end query latency, which does not permit instance-level optimization. 
\cut{ 
To do so, our work employs fine-grained modeling and a hierarchical MOO solution. }
\section{System Overview}
\label{sec:sys-overview}

In this section, we provide the background on MaxCompute~\cite{maxcompute}, the big data query processing system at Alibaba, and our proposed extended architecture for resource optimization. 

\vspace{-0.05in}
\subsection{Background on MaxCompute}

In MaxCompute, a submitted user job is represented hierarchically using stages and operators, which will then be executed by parallel instances. As shown in Fig.~\ref{fig:mc-job-lifecycle},  
a \textit{job} is a Directed Acyclic Graph (DAG) of stages, where the edges between stages are inter-machine data exchange (shuffling) operations.
A \textit{stage} is a DAG of operators, where edges are intra-machine pipelines without data shuffling.
The input data of each stage is partitioned over different machines, where each partition is run as an {\it instance} of the stage in a container.


Fig.~\ref{fig:mc-job-lifecycle} shows the functionality of query optimizers and the scheduler in the lifecycle of a submitted job.   

\textbf{Cost-Based Optimizer} (CBO): MaxCompute's Cost-Based Optimizer (CBO) is a variant of the Cascades optimizer~\cite{cascades}. It follows traditional SQL optimization based on cardinality and cost estimation and generates a Physical Operator Tree (POT), which is a DAG of stages where each stage is a DAG of operators.   

\textbf{\rv{History}-Based Optimizer} (HBO): To facilitate resource optimization, a \rv{History}-Based Optimizer (HBO) gives an initial attempt to recommend a \textit{partition count} (number of instances) for each stage and a \textit{resource plan} (the number of cores and memory needed) for all instances of the stage based on previous experiences. 
This history-based approach is known to be suboptimal because it does not consider the machines to run these queries and their current states. Both hardware characteristics and system states affect the latencies of instances, and an optimal solution should try to minimize the maximum latency of parallel instances, in addition to cost objectives.
Moreover, HBO needs expensive engineering efforts, especially when workloads change and the system upgrades. 
We will address these issues in our new system design. 

\textbf{Scheduler}: During job execution, the granularity of scheduling is a stage: a stage that has all dependencies met is handed over to the Fuxi scheduler~\cite{fuxi-vldb14}. Fuxi uses a heuristic approach to recommending a \textit{placement plan} that sends instances to machines, and each instance is assigned to a container on a machine with a previously-determined {\it resource plan} for CPU and memory. 
These decisions, however, are made without being aware of the latency of each instance. 
As such, an instance with a potential longer running time (e.g., due to larger input size) may be sent to a heavily loaded machine while another instance with less data may be sent to an idle machine, leading to overall poor stage latency (the maximum of instance latencies). A detailed example is given later in Fig.~\ref{fig:toy}.

%

\begin{figure*}
	\centering
    \includegraphics[width=.95\textwidth,height=4.55cm]{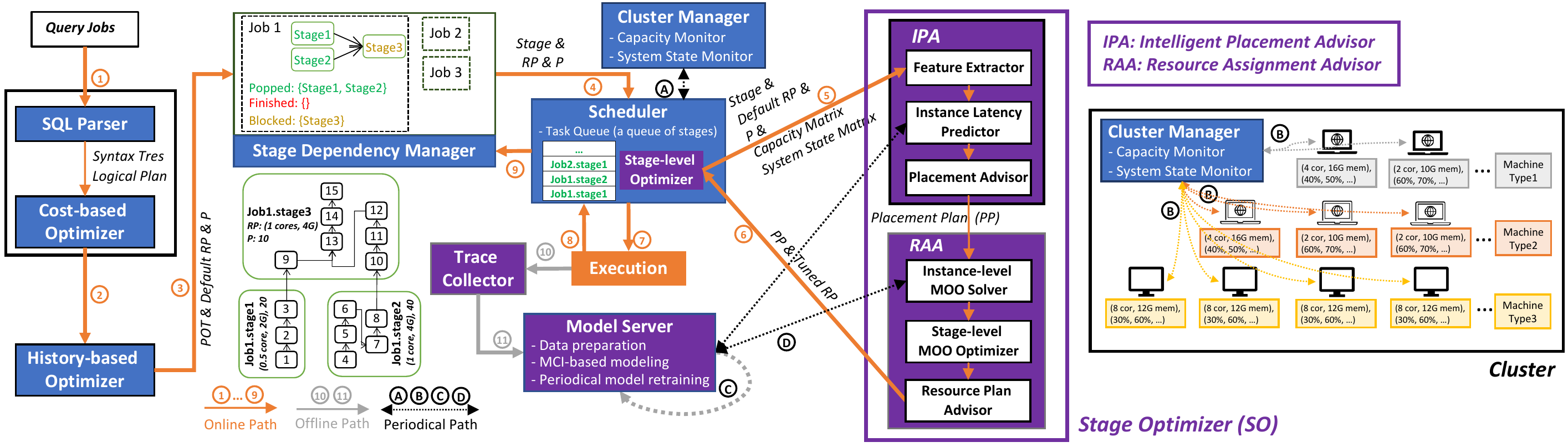}
    \vspace{-0.1in}
    \caption{Extended system architecture for resource optimization}
    \vspace{-0.1in}
    \label{fig:system-arch}
    \vspace{-0.08in}    
\end{figure*}


{\bf Workload and Cluster complexity.} 
Production clusters take workloads with a vast variety of characteristics. 
Workload A from an internal department includes 1 - 64 stages in each job. A stage may involve 2 - 249 operators and be executed by 1 - 42K instances, where an instance could take sub-seconds up to 1.4 hours to run.
Further, production clusters consist of heterogeneous machines. For example, we observed 5 different hardware types when executing workload A, where each hardware type includes 30 - 7K machines.
Moreover, the system states in each machine vary over time.
Take CPU utilization for example. The average CPU utilization varies from 32\%-83\%, and the standard deviation ranges from 6\%-23\%. 
Finally, each machine could run multiple containers, and a container runs an instance with a quota of CPU and memory based on a resource plan. For workload A, we observed 17 different resource plans for containers. Since there is no perfect isolation between containers and the host OS~\cite{dist-sys-book}, containers with the same resource plan could perform differently on machines with different hardware or system states, making the running environment more complex.

\vspace{-0.1in}
\subsection{System Design for Resource Optimization}
We next show our design for resource optimization by extending the MaxCompute architecture. 
Our work aims to support \textit{multi-objective resource optimization} (MORO). 
Here, we can support any user objectives (as long as we can obtain training data for them). 
For ease of composition, our discussion below focuses on minimizing both the stage latency (maximum latency among its instances) and 
\rv{the cloud cost (a weighted sum of CPU-hour and memory-hour)},
the two common objectives of our users. 

To achieve MORO, the resource optimizer needs to make three decisions: 
(1)~the \textit{partition count} of a stage; 
(2)~the \textit{placement plan} (PP) that maps the instances of a stage to the available machines; 
(3)~the \textit{resource plan} (RP) that determines the resources (the number of cores and memory size) assigned to each instance on a given machine. 
Our design is guided by two principles: 

{\bf Simplicity and Efficiency.}
An optimal solution to MORO may require examining all possibilities of dividing the input data into $m$ instances and arranging them to run on some of the $n$ machines, each using one of the $r$ possible resource configurations. 
In production clusters, both $m$ and $n$ could be 10's of thousands, while all the RO decisions must be made well under a second.  
Hence, it is infeasible to run 
an exhaustive search for optimal solutions.

Our design principle is to break MORO into a series of simpler problems, each of which can run very fast. 
First, we keep the \rv{History-Based Optimizer} (HBO) that uses past experiences to recommend a \textit{partition count} for a stage and an initial \textit{resource plan} for its instances. There is a merit of learning such configurations from past best-performing runs of recurring jobs (which dominate production workloads). While the recommendations may not be optimal, they serve as a good initial solution. 
Second, given the output of HBO, we design an \textbf{Intelligent Placement Advisor} (IPA) that determines the \textit{placement plan} (PP), mapping the instances  to machines, by predicting latencies of individual instances.
Third, our \textbf{Resource Assignment Advisor} (RAA) will fine-tune the \textit{resource plan} (RP) for each instance, after it is assigned to a specific machine, to achieve the best tradeoff between stage latency and cost. 
\rv{We give a theoretical justifcation of our approach in Section~\ref{sec:so}.}

{\bf Fine-grained Modeling and Hierarchical MOO.}
As discussed earlier, instance-level recommendations for PP and RP are key to minimizing latency and cost. To do so, we design a fine-grained model that predicts the latency of each instance based on the workload characteristics, hardware,  machine states, and resource plan in use. 
The model will be used in IPA to develop the PP that minimizes the maximum instance latency (while using the same RP for all instances). The instance-level model will be further used in RAA to improve the RP by solving a hierarchical MOO problem: We first compute the instance-level MOO  solutions that minimize the latency and cost of each individual instance. 
We then combine the instance-level MOO solutions into stage-level MOO solutions, and recommend one of them that determines the instance-specific resources with the best tradeoff between stage latency and cost. 


Fig.~\ref{fig:system-arch} presents the detailed architecture that extends MaxCompute with three new components (colored in purple): 
\textbf{The trace collector} collects runtime traces, including 
the query traces (the query plan and operator features such as cardinality\cut{and cost estimation}), 
instance-level traces (input row number and data size of an instance, and the resource plan assigned to it), 
and machine-level traces (machine system states and hardware type).
\textbf{The model server} featurizes the collected traces and internally learns an instance-level latency model. The internal model gets updated periodically by retraining or fine-tuning when the new traces are ready. It will serve as an instance-level latency predictor for resource optimization of online queries.
\textbf{The Stage-level Optimizer} (SO) consists of the IPA and RAA, as described above. For a stage to be scheduled, it calls the predictive model to estimate the latency of each instance on any available machine, and determines the PP by minimizing stage latency and then the RP by minimizing both stage latency and cost.

\section{Fine-Grained Modeling}
\label{sec:modeling}

In this section, we present how to build fine-grained instance-level models that can be used to minimize both latency and cost of each stage. To suit the complexity of big data systems,  our models capture all relevant systems aspects to support the resource optimizer to make effective recommendations.

\cut{ 
\begin{figure*}
\centering
\hspace{-0.3in}
\begin{minipage}{.72\textwidth}
\begin{flushleft}
   \includegraphics[height=6.8cm,width=.98\linewidth]{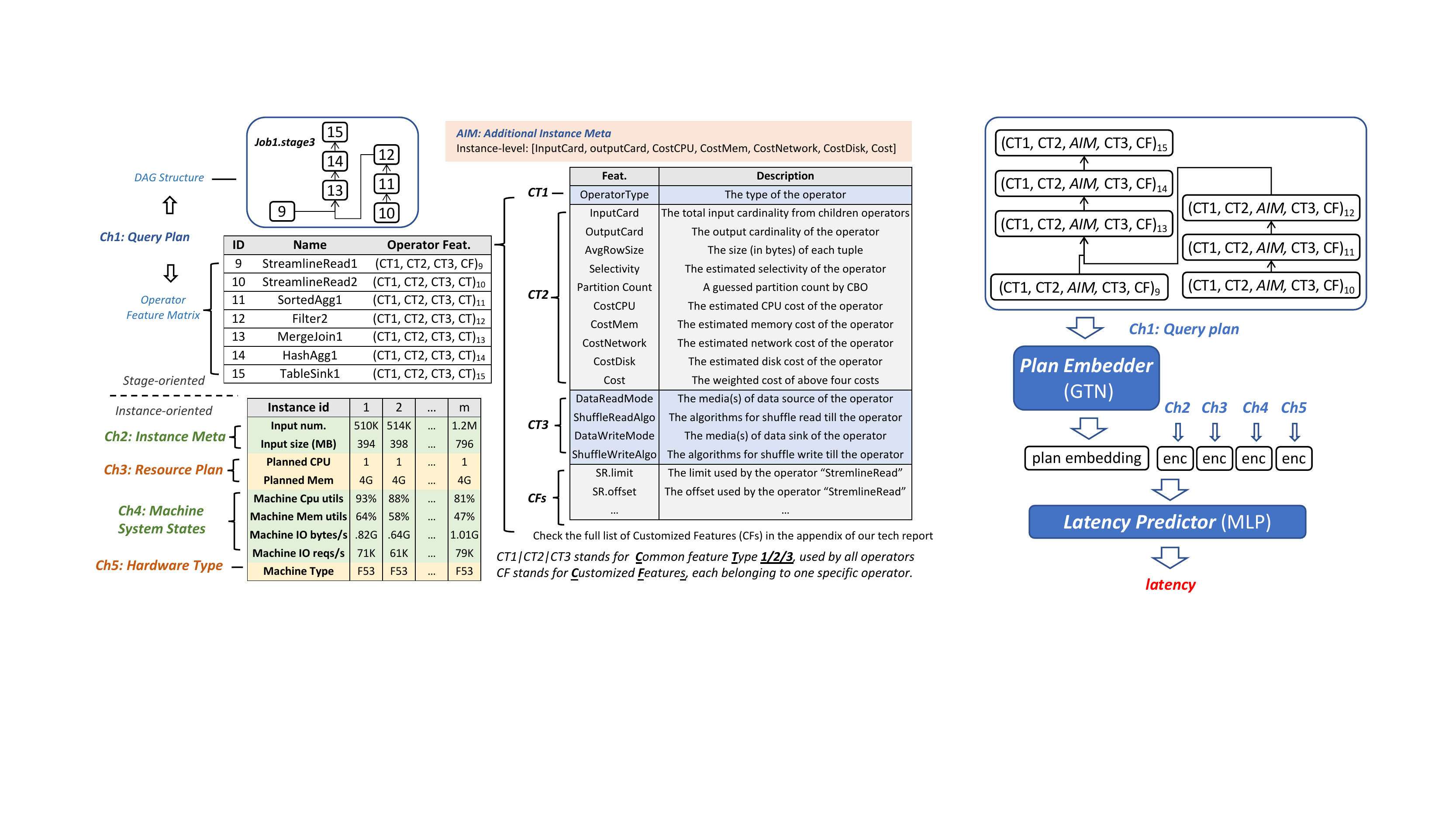}
  \vspace{-0.1in}
  \captionof{figure}{The multi-channel coverage from basic features}
  \vspace{-0.1in}
  \label{fig:mci}
\end{flushleft}  
\end{minipage}%
\begin{minipage}{.24\textwidth}
\begin{flushright}
   \includegraphics[height=5.8cm,width=.9\linewidth]{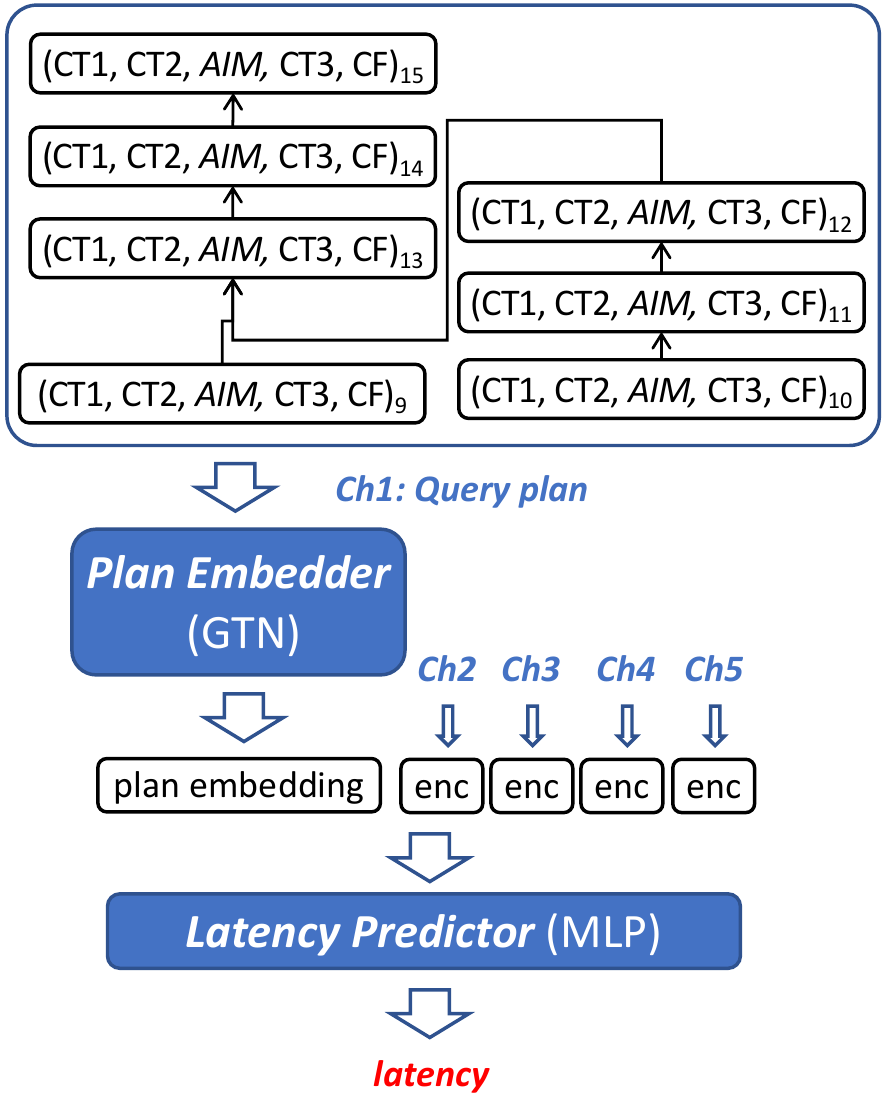}
   \captionof{figure}{MCI-based modeling framework}   
  \label{fig:mci-model}
\end{flushright}   
\end{minipage}
\end{figure*}
}

\begin{figure*}
\centering
\begin{minipage}{.73\textwidth}
\begin{flushleft}
   \includegraphics[height=5.1cm,width=.95\linewidth]{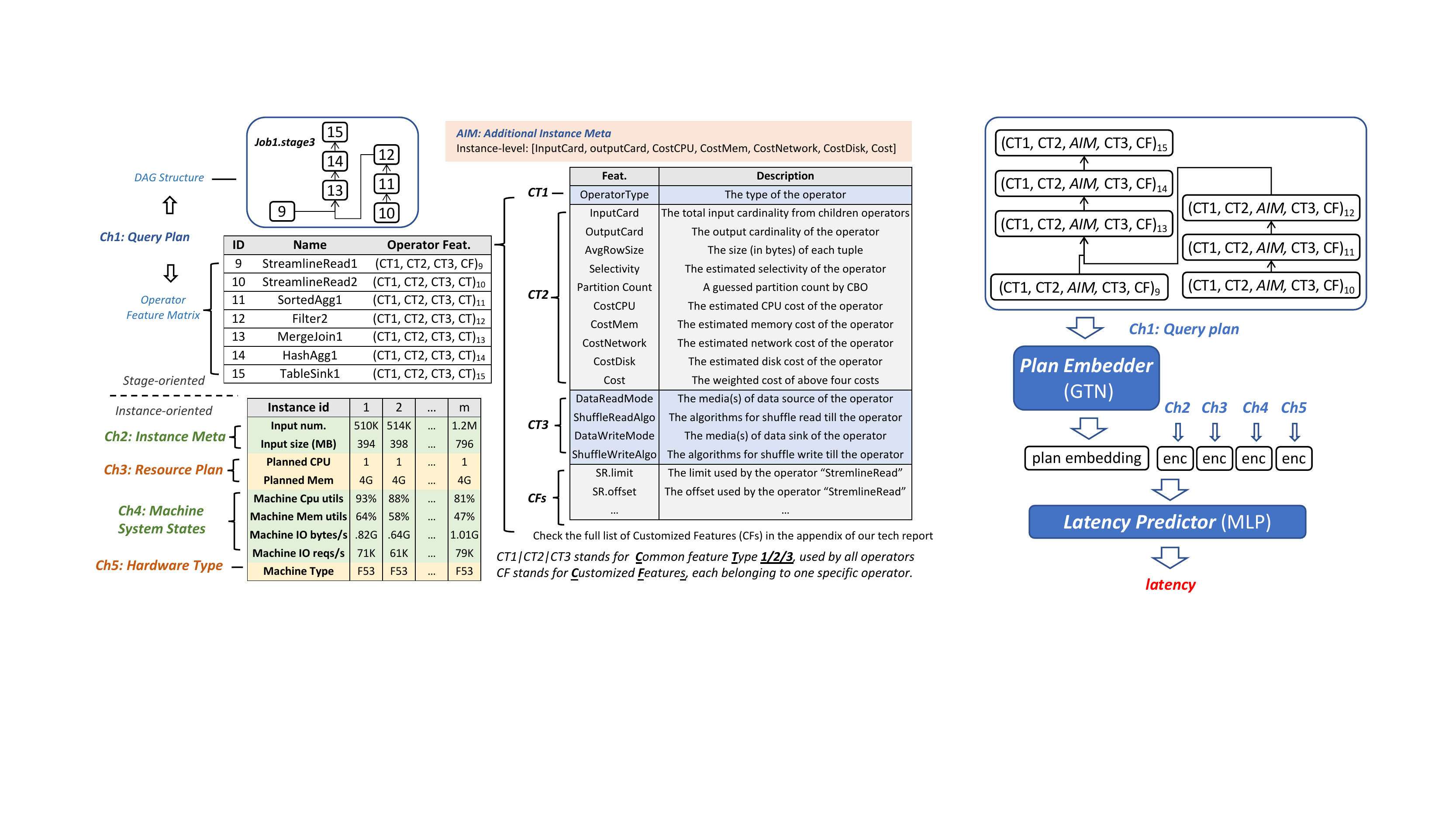}
  \vspace{-0.1in}
  \captionof{figure}{The multi-channel coverage from basic features}
  \vspace{-0.1in}
  \label{fig:mci}
\end{flushleft}  
\end{minipage}%
\begin{minipage}{.23\textwidth}
\begin{flushright}
   \includegraphics[height=4.7 cm,width=.95\linewidth]{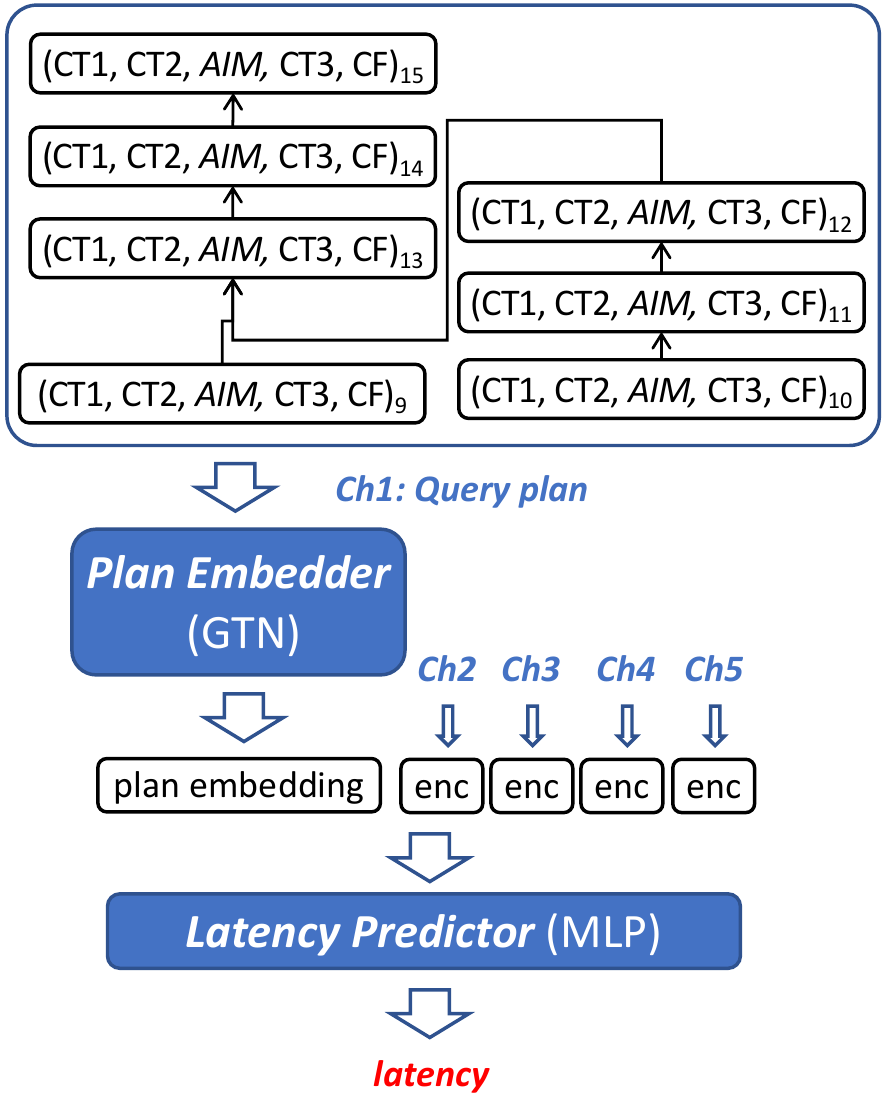}
   \vspace{-0.1in}
   \captionof{figure}{MCI-based modeling framework}   
  \vspace{-0.1in}
  \label{fig:mci-model}
\end{flushright}   
\end{minipage}
\vspace{-0.1in}
\end{figure*}

\subsection{Multi-Channel Coverage}

It is nontrivial to predict the latency of a single instance due to many factors in big data systems. 
Those factors include 
the characteristics of the (sub)query plan running in the instance, 
data characteristics, 
resources in use, 
current system states, 
and hardware properties.
Each factor alone could affect the latency of an instance, and the coeffects of multiple factors can make the latency pattern more complex.
Therefore, we propose the idea of \textit{multi-channel inputs} (MCI) to capture all of the above factors in our model.


\textbf{Multi-channel Inputs.} 
For a given stage, we extract features from all available runtime traces including the query plan, resource plan, instance-level metrics,  hardware profile, and system states. We design multi-channel inputs (MCI) to group the features into five channels that characterize different factors.
Fig.~\ref{fig:mci} shows the overview of the MCI design for the instances of a stage. 

\underline{Channel 1:} {\em Stage-oriented features (query plan)}.
Channel 1 introduces the stage-oriented features that are shared among instances in a stage. 
It captures the operator characteristics in an operator feature matrix (see Fig.~\ref{fig:mci})
and operator dependencies in a DAG structure (Fig.~\ref{fig:mci-model}).
The characteristics of each operator are featurized by three common feature types (CT1-CT3) and the customized features (CFs).
CT1 identifies the operator type and is represented as a categorical variable. 
CT2 captures the statistics from CBO and HBO, including the cardinality, selectivity, average row size, partition count, and cost estimation.
CT3 captures the IO-related properties, including the location of data (local disk or network) and strategies for shuffling.
CFs are tailored for the unique properties of operators, and each feature belongs to one particular operator.
 
\underline{Channels 2-5:} {\em Instance-oriented features}.  
Channel 2-5 characterizes individual instances (Fig.~\ref{fig:mci}).
{\it Channel 2, instance meta,} includes an instance's input row number and input size captured from the underlying storage system after the data partition is determined. 
{\it Channel 3, resource plan,} featurizes properties of the container that runs an instance, in terms of the CPU cores and memory size.
{\it Channel 4, machine system states,} records the CPU utilization, memory utilization, and IO activities of a machine to capture its running states. Since the system states represented by count-based measurements could be infinite, we discretize the system states to reduce the computation complexity. 
{\it Channel 5, hardware type}, uses the machine model to distinguish a set of hardware types. 

\underline{Augmented Channel 1:}  \textit{Additional Instance Meta (AIM)}.
In big data systems, CBO produces the query plan without considering the characteristics of individual instances.
Therefore, instances in a stage share the same query plan features in Ch1, even though their input row numbers can differ significantly. A model built on such features may have difficulty predicting different instance latencies because most of their features are identical. 

To address this issue, for each instance, we seek to enrich its query plan features with instance-level characteristics. More specifically, we augment the query plan channel of an instance by adding {\it additional instance meta (AIM)} features for each operator. AIM consists of the estimated operator input/output cardinality and costs of an individual instance.
It is derived using the stage-level selectivity (Ch1), the instance meta (Ch2), and the cost model in CBO.


Take  {\it job1.stage3} in Fig.~\ref{fig:mci} for example. 
We first get the precise input cardinalities for operators 9 and 10 in the instance by leveraging the instance meta (Ch 2) in the corresponding data partitions.
Accordingly, we calculate the output cardinality of 9 and 10 by multiplying the input cardinality and the operator selectivity.
Then we derive the instance-level cardinality for the remaining operators based on the operator dependency.
Finally, we recompute the operator costs in an instance by reusing CBO's cost model. 
Specifically, we substitute the stage-level cardinality with the instance-level cardinality, set the partition count to one, and call the CBO's cost model to derive the operator costs for an instance.

Note that the above approach assumes that instances in a stage share the same selectivities. 
Although it may not always be true, our evaluation results show that the model built under this assumption could approximate the best performance using the (unrealistic) ground-truth instance-level selectivities. 


%

Finally, recent ML models of DBMSs~\cite{qppnet-vldb19,tree-lstm-acl15}  also learn data characteristics by encoding/embedding tables. In a production system, however, many tables are not reachable due to access control for security reasons, and the overhead of representing a distributed table is much higher than in a single machine. Therefore, we assume the data characteristics are already digested by CBO in our work.


\subsection{MCI-based Models}

Now we propose our MCI-based models to learn instance latency. 

\minip{MCI-based Modeling Framework.}
In the output of MCI, features in the query plan channel are represented as a DAG (Fig.~\ref{fig:mci-model}), while the instance-oriented features are in the tabular form.
To incorporate these heterogeneous data structures in model building, we design two model components, as shown in Fig.~\ref{fig:mci-model}.
The {\it plan embedder} constructs a plan embedding of the plan features in the form of an arbitrary DAG structure. 
It first encodes the operators into a uniform feature space by padding zeros for the unmatched customized features. 
To embed the query plan, it then applies a Graph Transformer Networks (GTN)~\cite{gtn-neurips19}  due to its ability to learn the DAG context in heterogeneous graphs with different types of nodes. 
The {\it  latency predictor} is a downstream model 
that concatenates the plan embedding with other instance-oriented features (channels 2-5) into a big vector, and feeds the vector to a Multilayer Perceptron (MLP) to predict the instance-level latency. 

\minip{Modeling Tools in Our Framework.}
Besides GTN, our modeling framework 
can accommodate other models designed for DBMSs, with necessary extensions to our graph-based MCIs. 
Thus, we can leverage different models and examine their pros and cons in our system. Due to space limitations, our extensions of  QPPNet~\cite{qppnet-vldb19} and TLSTM~\cite{tlstm-cost-estimator-vldb19} are deferred to \techreport{Appendix~\ref{appendix:model}}{~\cite{tech-report}} .

\section{Stage-level optimization}
\label{sec:so}
In this section, we present our {\it Stage-level Optimizer (SO)} that minimizes both stage latency and cost based on instance-level models.


\begin{figure*}
\centering
\begin{minipage}{.28\linewidth}
\begin{flushleft}
\centering
  \includegraphics[width=1.\linewidth,height=3.8cm]{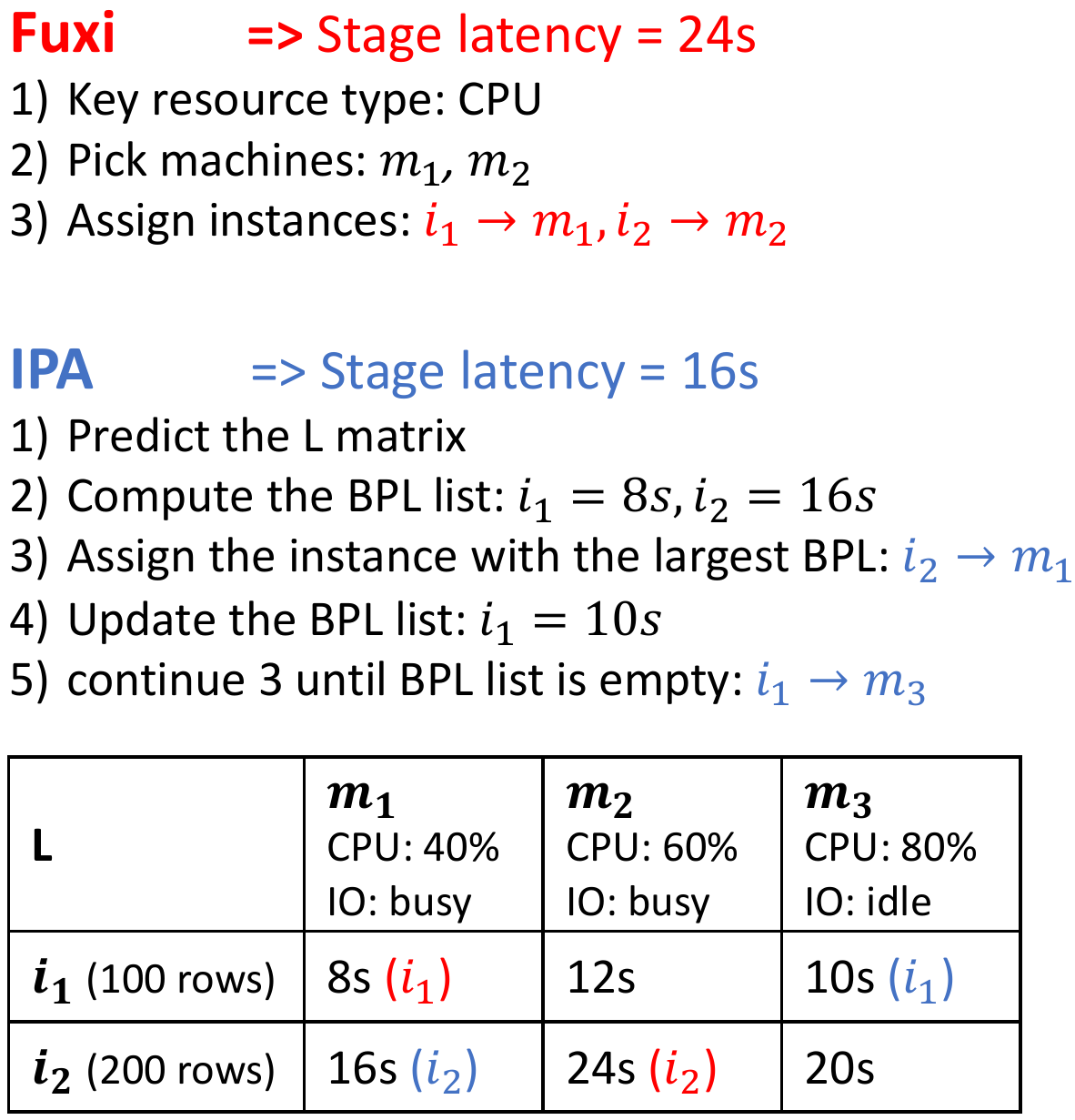}
  \vspace{-0.2in}
  \captionof{figure}{Example of IPA}
  \label{fig:toy}
\end{flushleft}   
\end{minipage}%
\hspace{0.1in}
\begin{minipage}{.33\linewidth}
\begin{flushleft}
  \includegraphics[width=1.\linewidth,height=3.8cm]{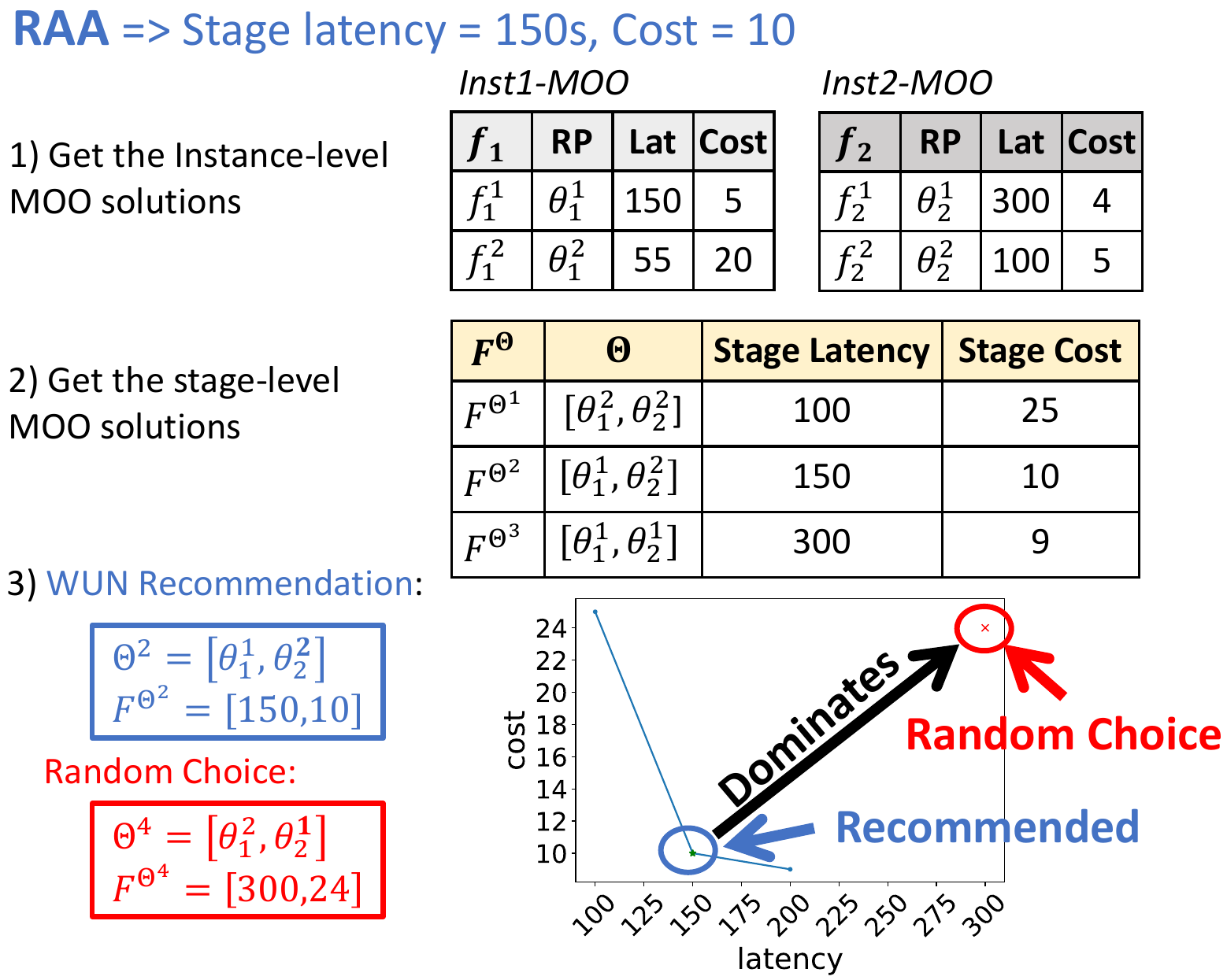}
	\vspace{-0.2in}
	\captionof{figure}{Example of RAA} 
	\label{fig:raa-example}
\end{flushleft}  	
\end{minipage}
\begin{minipage}{.33\textwidth}
\begin{flushright}
\centering
   \includegraphics[width=1.\linewidth,height=3.8cm]{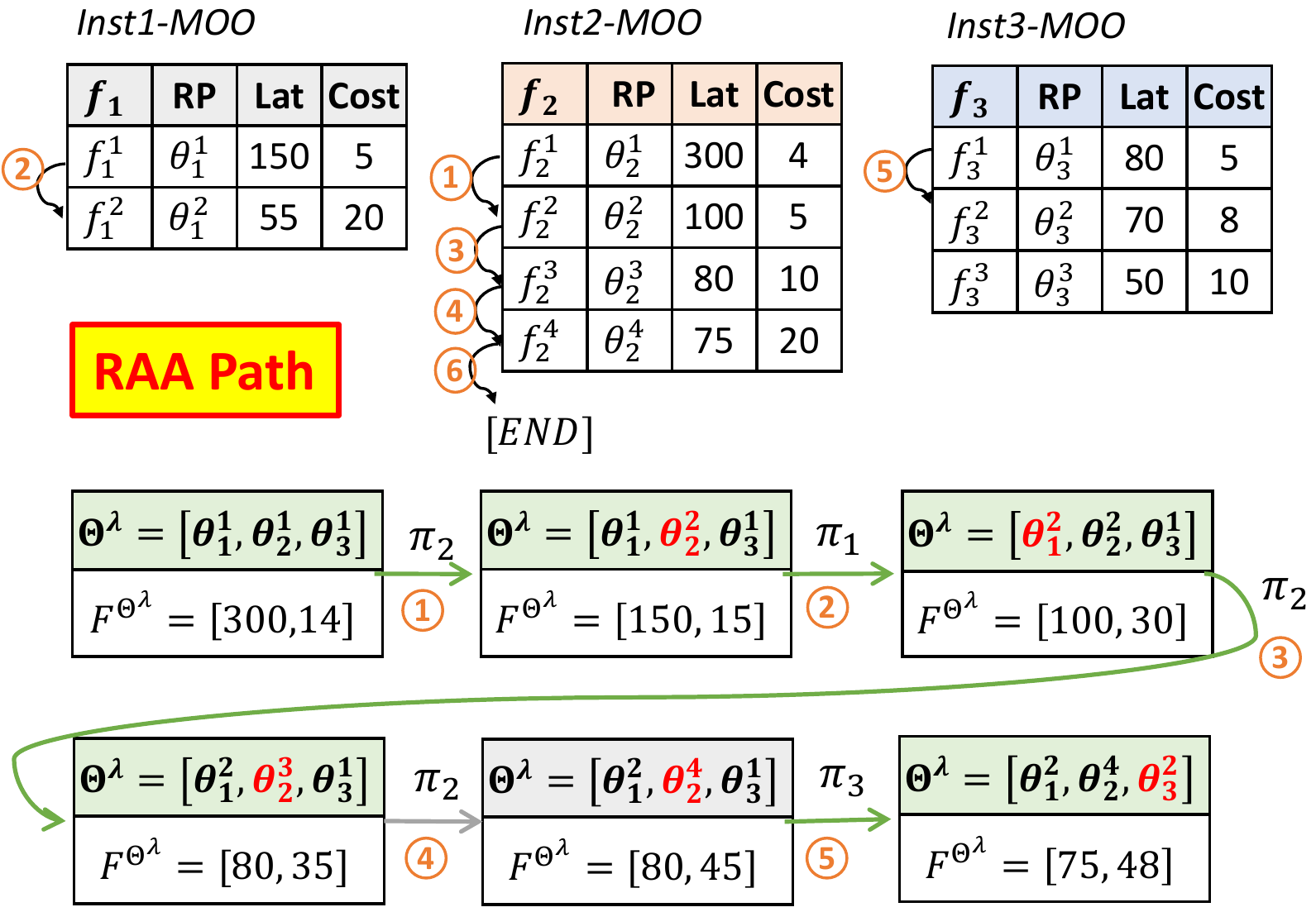}
    \vspace{-0.2in}
    \captionof{figure}{Example of RAA path} 
    \label{fig:raa-path}
\end{flushright}    
\end{minipage}
\vspace{-0.1in}
\end{figure*}

During execution, once a stage is handed to the scheduler, two decisions are made: 
the {\it placement plan (PP)} that maps instances to machines, and the {\it resource plan (RP)} that determines the CPU and memory resources of each instance on its assigned machine.
The current Fuxi scheduler~\cite{fuxi-vldb14} decides a PP for $m$ instances as follows: 
(1) Identify the key resource (bottleneck) in the current cluster, e.g., CPU or IO.
(2) Pick $m$ machines with top-$m$ lowest resource watermarks.
(3) Assign instances, in order of their instance id, to the $m$ machines, and use the same resource plan for each instance as suggested by HBO.
However, its negligence of latency variance among instances leads to suboptimal decisions for PP and RP: 

\underline{Example 2.} 
Figure~\ref{fig:toy} shows how Fuxi gets a suboptimal placement plan in a toy example of sending a stage of 2 instances ($i_1$, $i_2$) to a cluster of 3 available machines ($m_1$, $m_2$, $m_3$).
Assume that the instance latency is proportional to its input row number on the same machine, and the current key resource in the cluster is the CPU.
In this example, Fuxi first picks $m_1$ and $m_2$ as machines with the top-2 lowest CPU utilization (watermarks) and assigns $i_1$ to $m_1$ and $i_2$ to $m_2$ respectively. The stage latency, i.e., the maximum instance latency, is 24s.
However, the {\em optimal placement plan} could achieve a 16s stage latency by assigning $i_1$ to $m_3$ and $i_2$ to $m_1$.
Further, using the same resources for $i_1$ and $i_2$ is not ideal. Instead, an {\em optimal resource plan} would be adding resources to $i_2$ and reducing resources for $i_1$ so as to reduce both latency and cost, indicating the need for instance-specific resource allocation.

\subsection{\rv{MOO Problem and Our Approach}}\label{sec:moo-problem}
\rv{
To derive the optimal placement and resource plans, we begin by providing the mathematical definition of multi-objective optimization (MOO) and present an overview of our approach. 
}

\minip{Instance-level MOO.} 
First, consider a given instance to be run on a specific machine. 
We use $f_1, ..., f_k$ to denote the set of predictive models of the $k$ objectives 
and $\theta$ to denote a {\em resource configuration} available on that machine. Then the instance-level multiple-objective optimization (MOO) problem is defined as:

{\definition{\textbf{Multi-Objective Optimization (MOO).}
		\begin{eqnarray}
		\label{eq:mult_obj_opt_def}
		\arg \min_{\theta} &f^\theta = f(\theta)=& {\left[
			\begin{array}{l}
			f_1(\theta) \\
			...\\
			f_k(\theta)
			\end{array}
			\right]} \\
		\nonumber s.t. & & \theta \in \Sigma \subseteq \mathbb{R}^d
		\end{eqnarray}
}}\noindent{where $\Sigma$ denotes the set of all possible resource configurations.}

A point $f' \in \mathbb{R}^k$ Pareto-dominates another point $f''$ iff $\forall i \in [1, k], f_i' \leq f_i''$  and $\exists j \in [1, k], f_j' < f_j''$. A point $f^*$ is \textbf{Pareto Optimal} iff there does not exist another point $f'$ that Pareto-dominates it. Then, the \textbf{Pareto Set} $f = [f^{\theta^1}, f^{\theta^2}, ...]$ includes all the Pareto optimal points in the objective space $\Phi \subseteq \mathbb{R}^k$, under the configurations [$\theta^1, \theta^2, \ldots$], and is the solution to the MOO problem.
		
\cut{For a given instance on a specific machine, our instance-level models offer a predictive model for each objective. Then we can use any existing MOO method to solve the instance-level MOO problem. In our work, we use the Progressive Frontier (PF) algorithm in our prior work, UDAO~\cite{spark-moo-icde21}, since it is shown to be the fastest.
}

\rv{
\minip{Stage-level MOO.}
We next consider the stage-level MOO problem over multiple instances.
Consider $m$ instances, ($x_1, ..., x_m$), and $n$ machines,  ($y_1$, ..., $y_n$).
We use $\tilde{x}_i$ to denote the characteristics of $x_i$ based on its  features of Ch1 and Ch2 in its multi-channel representation, and $\tilde{y}_j$ to denote the features of Ch4 and Ch5 of machine $y_j$. 
Then consider two sets of variables, $B$  and $\Theta$: 
\begin{enumerate}
\item 
$B \in \mathbb{R}^{m\times n}$ is the binary assignment matrix, where $B_{i,j} = 1$ when $x_i$ is assigned to $y_j$, and $\sum_j B_{i,j} = 1$, $\forall i = 1...m$.
\item 
$\Theta = [\Theta_1, \Theta_2, ..., \Theta_m]$, denotes the collection of resource configurations of $m$ instances, where $\forall i \in [1, \ldots, m], \Theta_i \in \Sigma_i^* \subseteq \mathbb{R}^d,$  and $\Sigma_i^*$ is set of possible  configurations of $d$ resources (e.g., $d=2$ for CPU and memory resources) for instance $i$.
\end{enumerate}

Given these variables, suppose that $f$ is the instance-level latency prediction model, i.e., $f(\tilde{x}_i, \Theta_i, \tilde{y}_j)$ gives the latency when $x_i$ is running on $y_j$ using the resource configuration $\Theta_i$.
Then the stage-level latency can be written as, 
$
L_{stage} =  \max_{i,j} B_{i,j} f(\tilde{x}_i, \Theta_i, \tilde{y}_j)
$.
The stage cost is the weighted sum of cpu-hour and memory-hour and can be written as, 
$
C_{stage}  = \sum_{i,j} B_{i,j}  f(\tilde{x}_i, \Theta_i, \tilde{y}_j) 
   (\mathbf{w} \cdot  \Theta_i^T)
$, 
where $\mathbf{w}$ is the weight vector over $d$ resources and 
$\mathbf{w} \cdot  \Theta_i^T$ is the dot product between $\mathbf{w}$ and $\Theta_i$.
Other objectives  can be written in a similar fashion. 
Then we have the Stage-level MOO Problem:

{\definition{\textbf{Stage-Level MOO Problem.}
\label{def:stagemoo}
\begin{eqnarray}
		\arg\min_{B, \Theta} & {\left[
			\begin{array}{l}
			L(B, \Theta) = \max_{i,j} B_{i,j}  f(\tilde{x}_i, \Theta_i, \tilde{y}_j) \\
			C(B, \Theta) = \sum_{i,j} B_{i,j}  f(\tilde{x}_i, \Theta_i, \tilde{y}_j)  (\mathbf{w} \cdot  \Theta_i^T ) \\
			...
			\end{array}
			\right]} \\
		\nonumber s.t. &  {\begin{array}{l}
			B_{i,j} \in \{0, 1\}, \,\,\, \forall i=1...m , \forall j=1...n \label{eq:con1}\\
			\sum_j B_{i,j} = 1, \,\,\, \forall i=1...m \\
			\sum_i B_{i,j} \Theta_i^1 \leq U_j^1 , \,\, \ldots, \,\,\sum_i B_{i,j} \Theta_i^d \leq U_j^d \,\, , \forall j=1...n \\
			\end{array}}
		\end{eqnarray}
}}where $U_j \in \mathbb{R}^d$ is the $d$-dim resource capacities on  machine $y_j$.
}

\rv{
\minip{Existing MOO Approaches.}
Given the above definition of stage-level MOO, one approach is to call existing MOO methods~\cite{marler2004survey,messac2003nc,Emmerich:2018:TMO,spark-moo-icde21} to solve it directly. However, this approach is facing a host of issues: 
(1)~The parameter space is too large. In our problem setting, both $m$ and $n$ can reach 10's of thousands. Hence, both $B \in \mathbb{R}^{m\times n}$ and 
$\Theta = [\Theta_1, \Theta_2, ..., \Theta_m]$, $\forall i, \Theta_i \in \mathbb{R}^d$, involve  $O(mn)$ and $O(md)$ variables, respectively, which challenge all MOO methods. 
(2)~There are also constraints specified in Def.~\ref{def:stagemoo}. Most MOO methods do not handle such complex constraints and hence may often fail to return feasible solutions.
We will demonstrate the performance issues of this approach in our experimental study.

\minip{Our MOO Approach.}
To solve the complex stage-level MOO problem while meeting stringent time constraints, we devise a novel MOO approach that proceeds in two steps:

\minip{\underline{Step 1 (IPA)}:} Take the resource configuration $\Theta_0$ returned from the HBO optimizer as the default and assign it uniformly to all instances, $\Theta_i = \Theta_0$, $\forall i \in [1,...,m]$.  Then minimize over $B$ in Def.~\ref{def:stagemoo} by treating $\Theta_0$ as a constant. 
\cut{
\begin{align}
		\label{eq:stage-ipa}
		\arg \min_{B} & {\left[
			\begin{array}{l}
			L(B, \Theta_0) = \max_{i,j} B_{i,j} f(\tilde{x}_i, \Theta_0, \tilde{y}_j) \\
			C(B, \Theta_0) = \sum_{i,j} B_{i,j} f(\tilde{x}_i, \Theta_0, \tilde{y}_j) (\mathbf{w} \cdot  \Theta_0^T ) \\
			...
			\end{array}
			\right]} 
\end{align}
}
		
\minip{\underline{Step 2 (RAA)}:} Given the solution from step 1, $B^*$, we  now minimize over the  variables $\Theta$ in Def.~\ref{def:stagemoo} by treating $B^*$ as a constant.
\cut{
\begin{align}
		\label{eq:stage-raa}
		\arg \min_{\Theta} & {\left[
			\begin{array}{l}
			L(B^*, \Theta) = \max_{i,j} B_{i,j}^* f(\tilde{x}_i, \Theta_i, \tilde{y}_j) \\
			C(B^*, \Theta) = \sum_{i,j} B_{i,j}^* f(\tilde{x}_i, \Theta_i, \tilde{y}_j) (\mathbf{w} \cdot  \Theta_i^T ) \\
			...
			\end{array}
			\right]} 
		\end{align}
}

The intuition behind our approach is that if we start with a decent choice of $\Theta_0$, as returned by HBO, we hope that step 1 will reduce stage latency via a good assignment of instances to machines by considering machine capacities and instance latencies. Then step 2 will fine-tune the resources assigned to each instance on a specific machine to reduce stage latency, cost, as well as other objectives. 
We next present our IPA and RAA methods that implement the above steps, respectively, prove their optimality based on their respective definitions, 
	and optimize them to run well under a second. 
}



\subsection{Intelligent Placement Advisor (IPA)}
\label{sec:ipa}

\rv{
\minip{Problem Setup.}
IPA determines the placement plan (PP) by way of minimizing the stage latency.} We assume that the total available resources in a cluster are more than a stage's need -- this assumption is likely to be met in a production system with 10's of thousands of machines\footnote{Otherwise, there is a scheduling delay for the stage with the admission control.}.
We further suppose that all the instances of a stage can start running simultaneously and hence the stage latency equals the maximum instance latency.

IPA begins by feeding the MCI features of the stage to the latency model and builds a latency matrix, $L\in \mathbb{R}^{m\times n}$, to include the predictions of running the instances on all available machines, i.e., $L_{i,j}$ is the latency of running instance $x_i$ on machine  $y_j$.
\rv{By setting $L_{i,j} = f(\tilde{x}_i, \Theta_0, \tilde{y}_j)$, ignoring the constant $\Theta_0$, and focusing on the stage latency objective in Def.~\ref{def:stagemoo}, we obtain:
\begin{align}
		\label{eq:ipa}
		\arg \min_{B} & {\left[
			\begin{array}{l}
			L(B) = \max_{i,j} B_{i,j} L_{i,j} 
			\end{array}
			\right]} \\
		\mbox{s.t.} & \;\; \sum_{j} B_{i,j} = 1, \forall i  
    		\;\; \mbox{and}
                 \;\; \sum_i B_{i,j} \leq \beta_j, \forall j
\end{align}
Here, we add the constraint $\sum_i B_{i,j} \leq \beta_j$ to meet the {\em resource capacity constraint} of each machine and a {\em diverse placement preference} (soft constraint) of sending instances to different machines to reduce resource contention.
Let $\beta_j$ denote the maximum number of instances a machine $y_j$ can take based on its capacity constraint ($U_j \in \mathbb{R}^d$) and diversity preference. We have
\begin{equation}
    \beta_j =  \min \{\lfloor U_j^1 / \Theta_0^1 \rfloor, \lfloor U_j^2 / \Theta_0^2 \rfloor, ..., \lfloor U_j^d / \Theta_0^d\rfloor, \alpha \} 
\end{equation}}where  $\alpha$ is a parameter that defines the maximum number of instances each machine can take based on the diverse placement preference. 
Basically, a smaller $\alpha$ shows a stronger preference of the diverse placement for a stage and we have $\alpha >= \lceil m/n \rceil$.

From Eq.~\eqref{eq:ipa}, 
given that the variable $B$ is a binary matrix,  it is an NP-hard {\em Integer Linear Programming (ILP)}  problem~\cite{algo-book}.

\minip{IPA Method.}
We design a new solution to IPA based on the following intuition: Since the stage latency is the maximum instance latency, intuitively, the placement plan wants to reduce the latency of the longer-running instances by sending them to the machines where they can run faster, potentially at the cost of compromising the latency of other short-running instances. 

Our IPA method works as follows: 
We first prioritize the instances by their {\it best possible latency} (BPL), where BPL is defined as the minimum latency that an instance  can achieve among all available machines. 
Then we keep sending the instance of the largest BPL to its matched machine and updating the BPL for instances when a machine cannot take more instances.
The full procedure is given in Algorithm~\ref{alg:appr-ipa-algo}.
The time complexity is $O(m(m+n)+d)$ using parallelism and vectorized computations in the implementation.
 
\rv{\minip{Optimality Result.}
Our proof of the IPA result is based on the following  \textit{column-order} assumption about the latency matrix $L$: For a given column (machine), when we visit the cells in increasing order of latency, let $s_1, \ldots, s_m$ denote the indexes of the cells retrieved this way. Then the column-order assumption is that all columns of the $L$ matrix share the same order, $s_1, \ldots, s_m$. See Fig.~\ref{fig:toy} for an example of $L$  matrix. This assumption is likely to hold because, in the IPA step, all machines  use the same amount of resources, $\Theta_0$, to process each instance. Hence, the order of latency across different instances is strongly correlated with the size of input tuples (or tuples that pass the filter) in those instances, which is independent of the machine used.  Empirically, we verified that this assumption holds over 88-96\% stages across three large production workloads, where the violations largely arise from the uncertainty of the learned model used to generate the $L$ matrix.

Below is our main optimality result of IPA. All proofs in this paper are left to \techreport{Appendix~\ref{appendix:ipa}}{\cite{tech-report}} 
due to space limitations. 

\begin{theorem}\label{thm:ipa}
	IPA achieves the single-objective stage-latency optimality under the column-order assumption.
\end{theorem}
}

 
\begin{algorithm}[t]
	\renewcommand{\algorithmiccomment}[1]{\bgroup//~{\it#1}\egroup}
	\caption{IPA Approach}
	\label{alg:appr-ipa-algo}
	\small
	\begin{algorithmic}[1]  
		\STATE $L$ = cal\_latency(model, X, Y), $S$ = cal\_max\_num\_inst().
		\STATE $Y^* = Y$, $X^* = X$, and $P = \{\}$. \hspace*{35mm} \COMMENT{Init}
		\STATE $BPL_{list}$ = cal\_bpl($L$, $X^*$, $Y^*$).
		\REPEAT
		\STATE $i_t, j_t$ = argmax($BPL_{list}$) 
		\STATE \hspace*{4mm} \COMMENT{get the instance and machine index pairs of the largest BPL}
		\STATE $P = P \bigcup \{x_{i_t} \rightarrow y_{j_t}\}$, $X^* = X^* - \{ x_{i_t} \}$, $S_{j_t} = S_{j_t} - 1$. 
		\IF {$Y^*$ is $\emptyset$ \AND $X^*$ is not $\emptyset$}
		\RETURN $\{\}$ \hspace*{34mm}\COMMENT{No solution found}
		\ENDIF		
		\STATE\algorithmicif\ {$S_{j_t} == 0$} \algorithmicthen\ $Y^* = Y^* - \{ y_{j_t}\}$, Recalculate the $BPL_{list}$.
		\STATE\algorithmicend\ \algorithmicif
		\UNTIL{$X^*$ is empty}
		\RETURN P
	\end{algorithmic}
\end{algorithm}

\cut{
{\it Step 1.} Calculate the predicted latency matrix $L$ and the maximum number of instances accepted on each machine $s_j$.

{\it Step 2.} Initialize the available machines as $Y^* = Y$, instance remained to be sent as $X^* = X$, and the placement plan $P = \{\}$.

{\it Step 3.} Get the BPL for all the remaining instances to have a BPL list. $BPL_i = \min_{\{j | Y_j \in Y*\} } L_i^j, \forall i \in \{i | x_i \in X^*\}$

{\it Step 4.} Get the id $i_t = arg\max_{\{i | x_i \in X^*\}}(BPL_i)$ for the instance of the largest BPL.

{\it Step 5.} Get the id $j_t = arg\min_{\{j | Y_j \in Y*\} } L_i^j$ for the machine that achieves BPL with $x_{i_t}$.

{\it Step 6.} $P = P \bigcup \{x_{i_t} \rightarrow y_{j_t}\}$. Update $X^* = X^* - \{ x_{i_t} \}$, $s_j = s_j - 1$. 

{\it Step 7.} When $s_j$ equals 0, further update $Y^* = Y^* - \{ y_{j_t}\}$ and recalculate the BPL list. Otherwise, skip the step.

{\it Step 8.} When $X^*$ is empty, return P and finish. Otherwise, if $Y^*$ is empty, return ``no solution found'', else go back to step 4.
}


\minip{Boosting IPA with clustering.} 
A remaining issue is that a stage can be up to 80K instances in a production workload, and the number of machines can also be tens of thousands. \cut{For example, solving the PP for a stage of 6716 instances over 8000  machines costs 2.4s, 
whereas a scheduler must give the response within 200 msec. 
}To further reduce the time cost, we exploit the idea that groups of machines or instances may behave similarly in the placement problem. 
Therefore, we boost the IPA efficiency by clustering both machines and instances without losing much stage latency performance. 

While there exist many clustering methods~\cite{book-data-clustering}, running them on many instances with large MCI features is still an expensive operation for a scheduler. This motivates us to design a customized clustering method based on MCI properties. 
An instance in a stage is characterized by its query plan (Ch1), instance meta (Ch2), resource plan (Ch3) and additional instance meta (AIM), where Ch1 and Ch3 are the same for all instances, hence not needed in clustering,   
and AIM fully depends on Ch1 and Ch2. So the key factors lie in Ch2, where the input row number and input size are correlated. Thus, we approximately characterize an \textbf{instance} only by its input row number and apply 1D  density-based clustering~\cite{book-data-clustering}. To represent a cluster, we choose the instance with the largest input row number to avoid latency underestimation. A \textbf{machine} in the cluster is characterized by its system states (Ch4) and hardware type (Ch5). We cluster them based on discretized values of Ch4 and Ch5.

Suppose that clustering yields $m'$ instance clusters and $n'$ machine clusters. Then the time complexity of IPA is $O(m\log m + n\log n + m'(m'+n') + d)$, 
where a sorting-based method is used for clustering both instances and machines. 
Since $m' \ll m, n' \ll n$, the complexity reduces to $O(m\log m + n\log n)$. 
Compared to other packing algorithms~\cite{LiNN14} that solve linear programming problems with quadratic complexity, the complexity of our algorithm is lower.

\subsection{Resource Assignment Advisor (RAA)}
\label{sec:raa}


After IPA determines the placement plan of a stage, each of its instances is scheduled to run on a specific machine. \rv{Then we propose a Resource Assignment Advisor (RAA) to tune the resource plan of each instance to solve a MOO problem, e.g., minimizing the stage latency, cloud cost, as well as other objectives.}

\rv{
\minip{Stage-level MOO.}
Consider the stage-level MOO problem defined in Def.~\ref{def:stagemoo}. By ignoring the constant $B$, we can rewrite it in the following abstract form using \textbf{aggregators ($g_1, ..., g_k$)}, each for one objective, to be applied to $m$ instances:
		\begin{align}
		\label{eq:stage-mult_obj_opt_def}
		\arg \min_{\Theta} & \,\,\, F^{\Theta} = F(\Theta)=& {\left[
			\begin{array}{l}
			F_1(\Theta) = g_1(f_1(\Theta_1), ..., f_1(\Theta_m)) \\ 
			...\\
			F_k(\Theta) = g_k(f_k(\Theta_1), ..., f_k(\Theta_m))
			\end{array}
			\right]} 
		\end{align} 
\noindent{where $\Theta = [\Theta_1, \Theta_2, ..., \Theta_m] = [\theta_1^{l_1}, \theta_2^{l_2}, ... \theta_m^{l_m}]$, $\theta_i^{l_i}$ denotes the $l_i$-th resource configuration of the $i$-th instance ($i \in \{ 1,...,m \}$), and the aggregator $g_j$ ($j$=1,...,$k$) is either a \texttt{sum} or \texttt{max}. 
}
}




As in the instance-level, we define the \textbf{stage-level resource configuration} as $\Theta = [\theta_1^{l_1}, \theta_2^{l_2}, ... \theta_m^{l_m}]$, with its corresponding $F^\Theta$ in the objective space $\Phi \subseteq \mathbb{R}^k$. Then we can define stage-level \textbf{Pareto Optimality} and the \textbf{Pareto Set} similarly as before.

Note that we are particularly interested in two aggregators. The first is \texttt{max}: the stage-level value of one objective is the maximum of instance-level objective values over all instances, e.g., latency. The second is \texttt{sum}: the stage-level value of one objective is the sum of instance-level objective values, e.g., cost.

\rv{
\underline{Example.} 
In Figure~\ref{fig:raa-example}, we have $f_1$ as the predictive model for latency ($g_1 = $ \texttt{max}) and $f_2$ for cost ($g_2 = $ \texttt{sum}). Suppose $\Theta = [\theta_1^1 ,\theta_2^2]$ as the stage-level resource configuration.  
Then we have the stage-level latency as $F_1(\Theta) = \max(150, 100) = 150$, the stage-level cost as $F_2(\Theta) = sum(5, 5) = 10$, and the stage-level solution $F^\Theta$ is $[150 ,10]$.}

\rv{
While one may consider using existing MOO methods to solve Eq.~\eqref{eq:stage-mult_obj_opt_def}, it is still subject to a large number ($O(md)$) of variables. To enable a fast algorithm, we next introduce our Hierarchical MOO approach developed in the divide-and-conquer paradigm. 

\definition{\textbf{Hierarchical MOO.} We solve the instance-level MOO problem for each of the $m$ instances of a stage separately, for which we can use any existing MOO method (in practice, the Progressive Frontier (PF) algorithm~\cite{spark-moo-icde21} since it is shown to be the fastest). Suppose that there are $k$ stage-level objectives, each with its \texttt{max} or \texttt{sum} aggregator to be applied to $m$ instances. Our goal is to efficiently find the stage-level MOO solutions from the instance-level MOO solutions, for which we use $f_i^j$ to denote the j-th \po solution in the i-th instance by using $\theta_i^j$. }

}

\minip{General hierarchical MOO solution.}
Suppose that there are $k$ user objectives, where $k_1$ objectives use the \texttt{max} and $k_2$ objectives use \texttt{sum}, $k_1 + k_2 = k$. 
\rv{
We are also given instance-level Pareto  sets. Then we want to select a solution for each instance such that the corresponding stage-level solution is Pareto optimal.
Algorithm~\ref{algo_stage_moo} gives the full  description.
After initialization, line $2$ obtains the lower and upper bounds for each of the $k_1$ \texttt{max} objectives. In line $3$ ($find\_all\_possible\_values$), we first find for each \texttt{max} objective all the possible values as one list, and then use the Cartesian product of these lists as the candidates for the $k_1$ \texttt{max} objectives.
}
\rv{
In lines $4$ to $10$, given one candidate, we try to 
find the corresponding Pareto-optimal solution for the $k_2$ \texttt{sum} objectives.
}
For the $k_2$  objectives using \texttt{sum}, due to the complexity of the sum operation, we can not afford the enumeration, which grows exponentially in the number of instances, 
$O(p_{max}^m)$
where $p_{max}$ is the maximum number of instance-level Pareto points among $m$ instances. 
To reduce the complexity, we resort to any existing MOO method, denoted by the function \textit{find\_optimal}, 
\rv{
 that (in line 6) for each instance selects one Pareto-optimal solution for the $k_2$  objectives.
}
At the end of the procedure, we add a filter to remove the non-optimal solutions.

\begin{algorithm}[t]
	\caption{General Hierarchical MOO}
		\SetKwData{Left}{left}\SetKwData{This}{this}\SetKwData{Up}{up} \SetKwFunction{findRange}{findRange}\SetKwFunction{findAllSolutions}{findAllSolutions}\SetKwFunction{enum}{enum}\SetKwFunction{findOptimal}{findOptimal}\SetKwFunction{calObjValues}{calObjValues}\SetKwFunction{filterDominated}{filterDominated} \SetKwInOut{Input}{input}\SetKwInOut{Output}{output}
	\label{algo_stage_moo}
	\small
	\begin{algorithmic}[1] 
		\REQUIRE {$f_i^j, i\in[1,m], j\in[1,p_i]$, where $p_i$ is the number of \po solutions in $i$-th instance
		and $ [\theta_i^1, ..., \theta_i^{p_i}]$}
		\STATE $PO_{\Theta} = [], PO_{F}=[]$
		\STATE min{M}List, max{M}List = {\it find\_range}($\bm{f}$)
		\rv{
		\STATE k1Combs = {\it find\_all\_possible\_values}($\bm{f}$, min{M}List, max{M}List)}	
		\FOR {$c$ \textbf{in} k1Combs}
			\FOR{$i$ \textbf{in} $m$}
				\STATE optimal\_solution, index = {\it find\_optimal}($c$, [$f_i^1, ... f_i^{p_i}$])
				\STATE $\Theta_i = \theta_i^{index}$
			\ENDFOR
			\rv{
			\STATE $PO_{\Theta}.append(\Theta)$, $PO_{F}.append(F(\Theta))$}
		\ENDFOR
		\rv{\STATE $PO_{F}, PO_{\Theta}$ = {\it filter\_dominated}($PO_{F}, PO_{\Theta}$)
		\RETURN $PO_{F}, PO_{\Theta}$}
	\end{algorithmic}
\end{algorithm}

%
%

\underline{Example.} 
\chenghao{
In Figure~\ref{fig:raa-example}, we calculate the lower and upper bounds of the stage-level latency as $\max(55, 100) = 100$ and $\max(150, 300)$ $= 300$ respectively. Then we enumerate all the possible stage-level latency values within the bounds ($100, 150, 300$) to collect potential \po solutions.
For latency=$100$, there is only one feasible $\Theta$ choice ($\Theta=[\theta_1^2, \theta_2^2]$) and hence the only solution is $[100, 25]$. Similarly for latency=$150$, we get another solution $[150, 10]$.
When latency=$300$, there are two solutions: $[300, 24]$ with $\Theta=[\theta_1^2, \theta_2^1]$ and $[300, 9]$ with $\Theta=[\theta_1^1, \theta_2^1]$.
The first one will be filtered because it is dominated by the latter one (line 6).
Finally, we further filter solutions being dominated in the chosen set (line 11) and get the stage-level MOO solutions as $[[100, 25], [150, 10], [300, 9]]$. 
}

\proposition \label{algo2_optimality} For a stage-level MOO problem, Algorithm~\ref{algo_stage_moo} guarantees to find a subset of stage-level Pareto optimal points.

\begin{algorithm}[t]
	\caption{\rv{RAA Path policy to construct $PO_{\Theta}$ and $PO_F$}}
	\label{alg:appr-raa}
	\small
	\begin{algorithmic}[1]  
		\REQUIRE {\rv{$p_i, \theta_i^j, f_i^j$, $\forall i\in [1,m], j\in[1, p_i]$, where we pre-sort $f_i^j$ and \\$\theta_i^j$ in the descending order of latency for each instance.}}
		\STATE \rv{$u_i^j = lat(f_i^j), \forall i\in [1,m], j\in[1,p_i]$. }
		\STATE \rv{$PO_{\Theta} = [], PO_F=[]$}, $\lambda = [1, 1, ..., 1]$, smax=inf 
		\STATE Build the max heap $Q$ with data points $[(u_i^1, i)], \forall i \in [1,m]$. 
		\REPEAT
		\STATE qmax, $i = Q.pop()$ 
		\IF {qmax < smax}
		\STATE \rv{$PO_{\Theta}.append(\Theta^\lambda)$, $PO_F.append(F(\Theta^\lambda))$}, smax = qmax
		\ENDIF
		\STATE $\lambda = \pi_i(\lambda)$
		\STATE \algorithmicif\ {$\lambda_i > p_i$} \algorithmicthen\ \Return{$PO_{\Theta}$, $PO_F$}
		\STATE
		\algorithmicend\ \algorithmicif
		\STATE $Q.push((u_i^{\lambda_i}, i))$ 
		\UNTIL{True}
	\end{algorithmic}
\end{algorithm}

\minip{Fast hierarchical MOO solution: RAA Path.} 
In the case that there are only two objectives, whose aggregators are \texttt{max} and \texttt{sum}, respectively,  we provide another algorithm that can resolve the Hierarchical MOO efficiently.
The key idea here is that the two objectives are indeed making tradeoffs. The increase in one objective often leads to a decrease in the other. With this observation, we design a new approach called "RAA path".


\cut{
Figure~\ref{fig:raa-path} shows an example of the Pareto sets of 3 instances in a stage, where each instance has 2, 4, 3 Pareto-optimal solutions, and each solution is sorted by latency in descending order.
We start the procedure by selecting the first Pareto solution for each instance, this leads to the $[\theta_1^1, \theta_2^1, \theta_3^1]$ and $F = [300, 14]$. Notice that $300$ is corresponding to $\theta_2^1$, so we replace $\theta_2^1$ by the entry below it in the same Pareto set. Now we have the second solution which is $[\theta_1^1, \theta_2^2, \theta_3^1]$ and $F = [150, 15]$. Continue in this manner, each time select the $\theta$ corresponding to the max latency and replace it with the item below it the same Pareto set. It terminates until there is no entry below the current $\theta$. 
}

\underline{Example.}
Figure~\ref{fig:raa-path} shows an example of the Pareto sets of 3 instances in a stage, where each instance has 2, 4, 3 Pareto-optimal solutions, and each solution is sorted by latency in descending order.
We start the procedure by selecting the first Pareto solution for each instance, this leads to $\Theta^\lambda = [\theta_1^1, \theta_2^1, \theta_3^1]$ and $F^\lambda = [300, 14]$. Notice that $300$ is corresponding to $\theta_2^1$, so we replace $\theta_2^1$ by the entry below it in the same Pareto set. Now we have the second solution as $\Theta^\lambda = [\theta_1^1, \theta_2^2, \theta_3^1]$ and $F^\lambda = [150, 15]$. Continue in this manner; each time, select the $\theta$ corresponding to the instance of the max latency and replace it with the item below it in the same Pareto set. It terminates until there is no entry below the current $\theta$. 


For a formal description, we use the following notation: 
(1)~$\lambda = [\lambda_1, \lambda_2, ..., \lambda_m]$ as a state, where $\lambda_i$ is the index of the Pareto point in instance $i$ with the $\lambda_i$-largest latency;
(2)~$\Theta^\lambda = [\theta_1^{\lambda_1},  \theta_2^{\lambda_2}, ..., \theta_m^{\lambda_m}]$;
(3)~$\pi_i: \lambda \rightarrow \lambda'$ is a {\it step} such that the $i^{th}$ dimension in state $\lambda$ is increased  by 1.
(4)~$p_i$ is the number of instance-level Pareto solutions for instance $i$.
Then Algorithm~\ref{alg:appr-raa} gives the RAA Path algorithm.

\proposition \label{algo3_optimality} For a stage-level MOO problem with two objectives using \texttt{max} and \texttt{sum}, respectively, 
\techreport{Algorithm~\ref{alg:appr-raa}}{RAA path}
guarantees to 
find the full set of stage-level Pareto optimal points at the complexity of 
$O( m \cdot p_{max}\log (m\cdot p_{max}))$.




\minip{RAA with Clustering.}
For efficiency, we run both the General hierarchical MOO and RAA Path methods by clustering the instances and machines, where $m$ is replaced by $m'<<m$ 
	in the complexity.

\minip{Resource plan recommendation. }
After getting the stage-level Pareto set, we reuse UDAO's Weighted Utopia Nearest (WUN) strategy to recommend the resource plan, which includes a configuration for each instance of the stage.
It recommends the resource plan whose objectives could achieve the smallest distance to the Utopia point, which is the hypothetical optimal in all objectives. \cut{The distance computation is also weighted by a weight vector for latency and cost based on previous traces to capture their importances.}




\section{Experiments}
\label{sec:expr}
%

This section presents the evaluation of our models and stage optimizer. 
Using production traces from MaxCompute, we first analyze our models and compare them to state-of-the-art modeling techniques. We then report the end-to-end performance of our stage optimizer using a simulator of the extended MaxCompute environment (detailed in\techreport{~\ref{appendix:simulator}}{~\cite{tech-report}}) by replaying the production traces. 


\cut{
We consider three representative workloads in a total of $\approx$ 0.62 million jobs, with $\approx$ 2 million stages and $\approx$ 0.12 billion instances. 
Each workload consists of productive jobs from a business department over five consecutive days in Alibaba. 
As shown in Fig.~\ref{tab:workload-stats}, workloads A, B, and C show various properties: 
\cut{
\begin{enumerate}
	\item the total number of jobs (41K-0.4M), stages (0.1M-1.0M), and instances (34M-50M).
	\item the average number of stages per job (2.4-5.0)
	\item the average number of instances per stage (85-1224).
	\item the average number of operators per stage (3.7-6.3)
	\item the average job latency (31-377s).
	\item the average stage latency (15-182s).
	\item the average instance latency (16-71s).
\end{enumerate} 
}
Workload A involves the most significant number of jobs (2-10x compared to B and C), while most are short-running jobs. 
Workload C has the least number of jobs but the longest job/stage/instance latencies, with its average job latencies 3-12.5x larger than A and B. It also has more instances on average, 12-14x times than A and B. 
}
{\bf Workload Characteristics.} See Table~\ref{tab:workload-stats} for descriptions of production workloads A-C. We give additional details in
 \techreport{Appendix~\ref{appendix:workload}}{~\cite{tech-report}}.

\cut{
{\bf Hardware Property.}
Our experiments are deployed over 3 machines, each with 2x Intel Xeon Platinum 8163 CPU with 24 physical cores, 500G RAM, and 8 GeForce GTX 2080 GPU cards.
}

\cut{
\begin{figure*}[t]
	\centering
	\vspace{-0.3in}
	\begin{tabular}{ccc}
	\hspace{-0.3in}
		\subfigure[\small{MCI profiling}]
		{\label{fig:expr-mci}\includegraphics[height=3.2cm,width=.35\textwidth]{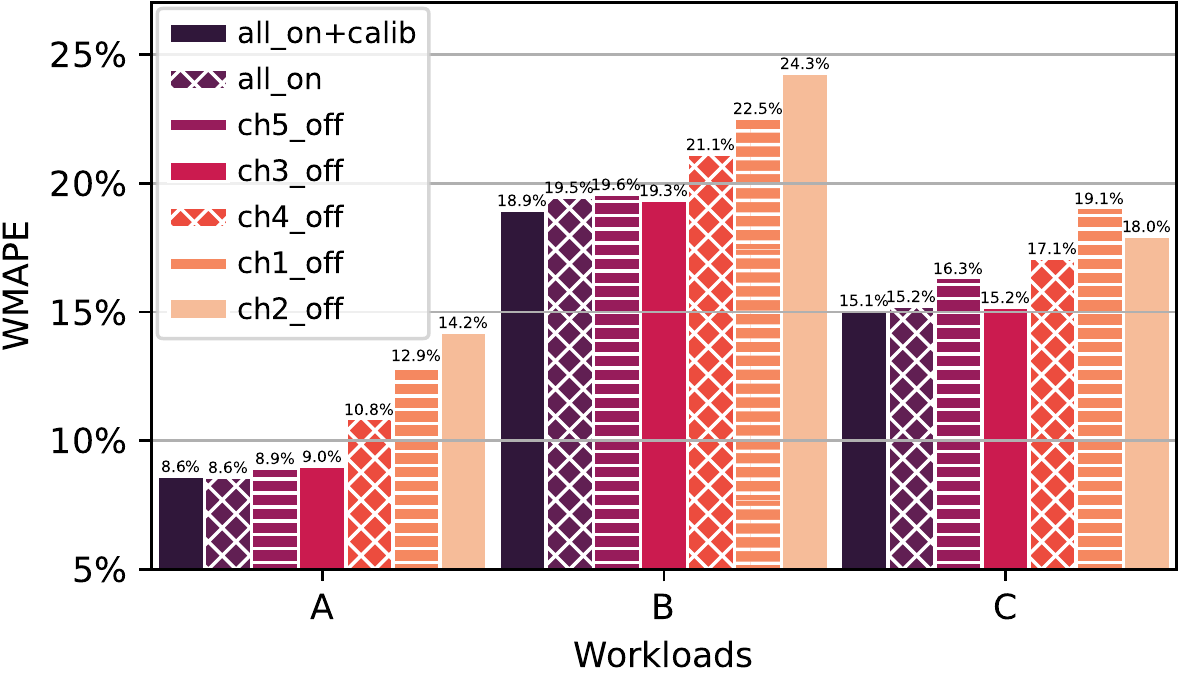}}

		\subfigure[\small{Comparison among different cardinality choices}]
		{\label{fig:expr-card}\includegraphics[height=3.2cm,width=.35\textwidth]{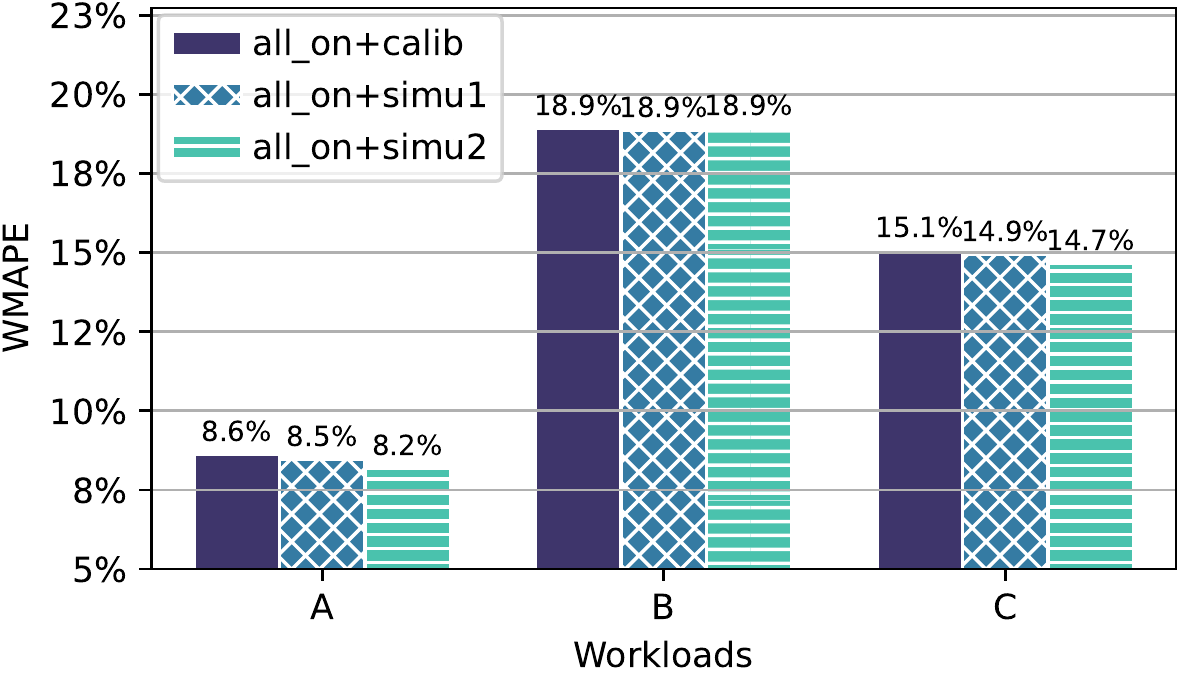}}
		
		\subfigure[\small{Comparison with single-machine SOTAs}]
		{\label{fig:sota-single-mach}\includegraphics[height=3.2cm,width=.35\textwidth]{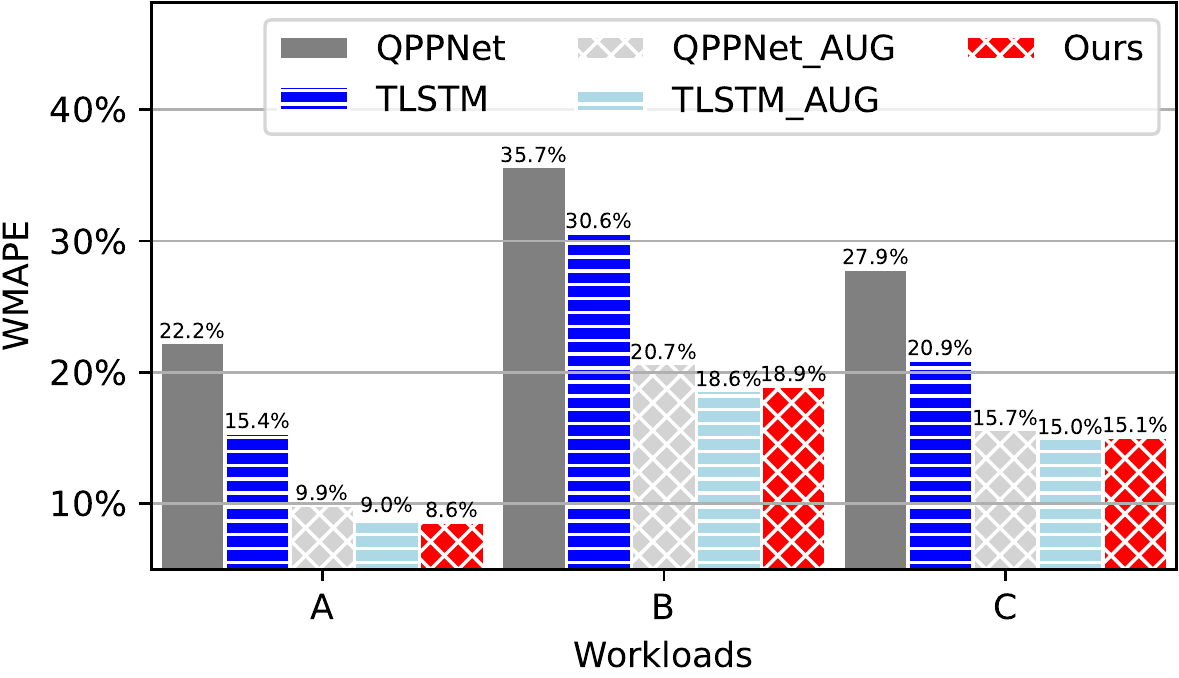}}				
		
				
	\end{tabular}
	\vspace{-0.25in}
	\caption{Performance of our instance-level models, compared to the state-of-the-art (SOTA) methods}
	\label{fig:expr-in-analyses}
	\vspace{-0.1in}
\end{figure*}
}

\begin{figure*}[t]
	\centering
	\begin{tabular}{ccc}
		\subfigure[\small{MCI profiling}]
		{\label{fig:expr-mci}\includegraphics[height=3.3cm,width=.32\textwidth]{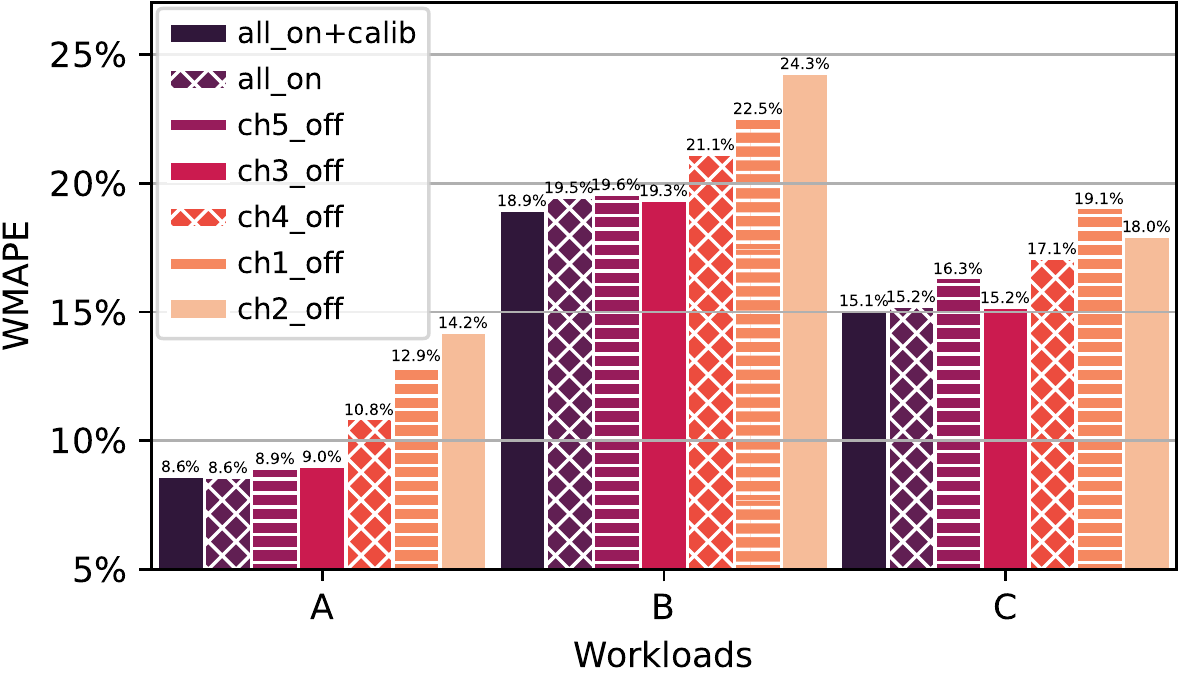}}

		\subfigure[\small{Comparison among different cardinality choices}]
		{\label{fig:expr-card}\includegraphics[height=3.3cm,width=.32\textwidth]{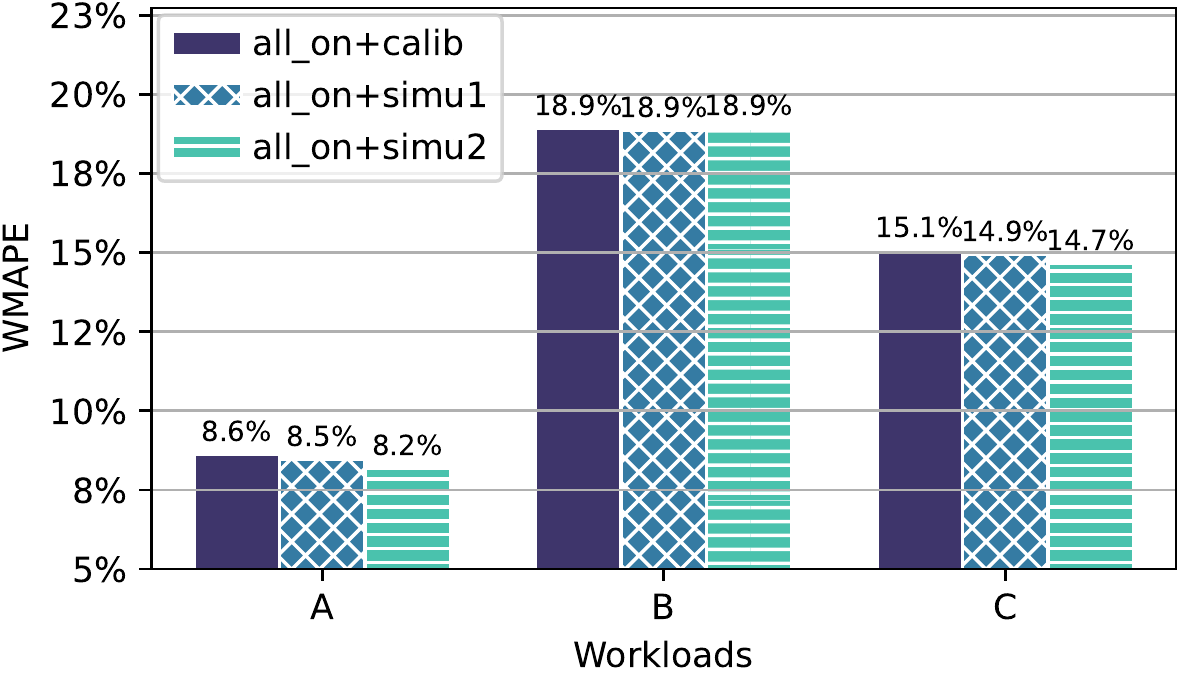}}
		
		\subfigure[\small{Comparison with single-machine SOTAs}]
		{\label{fig:sota-single-mach}\includegraphics[height=3.3cm,width=.32\textwidth]{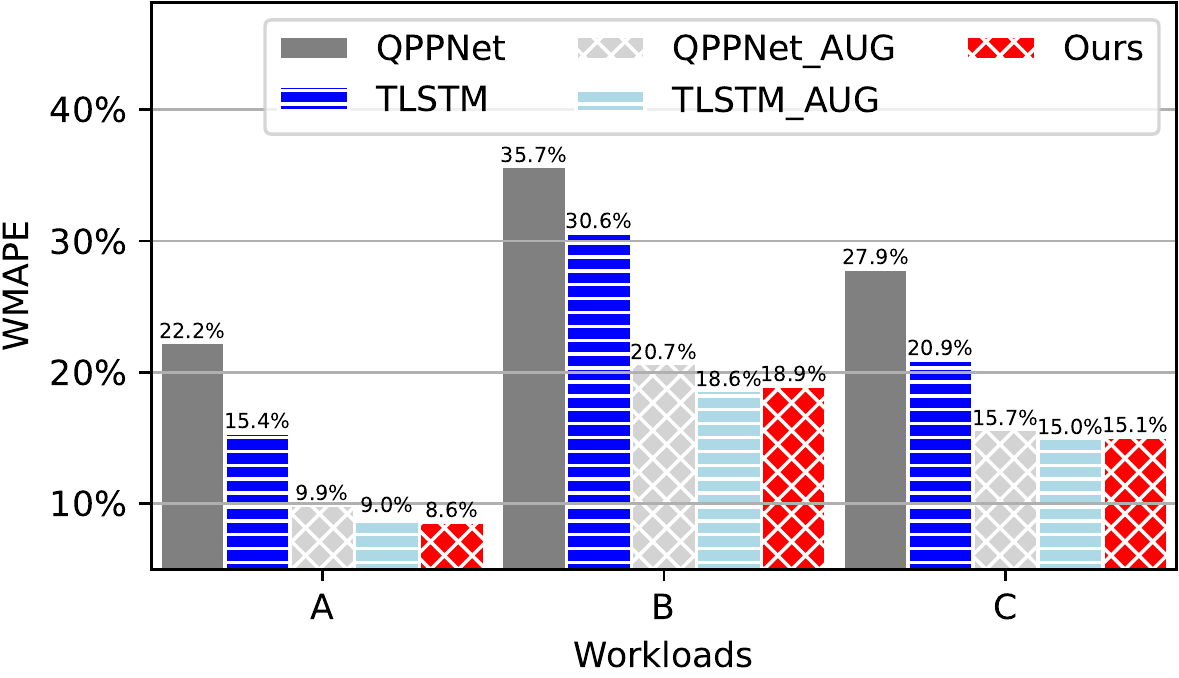}}				
		
				
	\end{tabular}
	\vspace{-0.15in}
	\caption{Performance of our instance-level models, compared to the state-of-the-art (SOTA) methods}
	\label{fig:expr-in-analyses}
	\vspace{-0.1in}
\end{figure*}

\begin{table}[t]
\ra{.7}
\footnotesize
\addtolength{\tabcolsep}{-2.5pt}
\begin{flushleft}
\begin{tabular}{cccccccccc}\toprule
\multicolumn{1}{c}{WL} & \multicolumn{1}{c}{\begin{tabular}[c]{@{}c@{}}Num.\\ Jobs\end{tabular}} & \multicolumn{1}{c}{\begin{tabular}[c]{@{}c@{}}Num. \\ stages\end{tabular}} & \multicolumn{1}{c}{\begin{tabular}[c]{@{}c@{}}Num.\\ Inst\end{tabular}} & \multicolumn{1}{c}{\begin{tabular}[c]{@{}c@{}}\#stages\\ /job\end{tabular}} & \multicolumn{1}{c}{\begin{tabular}[c]{@{}c@{}}\#insts\\ /stage\end{tabular}} & \multicolumn{1}{c}{\begin{tabular}[c]{@{}c@{}}\#ops\\ /stage\end{tabular}} & \multicolumn{1}{c}{\begin{tabular}[c]{@{}c@{}}Avg Job \\ Lat(s)\end{tabular}} & \multicolumn{1}{c}{\begin{tabular}[c]{@{}c@{}}Avg Stage \\ Lat(s)\end{tabular}} & \multicolumn{1}{c}{\begin{tabular}[c]{@{}c@{}}Avg Inst \\ Lat(s)\end{tabular}} \\

\midrule
A  & 405K       & 970K         & 34M        & 2.40      & 35.45      & 3.71      & 30.97      & 14.64        & 16.85       \\
B  & 173K       & 858K         & 36M        & 4.95      & 42.02      & 6.27      & 120.15     & 39.72        & 15.63       \\
C  & 41K        & 100K         & 50M        & 2.42      & 505.51     & 5.31      & 376.83     & 181.88       & 71.08     \\
\bottomrule
\end{tabular}
\end{flushleft}
\caption{Workload statistics for 3 workload over 5 days}
\vspace{-0.1in}
\label{tab:workload-stats}	
\vspace{-0.2in}
\end{table}

\begin{figure}
	\begin{flushright}
	    \setlength{\belowcaptionskip}{0pt}  
	    \setlength{\abovecaptionskip}{3pt}      
        \includegraphics[width=.95\linewidth,height=1.7cm]{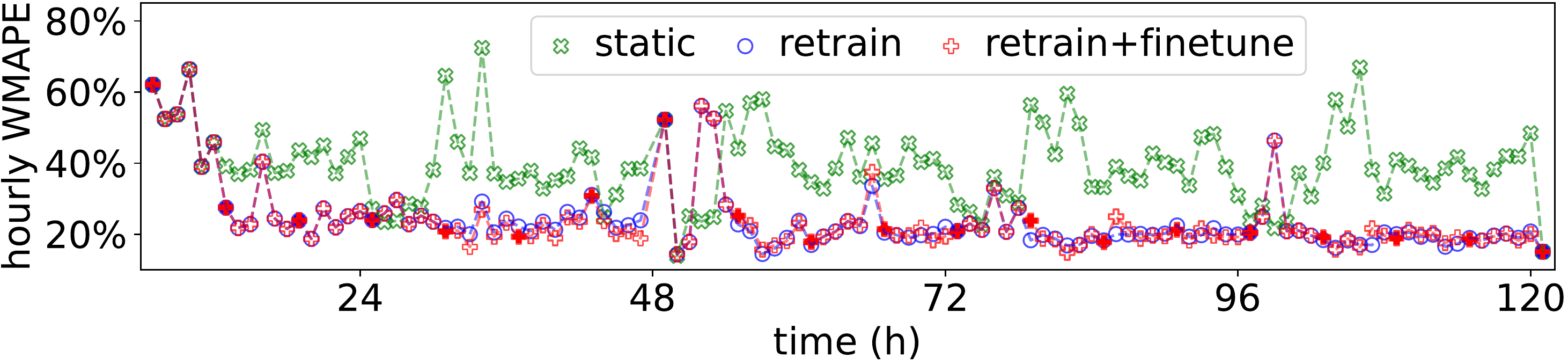}
        \vspace{-0.05in}
        \captionof{figure}{\rv{WMAPEs over time (workload C)}}
        \vspace{-0.05in}
        \label{fig:hourly-wmape}
	\end{flushright}      
	\vspace{-.25in}
\end{figure}

\subsection{Model Evaluation}
\label{sec:expr-model}

To train models, we partition the traces of workloads A-C into training, validation, and test sets, and tune the hyperparameters of all models using the validation set. 
As the default, we train MCI+GTN over all channels augmented by AIM for the instance latency prediction. 
More details about training  are given in \techreport{Appendix~\ref{appendix:model-setup}}{~\cite{tech-report}}.


We report model accuray on 5 metrics: (1)~weighted mean absolute percentage error (WMAPE), (2)~median error (MdErr), (3)~95 percentile Error (95\%Err), (4)~Pearson correlation (Corr),  and (5)~error of the cloud cost (GlbErr).
We choose WMAPE {\it  as the primary, also the hardest, metric} because it assigns more weights to  long-running instances, which are of more importance in resource optimization.

\underline{Expt 1:}  {\it Performance Profiling.} We report the performance of our best models for each workload in Table~\ref{tab:expr-sisl}.
First, our best model achieves 9-19\% WMAPEs and 7-15\% MdErrs over the three workloads with our MCI features (Ch1-Ch5 plus AIM).
Second, WMAPE is a more challenging metric than MdErr. The distribution of the instance latency in production workloads is highly skewed to short- and median-running instances. Therefore, MdErr is determined by the majority of the instances but does not capture well those long-running ones.
Third, the error rate of the total cloud cost is 3-4.5x smaller than WMAPEs, because individual errors of single instances could have a canceling effect on the global resource metric, hence better model performance on the global metric.

Our breakdown analyses further show that IO-intensive operators and the dynamics of system states are the two main sources of model errors.
First, by training a separate model for the instance-level operator latency and calculating the error contribution of each operator type, we find the top-3 most inaccurate operators as \verb|StreamLineWrite|, \verb|TableScan|, and \verb|MergeJoin|, all involving frequent IO activities. It is likely that current traces in MaxCompute miss the features to fully characterize IO operations. 
As for the system dynamics, we train a separate model by augmenting system states with the average system metrics during the lifetime of an instance (which is not realistic to get before running an instance). 
But this allows us to show that the model further reduces WMAPE and MdErr to 6-12\% and 5-10\%, pointing the error source  to the system dynamics.
The comprehensive results are shown in  \techreport{Appendix~\ref{appendix:expt-breakdown-analyses}}{~\cite{tech-report}}.

\underline{Expt 2:} {\it Multi-channel Inputs}.
 We now investigate the importance of each channel to model performance. 
We train separate models with different input channel choices, including (1) the leave-one-out choices for a channel \verb|x| (\verb|Chx_off|), (2) the five basic channels \verb|all_on|, and (3) the five basic channels and the AIM \verb|all_on+calib|.

Fig.~\ref{fig:expr-mci} shows our results. 
The top-3 important features are the instance meta (Ch2), the query plan (Ch1), and the system states (Ch4). By turning off each, WMAPE gets worse by 18-66\%, 16-50\%, and 9-27\%, respectively, compared to \verb|all_on|. 
The effect of the hardware type (Ch5) is not significant, likely because the hardware types used are all high-performance ones. 
The effect of the resource plan (Ch3) is small here due to its sparsity in the feature space; e.g., only 26 different resource plans are observed in workload B. But their effects will be different once they are tuned widely in MOO. 
\underline{Expt 3:} {\it Impact of Cardinality.} 
We train separate models by deriving AIMs from different cardinalities: 
(1)~\verb|all_on+calib| represents the stage-oriented cardinality from MaxCompute's CBO;
(2)~\verb|all_on+simu1|  applies the {\it ground-truth} stage-level cardinalities while assuming operators in different instances  share the stage selectivities;
(3)~\verb|all_on+simu2| applies the {\it ground-truth} instance-level cardinalities for operators.
Fig.~\ref{fig:expr-card} shows at most a 0.2\% and 0.4\% WAMPE can be reduced by using \verb|all_on+simu1| and \verb|all_on+simu2|, respectively.
This means that improving cardinality estimation alone cannot improve much latency prediction in big data systems, which is consistent with CLEO's observation \cite{cleo-sigmod20}. 

\underline{Expt 4:} {\it Comparison with QPPNet}~\cite{qppnet-vldb19} {\it and TLSTM}~\cite{tlstm-cost-estimator-vldb19}. 
We next compare different modeling tools, including (1) original QPPNet, (2) original TLSTM, (3) our extension, MCI-based QPPNet, (4) MCI-based TLSTM, and (5) MCI+GTN (our new graph embedder). 
Fig.~\ref{fig:sota-single-mach} shows our results.
First, QPPNet and TLSTM achieve 22-36\% and 15-31\% WMAPE, respectively, 2-3x larger than our best model, MCI+GTN.
Second, using our MCI, MCI+QPPNet and MCI+TLSTM can improve WMAPE by 12-15\% and 6-12\%, respectively, while MCI+TLSTM and MCI+GTN achieve close performance.
Thus, our MCI framework offers both good modeling performance and  extensibility to adapt SOTAs from a single machine to big data systems.

\underline{Expt 5:} {\it Training Overhead and Model Adaptivity}.
We next report on the training overhead:  
The data preparation stage costs $\sim$2 minutes every day to collect new traces from each department using MaxCompute jobs (with a parallelism of 219). We collected 5-day traces from three internal departments, totaling 0.62M jobs, 2M stages, and 407G bytes. 
In the training stage, we run 24h periodic retraining (\textit{retrain}) at midnight of each day when the workloads are light, which takes 3.5 hours on average over 16 GPUs (including hyperparameter tuning). Optionally, we run fine-tuning every 6 hours (\textit{retrain+finetune}), at the average cost of 0.9h each. 

To study model adaptivity, we design two settings of workload drifts: (a) a realistic setting where queries from 5 days are injected in temporal order; (b) a hypothetical worst case where queries are injected in decreasing order of latency. In both settings, we compare a \textit{static} method that trains the model using the data from the first 6 hours and never updates it afterward, with our 
	\textit{retrain} and \textit{retrain+finetune} methods. 
Our results, as shown in Fig.~\ref{fig:hourly-wmape} for workload C, include: 
(1) The static approach can reach up to 72\% WMAPE in some hours of a day, demonstrating the presence of workload drifts in our dataset.  
(2) \textit{retrain} and \textit{retrain+finetune}  adapt better to the workload drifts, keeping errors mostly in the range of 15-25\% after days of training, while having tradeoffs between them: If the workload patterns in 24h windows are highly regular, \textit{retrain} works slightly better than \textit{retrain+finetune} by avoiding overfitting to insignificant local changes (workload A). Otherwise, \textit{retrain+finetune} works better by adapting to significant local changes (workload B).
More discussion is given in \techreport{Appendix~\ref{appendix:perf-over-workload-shift}}{~\cite{tech-report}}. 

\chenghao{
Additional experiments are in \techreport{Appendix~\ref{appendix:expt-modeling-ss} and~\ref{appendix:expt-modeling-cleo}}{~\cite{tech-report}}.
}

\begin{table*}
\begin{minipage}{0.64\linewidth}
\ra{0.7}
\addtolength{\tabcolsep}{-2.8pt}
\newrobustcmd{\B}{\bfseries}
\newrobustcmd{\BR}{\color{red}}
\newrobustcmd{\BL}{\color{blue}}
\footnotesize
\begin{flushleft}
\begin{tabular}{@{}rrclrclrclccc@{}}\toprule
& \multicolumn{3}{c}{Coverage}
& \multicolumn{3}{c}{$Lat_{s}^{(in)} \downarrow$} 
& \multicolumn{3}{c}{$Cost_{s} \downarrow$} 
& \multicolumn{3}{c}{$avg(T_{s})$ (ms) / $\max(T_{s})$ (ms)} \\
\cmidrule{2-4} \cmidrule{5-7} \cmidrule{8-10} \cmidrule{11-13}
SO choice & A & B & C & A & B & C & A & B & C & A & B & C 
\\ \midrule
\verb|IPA(Org)| &100\% &100\% &100\%  & 11\%   & 20\% & 51\% & 5\%  & 9\% & 15\% & 280 / 2.0K & 17 / 18 & 1.4K / 1.8K   \\ 
\verb|IPA(Cluster)| &100\% &100\% &100\% & \B 12\%  & \B 17\% & \B 50\% & \B 4\%  & \B 7\%  & \B 14\% & \B 15 / 24 & \B 10 / 10 &\B 33 / 36 \\ 
\midrule 
\verb|IPA+RAA(W/O_C)| &100\% &100\% &100\% & 8\% & 79\% & 58\% & 31\% & 76\% & 75\% & 2.5K / 19K & 177 / 220 & 3.5K / 6.3K \\ 
\verb|IPA+RAA(DBSCAN)| &100\% &100\% &100\%& 27\% & 69\% & 67\% & 21\% & 64\% & 74\% & 223 / 937 & 132 / 136 & 258 / 452 \\ 
\verb|IPA+RAA(General)| &100\% &100\% &100\% & 36\% & 80\% & 76\% & 29\% & 75\% & 75\% & 100 / 241 & 20 / 23 & 167 / 229 \\ 
\B \texttt{IPA+RAA(Path)} &100\% &100\% &100\% & \B36\% & \B80\% & \B76\% & \B29\% & \B75\% & \B75\% & \B98 / 226   & \B17 / 18 & \B156 / 224 \\ \midrule 
\BR \texttt{EVO} & \BR   0\%& \BR 82\%& \BR 0\%    & \BR   --\%& \BR -36\%& \BR --\%  & \BR   --\%& \BR 66\%& \BR --\%      & \BR   -- / --& \BR 21K / 24K& \BR -- / -- \\ 
\BR \texttt{WS(Sample)} & \BR   90\%& \BR 85\%& \BR 82\%        & \BR   -140\%& \BR 48\%& \BR -74\%       & \BR   -107\%& \BR 49\%& \BR -52\%       & \BR   7.4K / 22K & \BR 465 / 753& \BR 9.8K / 12K \\ 
\BR \texttt{PF(MOGD)} & \BR   99\%& \BR 100\%& \BR 98\%       & \BR   -15\%& \BR 49\%& \BR 65\%    & \BR   24\%& \BR 56\%& \BR 75\%        & \BR   2.7K / 4.0K & \BR 2.1K / 2.5K & \BR 1.5K / 2.3K \\ \midrule 
\BL \texttt{IPA+EVO} & \BL   98\%& \BL 100\%& \BL 98\%       & \BL   -27\%& \BL 69\%& \BL 43\%    & \BL   22\%& \BL 75\%& \BL 71\%        & \BL   3.0K / 7.3K & \BL 2.4K / 2.6K& \BL 5.0K / 5.9K \\ 
\BL \texttt{IPA+WS(Sample)} & \BL   100\%& \BL 100\%& \BL 100\%     & \BL   -36\%& \BL 76\%& \BL -1\% & \BL   -19\%& \BL 72\%& \BL 32\%    & \BL   3.5K / 10K & \BL 517 / 741 & \BL 12K / 18K \\ 
\BL \texttt{IPA+PF(MOGD)} & \BL   100\%& \BL 100\%& \BL 100\%     & \BL   -0.4\%& \BL 51\%& \BL 69\%     & \BL   26\%& \BL 56\%& \BL 75\% & \BL   1.6K / 2.5K & \BL 1.2K / 1.6K & \BL 1.2K / 2.2K \\ 
\bottomrule
\end{tabular}
\end{flushleft}
\caption{\rv{Average Reduction Rate (RR) against Fuxi in 29 subworkloads within 60s}}
\label{tab:expr-so-sub-workloads}
\end{minipage}
\begin{minipage}{0.35\linewidth}
\begin{minipage}{1.\textwidth}    
\ra{.6}
\centering
\setlength{\belowcaptionskip}{0pt}  
\setlength{\abovecaptionskip}{0pt}  
\footnotesize
\newrobustcmd{\B}{\bfseries}
\addtolength{\tabcolsep}{-2.5pt}
\begin{tabular}{@{}cccccc@{}}
\toprule
\multicolumn{1}{c}{WL} & \multicolumn{1}{c}{WMAPE} & \multicolumn{1}{c}{MdErr} & \multicolumn{1}{c}{95\%Err} & \multicolumn{1}{c}{Corr} & \multicolumn{1}{c}{GlbErr} \\
\midrule
 A     & 8.6\%  & 7.4\%  & 62.4\%  & 96.6\% & 1.9\%  \\
 B & 19.0\% & 15.1\% & 71.5\%  & 96.4\% & 5.4\%  \\
 C & 15.1\% & 12.7\% & 97.3\%  & 98.4\% & 5.1\%  \\ 
\bottomrule
\end{tabular}	
\caption{Modeling Performance }
\label{tab:expr-sisl}	
\end{minipage}
\begin{minipage}{1.\textwidth}
\ra{0.7}
\begin{flushright}
\footnotesize
\addtolength{\tabcolsep}{-2.5pt}
\newrobustcmd{\B}{\bfseries}
\begin{tabular}{@{}lcc@{}}
\toprule
 SO scenario        & Stage   Lat (in):    & Cost:       \\
\midrule 
\B IPA (noise-free)    & \B10\%, \B15\%, \B44\%      &\B 3\%, \B 7\%, \B 12\%   \\ %
\B \status{IPA (noisy)}       & \status{\B9\%, \B10\%, \B42\% }      & \status{\B3\%, \B6\%, \B12\%  } \\ %
\B IPA+RAA (noise-free)     & \B 37\%, \B 58\%, \B 72\%      & \B 43\%, \B 56\%, \B 78\% \\ %
\B \status{IPA+RAA (noisy)}          & \status{\B37\%, \B55\%, \B72\% }     &\status{ \B42\%, \B56\%, \B78\% }\\ %
\midrule
Bootstrap Model    & Stage   Lat (in):    & Cost:       \\
\midrule 
GTN+MCI	&	34\%,49\%,68\%	&	48\%,41\%,71\% \\
TLSTM	&	{\B 17\%},{\B 2\%},65\%	&	43\%,{\B 31\%},70\% \\
QPPNet	&	34\%,{\B 1\%},63\%	&	46\%,{\B 30\%},68\% \\
\bottomrule
\end{tabular}
\caption{\rv{Average RR over 2M stages }}
\label{tab:e2e-noise}	
\end{flushright}
\end{minipage}
\end{minipage}
\vspace{-0.3in}
\end{table*}

\vspace{-0.06in}
\subsection{Resource Optimization (RO) Evaluation}
\label{sec:expr-so}

We next evaluate the end-to-end performance of our Stage Optimizer (SO) against the current HBO and Fuxi scheduler~\cite{fuxi-vldb14}, \rv{as well as other MOO methods} using production workloads A-C. 
As we cannot run experiments directly in the production clusters, we developed a  simulator of the extended MaxCompute and replayed the query traces to conduct our experiments (detailed in\techreport{~\ref{appendix:simulator}}{~\cite{tech-report}}).

%

\cut{
As we cannot run experiments directly in the production clusters for workloads A-C, we developed a  simulator of the extended MaxCompute environment for our experiments. To simulate the latency of a stage after our SO determines the PP and RP, we pre-train a Gaussian Process Regression (GPR) model to learn the actual latency distribution against the predicted latency from a set of (predicted, actual) pairs generated from a bootstrap model (by default, MCI+GTN).
During simulation, the GPR model takes a predicted latency and generates a Gaussian distribution ($N(\mu,\sigma)$) as the output. Then our simulator will sample from the distribution within $\mu\pm3\sigma$ to simulate the actual latency.
}

We consider the following metrics in resource optimization (RO): 
\rv{(1) \textit{coverage}, the ratio of stages that receive  feasible solutions within 60s;}
(2) $Lat_{s}^{(in)}$, the average stage latency that includes the RO time;
(3) $Cost_{s}$, the average cloud cost of all stages in a workload;
(4) $T_{s}$, the RO time cost.
Below, we first present  a microbenchmark of our methods and \rv{other MOO methods} using 29 subworkloads in Table~\ref{tab:expr-so-sub-workloads}, and then report our net benefit over the full dataset in Table~\ref{tab:e2e-noise}.

\underline{Expt 6:} {\it IPA Only.} 
We first turn on only the IPA module in the stage optimizer (no  RAA) and compare the two options of IPA to the Fuxi scheduler over the 29 subworkloads.
As shown in Tab.~\ref{tab:expr-so-sub-workloads}, the clustered version \verb|IPA(Cluster)| reduces the stage latency (including the solving time) by \rv{12-50\%} and cost by \rv{4-14\%}, with the average time cost between \rv{10-33} msec. 
Without clustering, \verb|IPA(Org)| achieves comparable reduction rates as \verb|IPA(Cluster)| but costs \rv{2-83x}  more time for solving.
\cut{
Table~\ref{tab:expr-so-sub-workloads} shows the average reduction rate (RR) IPA achieved in stage latency and cost, where \verb|IPA(Cluster)| reduces the stage latency (including the solving time) by \rv{12-50\%} and cost by \rv{4-14\%}, with the average time cost between \rv{10-33} msec. 
Without clustering, \verb|IPA(Org)| costs \rv{2-83x} more RO time compared to \verb|IPA(Cluster)|.
}


\cut{
We then dive into 29 sub-workloads, subsampled in different system busy and idle periods. 
Fig.~\ref{fig:ipa-1} shows that IPA reduces more stage latencies and \rv{cloud cost} for workload C (with more long-running instances,) than workload B (median-length stages) and workload A (short-running stages).
Fig.~\ref{fig:ipa-2} shows that  IPA helps reduce more latency and cost when a cluster has an idle environment because it can assign more machines of high capacity to a stage and hence reduce the stage latency more.
}

\underline{Expt 7:} {\it  IPA+RAA}. 
We now run \verb|IPA(Cluster)| with RAA of four choices.
\verb|IPA+RAA(Path)| solves the resource plan by applying RAA Path over instance and machine clusters and achieves the best performance. It reduces the stage latency by \rv{36-80\%} and \rv{cloud cost} by \rv{29-75\%}, with an average time cost between \rv{17-156ms} (including IPA).
\verb|IPA+RAA(W/O_C)| does RAA without clustering and suffers high overhead (up to \rv{19s} for a stage). 
\verb|IPA+RAA(DBSCAN)| applies DBSCAN for the instance clustering, which incurs 
up to \rv{937} msec for a stage, hence inefficient for production use.
\verb|IPA+RAA(General)| shows the performance of our general hierarchical MOO approach, which is slightly worse than \verb|IPA+RAA(Path)| (in this 2D MOO problem) in running time while offering a similar reduction of latency and cost. 
Details of the four choices are in~\techreport{Appendix~\ref{appendix:raa} \&~\ref{appendix:expt-ipa+raa}}{~\cite{tech-report}}.

\underline{Expt 8:} {\it MOO baselines.} 
We next compare to SOTA MOO solutions, \texttt{EVO}~\cite{Emmerich:2018:TMO}, \texttt{WS(Sample)}~\cite{marler2004survey}, and \texttt{PF(MOGD)}~\cite{spark-moo-icde21}, using Def.~\ref{def:stagemoo}. See \techreport{Appendix~\ref{appendix:formula}}{~\cite{tech-report}} for their implementation details.
As shown in the 3 red rows of Table~\ref{tab:expr-so-sub-workloads}, 
(1) over 29 sub-workloads, none of them guarantees to return all results within 60s;  
(2) their latency and cost reduction rates  are all dominated by \texttt{IPA+RAA(Path)}, and even lose to the Fuxi scheduler on some workloads; 
(3) their solving time is 1-2 magnitude higher than our approach, making them infeasible to be used by a cloud scheduler.
As an alternative, we apply IPA to solve the $B$ variables and these MOO methods to  solve only $\Theta$  based on Eq.~\eqref{eq:stage-mult_obj_opt_def}. 
As shown in the last 3 blue rows in Table~\ref{tab:expr-so-sub-workloads}, they are still inferior to \texttt{IPA+RAA(Path)} in both latency and cost reduction and in running time (mostly taking 1-6 sec to complete).

\cut{
{\ }\\
{\bf Expt 11: end-to-end comparison with noisy predictions}\\ 
\noindent

Now we run the entire 2M stages over 3 departments in both noise and noise-free cases. 
To simulate the effects from the model inaccuracy, the execution simulator pre-trains a Gaussian Process Regression (GPR) model to learn the actual latency distribution against the predicted latency from a trained model (named bootstrap model), which by default, uses MCI+GTN). 
For each predicted instance latency, GPR generates a gaussian distribution ($N(\mu,\sigma)$) of the actual latency, and sample a simulated actual from the distribution within $\mu\pm3\sigma$.

Table~\ref{tab:e2e-noise} shows that in both noise-free and noisy settings, SO significantly reduces stage latency and cost compared to Fuxi with the MCI+GTN model for latency prediction.
}

\underline{Expt 9:} {\it Net Benefits.} 
We next compare our SO (IPA+RAA) against Fuxi's scheduling results by considering the model effects: the 
{\em noise-free} case means that the predicted latency is the true latency, while the {\em noisy} case captures the fact that the true latency is different from the predicted one.  We run the entire 2M stages over 3 departments in both noisy and noise-free cases. For the noisy case, we turn on the actual latency simulator (GPR model) to simulate the actual latency. Specifically, given a predicted instance latency, GPR generates a Gaussian distribution ($N(\mu,\sigma)$) of the actual latency, and samples from the distribution within $\mu\pm3\sigma$.
Table~\ref{tab:e2e-noise} shows that in both noise-free and noisy settings, SO (IPA+RAA) significantly reduces stage latency and \rv{cloud cost} compared to Fuxi, with the fine-grained MCI+GTN model for latency prediction. 

\underline{Expt 10:} {\it Impact of Model Accuracy.}
Finally, to quantify the impact of model accuracy on resource optimization, we use GPR models pre-trained over three bootstrap models (MCI+GTN, TLSTM, QPPNet).  Note that in terms of model accuracy, we have  (MCI+GTN > TLSTM > QPPNet), as shown in Fig.~\ref{fig:sota-single-mach}. 
We borrow Fuxi's ground-truth placement plan and ask RAA for resource plans of each stage under the same condition of the system states as Fuxi. To be fair in the comparison, we dropped the scheduling delays for all stages and set the stage latency as the maximum instance latency.
Table~\ref{tab:e2e-noise} compares the RAA's reduction rate (RR) among different bootstrap models. These results show that indeed, better reduction rates are achieved by using a more accurate model, which validates the importance of having a fine-grained accurate prediction model.


\section{Conclusions}
\label{sec:conclusion}

We presented a MaxCompute~\cite{maxcompute} based big data system that supports multi-objective resource optimization via fine-grained instance-level modeling and optimization. 
To suit the complexity of our system, we developed fine-grained instance-level models that encode all relevant information as multi-channel inputs to deep neural networks. By exploiting these models, our stage optimizer employs a new IPA module to derive a latency-aware placement plan to reduce the stage latency, and a novel RAA model to derive instance-specific resource plans to further reduce stage latency and cost in a hierarchical MOO framework.
Evaluation using production workloads shows that   
(1) our best model achieved 7-15\% median error and 9-19\% weighted mean absolute percentage error; 
(2) compared to the Fuxi scheduler~\cite{fuxi-vldb14}, IPA+RAA achieved the reduction of \rv{37-72\%}
latency and \rv{43-78\%} cost while running in \rv{0.02-0.23s}.


\begin{acks}
This work was partially supported by the European Research Council (ERC) Horizon 2020 research and innovation programme (grant n725561), 
Alibaba Group through Alibaba Innovative Research Program,
and China Scholarship Council (CSC). 
We also thank 
Yongfeng Chai,
Daoyuan Chen, 
Xiaozong Cui,
Botong Huang,
Xiaofeng Zhang, 
and Yang Zhang
from the Alibaba group for the discussion and help throughout the project.
\end{acks}

\clearpage

\balance
\bibliographystyle{ACM-Reference-Format}
\bibliography{refs/model+opt,refs/bigdata,refs/optimization,refs/db}

\techreport{
\clearpage
\appendix
\section{Problem Formulation}
\label{appendix:formula}

\begin{table*}[t]
\ra{1.2}
\small
\centering
	\begin{tabular}{cl}\toprule
	Symbol & Description \\\midrule
$d$ & the number of resource types in a resource configuration (RC). $d=2$ in our case.
\\
$m$ & the number of instances.
\\
$n$ & the number of machines.
\\
$x_i$ & the $i$-th instance
\\
$\tilde{x}_i$ & the features on channel 1 and channel 2 and AIM towards $x_i$
\\
$y_j$ & the $j$-th machine
\\
$\tilde{y}_j$ & the features on channel 4 and channel 5 towards $y_j$
\\
$B$ & the placement matrix $B \in \{0, 1\}^{m\times n}$
\\
$B_{i,j}$ & the $i$-th row, $j$-th column of $B$ ($B_{i,j} \in \{0, 1\}$). $B_{i,j} = 1$ when $x_i$ is sent to $y_j$.
\\
$\Theta$ & the second set of parameters, a.k.a. the resource plan matrix $R \subseteq \mathbb{R} ^{m\times 2}$
\\
$\Theta_i$ & the resource plan of the $i$-th instance $x_i$. $\Theta_i \in \mathbb{R}^2$. In our case, $\Theta_i = (req\_cores, req\_mem)$
\\
$f$ & the instance-level latency model, mapping from channels 1-5 and AIM to latency estimation.
\\\midrule
{\bf clustering-related}	&
\\
$m'$ & the number of instance clusters.
\\
$n'$ & the number of machine clusters.
\\
$X'_i$ & the $i$-th instance cluster
\\
$x'_i$ & the representative instance of $i$-th instance cluster
\\
$\vert X'_i \vert$ & the number of instances in $i$-th instance cluster
\\
$Y'_j$ & the $j$-th machine cluster
\\
$y'_j$ & the representative machine of $j$-th machine cluster
\\
$\vert Y'_j \vert$ & the number of machines in $j$-th machine cluster
\\
$B'$ & the placement matrix to indicate the number of instances assigned from each instance cluster to each machine cluster. $B' \in \mathbb{N}^{m'\times n'}$
\\
${B'}_{i,j}$ & the $i$-th row, $j$-th column of $B'$. ${B'}_i^j$ indicates the number of instances in $X'_i$ when they are sent to $Y'_j$.
\\
$\Theta'$ & the resource plans of instances in all instance clusters
\\
${\Theta'}_{i,j}$ & the resource plan of the representative instances in $i$-th instance cluster $x'_i$ on the $j$-th machine cluster. ${\Theta'}_{i,j} \in \mathbb{R}^d$
\\\bottomrule
	\end{tabular}
\caption{Notation for the MOO problem}
\label{tab:formula-notations}
\end{table*}

Consider $m$ instances ($x_1, ..., x_m$) and $n$ machines ($y_1$, ..., $y_n$).
Denote $\tilde{x}_i$ as the features of Ch1 and Ch2 of instance $x_i$, and $\tilde{y}_j$ as the features of Ch4 and Ch5 of machine $y_j$. 
Given two set of variables $B$  and $\Theta$, assuming that $f$ is the instance-level latency prediction model, we have $L_{i,j} = f(\tilde{x}_i, \Theta_i, \tilde{y}_j)$ as the latency when $x_i$ is running on $y_j$ using the resource configuration $\Theta_i$.

Denote two objectives as the stage-level latency 
$$
L_{stage} = L(B, \Theta) = \max_{i,j} B_{i,j} L_{i,j} = \max_{i,j} B_{i,j} f(\tilde{x}_i, \Theta_i, \tilde{y}_j)
$$
and the stage-level total cost (the weighted sum of cpu-hour and memory-hour) 
$$
C_{stage} = C(B, \Theta)= \sum_{i,j} B_{i,j} L_{i,j} (\mathbf{w} \cdot  \Theta_i^T) = \sum_{i,j} B_{i,j} f(\tilde{x}_i, \Theta_i, \tilde{y}_j) (\mathbf{w} \cdot  \Theta_i^T)
$$
where $w_1$ and $w_2$ are two constants.
$\mathbf{w}$ is a weight vector.
Similarly, other objectives such as CPU-hour, memory-hour can be defined.

Consider $m$ instances ($x_1, ..., x_m$) and $n$ machines ($y_1$, ..., $y_n$).
Given two set of variables $B$  and $\Theta$,  $f$ is instance-level latency when $x_i$ is running on $y_j$ using the resource configuration $\Theta_i$.
We have a constrained Multi-Objective Optimization (MOO) problem.
\begin{align}
	\label{eq:stage-moo-def}
	\arg\min_{B, \Theta} & {\left[
			\begin{array}{l}
			L(B, \Theta) = \max_{i,j} B_{i,j}  f(\tilde{x}_i, \Theta_i, \tilde{y}_j) \\
			C(B, \Theta) = \sum_{i,j} B_{i,j}  f(\tilde{x}_i, \Theta_i, \tilde{y}_j)  (\mathbf{w} \cdot  \Theta_i^T ) \\
			...
			\end{array}
			\right]} \\
		\nonumber s.t. &  {\begin{array}{l}
			B_{i,j} \in \{0, 1\}, \,\,\, \forall i=1...m , \forall j=1...n\\
			\sum_j B_{i,j} = 1, \,\,\, \forall i=1...m \\
			\sum_i B_{i,j} \Theta_i^1 \leq U_j^1 , \,\, \ldots, \,\,\sum_i B_{i,j} \Theta_i^d \leq U_j^d \,\, , \forall j=1...n \\
			\sum_i B_{i,j} \leq \alpha \,\,\, , \forall j=1...n
	\end{array}}
\end{align}
where $U_j \in \mathbb{R}^d$ is the capacity of $d$ resources on machine $y_j$ and $\alpha$ defines the maximum number of instances each machine can take based on the diverse placement preference. 
Basically, a smaller $\alpha$ shows a stronger preference of the diverse placement for a stage and we have $\alpha >= \lceil m/n \rceil$. 

To simplify the problem, we take two steps:

\minip{Step 1:} Take the resource configuration $\Theta_0$ returned from the HBO optimizer as the default and assign it uniformly to all instances, $\Theta_i = \Theta_0$, $\forall i \in [1,...,m]$.  Then minimize over $B$ in Eq.~\eqref{eq:stage-moo-def} by treating $\Theta_0$ as a constant. Then the objective functions are shown in Eq.~\eqref{eq:stage-ipa}. 
		
\begin{align}
		\label{eq:stage-ipa}
		\arg \min_{B} & {\left[
			\begin{array}{l}
			L(B, \Theta_0) = \max_{i,j} B_{i,j} f(\tilde{x}_i, \Theta_0, \tilde{y}_j) \\
			C(B, \Theta_0) = \sum_{i,j} B_{i,j} f(\tilde{x}_i, \Theta_0, \tilde{y}_j) (\mathbf{w} \cdot  \Theta_0^T ) \\
			...
			\end{array}
			\right]} 
\end{align}

\minip{Step 2:} Given the solution from step 1, $B^*$, minimize over $\Theta$ by treating $B^*$ as a constant, as shown in Eq.~\eqref{eq:stage-raa}.

\begin{align}
		\label{eq:stage-raa}
		\arg \min_{\Theta} & {\left[
			\begin{array}{l}
			L(B^*, \Theta) = \max_{i,j} B_{i,j}^* f(\tilde{x}_i, \Theta_i, \tilde{y}_j) \\
			C(B^*, \Theta) = \sum_{i,j} B_{i,j}^* f(\tilde{x}_i, \Theta_i, \tilde{y}_j) (\mathbf{w} \cdot  \Theta_i^T ) \\
			...
			\end{array}
			\right]} 
		\end{align}

The alternative formula forms are also shown in Eq.~\eqref{eq:ipa} for the step 1 and Eq.~\eqref{eq:stage-mult_obj_opt_def} for the step 2.
The intuition behind our approach is that if we start with a decent choice of $\Theta_0$, as returned by HBO, we hope that step 1 will reduce stage latency via a good assignment of instances to machines by considering machine capacities and instance latencies. Then step 2 will fine-tune the resources assigned to each instance on a specific machine to reduce stage latency, cost, as well as other objectives.

\subsection{Formulation of Different MOO Problems}


%

We consider solving the MOO problem in two scenarios:
\begin{enumerate}
	\item {\bf Plan A: Optimization over \bm{$B$} and \bm{$\Theta$}.} Take the entire parameter space with $B$ and $\Theta$ as variables.
	\item {\bf Plan B: Optimization over \bm{$\Theta$}.} Use the solution from step 1, and take Step 2 with $\Theta$ as the variable.
\end{enumerate}

\subsubsection{Plan A: Optimization over \bm{$B$} and \bm{$\Theta$}}~

We consider two methods implemented for the plan A. 

The first method solves the {\bf vanilla plan A} by 
strictly applying the MOO definition in Eq.~\eqref{eq:stage-moo-def}. However it suffers from the huge amount fo variables and constraints, where we get 
$m*(n+d)$
variables, and 
$m+(d+1)*n$
constraints. 
Notice both $m$ and $n$ can be several thousands, the problem could have millions of parameter and thousands of constraints.

To reduce the dimensionality, we consider the {\bf clustered plan A} that solves the variables after clustering instances and machines as our IPA did.
After clustering, assuming there are $m'$ instance clusters $({X'_1}, ..., {X'_{m'}})$ and $n'$ machine clusters $({Y'_{1}}, ..., {Y'_{n'}})$, where ${x'_{i}}$ is the representative instance of instance cluster $X'_i$, ${y'_{j}}$ is the representative machine of machine cluster $Y'_j$. Denote $\tilde{x}'_i$ as the features of Ch1 and Ch2 of ${x'_{i}}$, and $\tilde{y'}_j$ as the features of Ch4 and Ch5 of $y'_j$. 
Given two sets of variables $B'$ and $\Theta'$, assuming that $f$ is the instance-level latency prediction model, we have $f(\tilde{x'}_i, {{\Theta'}_{i,j}}, \tilde{y'}_j)$ as the latency when $x'_i$ is running on $y'_j$ using the resource configuration ${{\Theta'}_{i,j}}$. We assume that instances in the same instance cluster behave similarly when they are scheduled to the machines from the same cluster.
Therefore, the constrained MOO problem is defined in Eq.~\eqref{eq:stage-moo-ori-cluster}.
	\begin{align}
		\label{eq:stage-moo-ori-cluster}
		\arg\min_{B', {\Theta'}} & {\left[
			\begin{array}{l}
            L'(B', \Theta') = \max_{i,j} {\mathcal{I}}({B'}_{i,j}) f(\tilde{x'}_i, {{\Theta'}_{i,j}}, \tilde{y'}_j) \\
			C'(B', \Theta') = \sum_{i,j} {B'}_{i,j} f(\tilde{x'}_i, {\Theta'}_{i,j}, \tilde{y'}_j)(\mathbf{w} \cdot  {\Theta'}_{i,j}^T) \\
			...
			\end{array}
			\right]} \\
		\nonumber s.t. &  {\begin{array}{l}
			\sum_j {B'}_{i,j} = \vert {X'}_{i} \vert, \,\,\, \forall i=1...m' \\
            \sum_i {B'}_{i,j} {{\Theta'}_{i,j}}^1 \leq \sum_r U_r^1 \,\,\, \forall j=1...n', r=1...{\vert {Y'}_{j} \vert} \\
			\sum_i {B'}_{i,j} {{\Theta'}_{i,j}}^2 \leq \sum_r U_r^2 \,\,\, \forall j=1...n', r=1...{\vert {Y'}_{j} \vert} \\
			...\\
			\sum_i {B'}_{i,j} {{\Theta'}_{i,j}}^d \leq \sum_r U_r^d \,\,\, \forall j=1...n', r=1...{\vert {Y'}_{j} \vert} \\
			\sum_i {B'}_{i,j} / {\vert {Y'}_{j} \vert}\leq {\alpha} \,\,\, \forall j=1...n'
			\end{array}}
	\end{align}
where ${\mathcal{I}}({B'}_{i,j})$ is an indicator function of ${B'}_{i,j}$. 
$\vert {X'}_{i} \vert$ denotes the number of instances in the instance cluster $X'_i$; 
$\vert {Y'}_{j} \vert$ denotes the number of machines in the machine cluster $Y'_j$; $\sum_r U_r \in \mathbb{R}^d$ is the total capacity of the $d$ resources on the machine cluster $Y'_j$; ${\alpha}$ defines the maximum number of instances each machine cluster can take based on the diverse placement preference and we have ${\alpha} >= \lceil m/n \rceil$.	

In clustered Plan A, the number of variables and the number of constraints are reduced to 
$m' * ((d + 1) * n')$ and $m' + ((d + 1) * n')$
, respectively. The constraints guarantee the requested resources of instances allocated in the same machine cluster do not exceed its capacity.

\subsubsection{Plan B: Optimization over $\Theta$}~

\cut{
Consider $m$ instances ($x_1, ..., x_m$) and $n$ machines ($y_1$, ..., $y_n$).
Denote $\tilde{x}_i$ as the features of Ch1 and Ch2 of instance $x_i$, and $\tilde{y}_j$ as the features of Ch4 and Ch5 of machine $y_j$. 
In Step $2$, $B^*$ is given, the optimization is over variables $\Theta$. Assuming that $f$ is the instance-level latency prediction model, we have $l_i^j = f(\tilde{x}_i, \Theta_i, \tilde{y}_j)$ as latency when $x_i$ is running on $y_j$ using the resource configuration $\Theta_i$.
The constrained MOO problem is defined as:
}

To reduce the parameter space in the original problem, we resolve the $B$ variable by IPA in the first step, and other MOO solutions to the step 2.
Given the solution from step 1, $B^*$, the MOO problem in Eq.~\eqref{eq:stage-moo-def} changes to the following Eq.~\eqref{eq:stage-moo-ori-b}.

\begin{align}
		\label{eq:stage-moo-ori-b}
		\arg \min_{\Theta} & {\left[
			\begin{array}{l}
			L(B^*, \Theta) = \max_{i,j} {B^*}_{i,j} f(\tilde{x}_i, \Theta_i, \tilde{y}_j) \\
    		C(B^*, \Theta) = \sum_{i,j} {B^*}_{i,j} f(\tilde{x}_i, \Theta_i, \tilde{y}_j)(\mathbf{w} \cdot  \Theta_i^T) \\
    		...
			\end{array}
			\right]} \\
		\nonumber s.t. &  {\begin{array}{l}
\sum_i B^*_{i,j} \Theta_i^1 \leq U_j^1 , \,\, \ldots, \,\,\sum_i B^*_{i,j} \Theta_i^d \leq U_j^d \,\, , \forall j=1...n \\
			\end{array}}
\end{align}
		
		

We also consider two methods implemented for the plan B. 
The first method {\bf vanilla plan B} solves Eq.~\eqref{eq:stage-moo-ori-b} directly, where the number of variables is $d * m$, and the number of constraints is $d * n$. Notice $m$ and $n$ can both be tens of thousands, the dimensionality of the problem is very high.

To further reduce the number of variables and constraints, we consider {\bf the clustered Plan B} that solves the problem after clustering.
After clustering, consider $m'$ instance clusters $({X'_1}, ..., {X'_{m'}})$ and $n$ machines ($y_1$, ..., $y_n$), where ${{x}'_{i}}$ is the representative instance of instance cluster ${X'_i}$. ${y_{j}}$ is the allocated machine of ${x}'_{i}$ after $B^*$ is given.
Denote $\tilde{x}'_i$ as the features of Ch1 and Ch2 of the representative instance ${{x}'_{i}}$ in instance cluster $X'_i$, and $\tilde{{y}_j}$ as the features of Ch4 and Ch5 of the allocated machine ${y_j}$. 

To align to the clustering ${B'}$ matrix in Plan A, we utilize ${B'}^*$ to denote the placement of instances in clustered Plan B. But the dimensions of ${B'}$ matrix in Plan A and ${B'}^*$ in Plan B are different. The differences occur in two aspects. Firstly, the instance clustering of Plan B and Plan A are different. In Plan A, as the elements in ${B'}$ matrix are variables, the instance clustering can only consider the instance's feature (e.g., cardinality). In Plan B, the placement plan ${B'}^*$ matrix is given, so the instance clustering can consider the feature of the instance and the machine state. It divided instance clusters of Plan A into more sub-clusters, which leads to different number of instance clusters (the number of rows in ${B'}$ and ${B'}^*$). Furthermore, ${B'}^*$ is given in Plan B, and ${{B'}^*}_{i,j}$ is a binary value, which indicates whether the representative instance of instance cluster $X'_i$ is scheduled on the specific machine $y_j$ rather than a machine cluster. Therefore, the number of columns of ${B'}$ matrix in Plan A and ${B'}^*$ in Plan B are different.

Similarly, the variables in Plan B are $\Theta'$, assuming that $f$ is the instance-level latency prediction model, we have ${L'}_{i,j} = f(\tilde{x'}_i, {{\Theta'}_{i,j}}, \tilde{y}_j)$ as the latency when $x'_i$ is running on ${y_j}$ using the resource configuration ${{\Theta'}_{i,j}}$. We assume groups of instances may behave similarly.
After clustering, the constrained MOO problem is defined as in Eq.~\eqref{eq:stage-moo-cluster-b}:

\cut{
\begin{eqnarray}
		\label{eq:stage-moo-ori}
		\arg \min_{\Theta'} & {\left[
			\begin{array}{l}
			L({{B'}^*}, \Theta') = \max_{i,j} {{B'}^*}_{i,j} f(\tilde{x'}_i, {{\Theta'}_{i,j}}, \tilde{y}_j) \\
			C({{B'}^*}, \Theta') = \sum_{i,j} {{B'}^*}_{i,j} f(\tilde{x'}_i, {{\Theta'}_{i,j}}, \tilde{y}_j)(\mathbf{w} \cdot  {\Theta'}_i^T) \vert {X'}_{i} \vert
			\end{array}
			\right]} \\
		\nonumber s.t. &  {\begin{array}{l}
			\sum_i {{\Theta'}_{i,j}}^1 {\gamma'}_i^j \leq U_j^1, \,\,\, \forall i=1...m', j=1...g \\
			\sum_i {{\Theta'}_{i,j}}^2 {\gamma'}_i^j \leq U_j^2, \,\,\, \forall i=1...m', j=1...g
			\end{array}}
		\end{eqnarray}
		
where $\vert {X'}_{i} \vert$ denotes the number of instances in the instance cluster $X'_i$; 
$g$ denotes the number of distinct machines for instances in all instance clusters; ${\gamma'}_i^j$ denotes the number of instances in cluster $X'_i$ allocated on the same machine $y_j$. Denote the indices of instances in cluster $X'_i$ as $I$, ${\gamma'}_i^j = \sum_i {{B'}^*}_{i,j}, i \in I$.
}

\begin{align}
		\label{eq:stage-moo-cluster-b}
		\arg \min_{\Theta'} & {\left[
			\begin{array}{l}
			L({{B'}^*}, \Theta') = \max_{i,j} {{B'}^*}_{i,j} f(\tilde{x'}_i, {{\Theta'}_{i,j}}, \tilde{y}_j) \\
			C({{B'}^*}, \Theta') = \sum_{i,j} {{B'}^*}_{i,j} f(\tilde{x'}_i, {{\Theta'}_{i,j}}, \tilde{y}_j)(\mathbf{w} \cdot  {\Theta'}_i^T) \vert {X'}_{i} \vert
			\end{array}
			\right]} \\
		\nonumber s.t. &  {\begin{array}{l}
			\sum_i {{\Theta'}_{i,j}}^1 {\gamma'}_i^j \leq U_j^1, ...,
			\sum_i {{\Theta'}_{i,j}}^d {\gamma'}_i^j \leq U_j^d, \,\,\, \forall j=1...n
			\end{array}}
		\end{align}
		
where $\vert {X'}_{i} \vert$ denotes the number of instances in the instance cluster $X'_i$; 
${\gamma'}_i^j$ denotes the number of instances in cluster $X'_i$ allocated on the same machine $y_j$. Denote the indices of instances in cluster $X'_i$ as $I$, ${\gamma'}_i^j = \sum_i {{B'}^*}_{i,j}, i \in I$.
In clustered Plan B, the number of variables and the number of constraints can be reduced to $d * m'$ and 
$d * n$.

\subsection{Implementation of Existing MOO Methods}

In our experiment, we consider $d=2$ types of resources as CPU and memory. Thus, we implemented the existing MOO methods for $d=2$. However, our implementation can be easily extended to the cases when $d > 2$.

\underline{Method 1: Evolutionary algorithm.}
We apply NSGA-II \cite{Deb:2002:FEM} in the family of Evolutionary (Evo) algorithms \cite{Emmerich:2018:TMO} for solving MOO problem of both clustered Plan A and Plan B. The objective functions and constraints are referred to the problem formulation mentioned above. We did hyperparameter tuning to choose a set of optimal population size and the number of iterations for both Plan A (\texttt{EVO}) and Plan B (\texttt{IPA+EVO}). The NSGA-II is implemented by the library Platypus \cite{Platypus}.
Here are the technical details:
\begin{enumerate}
	\item {\bf Complex functions and constraints:} When calculating the stage-level function values, we need to call the predictive model, where the variables are part of the input to get predictions. Moreover, the constraints in Plan A are quadratic, which is more complex than a linear representation. 
	\item {\bf Evaluation time:} As we discussed before, the number of variables and constraints are related to the number of instances. The time cost of different stages could be various if they have a very different number of instances. Based on our observation, a stage could take over $10$ minutes to return solutions, which is not acceptable under the requirement of \textit{efficiency} and we have 0.1M-1M stages of running. Therefore, we limit the time cost to $60$ seconds for solving a stage MOO problem.
	\item {\bf Hyperparameter tuning:} Population size and the number of iterations are two hyperparameters influencing optimization performance. To select one set of optimal hyperparameters, under the time limit of $60$ seconds, we try $12$ sets of hyperparameters for the workload with $103$ stages which include the number of instances varies from $1$ to over $6$k. Finally, the hyperparameter is selected based on the higher coverage of feasible stages and a higher reduction rate on the stage-level latency and cost.
	\item {\bf Randomness:} NSGA-II is an evolutionary algorithm where the population is initialized randomly, and genetic operations are executed under a probability (e.g., crossover and mutation). To make the results reproducible, we fix the randomness of each iteration related to population generation and genetic operations.
\end{enumerate}

\vspace{0.05in}
\underline{Method 2: Weighted Sum.}
We applied \texttt{WS(Sample)}~\cite{marler2004survey} to reduce a MOO problem into a single-objective optimization (SOO) problem by minimizing the weighted sum of the stage latency and cost.
We applied the random sampling approach to solving the SOO problem in our implementation.
For each stage, we get the corresponded MOO solutions over a set of weights over latency and cost by the following steps:
\begin{enumerate}
	\item Sample up to 100K variable choices randomly.
	\item Filter the variables that violate the constraints of resource capacity or the diverse preference)
	\item Evaluate the objective values over each filtered variable.
	\item Iterate the objective weights over latency and cost and return the variables that minimize the weighted sum of the objectives respectively.
	\item Form the MOO solutions based on the returned variables.
\end{enumerate}

\texttt{WS(Sample)} treats the resource plan $\Theta$ and the placement plan $B$ as the variables in plan A.
In plan B, \texttt{IPA+WS(Sample)} solves only $\Theta$ with the placement plan returned by our \texttt{IPA(Cluster)}.

\vspace{0.05in}
\underline{Method 3: Progressive Frontier (PF).}
We implemented our previous work \texttt{PF(MOGD)}~\cite{spark-moo-icde21} in the parallel version (PF-AP) to solve the MOO problem progressively.
\texttt{PF(MOGD)} transforms the MOO problem into a set of single-objective constraint optimization (CO) problems and solves them via the Multi-Objective Gradient Descent (MOGD) solver. 
When doing gradient descent for solving the variables, the MOGD solver treats $\Theta$ and $B$ as float matrices and rounds their values to integers in their corresponding domain spaces after every backward step.

\cut{
\minip{Relax Constraints}
Notice that the 5-6th constraints have quadratic forms of variables, and hence the problem is difficult even for a solver to find a feasible solution.
Therefore, we target a MOO problem in a tighter space with all the original quadratic constraints replaced by linear ones to make the problem more feasible. 
\begin{eqnarray}
		\label{eq:stage-moo-solver}
		\min_{B, \Theta} & {\left[
			\begin{array}{l}
			L(B, \Theta) = \max_{i,j} B_i^j f(\tilde{x}_i, \Theta_i, \tilde{y}_j) \\
			C(B, \Theta) = \sum_{i,j} B_i^j f(\tilde{x}_i, \Theta_i, \tilde{y}_j) (w_1 \Theta_i^1 + w_2 \Theta_i^2 )
			\end{array}
			\right]} \\
		\nonumber s.t. &  {\begin{array}{l}
			B_i^j \in \{0, 1\}, \,\,\, \forall i=1...m , \forall j=1...n \label{eq:con1}\\
			\sum_j B_i^j = 1, \,\,\, \forall i=1...m \\
			\Theta_i^1 \in [0.5, 1], \,\,\, \forall i=1...m \\ 
			\Theta_i^2 \in [1G, 12G], \,\,\, \forall i=1...m \\
			\sum_i B_i^j \leq s_j \,\,\, \forall j=1...n
			\end{array}}
		\end{eqnarray}
where $s_j = \min (\lfloor U_j^1 / 1 \rfloor, \lfloor U_j^2 / 12G \rfloor, \alpha)$ refers to the minimum number of instances $y_j$ can take consider the resource capacity and the diverse preference. 
}

\section{Workload Description}
\label{appendix:workload-desc}

\begin{figure}[t]
	\centering
  \includegraphics[height=7cm,width=.7\linewidth]{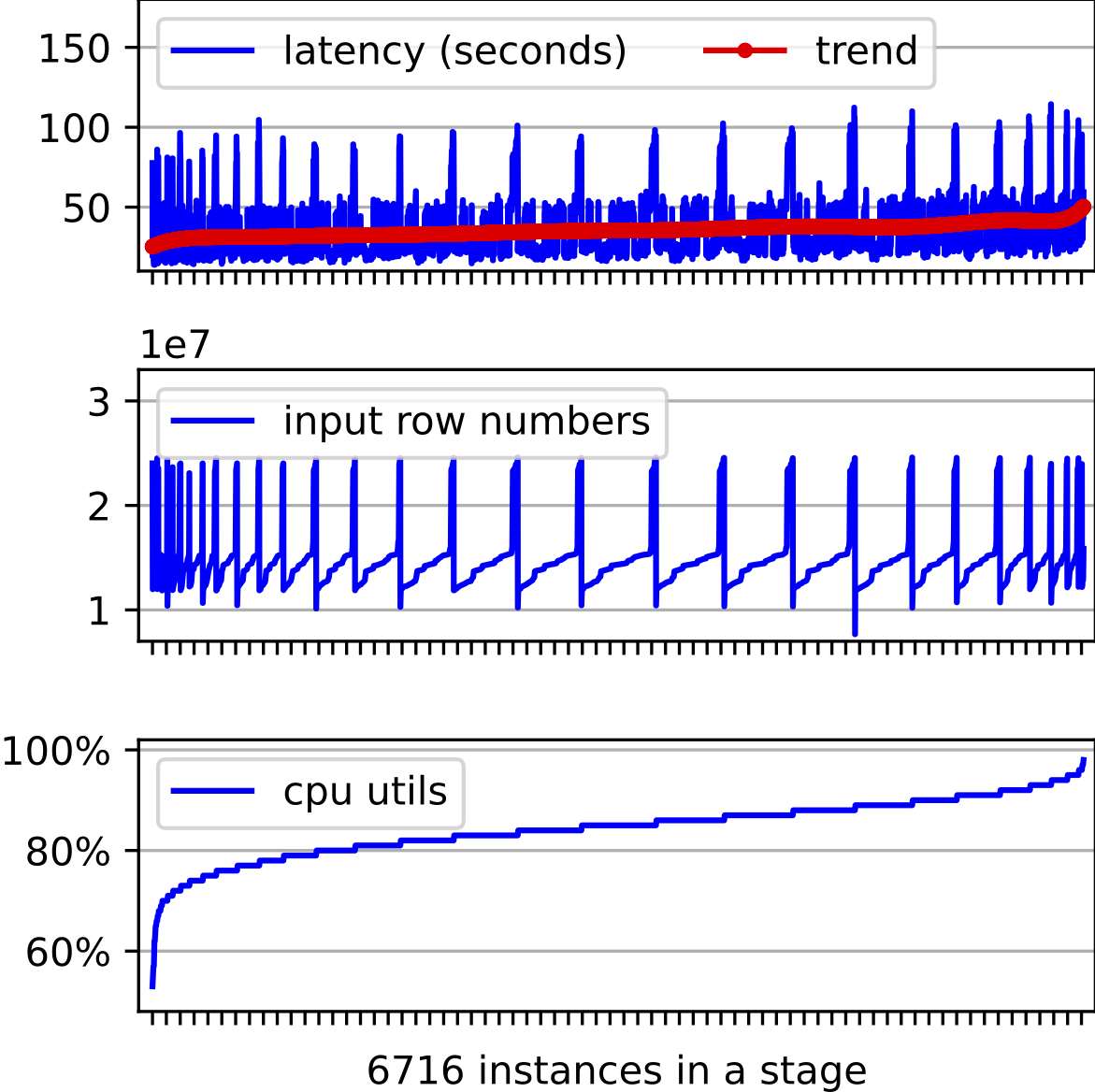}
  \captionof{figure}{Latency variation of 6716 instances in a stage}
   \vspace{-0.1in}
  \label{fig:inst-profiling}
\end{figure}

\underline{Example 1.}
Figure~\ref{fig:cdf-job-stage-num-low} and~\ref{fig:cdf-stage-inst-num-low} show that in a production trace of 0.62 million jobs, there are 1.9 million stages in total, with up to 64 stages in each job, and 121 million instances, with up to 81430 instances in a stage. For a particular stage with 6716 instances, Figure~\ref{fig:stage-lats} shows that the latencies of different instances vary a lot. 


Figure~\ref{fig:inst-profiling} shows the
instances sorted by (a) CPU utilization and (b) by input row number. 
Even though the instances share the same (sub)query structure and the same resource plan, the instances exhibit great variation of latencies, which could be caused by different instance or machine characteristics.
(1) For the instances running on machines with similar system states (e.g., same CPU utilization), their latency is positively correlated to the input row numbers. 
(2) For those instances that run on machines with different CPU utilization, we observe that the latency has an increasing trend 
as the CPU utilization rises up, which is due to the lack of perfect isolation in big data systems.

If a performance model captures only the overall stage latency, i.e., the maximum instance latency, when the resource optimizer is asked to reduce latency, it will assign more resources uniformly to all the instances (as they are not distinguishable by the model). For those short-running instances, the extra resources do not contribute to the stage latency while incurring a higher cost. Instead, an optimal solution would be to assign more resources only to the long-running instances while maintaining or reducing resources for the short-running ones. Such decisions require fine-grained instance-level models as well as instance-specific resource plans. 

\section{More Details on Modeling}
\label{appendix:model}
In this section, we present how we build fine-grained instance-level models to support minimizing both latency and cost of each stage (the granularity of scheduling)  in big data systems. 

Most notably, choosing the latency of a single instance as our fine-modeling target distinguishes our work from other existing work that captures end-to-end query latency~\cite{spark-moo-icde21,udao-vldb19} or operator latency including multiple instances on different machines~\cite{cleo-sigmod20}. These coarse-grained models not only are subject to highly-variable predictions but also miss opportunities for instance-level recommendations -- our fine-grained models aim to solve both problems.   

In addition, existing works considered only part of the system issues in resource management.
For example, the scheduling algorithms~\cite{fuxi-vldb14,Yarn,trident-vldb21} treat scheduling tasks as blackboxes without considering the complexity of the workloads. 
Recent works~\cite{qrop-icde18,cleo-sigmod20,plan-aware-resource-opt-hotcloud20} optimize the resource plan by leveraging the query characteristics, but without taking into account the cluster complexity. 
In contrast, our work aims to capture all relevant systems issues in modeling so that the resource optimizer can leverage the model to make effective recommendations.

\minip{Modeling Tools in Our Framework.}
Besides our GTN based model, we next show that our modeling framework 
can accommodate other modeling tools designed for DBMSs, with necessary extensions. 
As thus, we can leverage different modeling tools in our framework and examine their pros and cons in experimental evaluation.

QPPNet~\cite{qppnet-vldb19} and TLSTM~\cite{tlstm-cost-estimator-vldb19} are two state-of-the-art methods that learn the query performance in a DBMS.
QPPNet learns operator-level neural networks named neural units for each query operator in a plan-structured model. The output of each neural unit includes a data vector and a latency channel. The input of each neural unit consists of the operator-oriented features and data vectors generated by its child nodes (if any). QPPNet returns the query latency prediction from the latency channel in the output of the root neural unit.
TLSTM encodes query operators into a uniform feature space and feeds the operator encodings to a Tree-Structured Long Short-Term Memory Network (TreeLSTM)~\cite{tree-lstm-acl15} to learn the query latency in a query plan tree. Figure~\ref{fig:tlstm} shows how TLSTM does stage embedding over a query job by taking use of its tree structure.

\begin{figure}[t]
	\centering
  \includegraphics[height=3cm,width=.7\linewidth]{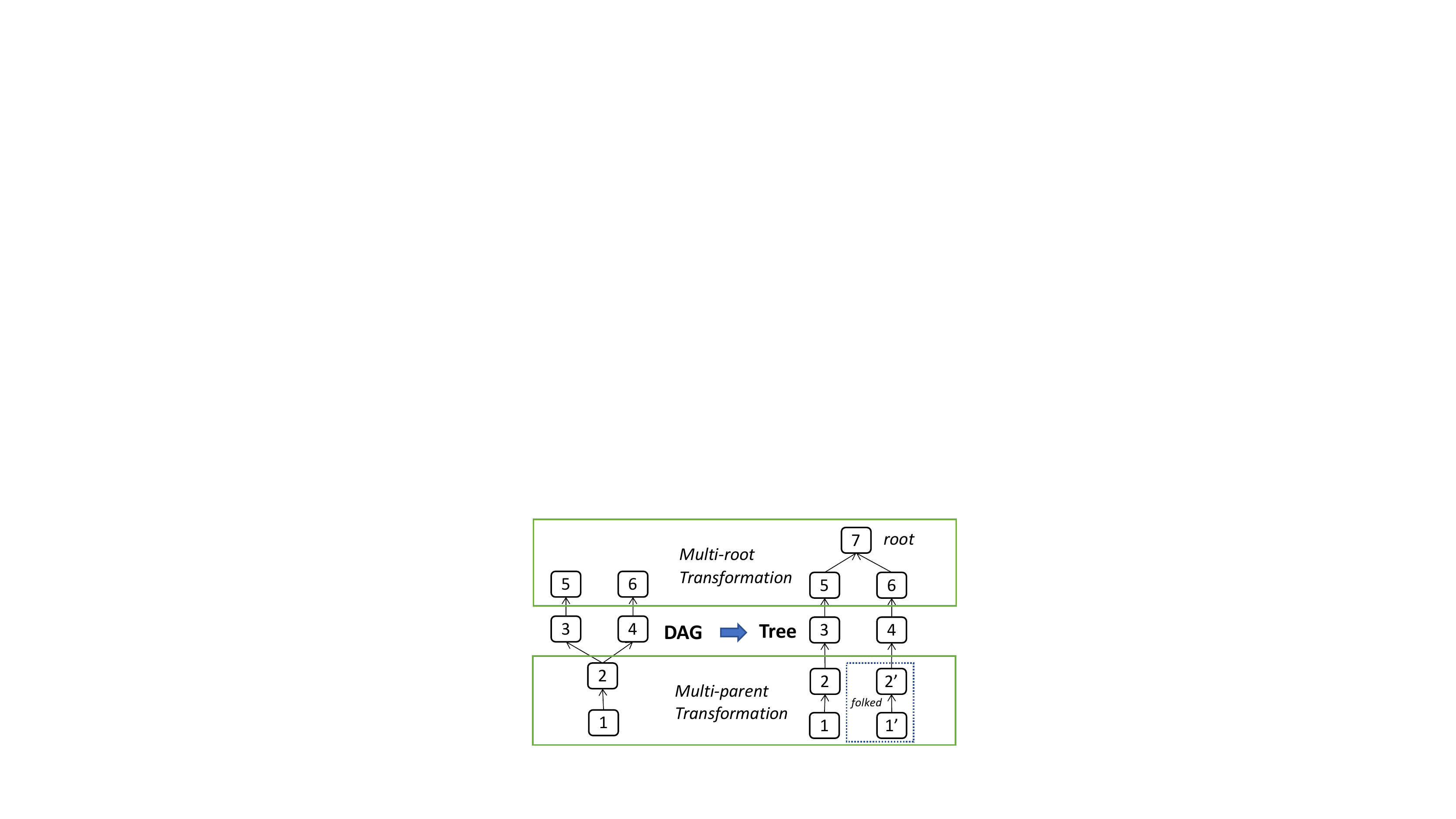}
  \captionof{figure}{Converting a DAG to a tree}
   \vspace{-0.1in}
  \label{fig:dag2tree}
\end{figure}

\begin{figure}[t]
	\centering
  \includegraphics[width=.7\linewidth]{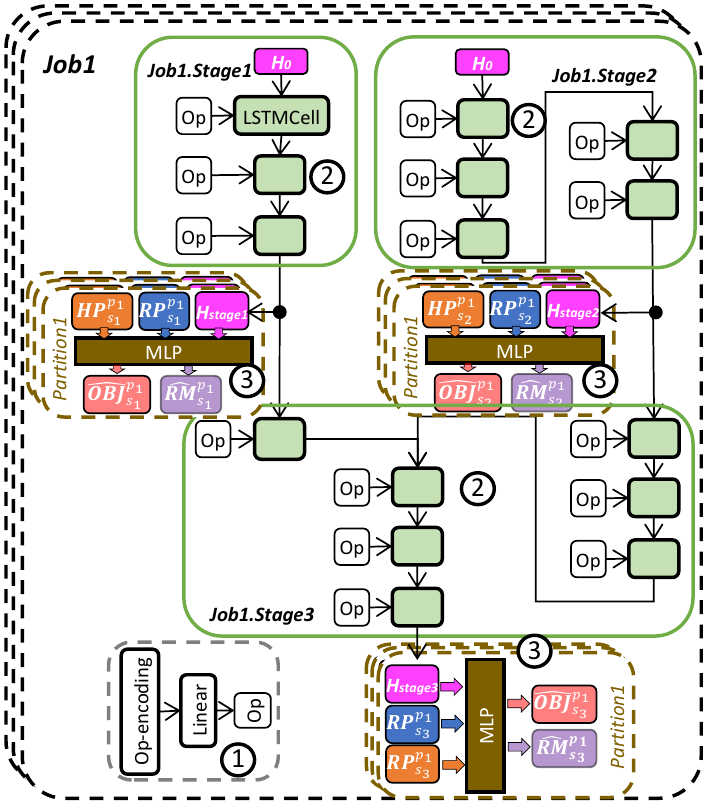}
  \captionof{figure}{Example of stage Embedding by TLSTM}
   \vspace{-0.1in}
  \label{fig:tlstm}	
\end{figure}

In our work, we extend QPPNet and TLSTM to support the MCI-based modeling in big data systems by solving two issues:

{\it (1) How to model arbitrary DAG structures in a tree-based model? }
Both QPPNet and TLSTM consider only tree-structured query plans, while a big data query plan can be an arbitrary DAG. We address the problem by converting an arbitrary DAG to a tree, as shown in Figure~\ref{fig:dag2tree}
We first consider a DAG as an ``extended" tree by treating output operators (out degree = 0) as the ``root(s)" and the input operators (in degree = 0) as the ``leaves". Hence the operator dependencies point from a child node to a parent node. 
Consequently, the ``extended" tree could have (i) nodes with multiple parents (multi-parent), and (ii) multiple ``roots" (multi-root). 
When a node has multiple parents, we fork the subtree rooted on the node and link each subtree to one parent. 
When multiple ``root" nodes are detected, we add one artificial node in the tree and point all the current ``root" nodes to it, such that the artificial node becomes the only root for the transformed tree.
After the transformations, the query plan could be converted into a tree and fits any tree-based model.

{\it (2) How to extend the model to leverage MCI features?}
QPPNet learns neural units in the plan structure and gets the latency prediction from the neural unit of the root operator. 
To leverage the instance-oriented features in channels 2-5, we broadcast them to all the neural units in the stage during preprocessing.
Consequently, each neural unit can have both the operator-oriented features (Ch1+AIM) and the instance-oriented features to finish the instance latency prediction.
Unlike QPPNet, TLSTM is quite compatible with our MCI framework after converting DAG plans to trees. The original TLSTM  implicitly treats the hidden layer of the root operator as a representation of the query and learns the latency for it.
In the extension, MCI+TLSTM applies the MCI framework with TreeLSTM as the plan embedder and the hidden layer of the root operator as the query embedding to learn the instance latency.

\minip{Customized Features in Figure~\ref{fig:mci}}

We will provide the list of customized features upon approval by Alibaba.

\section{More Details on IPA}
\label{appendix:ipa}

\subsection{Proof for Theroem~\ref{thm:ipa}}

Denote $m$ and $n$ as the number of instances and machines, respectively. Denote $L \in \mathbb{R}^{m\times n}$ as the latency matrix, $L_i$ as the i-th row in $L$, and $L_{ij}$ as the latency in the $i$-th row and the $j$-th column.
And further denote IPA-$L$ as the procedure of applying IPA over a latency matrix L and $IPA(L)$ as the stage latency achieved accordingly.

\cut{
We assume that each stage has the {\bf properties $\mathcal{P}$} as:
\begin{enumerate}
	\item Instances in a stage request the same CPU and memory.
	\item A stage can always find a feasible solution because the cluster's capacity is always more than a stage's need.
	\item Each column of the $L$ matrix shares the same order\footnote{This property is verified over 94\% of the 2M stages, and the remaining 6\% violation is due to the model uncertainty.}.
\end{enumerate}
}

Recall that instances in a stage request the same CPU and memory when applying IPA and IPA can always find a feasible solution for a stage because the cluster's capacity is always more than a stage's need. 

We consider each machine a 1-instance machine whose capacity can only afford one instance. Otherwise, we drop machines out of resource capacity and flat each k-instance machine (k>1) into k 1-instance machines. 
Since a stage can always find a feasible solution in the cluster, we have $m \leq n$.


Under the {\it column-order} assumption $\mathcal{C}$, let us denote $\tilde{L}$ as constructed by sorting rows of $L$ in the non-ascending order such that $\tilde{L}_{1j} \geq \tilde{L}_{2j} \geq ... \geq \tilde{L}_{mj}, \forall j$.

\begin{proposition}\label{pps:ipa-tilde}
	Under the column-order assumption $\mathcal{C}$, IPA-$L$ and IPA-$\tilde{L}$ return the same placement plan and IPA-$\tilde{L}$ returns the instance-machine matches in the descending order of the rows.
\end{proposition}

\begin{proof}
Let us denote the machine sequence $A=(a_1, a_2, ..., a_m)$ such that the best possible latency (BPL) for row $1, 2,..., m$ in $L$ is achieved at $L_{1a_1}, L_{2a_2}, ..., L_{ma_m}$. 
Similarly, we denote the machine sequence $\tilde{A}=(\tilde{a}_1, \tilde{a}_2, ..., \tilde{a}_m)$ as the machine sequence that BPLs are achieved for row $1,2,...,m$ in $\tilde{L}$.
Since $\tilde{L}$ is constructed by sorting rows of $L$, we have $\tilde{A}$ as one of the permutations of $A$.

Notice that after each instance-machine match, IPA will disable the responding row and column in the latency matrix, and recursively choose the maximum (recalculated) BPL as the instance-machine match from the remaining latency matrix.

Consider after the $t$-th ($0\leq t \leq m-1$) match, where we have the remaining latency matrices $L, \tilde{L} \in \mathbb{R}^{(m-t)\times (n-t)}$.
Assume $(i, a_i)$ is the first instance-machine match picked by IPA in $L$, such that $L_{ia_i}$ is the maximum among BPLs of all remaining rows.
Since each entry in $L$ and $\tilde{L}$ represents the running latency, we further assume each entry in $L$ and $\tilde{L}$ has a unique value, and hence the column order in $\tilde{L}$ is descending.

Now we show that IPA always picks the BPL at the first row in $\tilde{L}$ at each step $t\in[0, m-1]$ as the first instance-machine match by proving $L_{ia_i} = \tilde{L}_{1\tilde{a}_1}$ over the remaining latency matrices $L$ and $\tilde{L}$.

When $t=m-1$, the remaining $L$ has only one row, and hence $\tilde{L} = L$. Thus, they return the same BPL in the first and only row.
When $t < m-1$, assume the first instance-machine match picked by IPA in $\tilde{L}$ is not in the first row. 
Without loss of generality, we assume the first match is $(2, \tilde{a}_2)$ in $\tilde{L}$. I.e., $\tilde{L}_{2\tilde{a}_2}$ is the maximum among all BPLs from each row of $\tilde{L}$ and hence $\tilde{L}_{2\tilde{a}_2} \geq \tilde{L}_{1\tilde{a}_1}$.
Notice that by following the descending column order in $\tilde{L}$ at column $\tilde{a}_1$, we have $\tilde{L}_{2\tilde{a}_1} < \tilde{L}_{1\tilde{a}_1}$. Since $(2, \tilde{a}_2)$ achieves the BPL as the minimum latency in row $\tilde{L}_2$, we have $\tilde{L}_{2\tilde{a}_2} \leq \tilde{L}_{2\tilde{a}_1}$.
Thus, we get $\tilde{L}_{2\tilde{a}_2} < \tilde{L}_{1\tilde{a}_1}$, which contradicts our assumption that $\tilde{L}_{2\tilde{a}_2} \geq \tilde{L}_{1\tilde{a}_1}$. Therefore, IPA picks $(1, \tilde{a}_1)$ as the first instance-machine match in $\tilde{L}$, which is originally from the first match $(i, a_i)$ in $L$.

Therefore, IPA-$\tilde{L}$ always picks the same instance-machine match as IPA-$L$ in each step, and hence IPA-$L$ and IPA-$\tilde{L}$ shall return the same placement plan.
Since IPA-$\tilde{L}$ recursively picks the first row in the remaining $\tilde{L}$ in each step, we also show that IPA-$\tilde{L}$ returns the instance-machine matches following the descending order of the rows.
\end{proof}

\begin{theorem}\label{thn:ipa-so}
Under the column-order assumption $\mathcal{C}$, 
	IPA-$\tilde{L}$ achieves the minimum stage latency, and hence for IPA-$L$, such that $IPA(\tilde{L}) = IPA(L) = l_m^*$, where $l_m^*$ denotes the groundtruth minimum stage latency of $L \in \mathbb{R}^{m\times n}$.
	
\cut{
	IPA-$\tilde{L}$ achieves Pareto-optimality, and hence for IPA-$L$, such that
	\begin{enumerate}[label=(\Alph*)]
		\item IPA achieves the minimum stage latency of $L \in \mathbb{R}^{m\times n}$ as $l_m^*$, where $l_m^*$ denotes the groundtruth minimum stage latency.
		\item IPA achieves the minimum stage cost among the placement plans that achieve $l_m^*$.
	\end{enumerate} 
	}
\end{theorem}

\begin{proof}
	
According to Proposition~\ref{pps:ipa-tilde}, we consider each column shares the same non-ascending order in $L_{m\times n}$.
Given an arbitrary $L_{m\times n} \in \mathbb{R}^{m\times n} (m\leq n)$ that satisfies $\mathcal{C}$, we show that $IPA(L_{m\times n}) = l_m^*$ by induction, where $l_m^*$ is the ground truth minimum stage latency in $L_{m \times n}$.
	
{\ }\\
{\it \underline{Base case.}} When $m=1$ and $\forall n \geq 1$, $L$ has only one row and the stage latency equals the latency of the only instance. Notice IPA picks the machine that achieves the minimum latency for the first instance. Therefore, $IPA(L_{1\times n}) = \min_j (L_{1j})$ is also minimum achievable stage latency $l_1^*$. Thus, $IPA(L_{m\times n}) = l_m^*$ for $m=1$ and $\forall n \geq 1$.

{\ }\\
{\it \underline{Induction.}} Assume when $m=k (k\geq 1)$ and $\forall n \geq k$, IPA achieves the groundtruth minimum stage latency, s.t., $IPA(L_{k\times n}) = l_k^*$.
Given $L_{(k+1)\times (n+1)} \in \mathbb{R}^{(k+1)\times (n+1)}$, we prove that IPA also achieves the minimum stage latency $l_{k+1}^*$, s.t., $IPA(L_{(k+1)\times (n+1)}) = l_{k+1}^*$. In the following text, we use $\hat{l}$ to represent $IPA(L_{(k+1)\times (n+1)})$ for the simplicity of reading, and hence we are going to prove $\hat{l} = l_{k+1}^*$.

Denote $L_{1j^*}$ as the first matched latency by IPA and $L_{k\times n} = L_{(k+1)\times (n+1)} / (L_{1\cdot}, L_{\cdot j^*})$. Thus, $L_{1j^*}$ is the minimum over $L_1$.
Notice that IPA will continue to be applied on $L_{k\times n}$. Hence we have $\hat{l} = \max(L_{1j^*}, IPA(L_{k\times n}))$.
Since $L_{k\times n}$ has $k$ rows and satisfies , IPA achieves the minimum stage latency $IPA(L_{k\times n}) = l_k^*$ according to the assumption. Therefore, we can rewrite $\hat{l} = \max(L_{1j^*}, l_k^*)$.

Assume $l_{k+1}^*$ is achieved by an optimal placement plan that provides a machine sequence $B = (b_1, b_2, ..., b_{k+1})$, s.t., instance $1$, $2$, ..., $k$+1 match machine $b_1$, $b_2$, ..., $b_{k+1}$, respectively. We further represent the stage latency achieved by a machine sequence $B$ over a latency matrix of $k+1$ rows as $l_{k+1}^B$. Thus, we have $l_{k+1}^B = l_{k+1}^*$.

Next, we category $B$ into three possible cases and prove $\hat{l} = l_{k+1}^*$ on each case.

\noindent
\underline{Case 1. $b_1 = j^*$ (when instance 1 matches machine $j^*$).} 

Since $l_{k+1}^B$ is the maximum instance latency picked over $k$+1 rows, we have $l_{k+1}^B = \max_i L_{ib_i} = \max(L_{1j^*}, \max_{i\geq 2}L_{ib_i})$. Notice the machine sequence $(b_2, ..., b_{k+1})$ constructs a placement plan over $L_{k\times n}$ and hence its corresponding stage latency $\max_{i\geq 2}L_{ib_i}$ cannot be less than the groundtruth minimum $l_k^*$. Thus, $\max_{i\geq 2}L_{ib_i} \geq l_k^*$. Therefore, we have 
$l_{k+1}^B = \max(L_{1j^*}, \max_{i\geq 2}L_{ib_i}) \geq \max(L_{1j^*}, l_k^*)$. 
Recall that $l_{k+1}^B = l_{k+1}^*$ is the groundtruth minimum and $\hat{l} = \max(L_{1j^*}, l_k^*)$. We get $\hat{l} \leq l_{k+1}^*$. Therefore, $\hat{l} = l_{k+1}^*$ based on the optimality of $l_{k+1}^*$.


\noindent
\underline{Case 2. $b_1 \neq j^*$ and machine $j^*$ is not picked by any instance.}

Consider another machine sequence $B' = (j^*, b_2, b_3, ..., b_{k+1})$ that picks machine $j^*$ for instance 1 and keeps the same machine choices for instance 2 to $k$+1 as $B$.
Since $L_{1j^*}$ is the minimum latency of $L_1$, we have $L_{1j^*} \leq L_{1b_1}$. Thus,
 $l_{k+1}^{B'} = \max(L_{1j^*}, \max_{i\geq2} L_{ib_i}) \leq \max(L_{1b_1}, \max_{i\geq2} L_{ib_i}) = l_{k+1}^B$. 
Since $l_{k+1}^B = l_{k+1}^*$ and $l_{k+1}^{B'} \leq l_{k+1}^B$, we shall have $l_{k+1}^{B'} = l_{k+1}^*$ based on the optimality of $l_{k+1}^*$.
Notice $B'$ achieves the groundtruth minimum latency and picks $L_{1j^*}$. Therefore, we have $\hat{l} = l_{k+1}^{B'} = l_{k+1}^*$ by the conclusion of case 1.

\noindent
\underline{Case 3. $b_1 \neq j^*$ and machine $j^*$ is picked by instance $i (i \neq 1)$}, such that $b_i = j^*$.
By the definition of the stage latency, $l_{k+1}^B = \max_i L_{ib_i} = \max(L_{ij^*}, \max_{i'\neq i}L_{i'b_{i'}})$.
Denote $L'_{k\times n} = L_{(k+1)\times (n+1)} / (L_{i\cdot}, L_{\cdot j^*})$ and its groundtruth minimum stage latency is $l'^*_{k}$.
Consider $\max_{i'\neq i}L_{i'b_{i'}}$ as the stage latency achieved over $L'_{k\times n}$ by picking machines in the sequence $(b_1, b_2, ..., b_{i-1}, b_{i+1}, b_{i+2}, ..., b_{k+1})$. Then we have $l'^*_{k} \leq \max_{i'\neq i}L_{i'b_{i'}}$ based on the optimality of $l'^*_{k}$.

We now claim that $l'^*_{k} \geq l^*_k$, where $l^*_k$ is the groundtruth minimum stage latency of $L_{k\times n} = L_{(k+1)\times (n+1)} / (L_{1\cdot}, L_{\cdot j^*})$.

First, $L_{k\times n}$ is elementwisely less than or equal to $L'_{k\times n}$. 
Since both latency matrices drop the column $j^*$ in $L_{(k+1)\times(n+1)}$ and the last $k-i+1$ rows of the two matrices are both from row $i+1$ to row $k+1$ in $L_{(k+1)\times(n+1)}$, the differences remain in the first $i-1$ rows.
Notice that the first $i-1$ rows in $L_{k\times n}$ are from row 2 to row $i$ in $L_{(k+1)\times(n+1)}$; and the first $i-1$ rows in $L'_{k\times n}$ are from row 1 to row $i-1$ in $L_{(k+1)\times(n+1)}$. 
Since each column in $L_{(k+1)\times(n+1)}$ keeps the non-ascending order, each of the first $i-1$ rows in $L_{k\times n}$ is elementwisely less than or equal to that of $L'_{k\times n}$,
and hence $L_{k\times n}$ is elementwisely less than or equal to $L'_{k\times n}$.

Second, we prove $l'^*_{k} \geq l^*_k$ by contradiction. 
Assume $l'^*_k < l_k^*$ and denote $B'=(b'_1, b'_2, ..., b'_k)$ as the machine sequence that achieves $l'^*_k$ in $L'_{k\times n}$, i.e., $l'^*_k = \max_{i'} L'_{i'b_{i'}}$. 
Since $L_{k\times n}$ is elementwisely less than or equal to $L'_{k\times n}$, we have $l'^*_k = \max_{i'} L'_{i'b_{i'}} \geq \max_{i'} L_{i'b_{i'}} = l_k^{B'}$. Therefore, $B'$ constructs one placement plan such that $l_k^{B'} \leq l'^*_k$.
Furthermore, we have $l_k^* \leq l_k^{B'}$ according to the optimality of $l_k^*$. Thus, $l_k^* \leq l_k^{B'} \leq l'^*_k$, which contradicts the assumption that $l_k^* > l'^*_k$.
Therefore, we have $l'^*_k \geq l_k^*$.

Recall that we have $l_{k+1}^* = L_{k+1}^B = \max (L_{ij^*}, \max_{i'\neq i}L_{i'b_{i'}})$ in case 3. We now show that $\hat{l} = l_{k+1}^*$ over the two subcases.

\noindent
\underline{Case 3.1. when $L_{ij^*} < \max_{i'\neq i}L_{i'b_{i'}}$}, $l_{k+1}^* = \max_{i'\neq i} L_{i'b_{i'}}$. (1) Notice $l'^*_k \geq l_k^*$ and $l'^*_{k} \leq \max_{i'\neq i}L_{i'b_{i'}}$. Thus, we have $l_{k+1}^* \geq l_k^*$.
(2) On the other hand, notice that $l_{k+1}^* = \max_{i'\neq i} L_{i'b_{i'}} \geq L_{1b_1} \geq L_{1j^*}$. Thus, we have $l_{k+1}^* \geq L_{1j^*}$. 
Considering (1) and (2), $l_{k+1}^* \geq \max(L_{1j^*}, l_k^*) = \hat{l}$. Based on the optimality of $l_{k+1}^*$, we have $\hat{l} = l_{k+1}^*$.

\noindent
\underline{Case 3.2. when $L_{ij^*} \geq \max_{i'\neq i}L_{i'b_{i'}}$}, $l_{k+1}^* = L_{ij^*}$. 
(1) Based on the definition of stage latency, we have $\max_{i'\neq i}L_{i'b_{i'}} \geq L_{1b_1}$. Thus, $L_{ij^*} \geq \max_{i'\neq i}L_{i'b_{i'}} \geq L_{1b_1}$. 
(2) Since each column has the non-ascending order, and $L_{1j^*}$ is the minimum of $L_1$, we have $L_{ij^*} \leq L_{1j^*} \leq L_{1b_1}$.
Thus, $l_{k+1}^* = L_{ij^*} = L_{1j^*} = L_{1b_1}$ based on (1) and (2).
Furthermore, since $l'^*_k \geq l_k^*$ and $l'^*_{k} \leq \max_{i'\neq i}L_{i'b_{i'}}$, we can derive $\max(L_{1j^*}, l_k^*) \leq \max(L_{ij^*}, l'^*_k) \leq \max(L_{ij^*}, \max_{i'\neq i}L_{i'b_{i'}})$. 
Recall that $\hat{l} = \max(L_{1j^*}, l_k^*)$ and $\max(L_{ij^*}, \max_{i'\neq i}L_{i'b_{i'}}) = l_{k+1}^*$.
Hence we have $\hat{l} \leq l_{k+1}^*$. According to the optimality of $l_{k+1}^*$, we have $\hat{l} = l_{k+1}^*$.

Therefore, we show that $IPA(L_{(k+1)\times (n+1))}) = \hat{l} = l_{k+1}^*$ holds for case 1-3 when $m=k+1$. And hence $IPA(L_{m\times n}) = l_m^*$ holds for $\forall m \geq 1, n \geq m$.
\end{proof}

\subsection{Clustering-based IPA}

By applying a clustering method, let us denote 
\begin{enumerate}
	\item $X' = [X'_1, ..., X'_{m'}]$ as the $m'$ instance clusters
	\item $\hat{X}' = [x'_1, ..., x'_{m'}]$ as the $m'$ instance cluster representatives.
	\item $\beta = [\beta_1, ..., \beta_{m'}]$ denote the number of instances in each instance cluster. $\beta_i = |X'_i|$
	\item $Y' = [Y'_1, ..., Y'_{n'}]$ as the $n'$ machine clusters
	\item $\hat{Y}' = [y'_1, ..., y'_{n'}]$ as the $n'$ machine cluster representatives.
	\item $s = [s_1, ..., s_{m'}]$ denotes the maximum instance each machine cluster can take. $s_{j} = \sum_{\{ j' | y_{j'} \in Y'_{j} \}}{s_{j'}} $
\end{enumerate}

\begin{algorithm}[t]
	\caption{Clustering-based IPA Approach}
	\label{alg:appr-ipa-algo-cluster}
	\small
	\begin{algorithmic}[1]  
		\REQUIRE {$X'$, $\hat{X}'$, $\beta$, $Y'$, $\hat{Y}'$, $s$}		
		\STATE $L'$ = cal\_latency(model, $\hat{X}'$, $\hat{Y}'$)	
		\STATE $Y^* = \hat{Y}'$, $X^* = \hat{X}'$, and $P = \{\}$. \hfill\COMMENT{Init}		
		\STATE $BPL_{list}$ = cal\_bpl($L'$, $X^*$, $Y^*$).
		\REPEAT
		\STATE $i_t, j_t$ = find\_index\_argmax($BPL_{list}$) \hfill\COMMENT{the instance and machine {\bf cluster ids} corresponding to the largest BPL value}.
		\STATE $\delta = \min(\beta_{i_t}, s_{j_t})$ \hfill\COMMENT{the num of instance can be sent}
		\STATE $X^{\delta}$ = get\_inst($X'_{i_t}$) \hfill\COMMENT{pick $\delta$ insts with top-$\delta$ input row num}
		\STATE $Y^{\delta}$ = get\_mach($Y'_{i_t}$) \hfill\COMMENT{pick $\delta$ available machs randomly}
		\STATE $P = P \bigcup \{x^\delta_k \rightarrow y^\delta_k |k\in [1,\delta] \}$, $s_{j_t} = s_{j_t} - \delta$, $\beta_{i_t} = \beta_{i_t} - \delta$, $X'_{i_t} = X'_{i_t} - X^\delta$, $Y'_{j_t} = Y'_{j_t} - Y^\delta$.		
		\IF {$Y^*$ is $\emptyset$ \AND $X^*$ is not $\emptyset$}
		\RETURN $P = \{\}$ \hfill\COMMENT{No solution found}
		\ENDIF			
		\STATE\algorithmicif\ {$\beta_{i_t} == 0$} \algorithmicthen\ $X^* = X^* - \{\hat{X}'_{i_t}\}$ \algorithmicend\ \algorithmicif
		\STATE\algorithmicif\ {$s_{j_t} == 0$} \algorithmicthen\ $Y^* = Y^* - \{\hat{Y}'_{j_t}\}$, Recalculate the $BPL_{list}$. \algorithmicend\ \algorithmicif		
				
		\UNTIL{$X^*$ is empty}
		\RETURN P
	\end{algorithmic}
\end{algorithm}

Algorithm~\ref{alg:appr-ipa-algo-cluster} shows the Pseudo code of the clustering based algorithm.
We use a density-based clustering~\cite{book-data-clustering} that calculates the kernel density estimation (KDE) of the inputs in the 1D space and uses the minima of the density distribution as the boundaries to divide instances into clusters.

It is worth mentioning that instances in the same cluster can still have smaller variances due to the differences in input row numbers. 
To avoid the case when the number of machines is not enough for running the entire instance cluster while the longer-running instances cannot be sent in line 9, 
our approach gives instances with a larger input row number a higher priority to be chosen in the same group. It also sorts the instances in the descending input row number order during clustering to reduce the time complexity during picking instances from the cluster.

\section{More details on RAA}
\label{appendix:raa}

\begin{table*}[t]
\ra{1.2}
\small
\centering
	\begin{tabular}{cl}\toprule
	Symbol & Description \\\midrule
{\bf instance-oriented}	&
\\
$k$ & the number of objectives. 
\\	
$d$ & the number of resource types in a resource configuration (RC). 
\\
$m$ & the number of instances.
\\
$i$ & the indicator to an instance.
\\
$j$ & the indicator to a \po solution in a Pareto set (list).
\\
$p_i$ & $p_i$ is the number of \po solutions in the $i$-th instance.
\\	
$\theta_i^{j}$ & the $j$-th RC corresponding to the Pareto-optimal solution for the $i$-th instance, $\theta_i^{j}\in \mathbb{R}^d$.
\\
$f_i^j$ & the $j$-th \po solution for the $i$-th instance. $f_i^j\in \mathbb{R}^k$.
\\
$\theta_i$ & the RC list, a list of RC corresponding to the Pareto set for the $i$-th instance. $\theta_i = [\theta_i^1, \theta_i^2, ..., \theta_i^{p_i}]$.
\cut{\\}
\cut{$f_i$ & the Pareto set for instance $i$, represented by a list of all \po solutions for the $i$-th instance. $f_i = [f_i^1, f_i^2, ..., f_i^{p_i}]$.}

\\\midrule
{\bf stage-oriented} & 
\\
$\bs{\theta}$ & a list of RC list over $m$ instances, $\bs{\theta} = [\theta_1, \theta_2, ..., \theta_m]$
\\
$\bs{f}$ & a list of Pareto set over $m$ instances, $\bs{f} = [f_1, f_2, ..., f_m]$
\\
$\lambda_i$ &  the index of the instance-level \po solution in the $i$-th instance. $\lambda_i \in [1, p_i]$
\\
$\lambda$ & the index of the stage-level resources configuration based on $\bm{\theta}$, aka., the {\it state}. $\lambda = [\lambda_1, \lambda_2, ..., \lambda_m] \in \Lambda$.
\\
$\pi_i$ & the {\it step} defined as to move a state $\lambda$ to another state $\lambda'$, such that $\lambda'_{i} = \lambda_i+1$ and $\lambda'_{i'} = \lambda_{i'}$ when $i' \neq i$
\\
$\Theta^{\lambda}$ & the stage-level resource configurations indicated by $\lambda$, $\Theta^\lambda = [\theta_1^{\lambda_1},  \theta_2^{\lambda_2}, ..., \theta_m^{\lambda_m}]$.
\\
$F^{\Theta^{\lambda}}$ / $F^{\lambda}$ & the stage-level objective achieved by the stage-level resource configurations $\Theta^{\lambda}$, $F^{\lambda} \in \mathbb{R}^k$
\cut{\\
$\Phi$ & the function mapping from the stage-level resource configurations space to the space of multi-objectives, $F^{\lambda} = \Phi({\Theta}^{\lambda})$.}
\\
\cut{$\Theta$ & a list of resource plans corresponding to all the stage-level \po solutions.}
$PO_{\Theta}$ & a list of resource plans corresponding to all the stage-level \po solutions.
\\
$PO_{F}$ & the Pareto set for the stage-level objectives, represented by a list of stage-level \po solutions.

\\\bottomrule
	\end{tabular}
\caption{Notation for the Hierarchical MOO problem}
\label{tab:notations}	
\end{table*}

\subsection{Optimized Versions of RAA}

Our Stage Optimizer supports several choices for using RAA. 

\minip{Clustering choices.} we provide choices for clustering including:
\begin{enumerate}
	\item \verb|RAA(W/O_C)| does not cluster instances before applying our algorithm, which could achieve the best stage-level latency (exclusive solving time) and cost. However, it has a higher overhead than others.
	\item \verb|RAA(DBSCAN)| applies a standard clustering approach, \verb|DBSCAN|, to divide instances into several clusters. Clustering is performed over the MCI features of these instances. 
	\item \verb|RAA(Fast_MCI)| applies our customized MCI-based strategy to cluster instances. It leverages the instance clusters from IPA's clustering result (based on fast 1D clustering) without additional time overhead.
Specifically, we subdivide each instance cluster in IPA into separate new clusters based on the specific machine clusters that the different instances are sent. 
Notice that this process naturally happens at line 9 in Algorithm~\ref{alg:appr-ipa-algo-cluster} when a subset of $\delta$ instances ($X^{\delta}$) from an instance cluster ($X^c_{i_t}$) are sent to a subset of machines ($Y^{\delta}$) from the same machine cluster ($Y^c_{j_t}$).
Therefore it incurs no additional time overhead to generate new instance clusters and find the instance with the largest input row number as its representative.
\end{enumerate}

Furthermore, when a clustering method is applied, each instance cluster would be reduced to one instance such that the latency of the instance cluster equals the latency of the representative instance and the cost of a cluster is the cost of the representative times the cluster size.
After clustering, for ease of composition, we also use ``instance'' to refer to an instance cluster (when the clustering method is applied) when confusion does not arise.

By default, we choose \verb|RAA(Fast_MCI)| for its efficiency. 

\minip{MOO solutions.} In this work, we provide two solutions to the hierarchical MOO problem:
\begin{enumerate}
	\item \verb|RAA(General)| applies a general hierarchical MOO solution (Algorithm~\ref{algo_stage_moo}) for the general stage-level MOO problems.
	\item \verb|RAA(Path)| applies a faster hierarchical MOO solution (Algorithm~\ref{alg:appr-raa}) for a specific case when the two objectives are stage-level latency and a summable cost metric (e.g.,  CPU-hours).
\end{enumerate}

\subsection{Additional Notation}

Table~\ref{tab:notations} shows the notation for our description of the MOO solutions.
Recall that a resource configuration (RC) refers to the configuration used for a container running on a specific instance, and it is an instance-specific configuration.
A stage-level resource plan is a list of RC's, one for each instance, and it is a stage-oriented solution.
The design reserves the subscript exclusively for indicating the instances and uses superscripts to indicate others.

\subsection{Proof of Proposition~\ref{algo2_optimality}}
\begin{table*}[t]
	\ra{1.2}
	\small
	\centering
	\begin{tabular}{cl}\toprule
		Symbol & Description \\\midrule
		$k_1$ & the number of objectives with \texttt{max} operator.
		\\
		$k_2$ & the number of objectives with \texttt{sum} operator.
		\\
		$f_{ijv}$ & the $v$-th objective value of $j$-th solution in $i$-th instance.
		\\
		$F_h$ & the stage-level $h$-th objective value in $k_1$ objectives $(h\in[1,...,k_1])$.
		\\
		$F_v$ & the stage-level $v$-th objective value in $k_2$ objectives $(v\in[1,...,k_2])$.
		\\
		$w_v$ & the weight setting of $v$-th objective, $v \in [1,...,k_2]$.
		
		\\\bottomrule
	\end{tabular}
	\caption{Additional notation for General Hierarchical MOO Problem}
	\label{tab:notations-5.1}	
\end{table*}

\noindent {\bf{Proposition \ref{algo2_optimality}}} For a stage-level MOO problem, Algorithm~\ref{algo_stage_moo} guarantees to find a subset of stage-level Pareto optimal points. 

\vspace{0.1in}
The corresponding notation of this proof is described in Table \ref{tab:notations} and \ref{tab:notations-5.1}. First  we give some details for the sub-procedure \textit{find\_optimal}.

The high level idea is that during the enumeration of the $k_1$ objectives (that use the \texttt{max} aggregate operator), once we have obtained some fixed values for the $k_1$ objectives, we can determine the $k_2$ objectives (that use the \texttt{sum} aggregate operator) as follows: we transform the $k_2$ objective values in each  instance-level Pareto solution into a single value by utilizing a given weight vector\cut{WeightSum (WS) method~\cite{marler2004survey}}. Then we select for each instance the instance-level solution 
$f_i^j$
corresponding to the best weighted value, and sum up 
[$f_i^1, ..., f_i^{p_i}$], $i \in [1 \ldots m]$]
to construct one stage-level Pareto solution, which contains the fixed $k_1$ objective values and the summed values for the $k_2$ objectives. Then by varying the weight vectors, we get multiple stage-level solutions, which is a hierarchical version of the WeightSum (WS) method~\cite{marler2004survey}.

As a concrete example, suppose that we have two instances and three objectives, one using the \texttt{max} aggregate operator and the other two using the \texttt{sum} operator. For instance 1, there are two instance-level Pareto solutions $[15, 10, 5], [20, 15, 2]$. For instance 2, there are also two Pareto solutions, $[30, 5, 15], [40, 10, 5]$. As illustrated in Figure~\ref{fig:raa-example} and the example in Section \ref{sec:raa}, 
The function \textit{find\_all\_possible\_values} for objective $1$ gives the list of possible values $[30, 40]$.
Under the value $30$ and the weight setting of $[0.5, 0.5]$ for the other two objectives, for instance 1, both solutions do not exceed $30$, and the weighted sum is $0.5*10+0.5*5 = 7.5$ (solution 1) and $0.5*15+0.5*2=8.5$ (solution 2) respectively. So the solution $[15, 10, 5]$ is selected as the optimal for instance 1. In the instance 2, only the solution $[30, 5, 15]$ does not exceed $30$, so $[30, 5, 15]$ is the optimal choice for instance 2. After that, a new stage-level solution is computed based on the selected instance-level solutions as, $[max(15, 30)=30, sum(10, 5)=15, sum(5, 15)=20]$. In the same way, under the value $40$ of objective 1 and the weight setting of $[0.5, 0.5]$ for the other two objectives, the optimal choice for instance 1 and instance 2 are $[15, 10, 5]$ and $[40, 10, 5]$ respectively. We can obtain the stage-level values as $[max(15, 40)=40, sum(10, 10)=20, sum(5, 5)=10]$. So, the stage-level MOO set is $[[30, 15, 20], [40,20,10]]$.
	If we vary the weight vector as $[0.1, 0.9]$, we can find other stage-level MOO solutions, such as $[30, 20, 17]$ and $[40,25,7]$. 

As a concrete example with more than one objectives using \texttt{max} to see details in step $3$ of Algorithm \ref{algo_stage_moo}, suppose that we have two instances and four objectives, two (objective $1$ and $2$) using the \texttt{max} aggregate operator and the other two (objective $3$ and $4$) using the \texttt{sum} operator. For instance $1$, there are two instance-level Pareto solutions $[15, 6, 10, 5], [20, 30, 15, 2]$. For instance $2$, there are also two Pareto solutions, $[30, 10, 5, 15], [40, 50, 10, 5]$. The function \textit{find\_all\_possible\_values} for objective $1$ and $2$ gives the lists of possible values $[30, 40]$ and $[10, 30, 50]$ respectively. Then, we use the the Cartesian product of these two lists as $[[30,10], [30,30], [30, 50], [40,$ $10], [40,30], [40,50]]$. Under the combination $[30, 10]$, for instance $1$, only the solution $1$ does not exceed the objective $1$ value $30$ and the objective $2$ value $10$. So the solution $[15, 6, 10, 5]$ is selected as the optimal for instance $1$. In the instance $2$, similarly, $[30, 10, 5, 15]$ is the optimal choice for instance $2$. After that, a new stage-level solution is computed based on the selected instance-level solutions as, $[max(15, 30)=30, max(6, 10)=10, sum(10, 5)=15, sum(5, 15)=20]$. 
In the same way, under the combination $[30,50]$ of objective $1$ and $2$, and the weight setting of $[0.5, 0.5]$ for the objective $3$ and $4$, for instance $1$, both solutions do not exceed the objective $1$ value $30$ and the objective $2$ value $50$, and the weighted sum is $0.5*10+0.5*5 = 7.5$ (solution $1$) and $0.5*15+0.5*2=8.5$ (solution $2$) respectively. So the solution $[15, 6, 10, 5]$ with lower weighted sum is selected as the optimal for instance $1$. In the instance $2$, both solutions do not exceed the value of objective $2$ ($50$), but only solution $1$ does not exceed the value of objective $1$ ($30$). Therefore, the solution $1$ $[30, 10, 5, 15]$ is selected as the optimal choice for instance $2$. After that, a new stage-level solution is computed based on the selected instance-level solutions as, $[max(15, 30)=30, max(6, 10)=10, sum(10, 5)=15, sum(5, 15)=20]$. Then, after traversing all combinations of objective $1$ and $2$ under the weight setting of $[0.5, 0.5]$, we get the stage-level MOO set $[[30, 10, 15, 20], [40, 50, 20, 10]]$.

Now we first prove for a special case: when there are only $k_2$ objectives using the \texttt{sum} operator.


\begin{lemma} \label{lemma_proof_5.1} 
	For a stage-level MOO problem with $k_2$ objectives that use the \texttt{sum} operator only, Algorithm~\ref{algo_stage_moo} guarantees to find a subset of stage-level Pareto optimal points.
\end{lemma}



In proving Lemma \ref{lemma_proof_5.1}, we observe that Algorithm~\ref{algo_stage_moo} is essentially a Weighted Sum procedure over Functions (WSF). Indeed we will prove the following two Propositions: 1) each solution returned by WSF is Pareto optimal; 2) the solution returned by the function \textit{find\_optimal}  is equivalent to the solution returned by WSF. Then it follows that the solution returned by Algorithm~\ref{algo_stage_moo} is Pareto optimal.

To introduce WSF, we first introduce the indicator function $x_{ij}$, $i \in [1, ..., m], j \in [1, ..., p_i]$, to indicate that the $j$-th solution in $i$-th instance is selected to contribute to the stage-level solution. $\sum_{j=1}^{p_i} x_{ij} = 1$ means that only one solution is selected for each instance. Then $x = [x_{1j_1}, ... x_{mj_m}]$ represents the 0/1 selection for all $m$ instances to construct a stage-level solution.

So for the $v$-th objective, its stage-level value could be represented as the function $H$ applied to $x$: 
\begin{equation}
\begin{aligned}
& F_v = H_v (x) = \sum_{i=1}^{m} \sum_{j=1}^{p_i} x_{ij} * f_{ijv}, \\
& \quad where \sum_{j=1}^{p_i} x_{ij} = 1, i \in [1, ..., m], j \in [1, ..., p_i], v \in [1, ... k_2]
\end{aligned}
\end{equation}

Now we introduce the Weighted Sum over Functions (WSF) as: \cut{(where $w_v$ is the weight setting of $v$-th objective, $v \in [1,k_2]$ )}
%
\begin{align}
& \arg\!\min_{x} (\sum_{v=1}^{k_2} w_v * H_v(x)) \\
\mbox{s.t.} & \;\; \sum_{v=1}^{v} w_v = 1 
\end{align}



Next, we prove for Lemma \ref{lemma_proof_5.1}. As stated before, It is done in two steps.

\begin{proposition}
	The solution constructed using $x$ returned by WSF is Pareto optimal.
\end{proposition}

\begin{proof}
	~\\
	\indent Assume that $x^*$ (correspoding to $[F_1^*, ..., F_{k_2}^*]$ ) is the solution of WSF.
	Suppose that an existing solution $[F_1^{'}, ..., F_{k_2}^{'}]$ (correspoding to $x^{'}$)  dominates $[F_1^*, ..., F_{k_2}^*]$. This means that $\sum_{v=1}^{k_2} w_v * H_v({x^{'}})$ is less than that of $x^*$. \par
	This contradict that  $x^*$ is the solution of WSF. So there is no $[F_1^{'}, ..., F_{k_2}^{'}]$ dominating $[F_1^*, ..., F_{k_2}^*]$. Thus, $[F_1^*, ..., F_{k_2}^*]$ is Pareto optimal.
\end{proof}

\begin{proposition}
	The optimal solution returned by the function \textit{find\_optimal} in Algorithm \ref{algo_stage_moo} is equivalent to the solution constructed using $x$ returned by WSF.
\end{proposition}

\begin{proof}
	~\\
	\indent Suppose $x^{'}$ is returned by WSF. The corresponding stage-level solution is $[F_1^{'}, ..., F_{k_2}^{'}]$
	\begin{equation}
	\begin{split}
	x^{'} & = \arg\!\min(\sum_{v=1}^{k_2} w_v * H_v(x))\\
	& = \arg\!\min(\sum_{v=1}^{k_2} w_v * (\sum_{i=1}^{m} \sum_{j=1}^{p_i} x_{ij} * f_{ijv}))\\
	& = \arg\!\min(\sum_{i=1}^{m} (\sum_{v=1}^{k_2} \sum_{j=1}^{p_i} (w_v * f_{ijv}) * x_{ij}))
	\end{split}
	\end{equation}
	
	For the solution $[F_1^{''}, ..., F_{k_2}^{''}]$ returned by the function \textit{find\_optimal} in Algorithm \ref{algo_stage_moo}, $x^{''}$ represents the corresponding selection. It is achieved by minimizing the following formula:
	
	\begin{equation}
	\begin{split}
	& \sum_{i=1}^{m} (WS_{ij} | j\in[1, p_i])\\
	& = \sum_{i=1}^{m} (\sum_{v=1}^{k_2} w_v * f_{ijv} | j\in[1, p_i])\\
	& = \sum_{i=1}^{m} (\sum_{v=1}^{k_2} \sum_{j=1}^{p_i}
	(w_v * f_{ijv}) * x_{ij})\\
	\end{split}
	\end{equation} 
	where $WS_{ij} = \sum_{v=1}^{k_2} w_v * f_{ijv}$. 
	
	So, we have:
	\begin{equation}
	\begin{split}
	x^{''} & = \arg\!\min(\sum_{i=1}^{m} (\sum_{v=1}^{k_2} \sum_{j=1}^{p_i} (w_v * f_{ijv}) * x_{ij}))
	\end{split}
	\end{equation}
	\noindent 
	 Since they are of the same form and achieve min at the same time, so we have $[F_1^{'}, ..., F_{k_2}^{'}] = [F_1^{''}, ..., F_{k_2}^{''}]$
	
\end{proof}
 
With these two propositions, we finish the proof of Lemma \ref{lemma_proof_5.1}.

Finally, we prove that for the general case that also involves $k_1$ objectives that use the \texttt{max} aggregate operator, Proposition \ref{algo2_optimality} holds as well.

\begin{proof}
	Recall that Algorithm \ref{algo_stage_moo} enumerates all combinations of stage-level $k_1$ objective values. Thus, there does not exist a (prefix of) stage-level $k_1$ solution 
	$ [F_1^{'},..., F_{k_1}^{'}]$
	that could not be found by algorithm \ref{algo_stage_moo}. Putting it together with Lemma \ref{lemma_proof_5.1}, we complete the proof of Proposition \ref{algo2_optimality}.
\end{proof}

Here we discuss {\bf the time complexity} when WS is utilized in the $find\_optimal$ function in Algorithm \ref{algo_stage_moo}.
Recalling the Algorithm \ref{algo_stage_moo}, it takes $O(m*p_{max})$ to find all the values within the lower and upper bounds of stage-level values for each \texttt{max} objectives. 
So, the Cartesian product of all $k_1$ lists takes $O((m*p_{max})^{k_1})$. Within the loop of each combination, 
Algorithm \ref{algo_stage_moo} varies $w$ weight vectors to generate multiple stage-level solutions. And under each weight vector, it takes $O(m*p_{max})$ to select the optimal solution for each instance based on WS. Thus, the overall time complexity is $O((m*p_{max})^{k_1} * w * (m*p_{max})) = O(w*(m*p_{max})^{k_1+1})$.

For the particular case of $k=2$ with one using \texttt{max} and the other using \texttt{sum} (the same case in Algorithm \ref{alg:appr-raa}), $k_1 = 1$ and $k_2 = 1$, which means the weight vector for $find\_optimal$ is fixed. Therefore, Algorithm \ref{algo_stage_moo} takes $O((m*p_{max})^{2})$.

\subsection{Proof of Proposition~\ref{algo3_optimality}}

\begin{table*}[t]
\ra{1.2}
\small
\centering
	\begin{tabular}{cl}\toprule
	Symbol & Description \\\midrule
$u_i^j$ & the latency in $f_i^j$.
\\
$o_i^j$ & the cost in $f_i^j$.
\\
$u_i$ & a list of latency in $f_i$. $u_i = [u_i^1, u_i^2, ..., u_i^{p_i}]$. 
\\
$o_i$ & a list of cost in $f_i$. $o_i = [o_i^1, o_i^2, ..., o_i^{p_i}]$. 
\\
$U^{\lambda}$ & the stage-level latency achieved by the stage-level configurations $\Theta^\lambda$,  $U^{\lambda} = \max_{i\in [1,m]}(u_i^{\lambda_i}) \in \mathbb{R}$ 
\\
$O^{\lambda}$ & the stage-level cost achieved by the stage-level configurations $\Theta^\lambda$, $O^{\lambda} = \sum_{i\in [1,m]}(o_i^{\lambda_i}) \in \mathbb{R}$ 
\\\bottomrule
	\end{tabular}
\caption{Additional notations for RAA path problem}
\label{tab:notations-5.2}	
\end{table*}

We show that Algorithm~\ref{alg:appr-raa} guarantees to find the full set of stage-level \po points
by induction. 
The notations in the proof are from Table~\ref{tab:notations} and~\ref{tab:notations-5.2}.

Assume that the stage-level Pareto set 
$PO_{F}$
has $z$ solutions. 
Denote 
$PO_{F} = [F^{{\lambda^{(1)}}}, F^{{\lambda^{(2)}}}, ..., F^{{\lambda^{(z)}}}]$, $PO_{\Theta} = [\Theta^{{\lambda^{(1)}}}, \Theta^{{\lambda^{(2)}}}, ..., \Theta^{{\lambda^{(z)}}}]$, 
and the state
${\lambda^{(t)}}\in \Lambda, t\in[1, z]$ indicates the $t$-th stage-level \po solution.
After pre-sorting, we assume that both the stage-level and instance-level \po solutions are sorted in the descending order of latency (and hence the ascending order of cost):
\begin{gather}
	U^{{\lambda^{(1)}}} > U^{{\lambda^{(2)}}} > ... > U^{{\lambda^{(z)}}} \label{eq:5.2-1} \\
	O^{{\lambda^{(1)}}} < O^{{\lambda^{(2)}}} < ... < O^{{\lambda^{(z)}}} \label{eq:5.2-2} \\
	u_i^1 > u_i^2 > ... > u_i^{p_i} \label{eq:5.2-3}, \forall i \in [1, m]\\
	o_i^1 < o_i^2 < ... < o_i^{p_i} \label{eq:5.2-4}, \forall i \in [1, m]
\end{gather}

To recall, Algorithm~\ref{alg:appr-raa} maintains a max-heap $Q$ with $m$ nodes built over the latencies among each of the $m$ instances at a state $\lambda$. 
At each RAA step from $\lambda$ to $\lambda'$, $Q$ pops out the maximum latency with its corresponding instance id $i$, and pushes back the $(\lambda_i+1)$-th \po solution in instance $i$. It stops the path when the $(\lambda_i+1)$-th solution does not exist in the popped instance $i$ ($\lambda_i = p_i$).


\begin{proof} Now we show that Algorithm~\ref{alg:appr-raa} covers $\lambda^{(1)}, ... \lambda^{(z)}$ along the path.
\\
\underline{Base case.} The state ${\lambda^{(1)}}$ is in RAA path. When $\lambda = [1,1, ..., 1]$, 
$$
O^{\lambda} = \sum_{i\in [1,m]}(o_i^1) \leq \sum_{i\in [1,m], j\in [1,p_i]}(o_i^j) 
$$ 
where "=" exists iff $o_i^j=o_i^1, \forall i$.
Therefore, no other point in the stage-level objective space could dominate $F^{\lambda}$ (on $O^{\lambda}$) and hence $F^{\lambda}$ is a stage-level \po point. Notice 
$$
U^{\lambda} = \max_{i\in [1,m]}(u_i^1) \geq \max_{i\in [1,m], j\in [1,p_i]}(u_i^j)
$$ 
Therefore $U^{\lambda} = \max_{{\lambda' \in \Lambda}} U^{{\lambda'}}$ is the maximum achievable stage latency and 
$\lambda = \arg \max_{{\lambda_q}} U^{{\lambda_q}} = {\lambda^{(1)}}$. Therefore $\lambda^{(1)}$ is on the RAA path as the starting point.
\\
\underline{Induction}: Assume that state ${\lambda^{(t)}}, t <= z-1$ is on the RAA path. Prove the state $\lambda^{(t+1)}$ will also be on the RAA path. 

Given $\lambda^{(t)}$ indicating the $t$-th stage-level \po solution, we use  $O^*$ and $U^*$ as shorthand for  
$O^* = O^{\lambda^{(t)}}, U^* = U^{\lambda^{(t)}}$. 

\underline{First}, let us construct $\hat{\lambda}$ using Eq.~\eqref{eq:5.2-5} and show that $\lambda^{(t+1)} = \hat{\lambda}$
\begin{align}\label{eq:5.2-5}
	\hat{\lambda}_i = \begin{cases}
		\lambda_i^{(t)} + 1 & \text{if} \;\; u_i^{\lambda^{(t)}_i} = U^* \\
		\lambda_i^{(t)}     & \text{otherwise} \;\; (u_i^{\lambda^{(t)}_i} < U^* )\\		
	\end{cases}, \forall i \in [1, m]
\end{align}

We prove two properties below.

(i) $O^{\hat{\lambda}}$ is the minimum stage-level cost when the stage-level latency is smaller than $U^*$.
Notice for an instance $i$ whose latency is not the largest at state $\lambda^{(t)}$, the instance-level minimum cost has already been achieved at $\hat{\lambda}_i = \lambda_i^{(t)}$ when the latency needs to be smaller than $U^*$ according to the Pareto-optimality (otherwise, $F^{\lambda^{(t)}}$ could have been dominated on the cost).
For an instance $i$ whose latency is equal to $U^*$ at state $\lambda^{(t)}$ ($u_i^{\lambda^{(t)}_i} = U^*$), its minimum cost shall be $o_i^{\lambda^{(t)}_i+1}$ at $\hat{\lambda}_i = \lambda_i^{(t)}+1$ when the latency needs to be smaller than $U^*$, according to Eq.~\eqref{eq:5.2-4}.
Therefore, at state $\hat{\lambda}$, the stage-level cost $O^{\hat{\lambda}}$ is the summation of the minimum instance-level achievable cost among all instances when their latencies are smaller than $U^*$. 

(ii) $U^{\hat{\lambda}}$ is the maximum stage-level latency that is smaller than $U^*$. Otherwise, $\exists \lambda' \in \Lambda, i \in [1, m]$, such that $u_{i}^{\hat{\lambda}_{i}} < u_{i}^{\lambda'_{i}} < U^*$.
Notice for an instance $i$ with $u_i^{\lambda^{(t)}_i} = U^*$, there is no instance-level \po solutions with latency between $u_i^{\hat{\lambda}_i}$ ($=u_i^{\lambda^{(t)}_i+1}$) and $U^*$ ($=u_i^{\lambda^{(t)}_i}$).
Thus, $u_{i}^{\lambda^{(t)}_{i}} < U^*$, and hence $u_{i}^{\hat{\lambda}_{i}} = u_{i}^{\lambda^{(t)}_{i}}$. 
Therefore, $u_{i}^{\hat{\lambda}_{i}} = u_{i}^{\lambda^{(t)}_{i}} < u_{i}^{\lambda'_{i}} < U^*$. According to the Pareto optimality in instance $i$, we have the cost as $o_{i}^{\lambda^{(t)}_{i}} > o_{i}^{\lambda'_{i}}$. Hence, there exists a \po solution in instance $i$ such that when the instance latency is smaller than $U^*$, it achieves a smaller cost than the $t$-th stage-level optimal solution at state $\lambda^{(t)}$, which violates the Pareto-optimality of $F^{\lambda^{(t)}}$. Therefore, $U^{\hat{\lambda}}$ is the maximum stage-level latency that is smaller than $U^*$. 

With (i) and (ii), we have $\lambda^{(t+1)} = \hat{\lambda}$ as constructed.

\underline{Second}, let us show that Algorithm~\ref{alg:appr-raa} can reach $\lambda^{(t+1)}$ by its policy. 
When $\lambda^{(t)}$ is on the path, the maximum latencies in the max-heap $Q$ is $U^*$. At each step, the policy will pop out an instance $i$ whose latency is $U^*$ achieved at $\lambda^{(t)}_i$ and push back $\lambda^{(t)}_i + 1$. 
Therefore, each step addresses one instance that matches the first line of Eq.~\eqref{eq:5.2-5}. 
Denote $q$ as the number of instances that matches the first line of Eq.~\eqref{eq:5.2-5}. Then we have $\lambda^{(t+1)}$ on the path after $q$ steps from $\lambda^{(t)}$.


Therefore, we have all the stage-level \po solutions in the RAA path by using Algorithm~\ref{algo3_optimality}.
\end{proof}

{\bf The time complexity} for sorting instance-level \po solutions costs $O(m \cdot p_{max} \log p_{max})$ over $m$ instances.
The algorithm uses $O(m\log m)$ to build a max-heap. At each step, it includes one pop and one push, with $O(\log m)$. It takes at most $\sum_i {p_i}$ steps to iterate all instance-level \po solutions, which is $O(\sum_i {p_i}) = O(m \cdot p_{max})$ iterations. Therefore the time complexity is $O(m \cdot p_{max} \log p_{max} + m \log m + m \cdot p_{max} \log m) = O((1+p_{max}) m \log m) = O(m \cdot p_{max} \log (m \cdot p_{max}))$.

\section{More on Performance Evaluation}
\label{appendix:evaluation}

\begin{figure}[t]
	\centering
    \includegraphics[width=.48\textwidth]{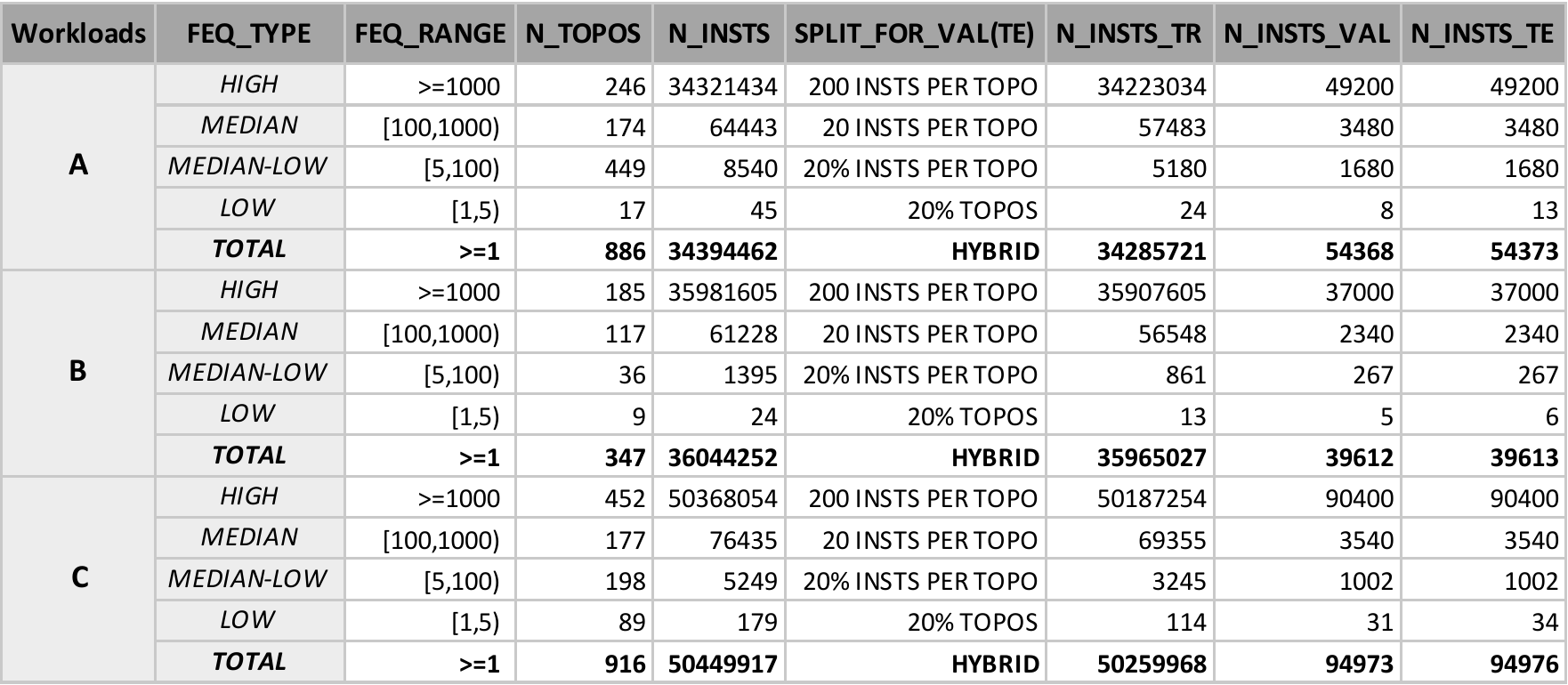}
    \caption{Data split for model evaluations}
    \label{fig:data-split}
\end{figure}

\subsection{Workload Characteristics}
\label{appendix:workload}

\begin{figure*}[t]
	\centering
	\vspace{-0.1in}
	\hspace{-6cm}
	
	\begin{tabular}{lcc}

		\subfigure[\small{CDF for the Number of Stages in a Job}]
		{\label{fig:cdf-job-stage-num}\includegraphics[height=3.0cm,width=5.5cm]{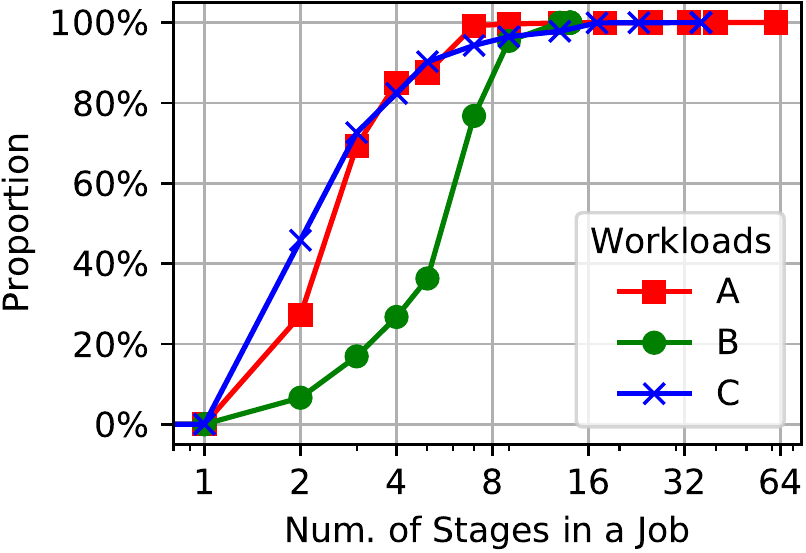}}

		&
		\subfigure[\small{CDF for the Number of Instances in a Stage}]
		{\label{fig:cdf-stage-inst-num}\includegraphics[height=3.0cm,width=5.5cm]{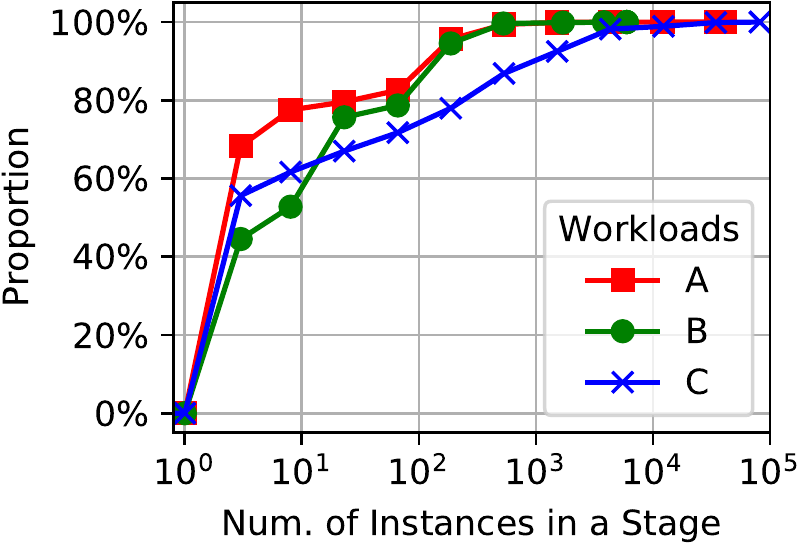}}

		&		
		\subfigure[\small{CDF for the number of operators in a Stage}]
		{\label{fig:cdf-stage-op-num}\includegraphics[height=3.0cm,width=5.5cm]{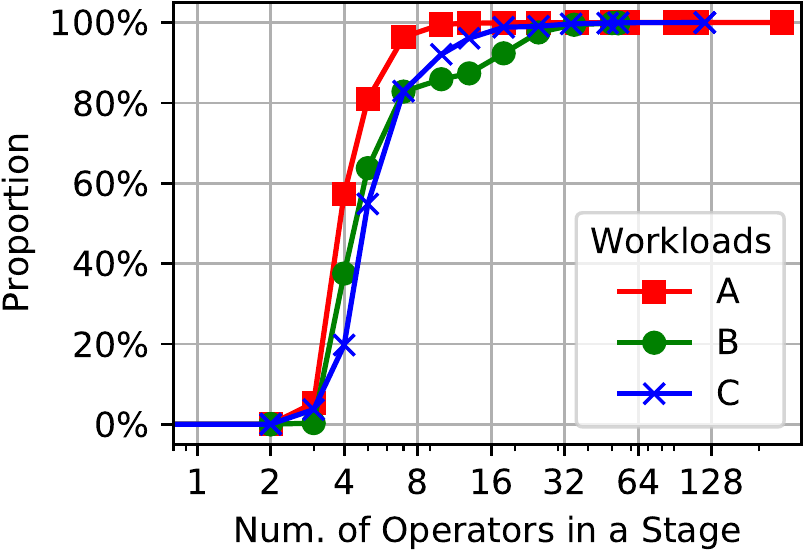}}
	
		\\	
		
		\subfigure[\small{CDF for the Job Latency}]
		{\label{fig:cdf-job-lat}\includegraphics[height=3.0cm,width=5.5cm]{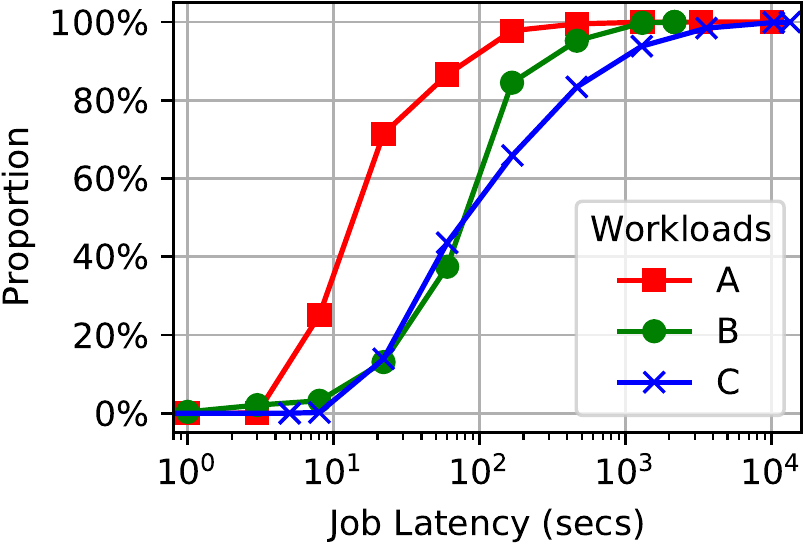}}

		&
		\subfigure[\small{CDF for the Stage Latency}]
		{\label{fig:cdf-stage-lat}\includegraphics[height=3.0cm,width=5.5cm]{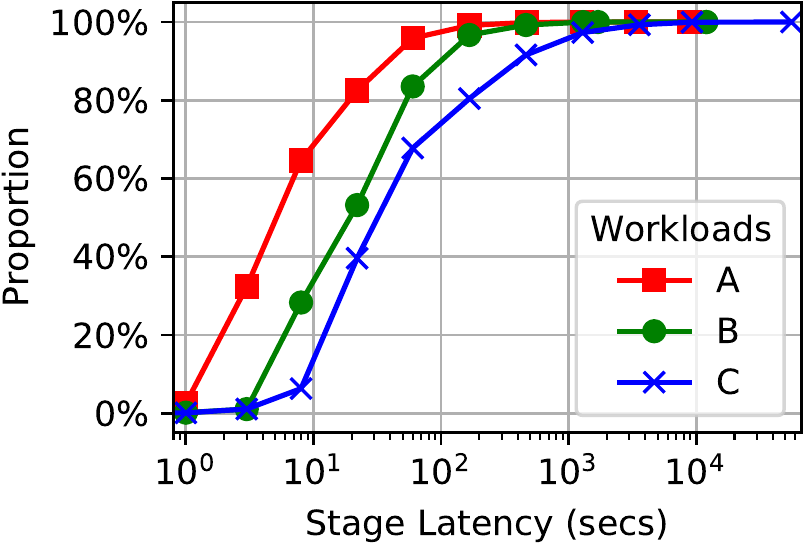}}

		&
		\subfigure[\small{CDF for the Instance Latency}]
		{\label{fig:cdf-inst-lat}\includegraphics[height=3.0cm,width=5.5cm]{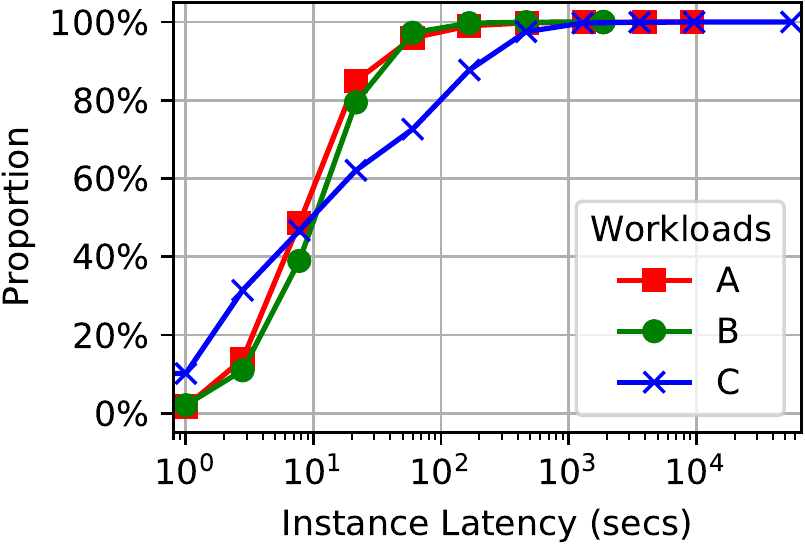}}
		
	\end{tabular}
	\caption{The Cumulative Distribution Function (CDF) for 6 metrics over 3 workloads}
	\label{fig:workload-cdfs}
\end{figure*}

We consider three representative workloads in a total of $\approx$ 0.62 million jobs, with $\approx$ 2 million stages and $\approx$ 0.12 billion instances. 
Each workload consists of productive jobs from a business department over five consecutive days in Alibaba. 
As shown in Fig.~\ref{tab:workload-stats}, workloads A, B, and C show various properties: 

\begin{enumerate}
	\item the total number of jobs (41K-0.4M), stages (0.1M-1.0M), and instances (34M-50M).
	\item the total size of collected traces (97GB-168GB)
	\item the average number of stages per job (2.4-5.0)
	\item the average number of instances per stage (85-1224).
	\item the average number of operators per stage (3.7-6.3)
	\item the average job latency (31-377s).
	\item the average stage latency (15-182s).
	\item the average instance latency (16-71s).
\end{enumerate} 

Workload A involves the most significant number of jobs (2-10x compared to B and C), while most are short-running jobs. 
Workload C has the least number of jobs but the longest job/stage/instance latencies, with its average job latencies 3-12.5x larger than A and B. It also has more instances on average, 12-14x times than A and B. 
Workload B has the most complex DAG topologies, with an average of 2x number of stages per job and an average of 1.2-1.7x number of operators per stage.

The statistics for data split in model training and evaluation are in Figure~\ref{fig:data-split}.
We show the cumulative distribution functions for the metrics over the three projects in Figure~\ref{fig:workload-cdfs}.

\subsection{Simulation Environment}
\label{appendix:simulator}

\begin{figure*}[t]
	\centering
    \includegraphics[height=8cm,width=.98\textwidth]{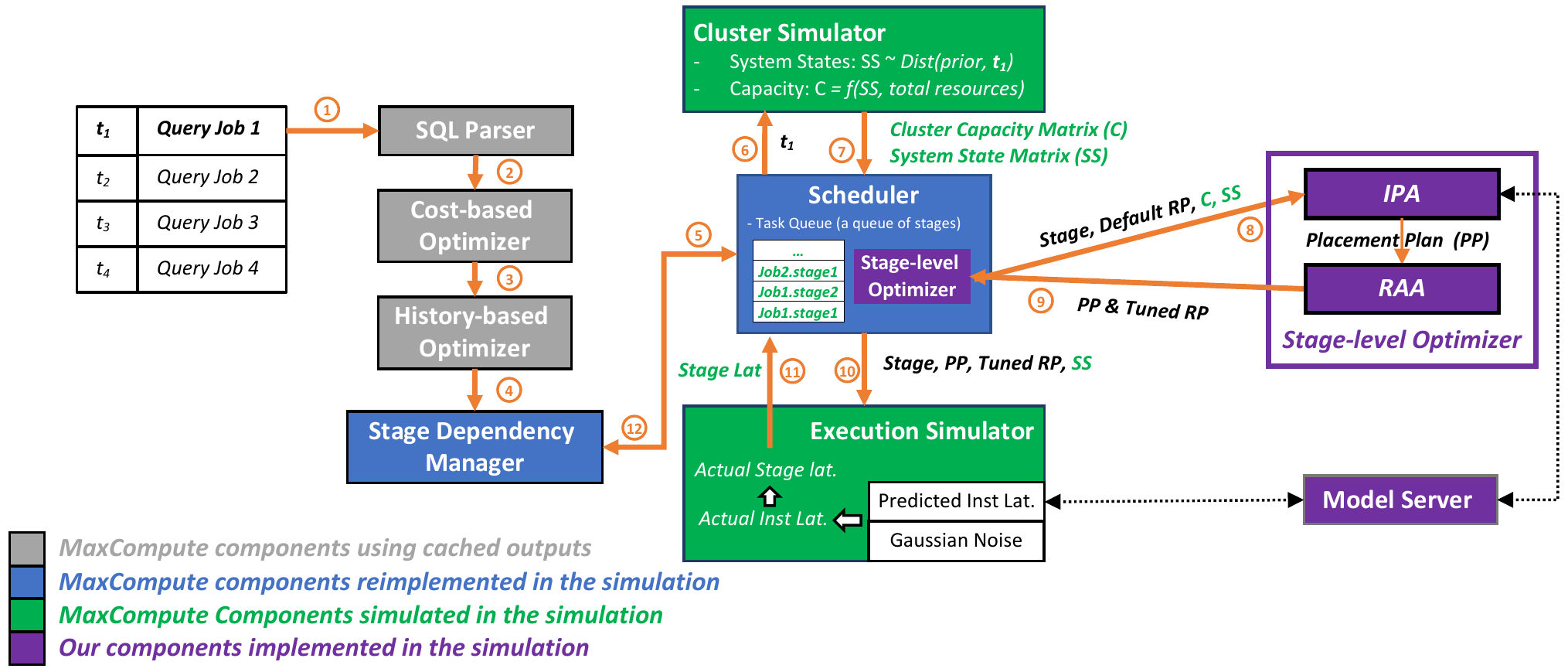}
    \caption{Simulation framework}
    \label{fig:simu-framework}
\end{figure*}

\minip{Hardware Property}
Our experiments are deployed over 3 machines, each with 2x Intel Xeon Platinum 8163 CPU with 24 physical cores, 500G RAM, and 8 GeForce GTX 2080 GPU cards.

\minip{Simulator}
Our simulator simulates the extended MaxCompute environment. 
First, to replace all the productive traces, it caches the query plans and instance meta information for the workload.
Second, it reimplements the stage dependency manager and an extensible scheduler that can support both Fuxi and SO.
Third, it generates the system states and capacity for each machine by sampling from prior knowledge. 
Lastly, it uses a Gaussian Processing Regression (GPR) model to sample the actual latency from a Gaussian distribution given the predicted latency. The GPR is pre-trained by learning the actual latency distribution against the latency predicted by a bootstrap model (by default, the MCI+GTN).
During simulation, the GPR model takes a predicted latency and generates a Gaussian distribution ($N(\mu,\sigma)$) as the output. Then our simulator will sample from the distribution within $\mu\pm3\sigma$ to simulate the actual latency. 
We show the diagram of the simulator in Figure~\ref{fig:simu-framework}.

\begin{figure}[t]
  \centering
	\begin{tabular}{c}
		\subfigure[\small{Department A}]
		{\label{fig:to-A}\includegraphics[width=0.98\linewidth]{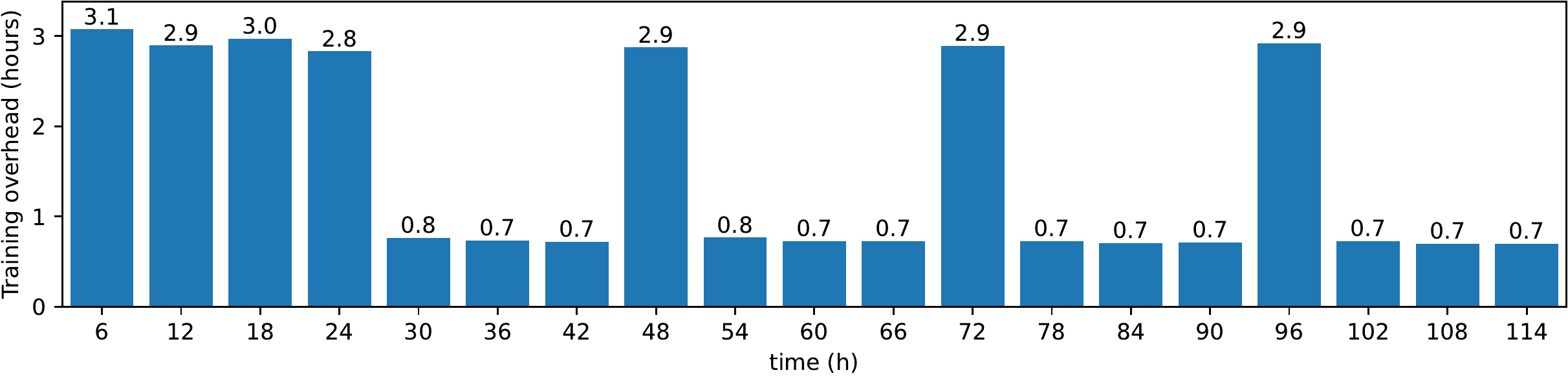}}
		\\
		\subfigure[\small{Department B}]
		{\label{fig:to-B}\includegraphics[width=0.98\linewidth]{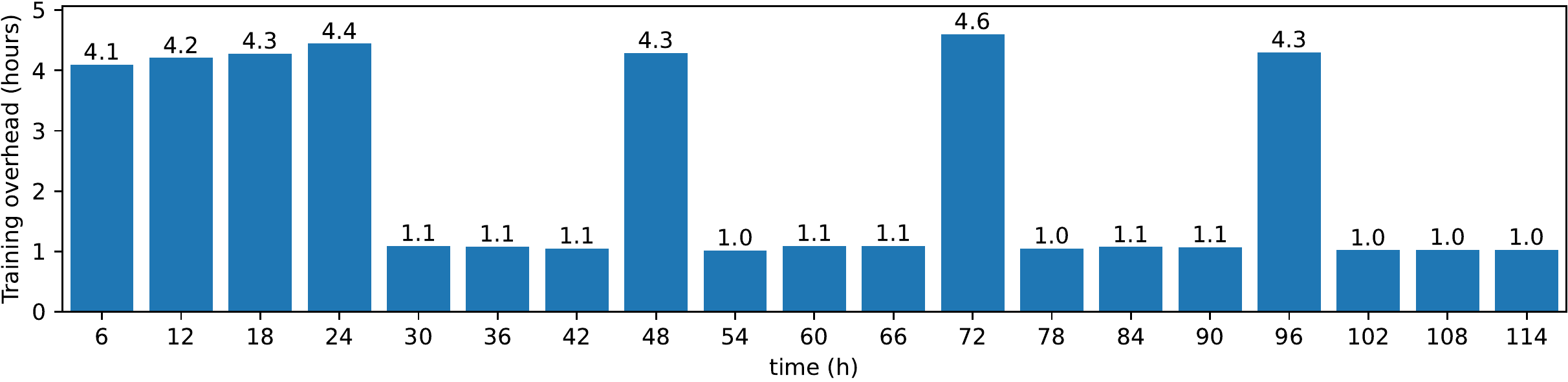}}
		\\		
		\subfigure[\small{Department C}]
		{\label{fig:to-C}\includegraphics[width=0.98\linewidth]{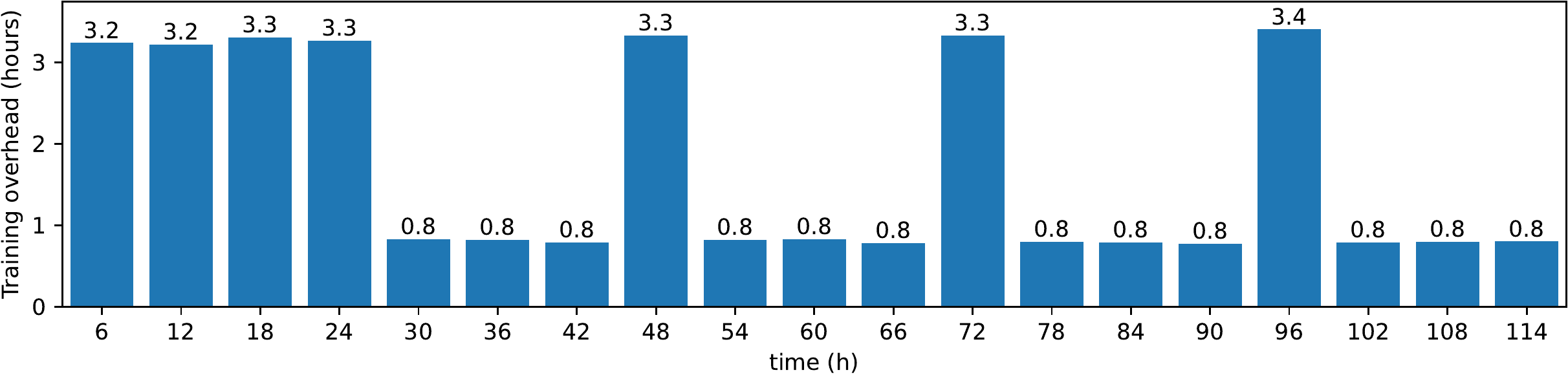}}
	\end{tabular}
  \captionof{figure}{Training overheads over 5 days}
  \label{fig:training-overhead}
\end{figure}

\subsection{Experimental Setup of Model Evaluation} 
\label{appendix:model-setup}

{\bf Data Split and Training for the Model Analyzing.} 
We follow the standard data split methods~\cite{petroni-etal-2021-kilt,imagenet-RussakovskyDSKS15} used in the modern big data era, where the validation and test data set are representative of a much smaller portion (40K-95K) of the total (35M-50M). 
Figure~\ref{fig:data-split} shows how we construct small and representative datasets for validation and testing over the three departments.
Notice that the number of instances in a stage has huge skewness as shown in Figure~\ref{fig:cdf-stage-inst-num-low}. 
To address the issue, we count the frequency of each DAG structure from all instances, and divide the DAG structures in each workload into four categories by their frequency.
We construct our test and validation set by applying different sampling strategies in each frequency type. 
For the high (>1K) and median (100-1K) type, we apply stratified sampling to pick 200 and 20 instances over each DAG structure; for the meidan-low (5-100) types, we randomly sample 20\% instances per structure; on the remaining low frequency type, we randomly sample 20\% over all instances.
Each of our workloads has separate 40K-95K instances for validation and testing. And we use the remaining instances for training. 
We follow the common rules for training and hyperparameter tuning in the machine learning literature.

\subsection{Training Strategies and Evaluation} \label{appendix:add-model-deploy}

\subsubsection{Training strategy.}
We employ two methods for training:
\begin{itemize}
	\item {\it Retraining:} We retrain the model every 6 hours on the first day when data is insufficient and switch to the 24-hour retraining schedule from the second day.
	\item {\it Incremental Updates:} We also provide incremental updates by fine-tuning the model every 6 hours from the second day.
\end{itemize}
Before deploying the model online, we pick 16 hyperparameters as candidates by exploring the hyperparameter space intensively via the grid-search approach. 
Once deployed, we train models over the 16 hyperparameters simultaneously on 16 GPUs for every model update.
The system switches to use the most recent model after a training or fine-tuning session finishes.

\subsubsection{Training overhead.}
\label{appendix:training-overhead}
Our experiments involve traces of 0.62M jobs, 2M stages, 0.12B instances, and 407GB raw data over five consecutive days in three departments.
The traces for training include:
\begin{itemize}
	\item query plans from CBO; 
	\item machine states and hardware information from the trace table at MaxCompute; 
	\item the instance-level input cardinality of each stage from an in-memory distributed storage system.
\end{itemize}

The data preparation stage costs 30 minutes to generate the data for training by running ETL queries over those traces at MaxCompute with a parallelism of 219.
Therefore, it costs $\sim$2 minutes to collect and update the training data every day for a department.

In the training stage, we execute the periodical 24h retraining (over up to 50M instances) at midnight each day when the workloads are light, and the incremental fine-tuning (over $\sim$2M instances) every 6 hours. 
Figure~\ref{fig:training-overhead} shows the training overhead of model updating over five consecutive days in three departments using 16 GPUs.
The average training overhead for retraining is 3.5 hours, and the average training overhead for fine-tuning is 0.9 hour.


\subsubsection{Modeling performance with the workload drift.}
\label{appendix:perf-over-workload-shift}

\begin{figure*}
  \centering
	\begin{tabular}{c}
		\subfigure[\small{Department A}]
		{\label{fig:hws-A}\includegraphics[width=0.98\linewidth,height=2.4cm]{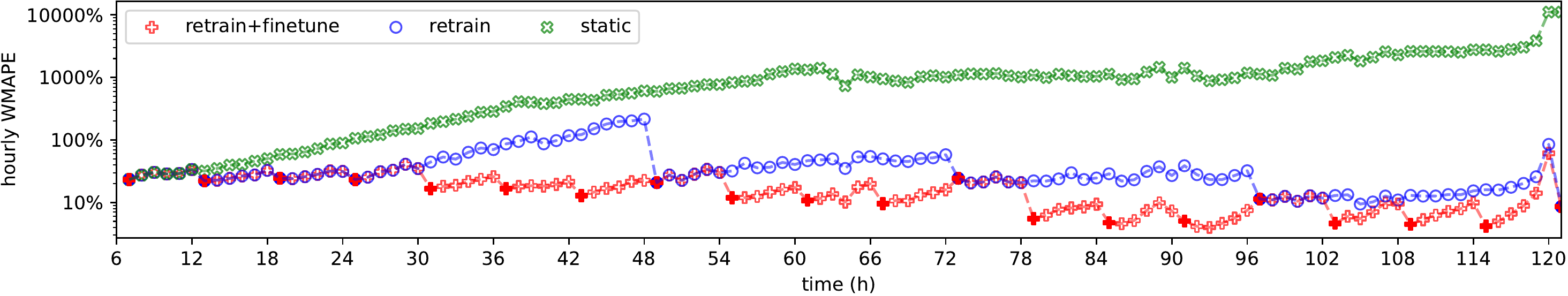}}
		\\
		\subfigure[\small{Department B}]
		{\label{fig:hws-B}\includegraphics[width=0.98\linewidth,height=2.4cm]{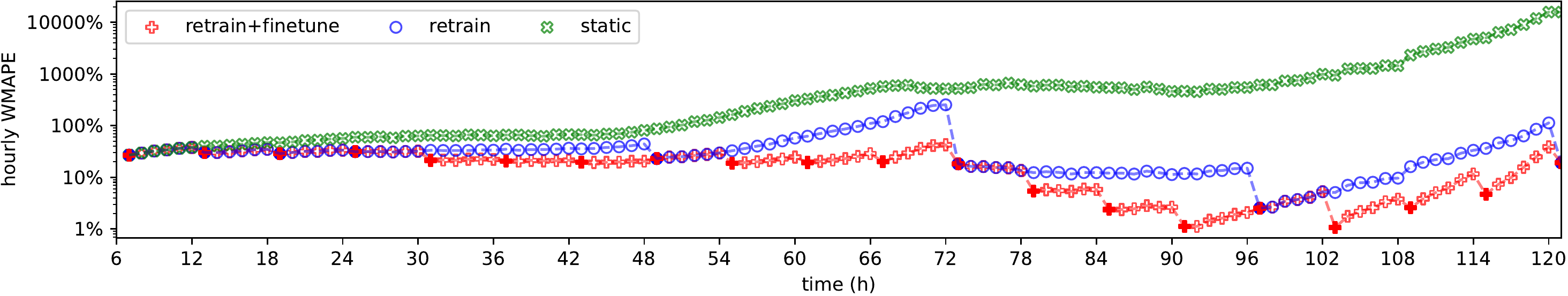}}
		\\		
		\subfigure[\small{Department C}]
		{\label{fig:hws-C}\includegraphics[width=0.98\linewidth,height=2.4cm]{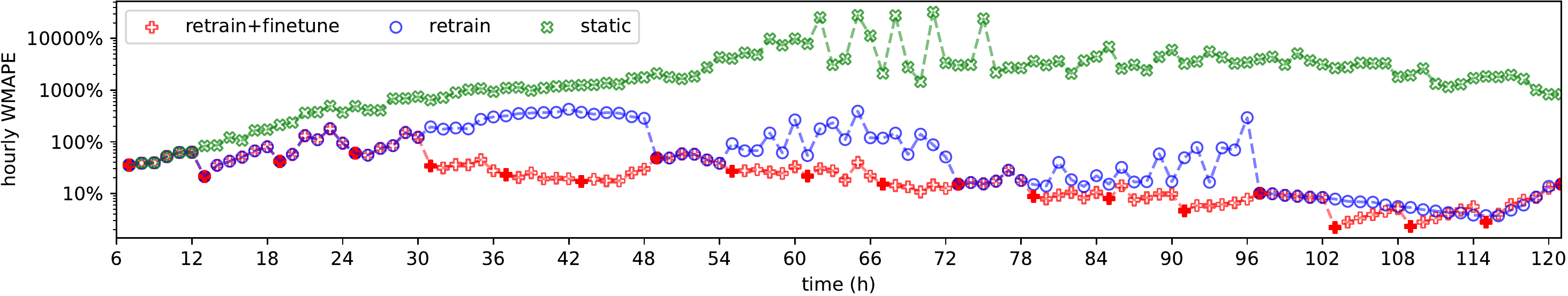}}
	\end{tabular}
  \captionof{figure}{Hourly WMAPE trends of three departments when injecting stages from long-running to short-running}
  \label{fig:hourly-wmape-slat}
\end{figure*}

\begin{figure*}
  \centering
	\begin{tabular}{c}
		\subfigure[\small{Department A}]
		{\label{fig:hwt-A}\includegraphics[width=0.98\linewidth,height=2.4cm]{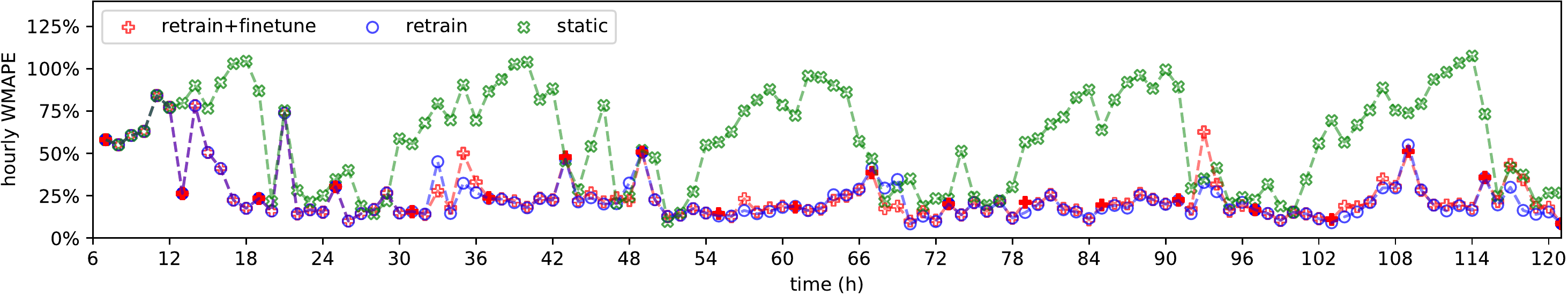}}
		\\
		\subfigure[\small{Department B}]
		{\label{fig:hwt-B}\includegraphics[width=0.98\linewidth,height=2.4cm]{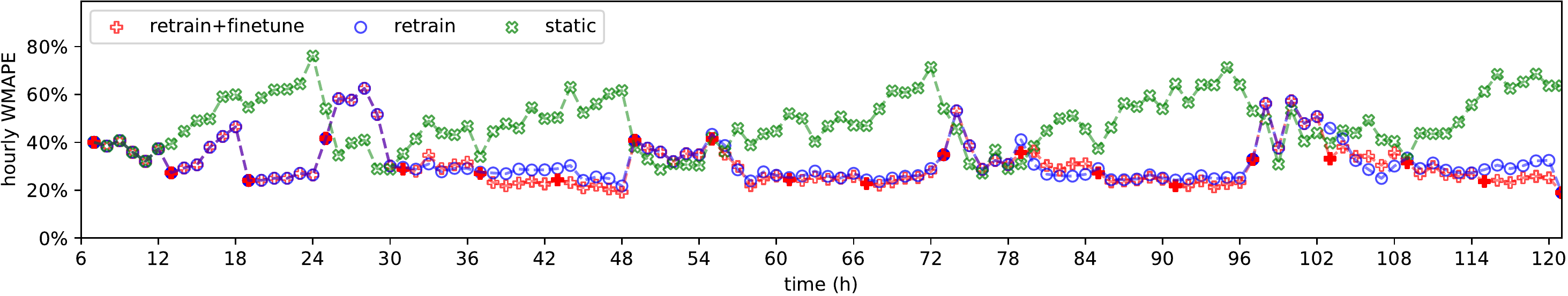}}
		\\		
		\subfigure[\small{Department C}]
		{\label{fig:hwt-C}\includegraphics[width=0.98\linewidth,height=2.4cm]{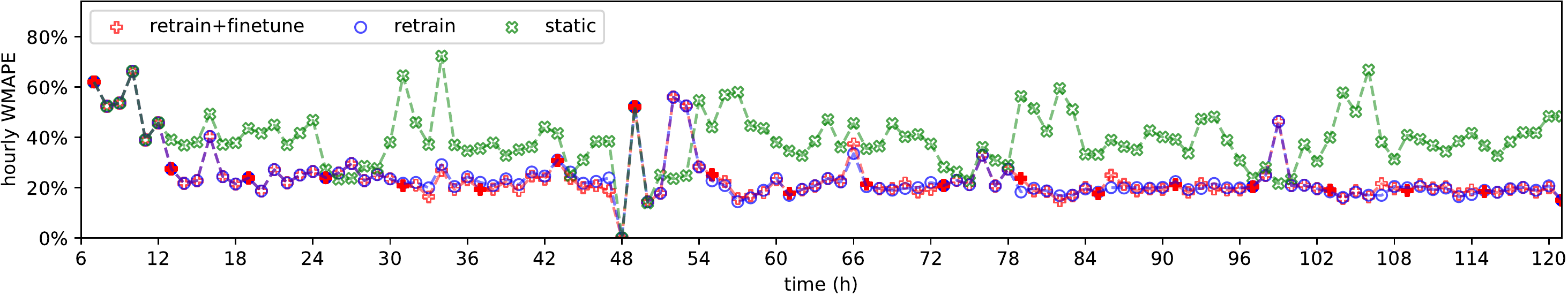}}
	\end{tabular}
  \captionof{figure}{Hourly WMAPE trends of three departments when injecting stages as the realistic setting}
  \label{fig:hourly-wmape-temporal}
\end{figure*}

\begin{figure*}[t]
	\begin{tabular}{ccc}
		\subfigure[\small{Department A}]
		{\label{fig:tws-A}\includegraphics[height=2.4cm]{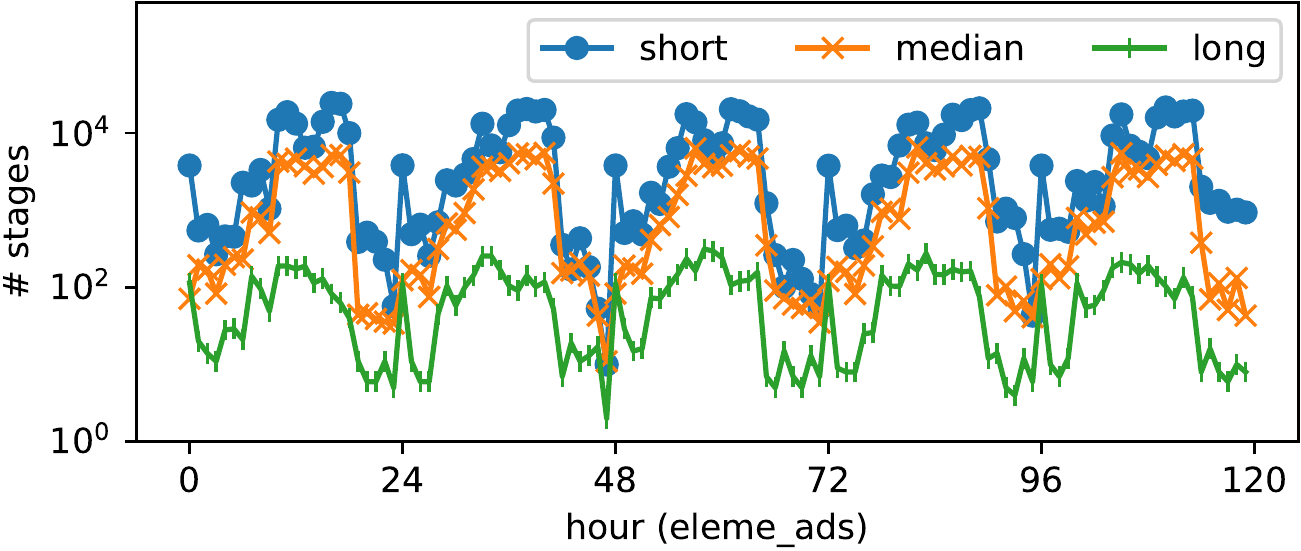}}
		&
		\subfigure[\small{Department B}]
		{\label{fig:tws-B}\includegraphics[height=2.4cm]{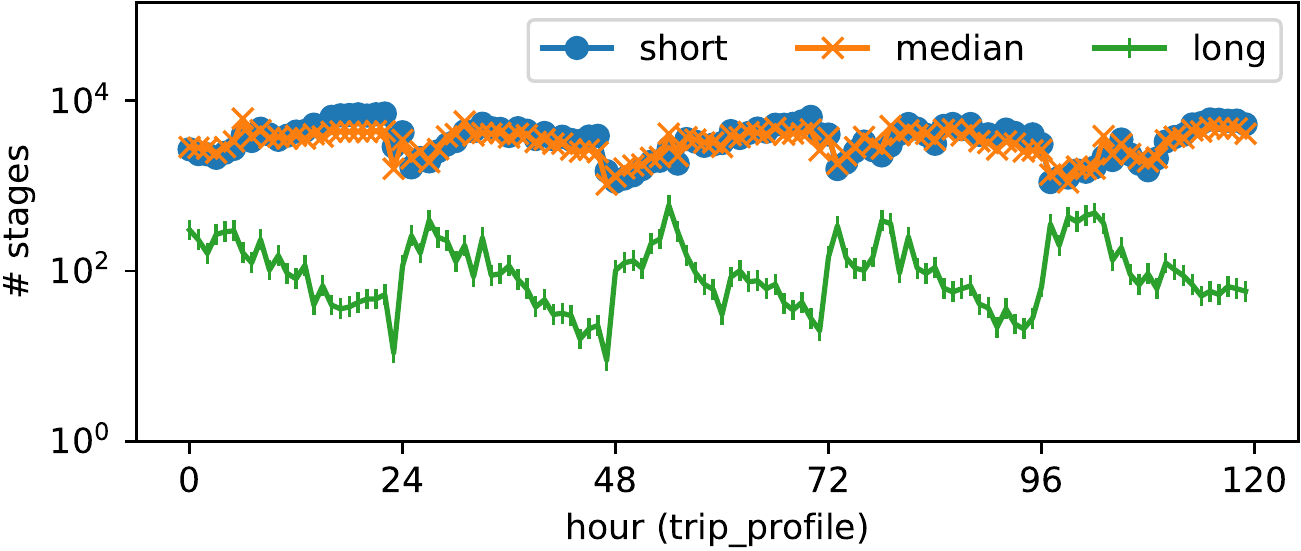}}
		&		
		\subfigure[\small{Department C}]
		{\label{fig:tws-C}\includegraphics[height=2.4cm]{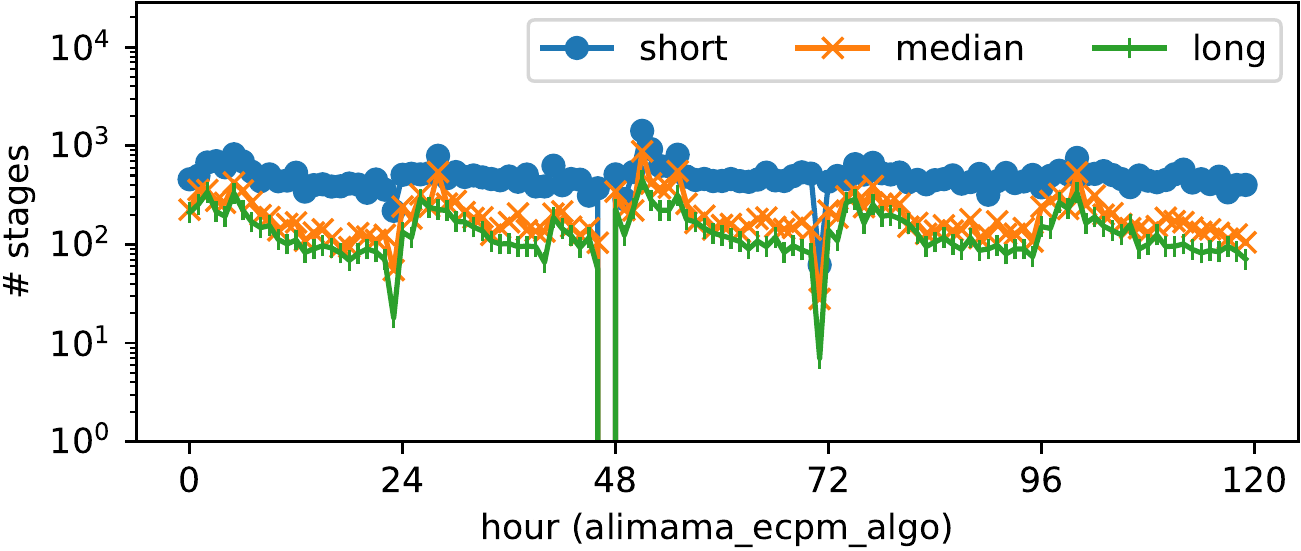}}
	\end{tabular}
	\caption{Hourly workload shifts}
	\label{fig:temporal-workload-shift}
\end{figure*}

To investigate the effects of the workload drift on the model performance, we design two experimental settings. One is the hypothetical worst, where we sorted stages in the descending order of the stage-level latency,  divided them evenly into 120 hourly (5 day) workloads, and injected them into the system in this order. The other is more realistic, where we injected jobs by obeying the job arriving order in the real world.
For each setting, we compare the following three training methods.
\begin{enumerate}
	\item The \textit{static} method: the model is trained over the data in the first 6 hours and never gets updated. Its WMAPE performance indicates the effect of workload drifts.
	\item The \textit{retrain} method: the model is retrained every 6 hours on the first day and gets periodically retrained every 24 hours from the second day.
	\item The \textit{retrain+finetune} method: the model has an incrementally fine-tuning every 6 hours from the second day in addition to the \texttt{retrain} method.
\end{enumerate}

\minip{Setting 1.} In the hypothetical worst case, the workload keeps drifting from the longest-running stage to the shortest-running stage over time. 
As shown in Fig.~\ref{fig:hourly-wmape-slat}, the~\textit{static} method shows the worst performance, with the WMAPE increasing up to 10000\%. 
The \textit{retrain} and \textit{retrain+finetune} methods both adapt to the workload drift, while the ~\textit{retrain+finetune} method shows a better adaptivity with more incremental updates. 
As shown in Figure~\ref{fig:hourly-wmape-slat}, the \texttt{retrain+finetune} method achieves an average of 6-10\% WMAPE on the last day, while the \texttt{retrain} method gets 7-17\%.

\minip{Setting 2.} In the realistic setting, the workload pattern is defined as the mix of stages with short-, median- and long-running latencies on an hourly basis. 
The drifts are not monotonic and more complex. As shown in Fig.~\ref{fig:temporal-workload-shift}, the drift patterns are mainly based on 24 hours and the differences between hours vary from time to time.
Therefore, if we only fine-tune every six hours, the \texttt{retrain+finetune} cannot adapt more to the hourly trend than the \texttt{retrain} method, as shown in Fig.~\ref{fig:hourly-wmape-temporal}. 
However, the \texttt{retrain} method can still catch the daily pattern and keep the average WMAPE low. For example, the hourly WMAPE in workload C mostly achieves as low as 15-25\% over the last day, which performs much better than the average WMAPE of 831-5083\% in the hypothetical worst setting.

\cut{
The workload drift is more complex and harder to adapt.
First, the performance of the \textit{static} model reflects a coarse periodical pattern in the workload drift. The WMAPE of the \textit{static} model works relatively well for the first 6 hours (28-38\% WMAPE) every day but performs badly (up to 100\%) around noon. However, the \texttt{retrain} and \texttt{retrain+refine} show better adaptivity and robust lower WMAPEs of 20-32\% and 20-31\%, respectively, in the last two days.
Secondly, the dramatic drifts have a big impact on the WMAPE. Figure~\ref{fig:hourly-wmape-temporal} shows the WMAPE jumps to high values in the first hours each day even though the model is just updated by retraining.
Figure~\ref{fig:temporal-workload-shift} shows the number of long-running, median-running, and short-running stages in each hourly workload. We observe a variety of stage-level latency distributions around the beginning of each day. 
Figure~\ref{fig:hourly-lat-cdf} shows a dramatic drift of the instance latency distribution during 42-51h. Notice that the CDFs of 49-51h are not as close as the previous hours. Therefore, our model failed to perform well with the assumption of smooth workload characteristics.
Lastly, when the workload drifts to certain relatively stable characteristics, the \texttt{retrain+finetune} still shows a better adaptivity than \texttt{retrain}, as shown during the 36-48h and 114h-120h in Figure~\ref{fig:hwt-B}.
}

\subsection{Modeling Targets} 

We group the targets of the latency measurements in big data systems into the following four categories.
\begin{enumerate}
	\item Single-instance Stage Latency (SiSL)
	\item Single-instance Operator Latency (SiOL)
	\item Multi-instance Stage Latency (MiSL)
	\item Multi-instance Operator Latency (MiOL)
\end{enumerate}

This paper aims at {\bf single-instance stage latency (SiSL)} as the modeling target. Besides, we model SiOL for the breakdown error profiling and MiSL for performance comparison. MiOL is the modeling target used by CLEO~\cite{cleo-sigmod20} to evaluate the exclusive cost of a query operator.

\begin{table}[t]
\centering
\small
\newrobustcmd{\B}{\bfseries}
\begin{tabular}{ccccccc}
\hline
\multicolumn{1}{l}{\B Target} & \B WL & \B WMAPE  & \B MdErr  & \B 95\%Err & \B Corr   & \B GlbErr \\ \hline
\multirow{3}{*}{\B SiSL} & A        & 8.6\%  & 7.4\%  & 62.4\%  & 96.6\% & 1.9\%  \\
                      & B        & 19.0\% & 15.1\% & 71.5\%  & 96.4\% & 5.4\%  \\
                      & C        & 15.1\% & 12.7\% & 97.3\%  & 98.4\% & 5.1\%  \\ \hline
\multirow{3}{*}{\B SiOL} & A        & 11.7\% & 6.2\%  & 52.4\%  & 98.6\% & 2.9\%  \\
                      & B        & 27.1\% & 2.3\%  & 46.5\%  & 95.5\% & 11.7\% \\
                      & C        & 18.1\% & 8.7\%  & 65.2\%  & 98.4\% & 6.5\%  \\ \hline
\multirow{3}{*}{\B ACT}  & A                 & 6.6\%          & 5.9\%          & 38.0\%           & 99.6\%        & 1.6\%           \\
                         & B                 & 14.5\%         & 9.2\%          & 40.6\%           & 96.6\%        & 2.0\%           \\
                         & C                 & 14.7\%         & 12.9\%         & 737.3\%          & 98.5\%        & 4.8\%           \\ \hline                
\multirow{3}{*}{\B ACT*} & A & 6.3\%  & 5.3\% & 60.2\%  & 90.1\% & 0.9\% \\
                     & B & 12.5\% & 7.8\% & 36.8\%  & 97.3\% & 0.9\% \\
                     & C & 10.9\% & 9.7\% & 942.0\% & 98.9\% & 3.6\% \\ \hline
\end{tabular}	
\caption{Modeling Performance for SiSL, SiOL, and ACTs}
\label{tab:model-performance-all}
\end{table}

\subsection{Breakdown of Model Errors}\label{appendix:expt-breakdown-analyses}

\begin{figure*}[t]
	\centering
	\vspace{-0.1in}
	\hspace{-6cm}
	
	\begin{tabular}{lcc}

		\subfigure[\small{7 operators make up 95\% of WMAPE in A}]
		{\label{fig:e2e-ipa-metrics}\includegraphics[height=3.0cm,width=5.5cm]{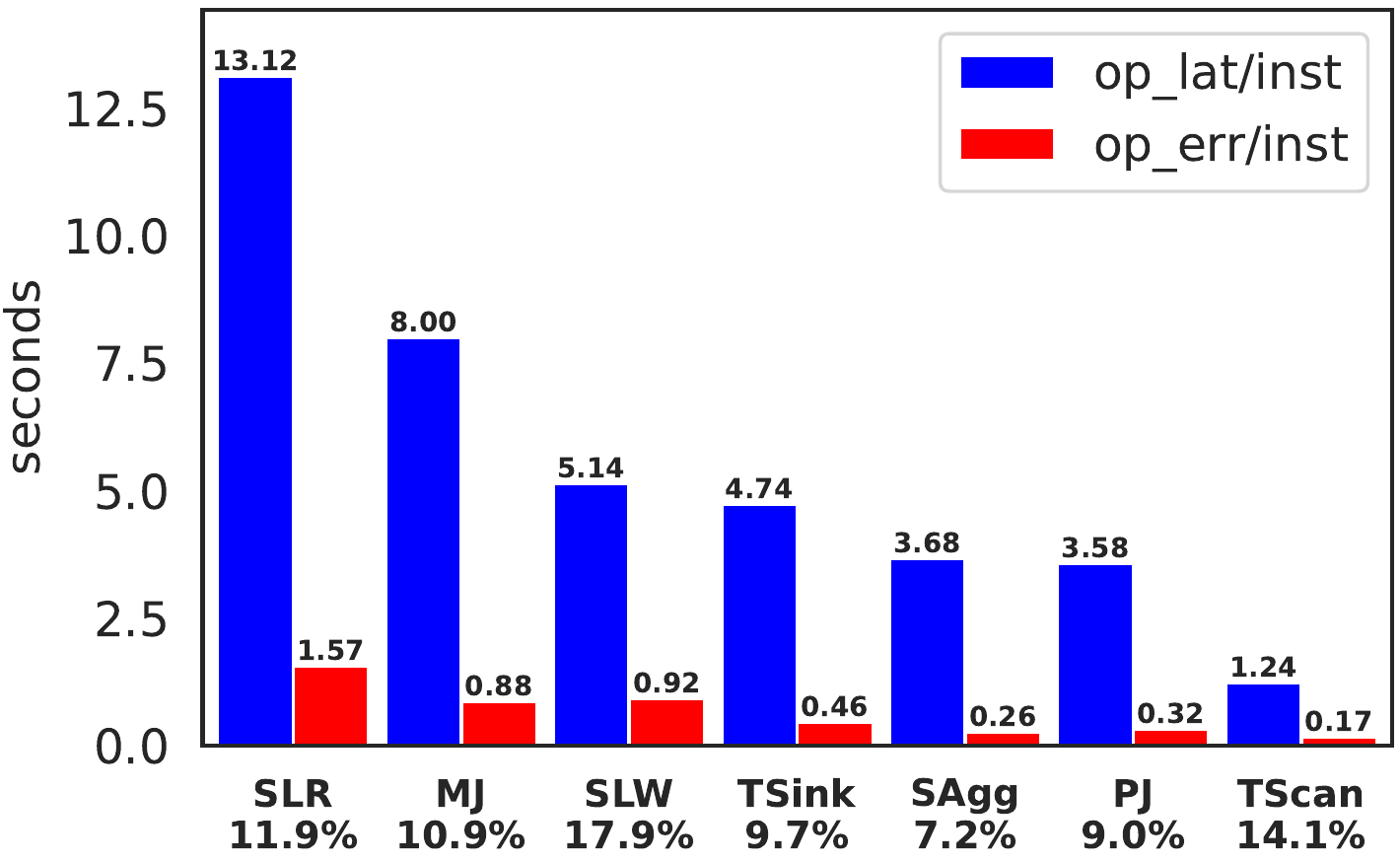}}

		&
		\subfigure[\small{6 operators make up 95\% of WMAPE in B}]
		{\label{fig:e2e-ipa-busy-idle}\includegraphics[height=3.0cm,width=5.5cm]{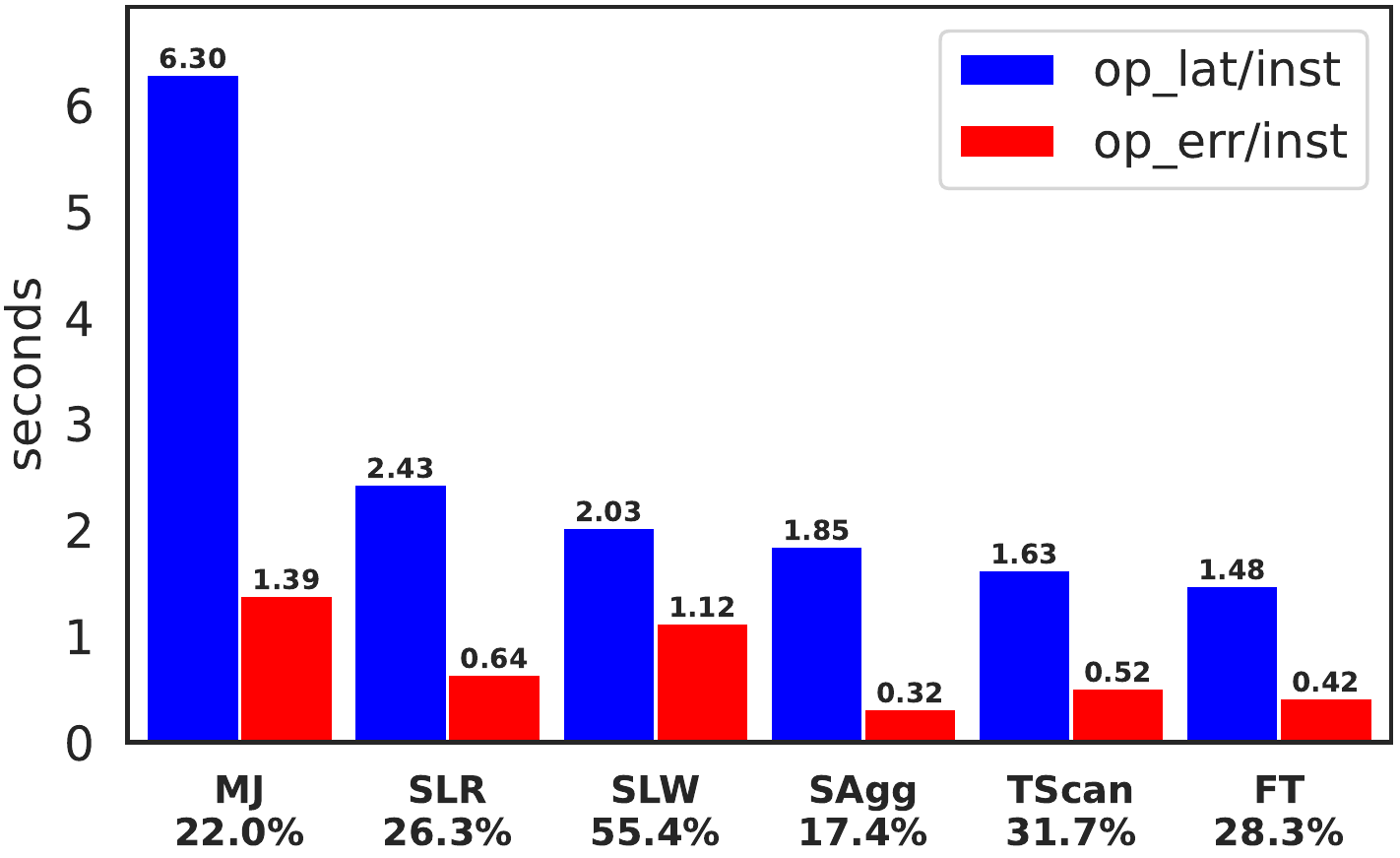}}

		&		
		\subfigure[\small{9 operators make up 95\% of WMAPE in C}]
		{\label{fig:e2e-ipa-fuxi-solving-time}\includegraphics[height=3.0cm,width=5.5cm]{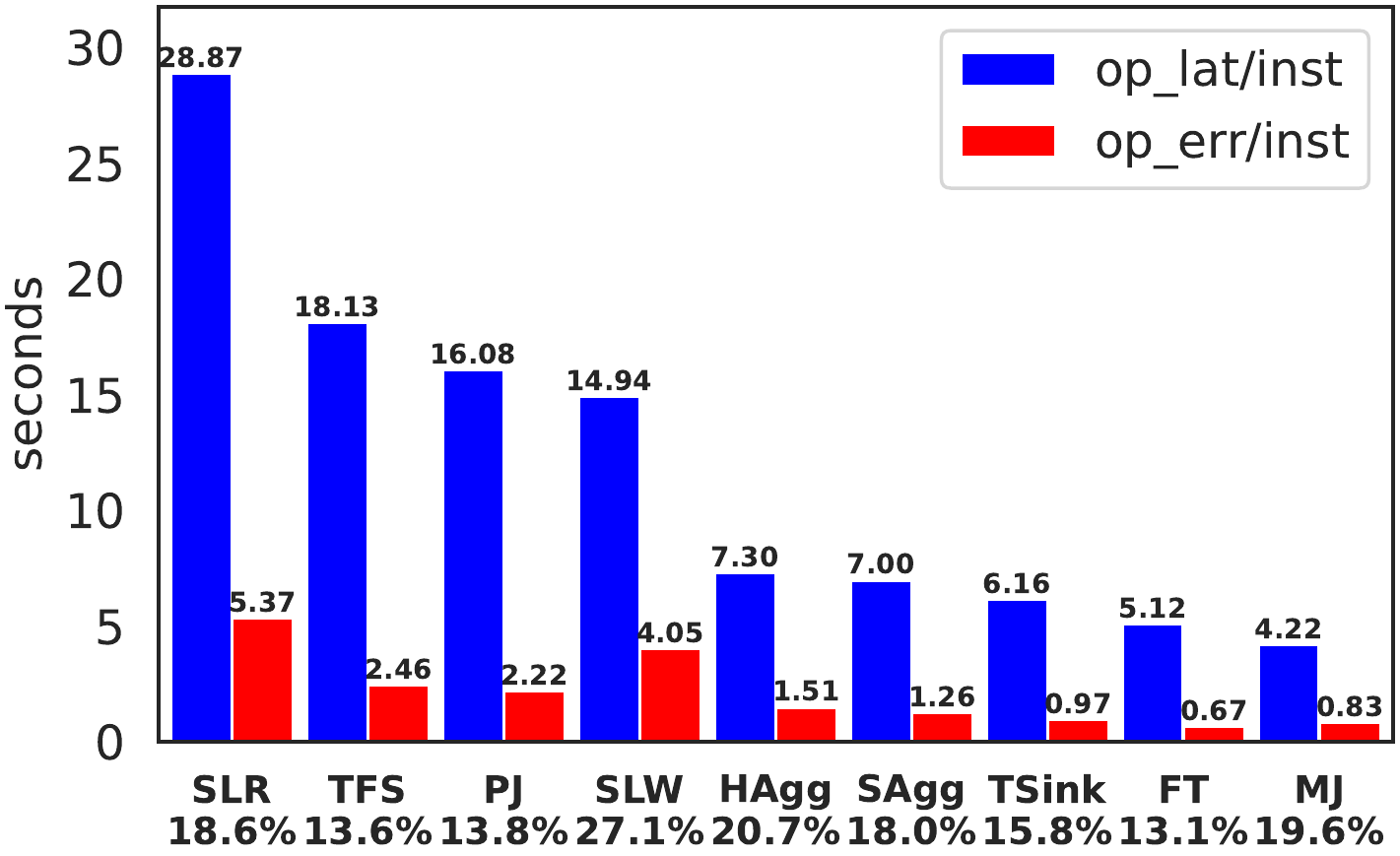}}
		
	\end{tabular}
	\caption{SiOL breakdown analyses in the test sets of the 3 workloads}
	\label{fig:siol-breakdown}

\end{figure*}

Our breakdown analyses show that the errors are mainly from the IO-intensive operators (due to lack of features to capture IO/network activities in the current system) and the dynamic system states.
\begin{enumerate}
	\item we train a separate model for SiOL by directly feeding GTN's node embedding features to the MLP. Our tuned model achieves 12-27\% WMAPEs and 2-8\% MdErrs on SiOL over the three workloads, as shown in Table~\ref{tab:model-performance-all}. 
	\item we average every operator's latency and error across all instances, and observe that IO-intensive operators make up 59-84\% of the WMAPE in the three workloads.
\begin{itemize}
	\item Workload A: 9.4\%/11.7\%
	\item Workload B: 21\%/27.1\%
	\item Workload C: 11\%/18.1\%
\end{itemize}
 	\item Figure~\ref{fig:siol-breakdown} shows that the top-k error contributed operators make up at least 95\% of the entire SiOL WMAPE in each workload. We sort the operators based on their average latencies in an instance and show the error rates of each operator in the x-axis ticks. We observe
\begin{itemize}
	\item IO-intensive operators take up most of the average latencies in an instance and account for most errors.
	\item IO-intensive operators have more significant individual error rates. E.g., the \verb|StreamLineWrite(SLW)| operator doing shuffle writing has the largest error rate in all workloads, with 17.9\% in A, 55.4\% in B, and 27.1\% in C.
\end{itemize}
	\item 
	We look into 127 instances from the same stage with an \verb|SLW| (StreamLineRead) operator in workload B. Besides the same query plan and resource plan, those instances run over machines with the same hardware types and similar system states with similar input row numbers for the \verb|SLW|; however, the time costs for SLWs vary from 34s to 117s.  
	\item When we switch to a CPU-oriented metric {\bf actual CPU time (ACT)} that measures the actual CPU cost in the lifetime of an instance, the WMAPE and MdErr reduce to 7-15\% and 6-13\%, respectively, as shown in Table~\ref{tab:model-performance-all}.
	\item When we augment the system states by adding the average system states during the lifetime of an instance (which is not available before running), our model towards actual cpu time {\bf (ACT*)} achieves additional improvement, with its WMAPE and MdErr reduced to 6-12\% and 5-10\%.
\end{enumerate}

\subsection{Impact of Discretized System States}\label{appendix:expt-modeling-ss}

\begin{figure*}[t]
	\centering
	\vspace{-0.1in}
	\begin{tabular}{lcc}

		\subfigure[\small{workload A}]
		{\label{fig:ss-dd-a}\includegraphics[height=3.0cm,width=5.5cm]{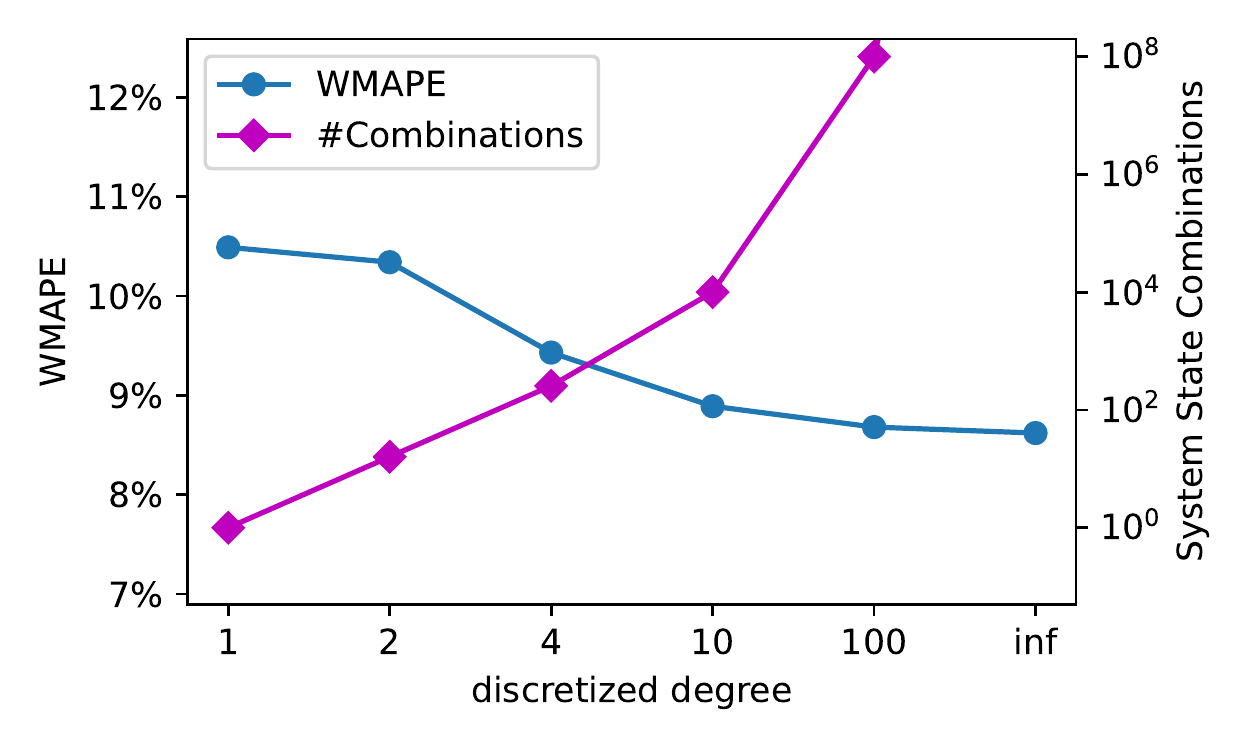}}

		&
		\subfigure[\small{workload B}]
		{\label{fig:ss-dd-b}\includegraphics[height=3.0cm,width=5.5cm]{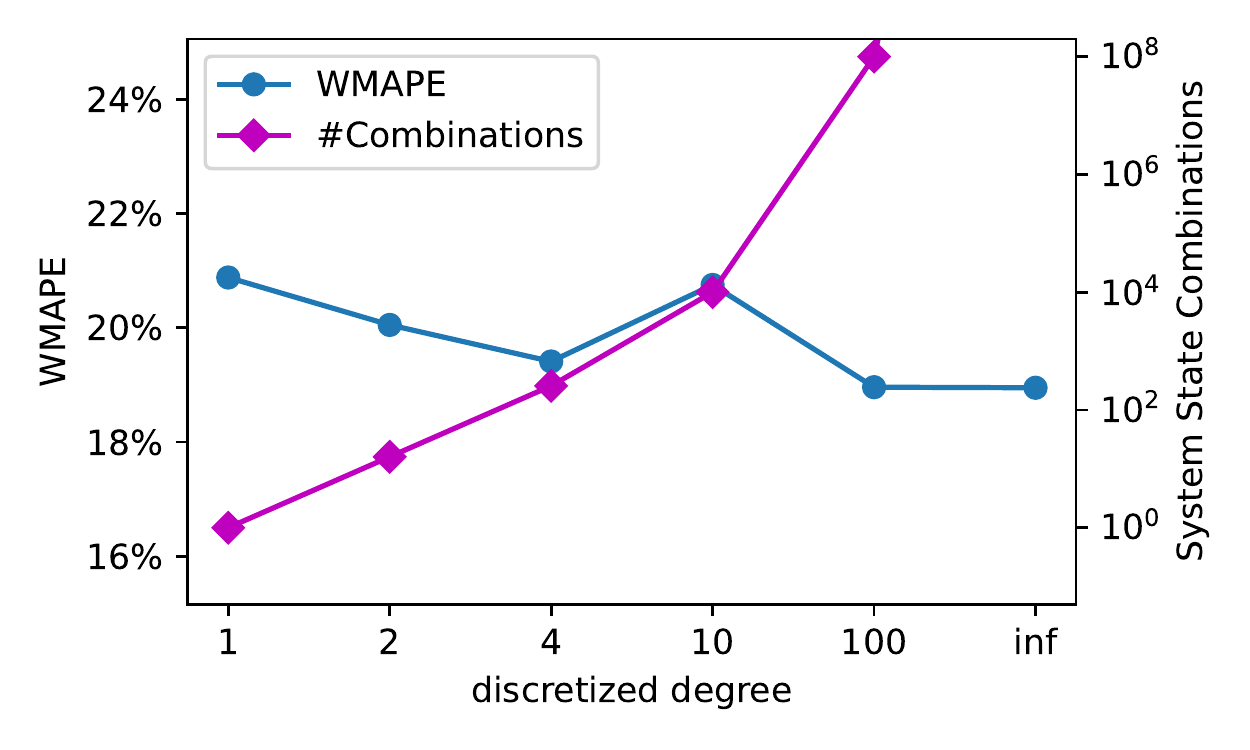}}

		&		
		\subfigure[\small{workload C}]
		{\label{fig:ss-dd-c}\includegraphics[height=3.0cm,width=5.5cm]{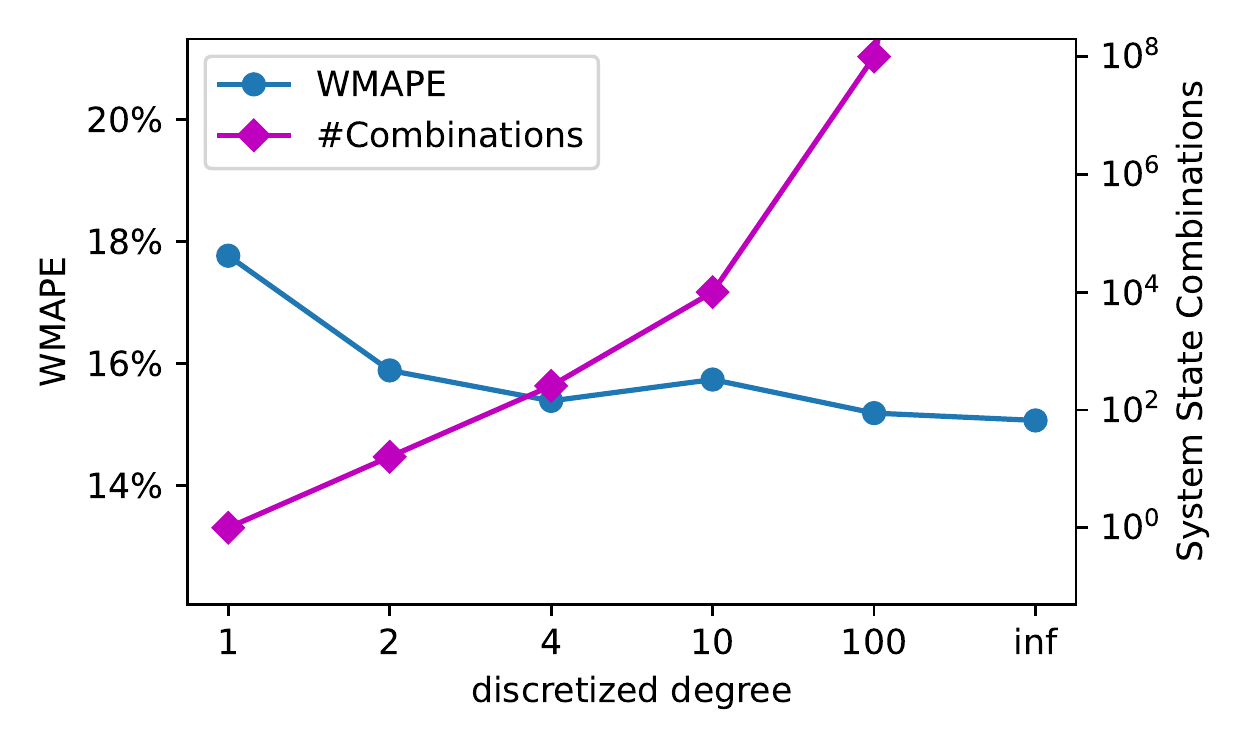}}
		
	\end{tabular}
	\caption{Tradeoffs between WMAPE and the num. of system state combinations}
	\label{fig:ss-dd}
\end{figure*}

\underline{Additional Expt 1:} {\it Impact of Discretized System States}.
We train separate models over system states with different discretized degrees (DD) and compare the model performances against optimization overhead in IPA. 
Figure~\ref{fig:ss-dd} shows the examples in three workloads. When increasing DD, WMAPE converges to a certain range, while the complexity of the resource optimization problem is exponentially increased. 
WMAPE increases when DD is increased from 4 to 10 in Figures~\ref{fig:ss-dd-b} and \ref{fig:ss-dd-c} due to the model overfitting.
We choose DD=10 for workload A, and DD=4 for both B and C when doing resource optimization by exploring the DD for each workload. 

\subsection{Comparison of Modeling Targets}\label{appendix:expt-modeling-cleo}

\underline{Additional Expt 2:} {\it Comparison of Modeling Targets}.
CLEO~\cite{cleo-sigmod20} learns a model for \textit{end-to-end query operator latency} across multiple instances. While the CLEO code is not publicly available and we lack source data to train query operator latency, we construct an indirect comparison via modeling the \textit{end-to-end stage latency across multiple instances} against our \textit{instance-level latency} to illustrate the issues of CLEO's  choice of the modeling target.
For the \textit{end-to-end stage latency}, our best method can achieve only 37-54\% WMAPE, much worse than 8.6-19\% WMAPE of our \textit{instance latency model},  showing that it is more challenging to predict end-to-end latency due to the high variability across different instances.
The modeling performance of the end-to-end stage latency is in Table~\ref{tab:model-performance-cleo}.

\begin{table}[t]
\centering
\small
\newrobustcmd{\B}{\bfseries}
\begin{tabular}{ccccccc}
\hline
\multicolumn{1}{l}{\B Target} & \B WL & \B WMAPE  & \B MdErr  & \B 95\%Err & \B Corr   & \B GlbErr \\ \hline
\multirow{3}{*}{MiSL} & A & 36.7\% & 25.5\% & 122.1\% & 90.1\% & 22.4\% \\
                      & B & 47.4\% & 29.6\% & 129.2\% & 84.2\% & 36.3\% \\
                      & C & 53.8\% & 39.2\% & 123.9\% & 63.7\% & 52.7\% \\ \hline
\end{tabular}	
\caption{Modeling Performance for MiSL}
\label{tab:model-performance-cleo}
\end{table}


\begin{table*}
\ra{0.7}
\addtolength{\tabcolsep}{-2pt}
\newrobustcmd{\B}{\bfseries}
\newrobustcmd{\BR}{\color{red}}
\newrobustcmd{\BL}{\color{blue}}
\begin{tabular}{@{}lrclrclrclrclrclrclrcl@{}}\toprule
& \multicolumn{3}{c}{Coverage}
& \multicolumn{3}{c}{$\overline{Lat_{stage}^{(in)}}$ (s)} &
 \multicolumn{3}{c}{$\overline{Cost_{stage}}$ (0.001\$)} & \multicolumn{3}{c}{$\overline{T_{stage}}$ (ms)} & \multicolumn{3}{c}{$\max(T_{stage})$ (ms)}\\
\cmidrule{2-4} \cmidrule{5-7} \cmidrule{8-10} \cmidrule{11-13} \cmidrule{14-16}
SO choice & A & B & C & A & B & C & A & B & C & A & B & C & A & B & C \\ \midrule
\verb|FUXI| &100\% &100\% &100\% & 9.5&34&22 & 50&27&234 & 2&1&3 & 8&2&4 \\ \midrule
\verb|IPA(Org)| &100\% &100\% &100\%  &  8.5&  28&  13	& 	49&  24&  206	& 	280&  17&  1.4K	& 	2.0K&  18&  1.8K\\
\B \texttt{IPA(Cluster)} &\B100\% & \B100\% &\B100\% & \B 8.4&\B29&\B14  &  \B50&\B25&\B208  & \B 15 & \B 10 &\B 30  & \B 24 & \B 10 & \B 36 \\
\midrule 
\verb|IPA+RAA(W/O_C)| &100\% &100\% &100\% &  8.7&7&11  &  36&6&60  & 2.5K & 177 & 3.5K & 19K & 220 & 6.3K \\
\verb|IPA+RAA(DBSCAN)| &100\% &100\% &100\% &  7.0&11&10  &  40&10&62  & 223   & 132 & 258  & 937 & 136 & 452 \\
\verb|IPA+RAA(General)| &100\% &100\% &100\% &  6.3&7&7  &  38&7&59  & 100 & 20  & 167   & 241 & 23 & 229  \\
\B \texttt{IPA+RAA(Path)} &\B100\% &\B100\% &\B100\% &  \B6.3&\B7&\B7  &  \B38&\B7&\B59  & \B98    & \B17  & \B156  & \B 226 & \B 18 & \B 224 \\ \midrule
\BR \texttt{EVO} & \BR   0\%& \BR 82\%& \BR 0\%    & \BR   --& \BR 38& \BR --  & \BR   --& \BR 3& \BR -- & \BR   --& \BR 21K& \BR -- & \BR   --& \BR 24K& \BR --\\
\BR{\texttt{WS(Sample)}}& \BR   90\%& \BR 85\%& \BR 82\%        & \BR   19& \BR 17& \BR 20        & \BR   5& \BR 4& \BR 6 & \BR   7.4K& \BR 465& \BR 9.8K & \BR   22K& \BR 753& \BR 12K\\
\BR{\texttt{PF(MOGD)}} & \BR   99\%& \BR 100\%& \BR 98\%       & \BR   10& \BR 18& \BR 9 & \BR   15& \BR 12& \BR 46      & \BR   2.7K& \BR 2.1K& \BR 1.5K        & \BR   4.0K& \BR 2.5K& \BR 2.3K\\ 
\midrule
\BL{\texttt{IPA+EVO}} & \BL   98\%& \BL 100\%& \BL 98\%       & \BL   11& \BL 11& \BL 13        & \BL   9& \BL 6& \BL 39        & \BL   3.0K& \BL 2.4K& \BL 5.0K        & \BL   7.3K& \BL 2.6K& \BL 5.9K\\
\BL{\texttt{IPA+WS(Sample)}} & \BL   100\%& \BL 100\%& \BL 100\%     & \BL   13& \BL 8& \BL 23 & \BL   66& \BL 7& \BL 156      & \BL   3.5K& \BL 517& \BL 12K        & \BL   10K& \BL 741& \BL 18K \\
\BL{\texttt{IPA+PF(MOGD)}} & \BL   100\%& \BL 100\%& \BL 100\%     & \BL   9.3& \BL 17& \BL 8  & \BL   42& \BL 12& \BL 59      & \BL    1.6K& \BL 1.2K& \BL 1.2K        & \BL   2.5K& \BL 1.6K& \BL 2.2K\\
\bottomrule
\end{tabular}
\caption{Metrics of simulating the 29 subworkloads}
\label{tab:expt-so-sub-workloads-raw}
\end{table*}

\subsection{Subworkloads for RAA Evaluation}\label{appendix:expt-sub-workloads}

We record the trend of CPU utilization every day and average them over five consecutive days, where we slide a 40 min time window to find periods of the highest and the lowest average CPU utilization for each workload.

We subsample jobs on the two time periods for every workload on each of the five consecutive days to form the 29 subworkloads (workload C submitted 0 jobs during its idle period on the second day), including 12K jobs, 51K stages, and 6M instances.

\cut{
We then dive into 29 sub-workloads, subsampled in different system busy and idle periods. 
Fig.~\ref{fig:ipa-1} shows that IPA reduces more stage latencies and \rv{cloud cost} for workload C (with more long-running instances,) than workload B (median-length stages) and workload A (short-running stages).
Fig.~\ref{fig:ipa-2} shows that  IPA helps reduce more latency and cost when a cluster has an idle environment because it can assign more machines of high capacity to a stage and hence reduce the stage latency more.
}
\subsection{More Results for Expt 6}\label{appendix:expt-ipa}

\begin{figure*}[t]
	\centering
	\vspace{-0.1in}
	\hspace{-6cm}
	
	\begin{tabular}{lcc}

		\subfigure[\small{reduction rates of $Lat_{stage}^{(ex)}$, $Lat_{stage}^{(in)}$, $Cost_{stage}$}]
		{\label{fig:ipa-appr-metrics}\includegraphics[height=3.0cm,width=5.8cm]{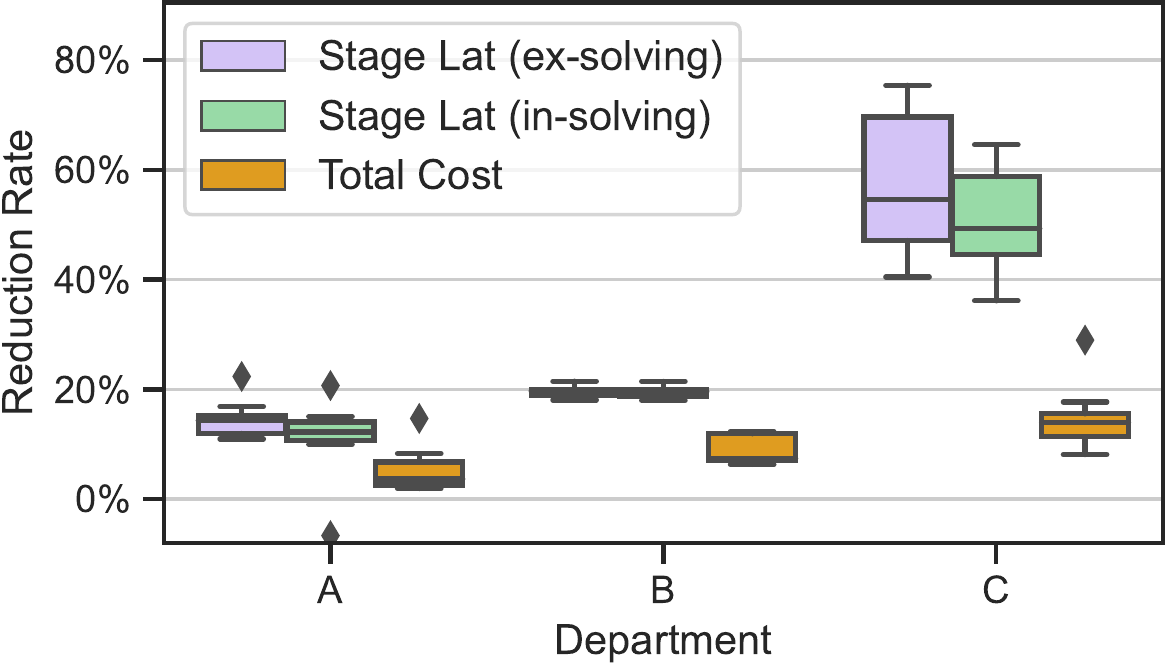}}

		&
		\subfigure[\small{reduction rates of $Lat_{stage}^{(ex)}$ on busy/idle}]
		{\label{fig:ipa-appr-busy-idle}\includegraphics[height=3.0cm,width=5.8cm]{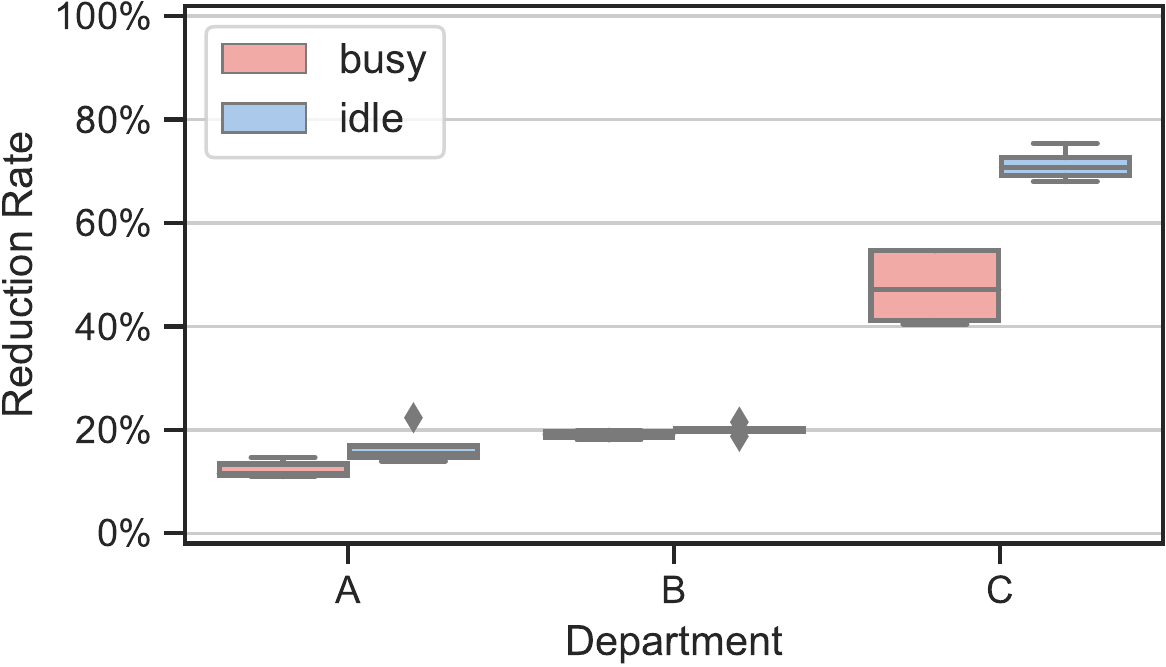}}

		&		
		\subfigure[\small{$T_{stage}$ for IPA(Cluster)}]
		{\label{fig:ipa-appr-solving-time}\includegraphics[height=3.0cm,width=5cm]{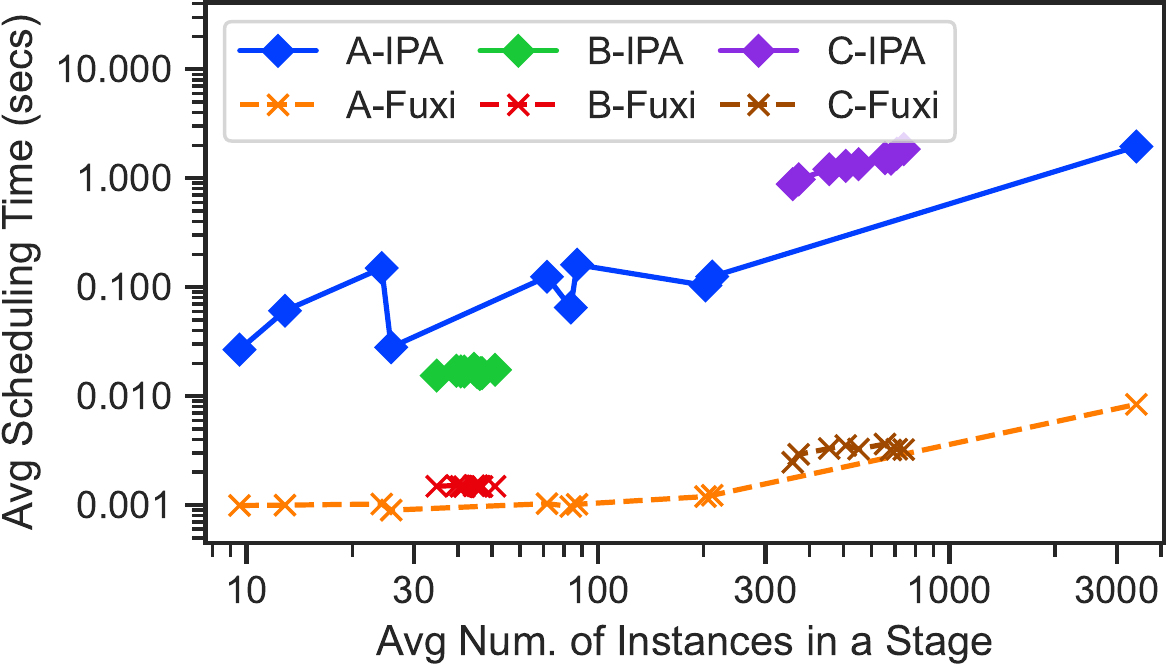}}
		
	\end{tabular}
	\caption{Comparative results between IPA(Org) and Fuxi over 29 workloads}
	\label{fig:ipa-appr}
\end{figure*}

\begin{figure*}[t]

	\centering
	
	\begin{tabular}{lccc}

		\subfigure[\small{reduction rates of $Lat_{stage}^{(ex)}$, $Lat_{stage}^{(in)}$, $Cost_{stage}$}]
		{\label{fig:ipa-cluster-metrics}\includegraphics[height=3.0cm,width=5.8cm]{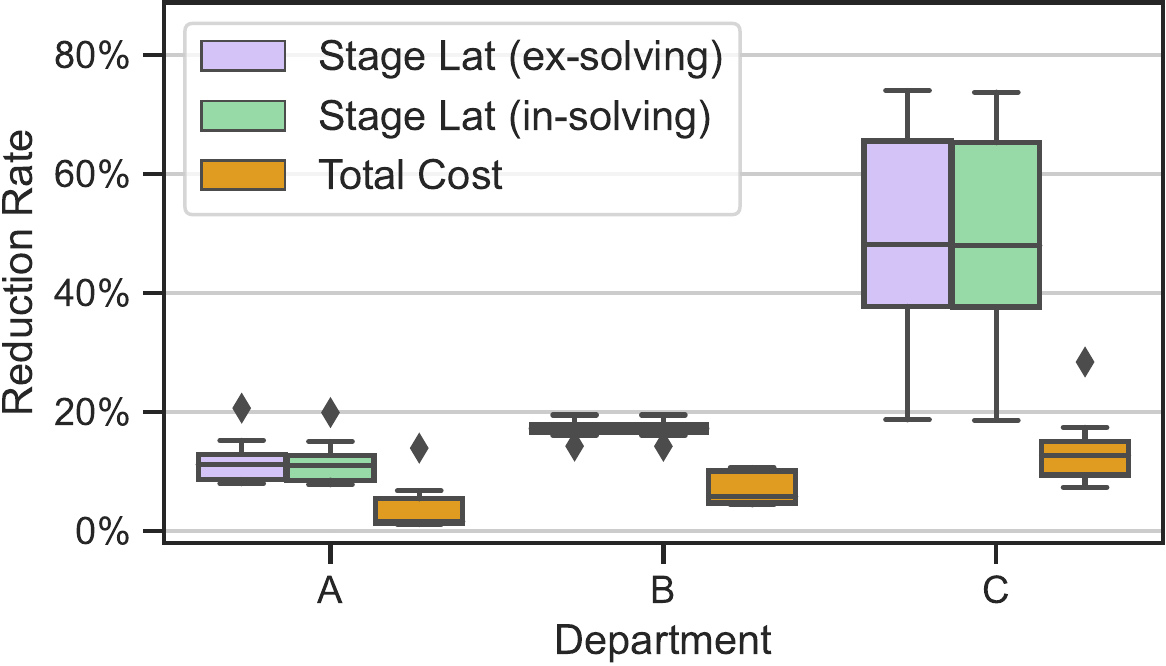}}

		&
		\subfigure[\small{reduction rates of $Lat_{stage}^{(ex)}$ on busy/idle}]
		{\label{fig:ipa-cluster-busy-idle}\includegraphics[height=3.0cm,width=5.8cm]{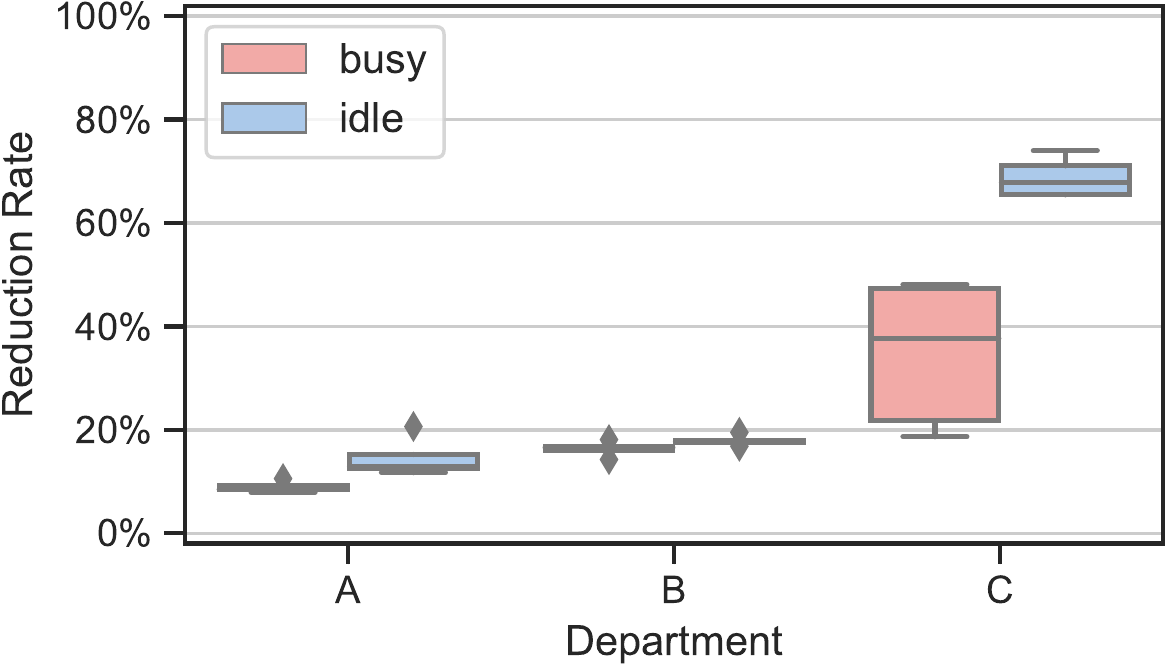}}

		&		
		\subfigure[\small{$T_{stage}$ for IPA(Cluster)}]
		{\label{fig:ipa-cluster-solving-time}\includegraphics[height=3.0cm,width=5cm]{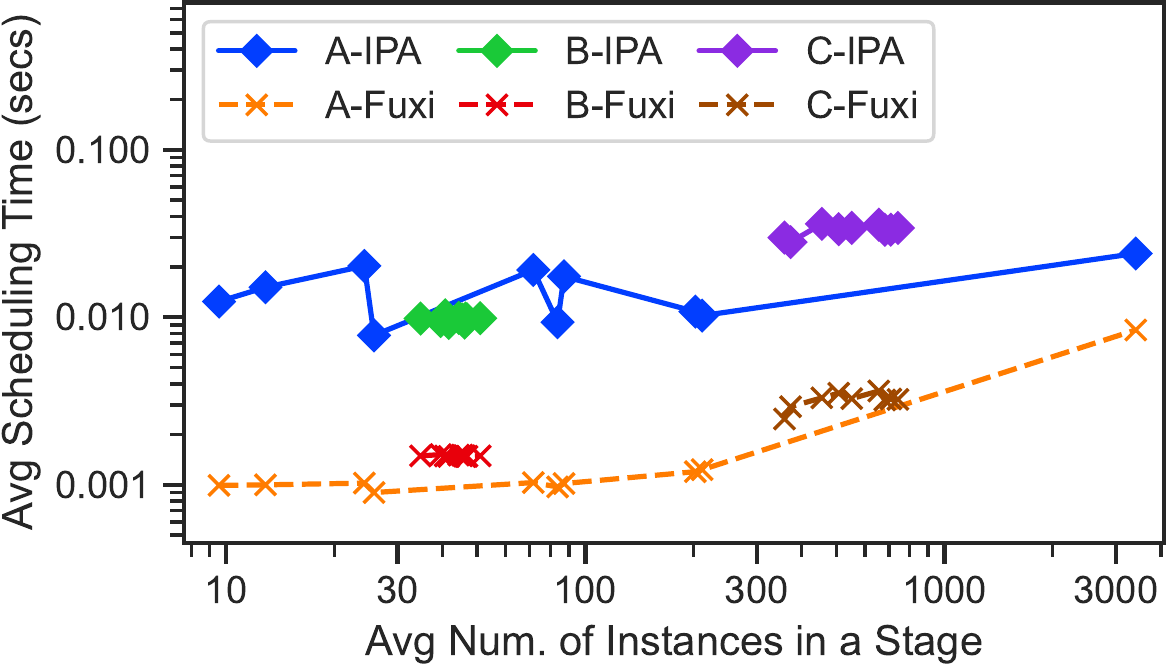}}

	\end{tabular}
	\vspace{-0.2in}
	\caption{Comparative results between IPA(Cluster) and Fuxi over 29 workloads}
	\label{fig:ipa-cluster}
\end{figure*}

We show more results on (1) stage-level latency reduction rates (both excluding and including the RO time overhead), (2) cloud cost reduction rates, (3) IPA time overhead for 
\verb|IPA(Org)| and \verb|IPA(Cluster)| in Figure~\ref{fig:ipa-appr} and Figure~\ref{fig:ipa-cluster}, respectively.

Fig.~\ref{fig:ipa-cluster-metrics} shows that IPA reduces more stage latencies and cloud cost for workload C (with more long-running instances,) than workload B (median-length stages) and workload A (short-running stages).
Fig.~\ref{fig:ipa-cluster-busy-idle} shows that  IPA helps reduce more latency and cost when stages are submitted to an idle environment because the cluster can assign more machines of high capacity to a stage.

We also show the average stage latency, cloud cost and solving time in Table~\ref{tab:expt-so-sub-workloads-raw}.

\subsection{More Results for Expt 7}\label{appendix:expt-ipa+raa}

\begin{figure*}[t]
	\centering
	\vspace{-0.1in}
	\hspace{-6cm}
	
	\begin{tabular}{lcc}

		\subfigure[\small{reduction rates of $Lat_{stage}^{(ex)}$, $Lat_{stage}^{(in)}$, $Cost_{stage}$}]
		{\label{fig:e2e-so(ic_off)-metrics}\includegraphics[height=3.0cm,width=5.8cm]{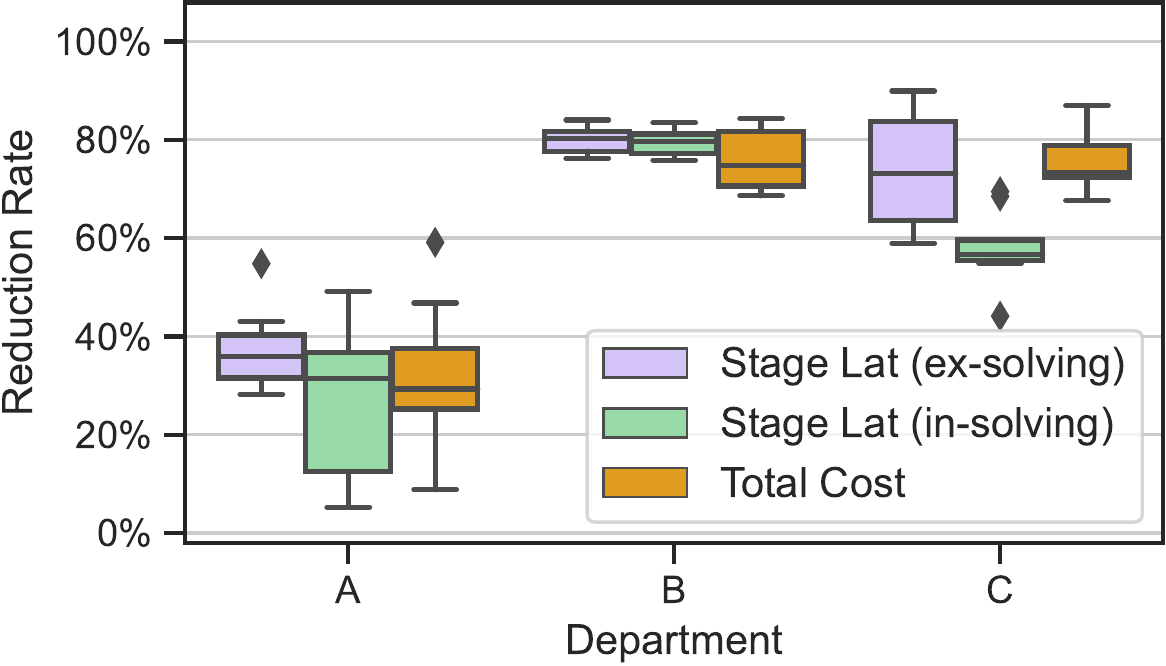}}

		&
		\subfigure[\small{reduction rates of $Lat_{stage}^{(ex)}$ on busy/idle}]
		{\label{fig:e2e-so(ic_off)-busy-idle}\includegraphics[height=3.0cm,width=5.8cm]{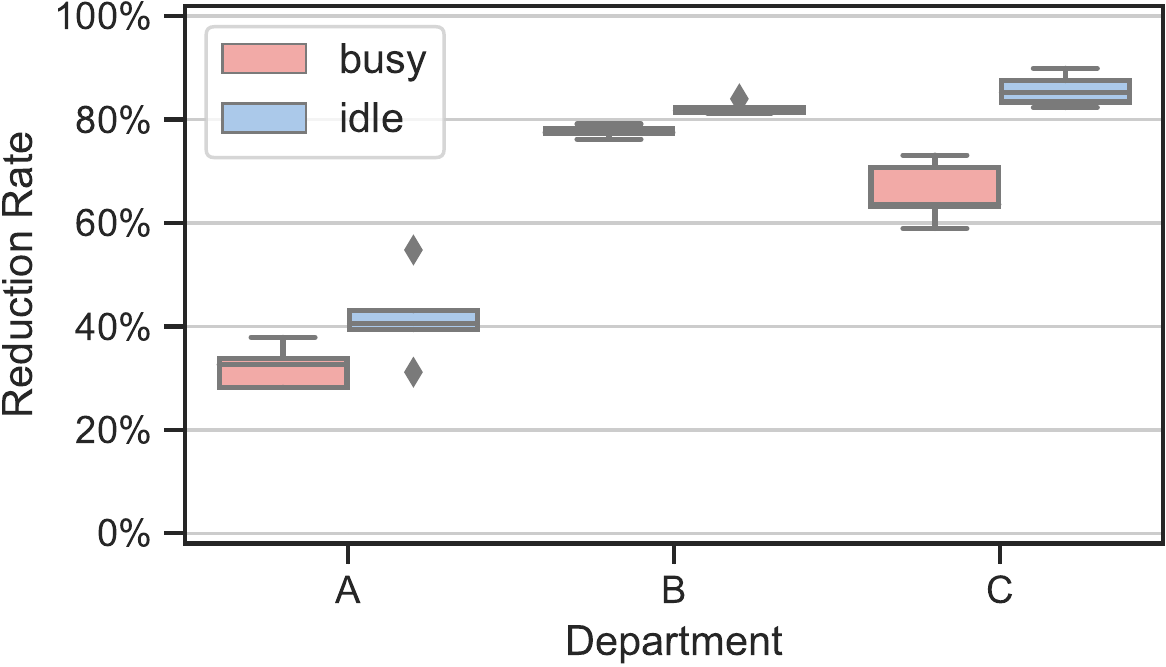}}

		&		
		\subfigure[\small{$T_{stage}$ for IPA(Cluster)}]
		{\label{fig:e2e-so(ic_off)-fuxi-solving-time}\includegraphics[height=3.0cm,width=5.cm]{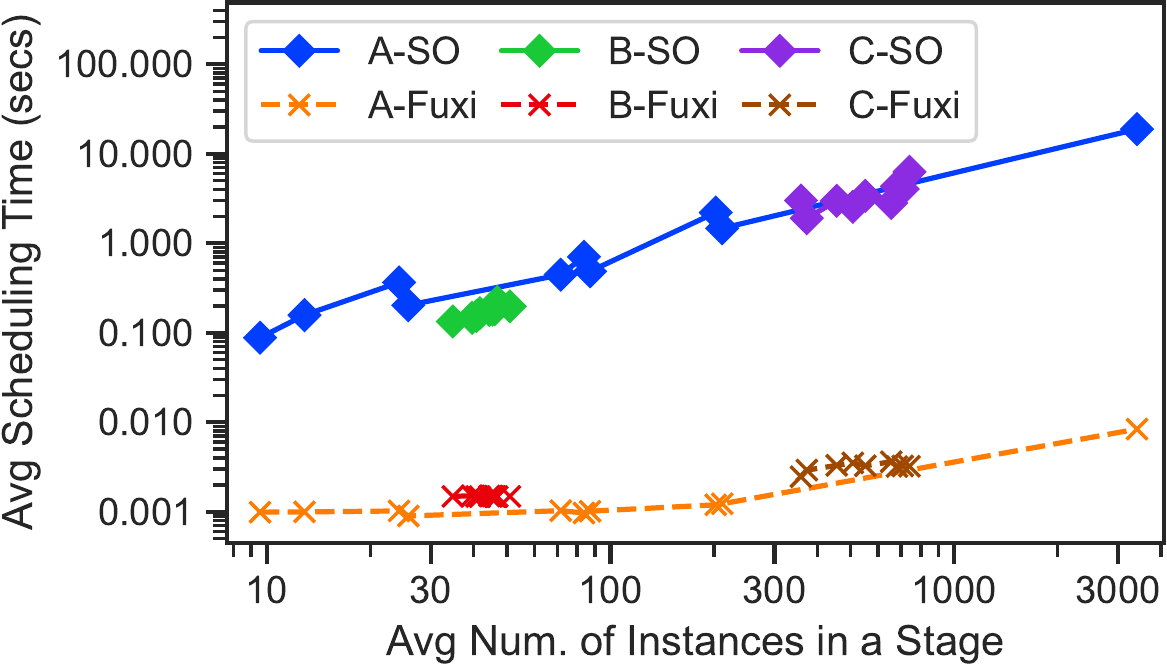}}
		
	\end{tabular}
	\caption{Comparative results between IPA+RAA(W/O\_C) and Fuxi over 29 workloads}
	\label{fig:ipa+raa_ic_o}
\end{figure*}

\begin{figure*}[t]
	\centering
	\vspace{-0.1in}
	\hspace{-6cm}
	
	\begin{tabular}{lcc}

		\subfigure[\small{reduction rates of $Lat_{stage}^{(ex)}$, $Lat_{stage}^{(in)}$, $Cost_{stage}$}]
		{\label{fig:e2e-so_DBSCAN_-metrics}\includegraphics[height=3.0cm,width=5.8cm]{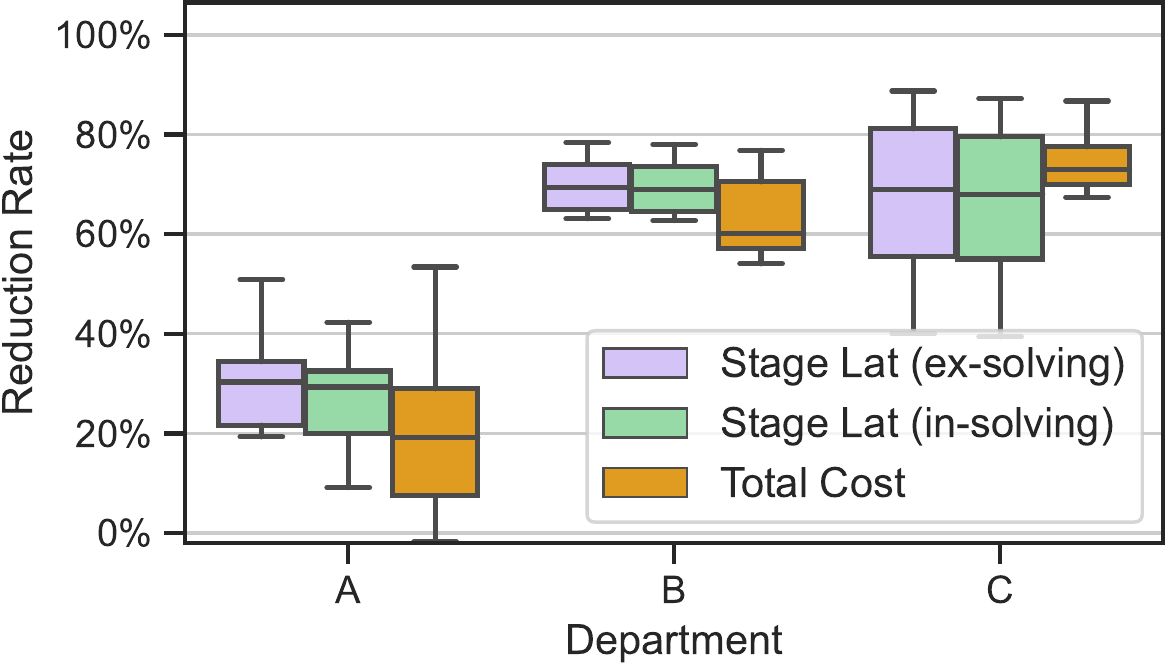}}

		&
		\subfigure[\small{reduction rates of $Lat_{stage}^{(ex)}$ on busy/idle}]
		{\label{fig:e2e-so_DBSCAN_-busy-idle}\includegraphics[height=3.0cm,width=5.8cm]{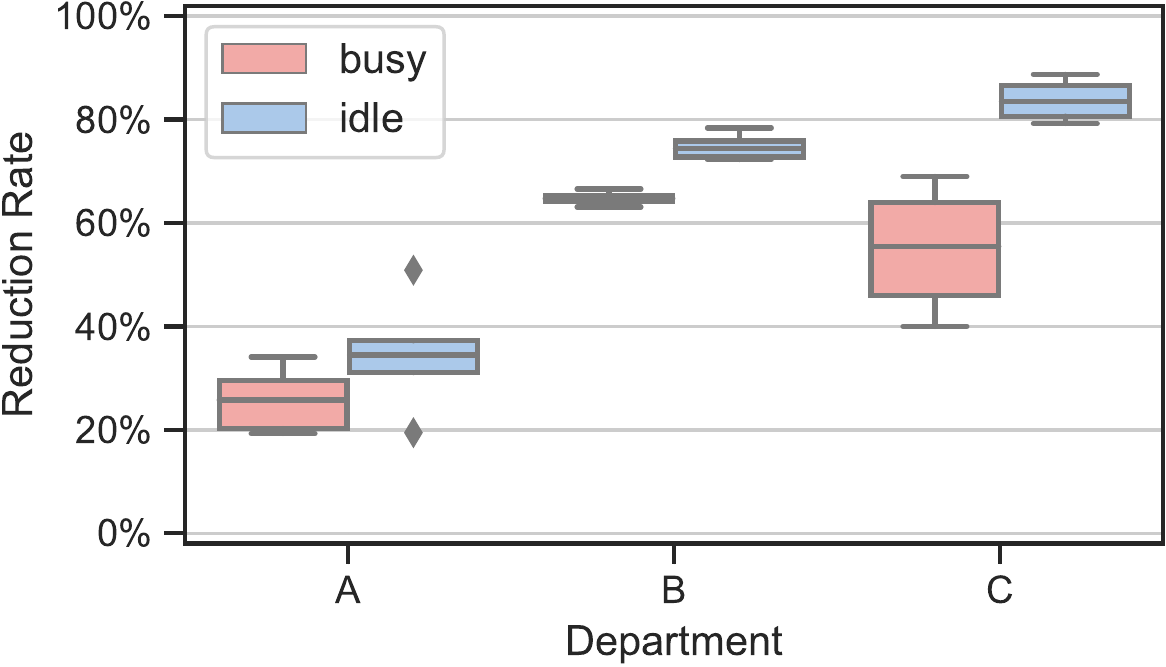}}

		&		
		\subfigure[\small{$T_{stage}$ for IPA(Cluster)}]
		{\label{fig:e2e-so_DBSCAN_-fuxi-solving-time}\includegraphics[height=3.0cm,width=5cm]{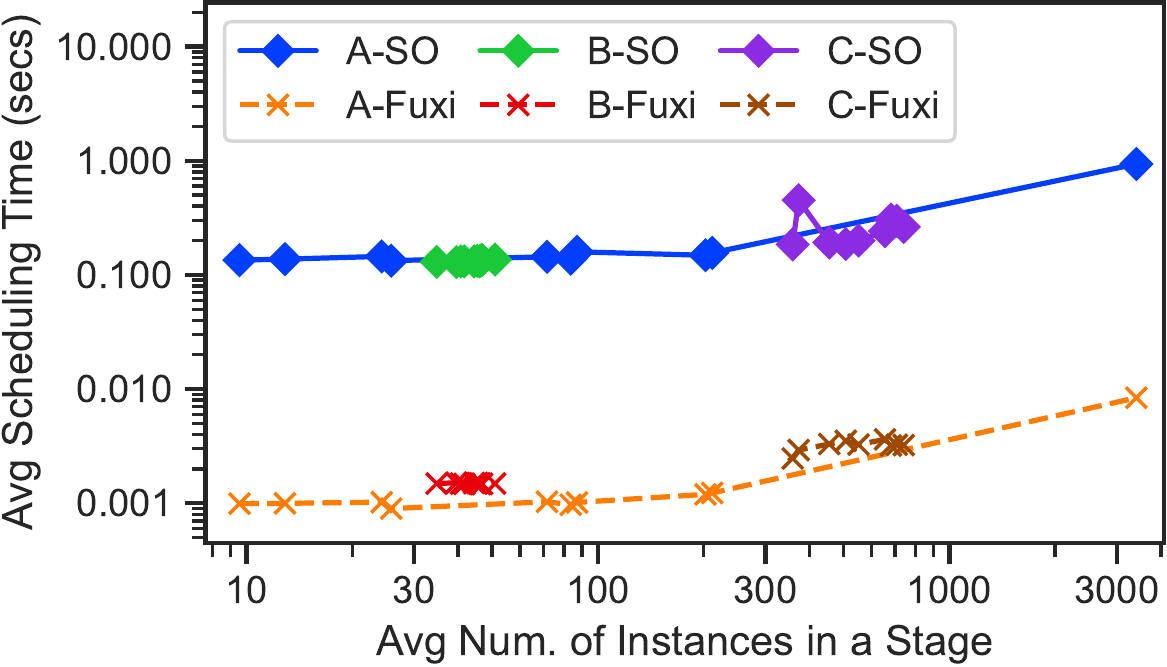}}
		
	\end{tabular}
	\caption{Comparative results between IPA+RAA(DBSCAN) and Fuxi over 29 workloads}
	\label{fig:ipa+raa_dbscan}
\end{figure*}

\begin{figure*}[t]
	\centering
	\vspace{-0.1in}
	\hspace{-6cm}
	
	\begin{tabular}{lcc}

		\subfigure[\small{reduction rates of $Lat_{stage}^{(ex)}$, $Lat_{stage}^{(in)}$, $Cost_{stage}$}]
		{\label{fig:e2e-so_General_-metrics}\includegraphics[height=3.0cm,width=5.8cm]{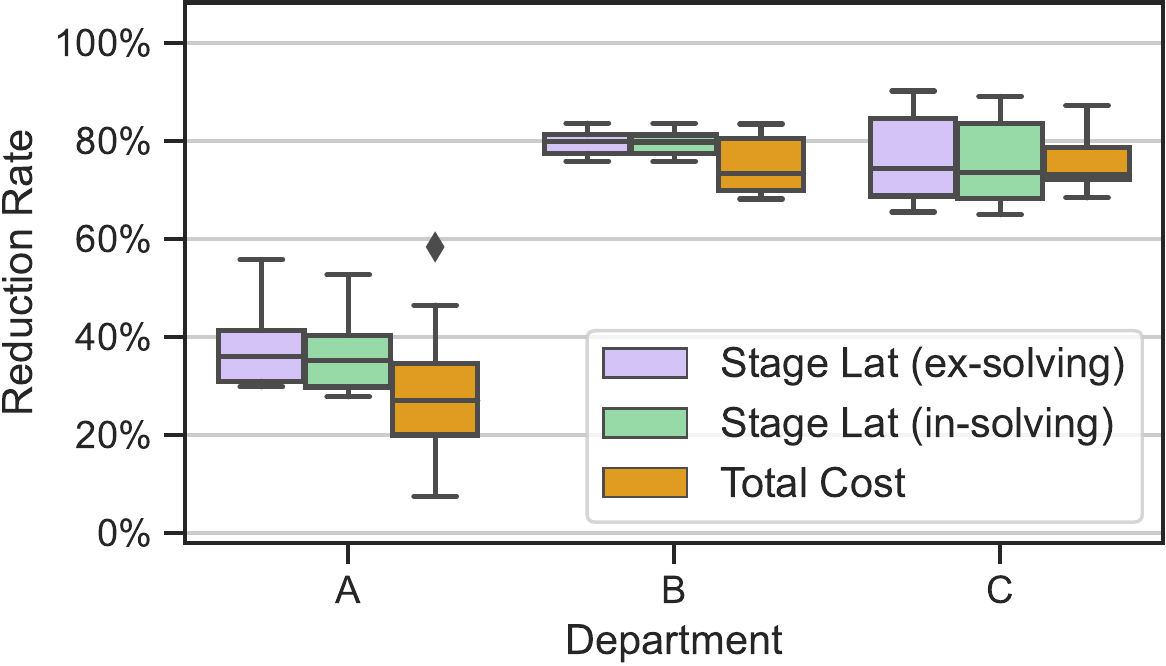}}

		&
		\subfigure[\small{reduction rates of $Lat_{stage}^{(ex)}$ on busy/idle}]
		{\label{fig:e2e-so_General_-busy-idle}\includegraphics[height=3.0cm,width=5.8cm]{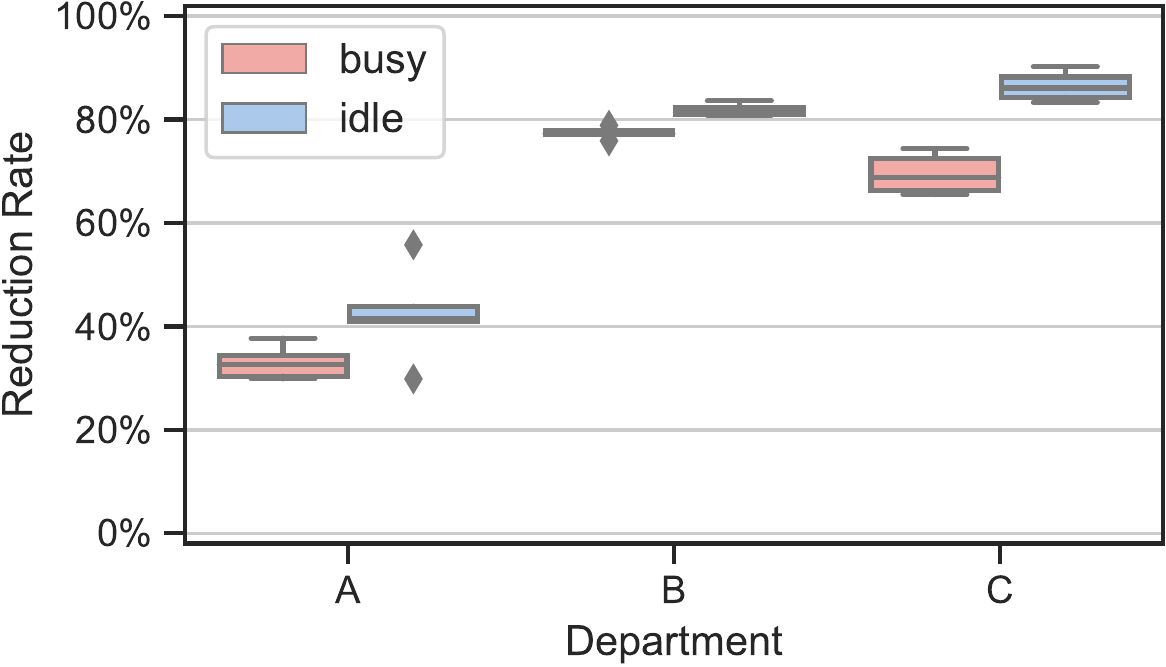}}

		&		
		\subfigure[\small{$T_{stage}$ for IPA(Cluster)}]
		{\label{fig:e2e-so_General_-fuxi-solving-time}\includegraphics[height=3.0cm,width=5cm]{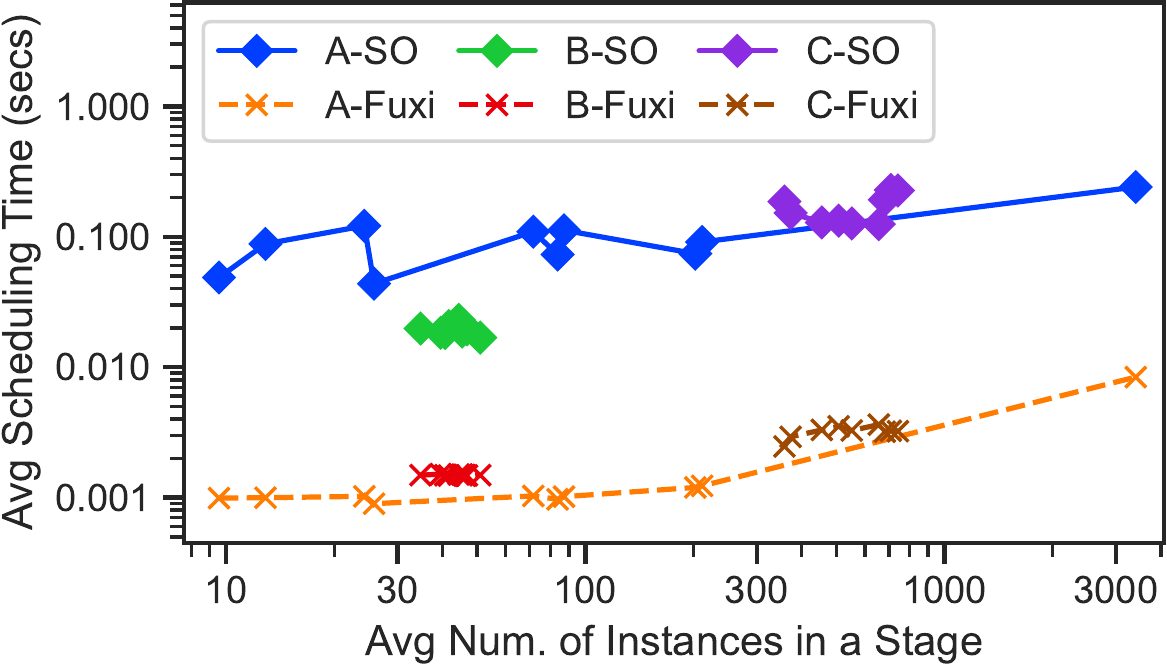}}
		
	\end{tabular}
	\caption{Comparative results between IPA+RAA(General) and Fuxi over 29 workloads}
	\label{fig:ipa+raa_inherit_general}
\end{figure*}

\begin{figure*}[t]
	\centering
	\vspace{-0.1in}
	\hspace{-6cm}
	
	\begin{tabular}{lcc}

		\subfigure[\small{reduction rates of $Lat_{stage}^{(ex)}$, $Lat_{stage}^{(in)}$, $Cost_{stage}$}]
		{\label{fig:e2e-so(inherit)-metrics}\includegraphics[height=3.0cm,width=5.8cm]{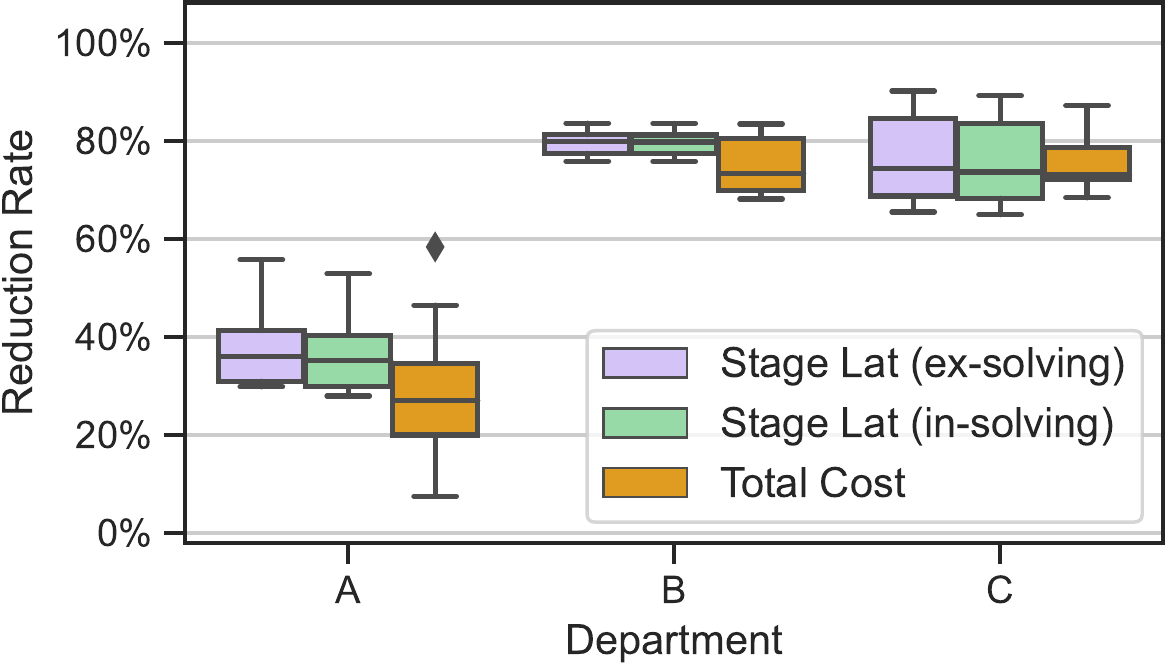}}

		&
		\subfigure[\small{reduction rates of $Lat_{stage}^{(ex)}$ on busy/idle}]
		{\label{fig:e2e-so(inherit)-busy-idle}\includegraphics[height=3.0cm,width=5.8cm]{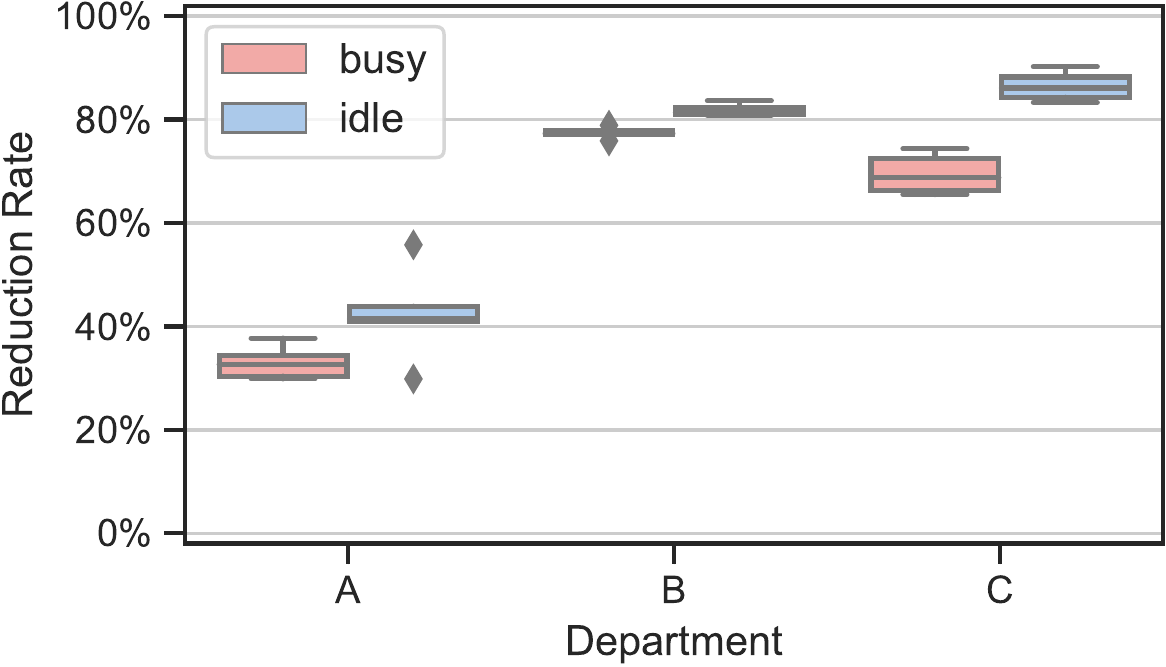}}

		&		
		\subfigure[\small{$T_{stage}$ for IPA(Cluster)}]
		{\label{fig:e2e-so(inherit)-fuxi-solving-time}\includegraphics[height=3.0cm,width=5cm]{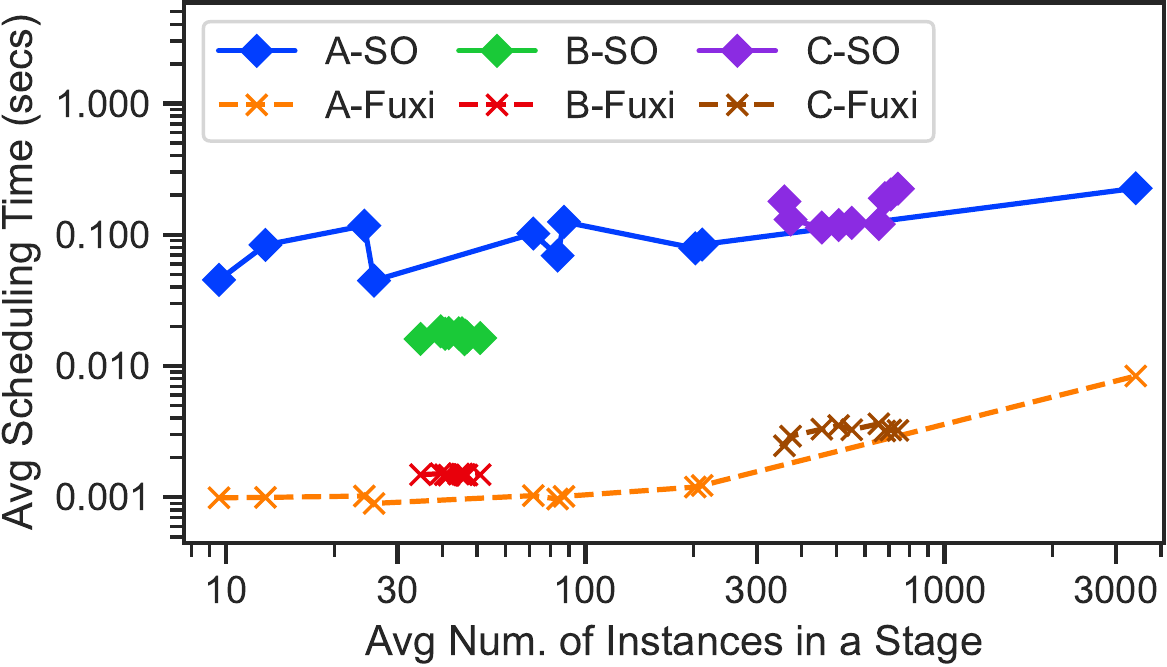}}
		
	\end{tabular}
	\caption{Comparative results between IPA+RAA(Path) and Fuxi over 29 workloads}
	\label{fig:ipa+raa_inherit_path}
\end{figure*}

We show more results on (1) stage-level latency reduction rates (both excluding and including the RO time overhead), (2) cloud cost reduction rates, (3) the solving time overhead for 
\verb|IPA+RAA(W/O_C)|, \verb|IPA+RAA(DBSCAN)|, \verb|IPA+RAA(General)|, \verb|IPA+RAA(Path)| in 
Figure~\ref{fig:ipa+raa_ic_o}, \ref{fig:ipa+raa_dbscan}, \ref{fig:ipa+raa_inherit_general}, and \ref{fig:ipa+raa_inherit_path} respectively, with \verb|IPA(Cluster)| as the default choice.

Among different workloads, \verb|IPA+RAA(Path)| reduces a lot more cloud cost than \verb|IPA(Cluster)| for workload C, which has more instances (505) per stage than B (42) and A (35), as shown in Fig.~\ref{fig:e2e-so(inherit)-metrics}.

Among the busy and idle cluster environments, \verb|IPA+RAA(Path)| shows a similar trend as \verb|IPA(Cluster)| to reduce more latency and cost when a cluster is idle, as shown in Fig.~\ref{fig:e2e-so(inherit)-busy-idle}.

We also show the average stage latency, cloud cost and solving time in Table~\ref{tab:expt-so-sub-workloads-raw}.

\subsection{More Analyses for Expt 8}

We also show the average stage latency, cloud cost and solving time in Table~\ref{tab:expt-so-sub-workloads-raw} in the last 6 rows for \texttt{EVO}, \texttt{WS(Sample)}, \texttt{PF(MOGD)}, \texttt{IPA+EVO}, \texttt{IPA+WS(Sample)}, and \texttt{IPA+PF(MOGD)}.

\subsection{More Analyses for Expt 10}

\begin{table}[t]
\ra{1}
\small
\addtolength{\tabcolsep}{1pt}
\newrobustcmd{\B}{\bfseries}
\begin{tabular}{lcc}
\toprule
Bootstrap Model    & Inst-latency (WMAPE) & Cost (GlbErr)       \\
\midrule                 
MCI+GTN    & 8.6\%, 19.0\%, 15.1\%      & 1.9\%, 5.4\%, 5.1\% \\
TLSTM    & 15.4\%, 30.6\%, 20.9\%     & 4.2\%, 14.2\%, 6.5\%  \\
QPPNet    & 22.3\%, 35.7\%, 27.9\%     & 11.8\%, 16.4\%, 14.3\% \\
\midrule
Bootstrap Model    & Stage   Lat (in):    & Cost:       \\ 
\midrule 
GTN+MCI	&	34\% ,49\% ,68\%	&	48\% ,41\% ,71\% \\
TLSTM	&	17\% ,2\%, 65\%	&	43\%, 31\%, 70\% \\
QPPNet	&	34\%, 1\%, 63\%	&	46\%, 30\%, 68\% \\
\bottomrule
\end{tabular}
\caption{Predictive Errors VS Reduction Rate}
\label{tab:bootstrip-model-accr-cost2}
\end{table}

Table~\ref{tab:bootstrip-model-accr-cost2} shows the model errors of instance latency and the cloud cost for (1) MCI+GTN, (2) TLSTM, (3) QPPNet together with the reduction rate achieved when those models are used as the bootstrap models. 
While the reduction rates of the latency in most stages decrease significantly when using a more inaccurate model TLSTM or QPPNet,
the reduction rates of the cloud cost only drop from 41\%-71\% to 31\%-70\% for TLSTM and 30\%-68\% for QPPNet.
The instance latency errors have a canceling effect on the cloud cost, and hence the model errors of cloud cost for TLSTM and QPPNet are still in a good range. 
Therefore, the impact of an inaccurate model on the reduction rate of the cloud cost is smaller than the stage-level latency.

\subsection{Breakdown Analyses on the End-to-End Performance}

\begin{table*}[t]
\ra{1}
\small
\addtolength{\tabcolsep}{2pt}
\newrobustcmd{\B}{\bfseries}
\begin{tabular}{lccc}
\toprule
Stages    & Short (<10s) & Median (10-100s) & Long (>100s)  \\
\midrule                 
 \% of stages that dominates the default RP in latency and cost     & 68\%  & 99\%  & 97\%  \\
Avg latency reduction (s)                                          & 0.9    & 11.7   & 201.8  \\
Avg latency reduction rate = (Avg latency reduction / Avg latency) & 46\% & 54\% & 65\% \\
Avg cloud cost reduction                                           & 0.002   & 0.019   & 1.697  \\
Avg cost reduction rate = (Avg cost reduction / Avg cost)          & 77\% & 62\% & 75\% \\
\bottomrule
\end{tabular}
\caption{Breakdown analyses on the end-to-end effects}
\label{tab:e2e-breakdown-cost2}
\end{table*}

\begin{figure*}[t]
	\centering
	\begin{tabular}{lcc}

		\subfigure[\small{Instances latency in Fuxi, IPA and IPA+RAA}]
		{\label{fig:e2e-bd1-eg1-lat}\includegraphics[height=6.5cm,width=5.5cm]{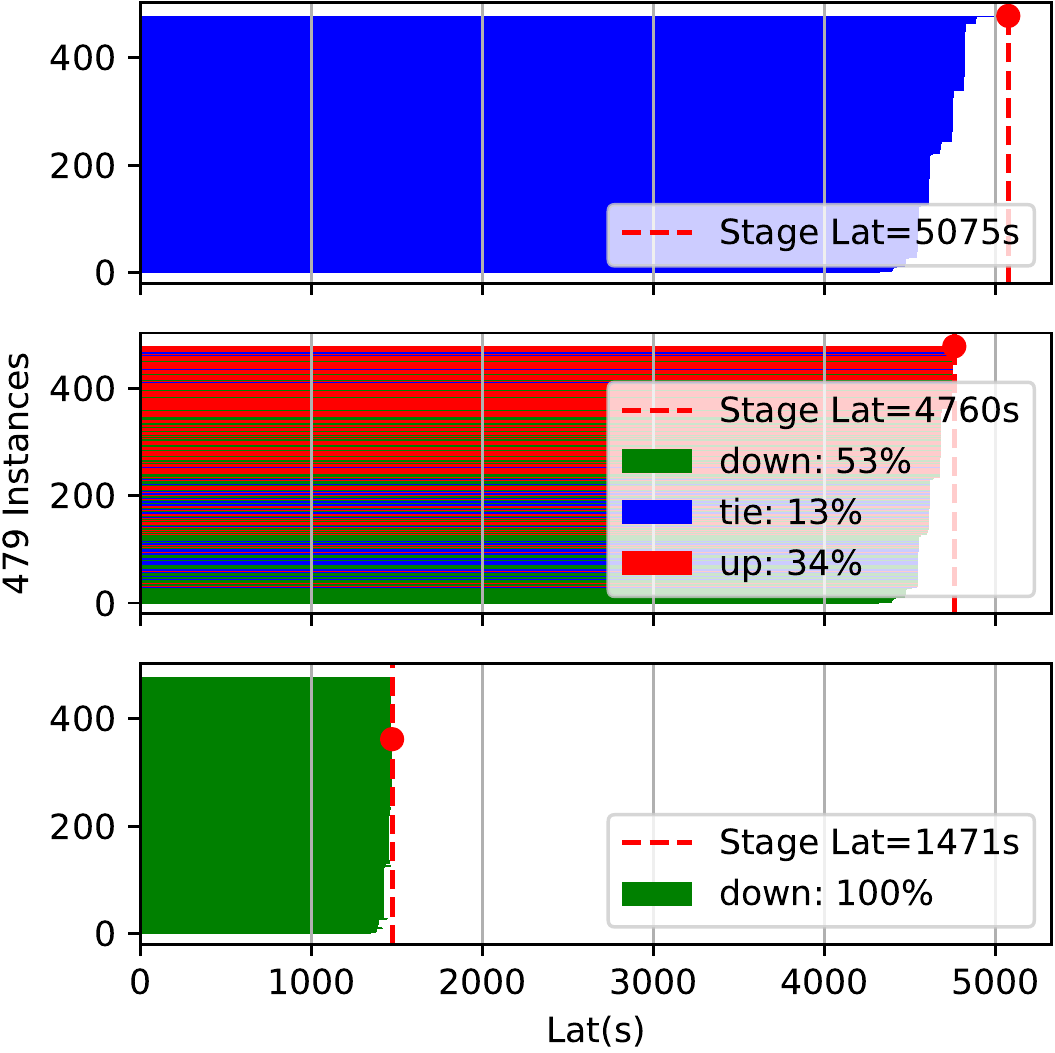}}

		&
		\subfigure[\small{Instances cost in Fuxi, IPA and IPA+RAA}]
		{\label{fig:e2e-bd1-eg1-cost}\includegraphics[height=6.5cm,width=5.5cm]{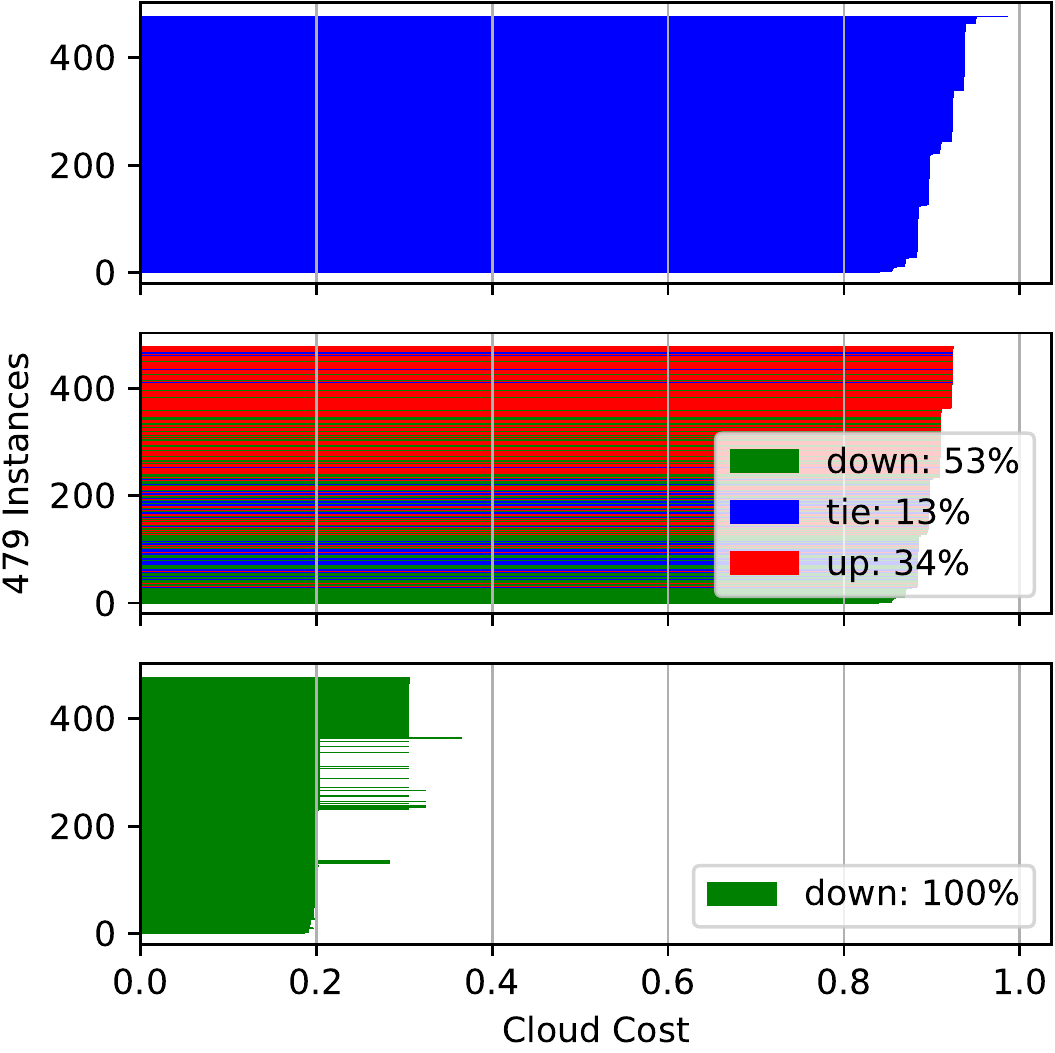}}

		&				
		\subfigure[\small{Instance cardinality, and resource plans}]
		{\label{fig:e2e-bd1-eg1-rp}\includegraphics[height=6.5cm,width=5.5cm]{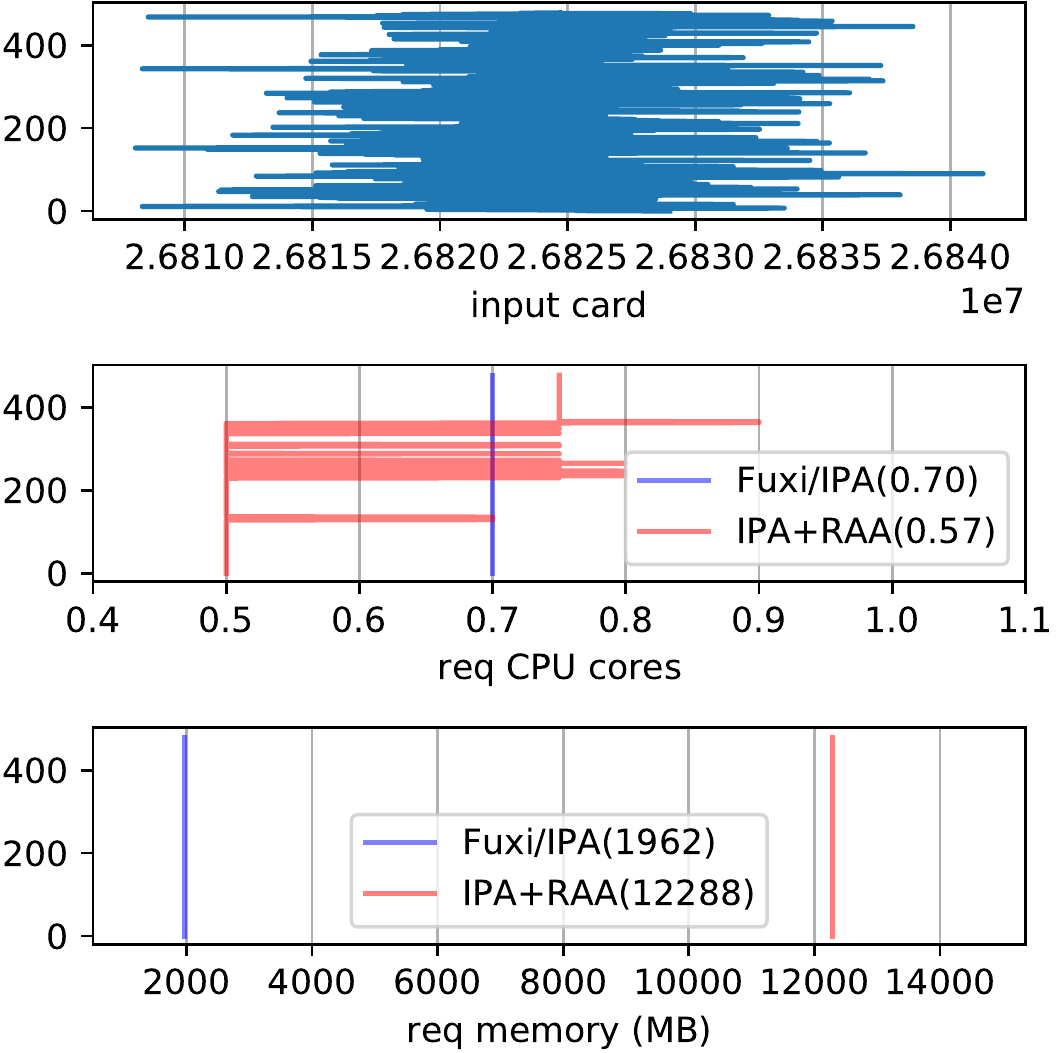}}
		
	\end{tabular}
	\caption{An example of stage breakdown analyses on the end-to-end performance}
	\label{fig:e2e-breakdown-eg1}
\end{figure*}

\begin{figure*}[t]
	\centering
	\begin{tabular}{lcc}

		\subfigure[\small{Instances latency in Fuxi, IPA and IPA+RAA}]
		{\label{fig:e2e-bd1-eg2-lat}\includegraphics[height=6.5cm,width=5.5cm]{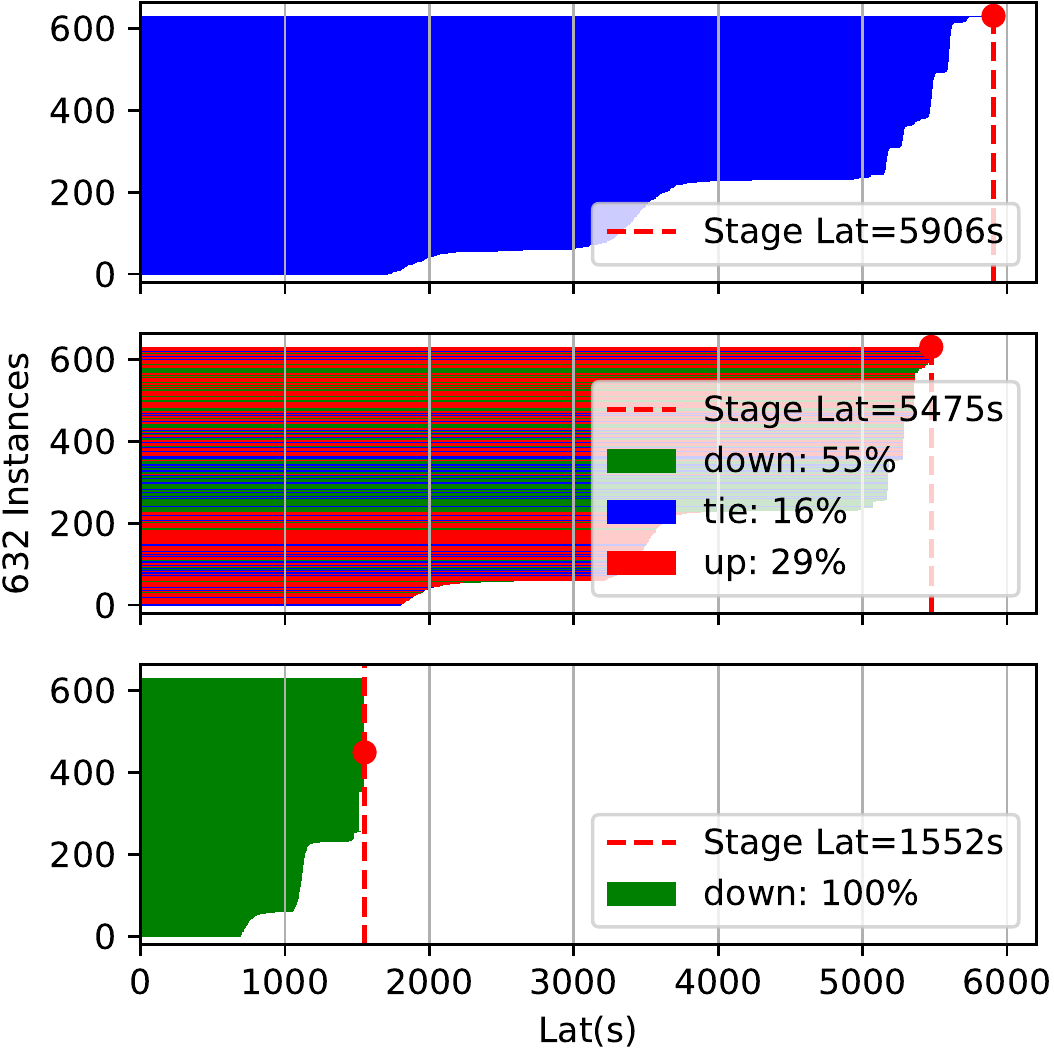}}

		&
		\subfigure[\small{Instances cost in Fuxi, IPA and IPA+RAA}]
		{\label{fig:e2e-bd1-eg2-cost}\includegraphics[height=6.5cm,width=5.5cm]{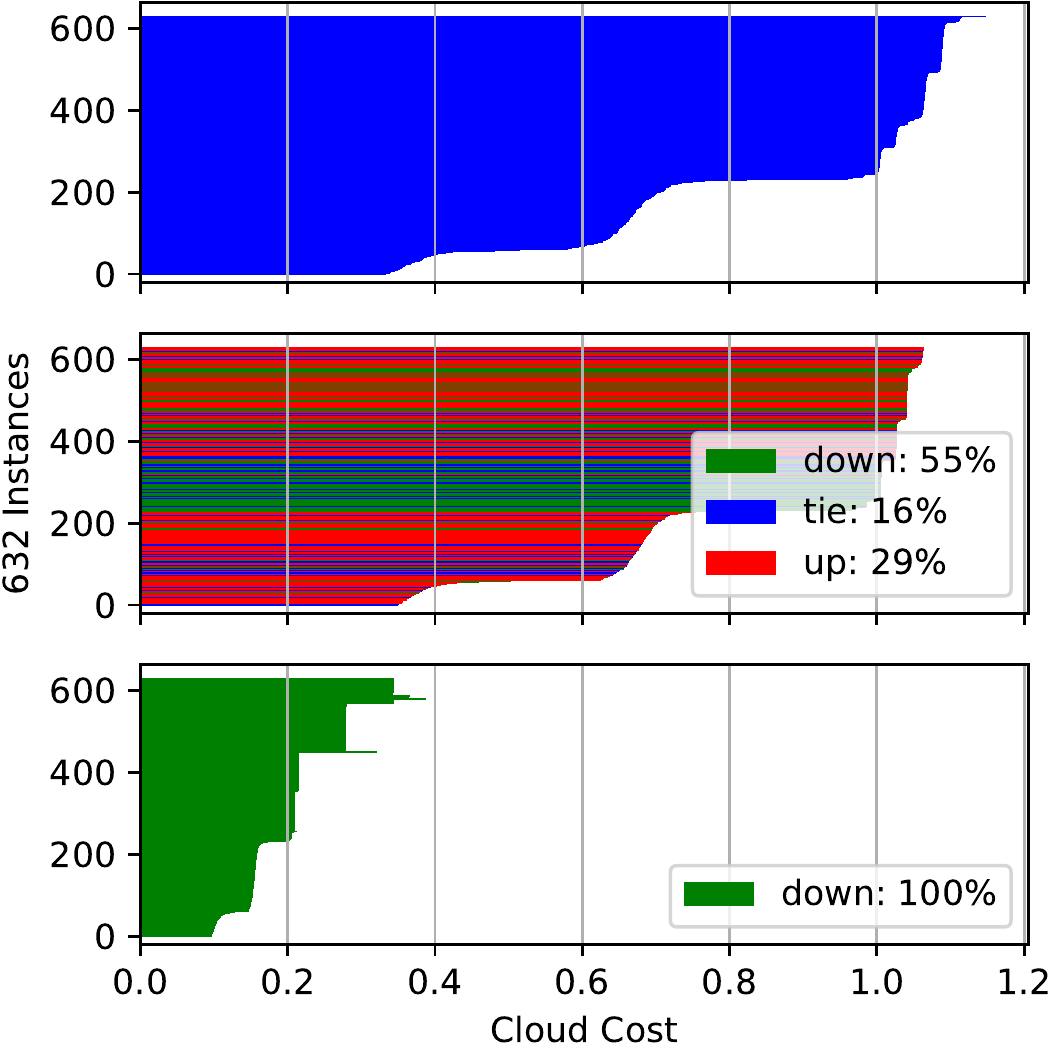}}

		&				
		\subfigure[\small{Instance cardinality, and resource plans}]
		{\label{fig:e2e-bd1-eg2-rp}\includegraphics[height=6.5cm,width=5.5cm]{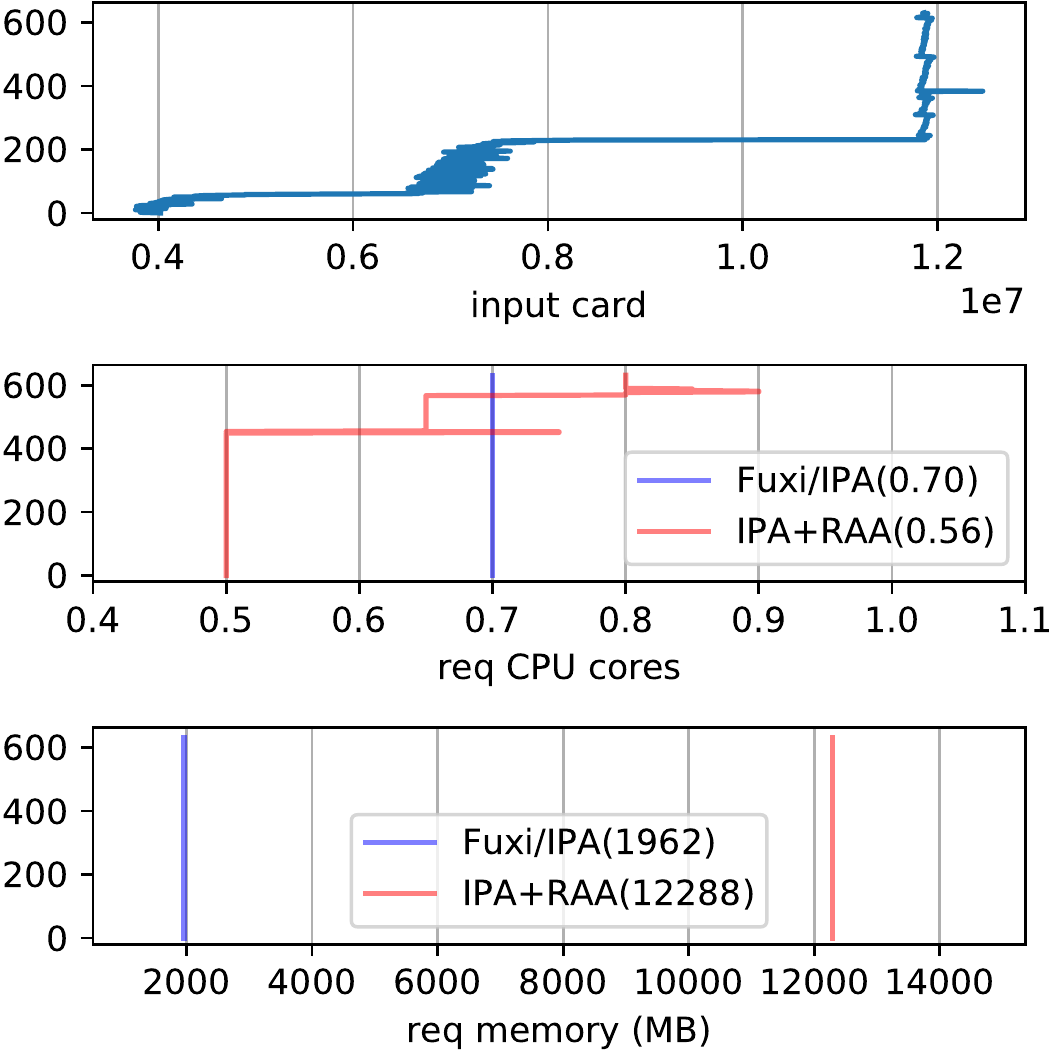}}
		
	\end{tabular}
	\caption{An example of uneven instance latency distribution}
	\label{fig:e2e-breakdown-eg2}
\end{figure*}

\begin{figure*}[t]
	\centering
	\begin{tabular}{lcc}

		\subfigure[\small{Instances latency in Fuxi, IPA and IPA+RAA}]
		{\label{fig:e2e-bd1-eg3-lat}\includegraphics[height=6.5cm,width=5.5cm]{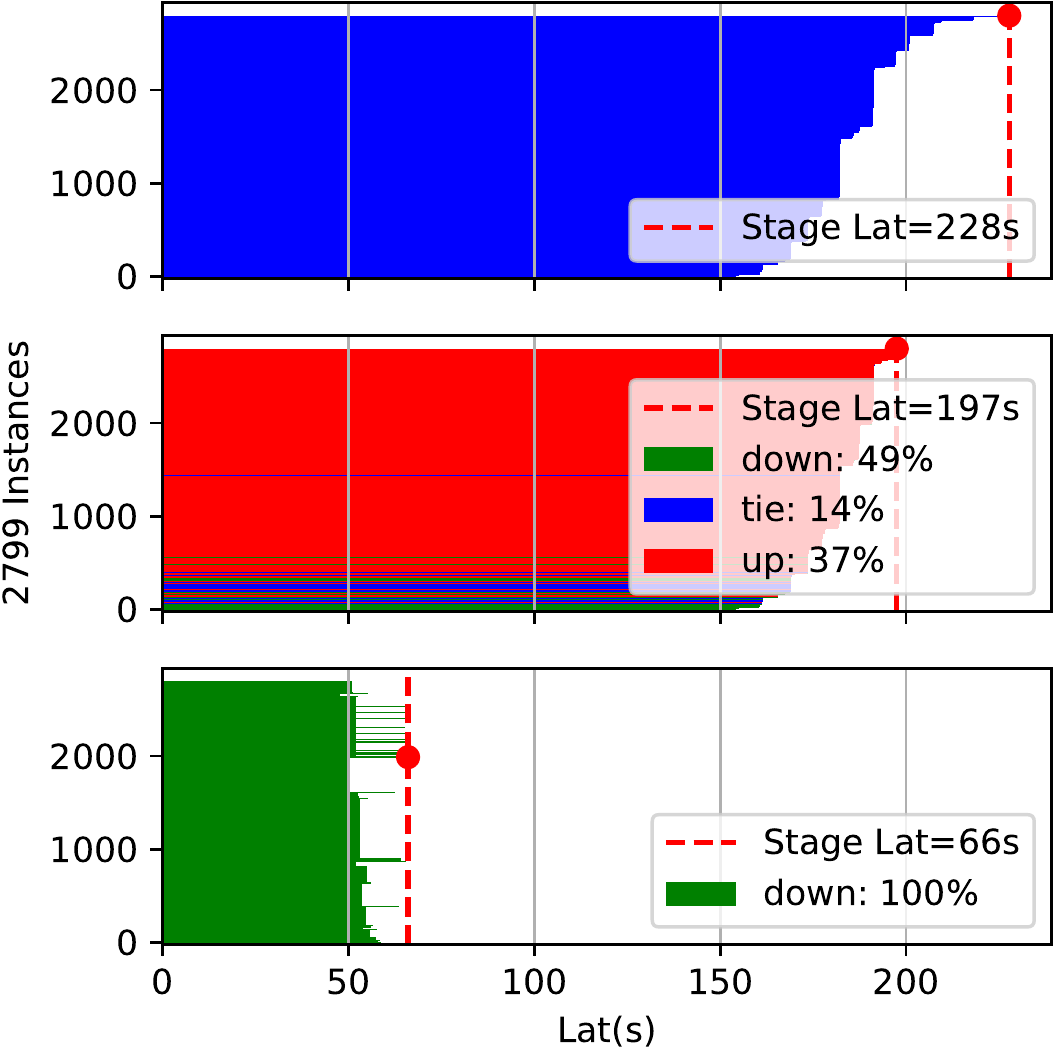}}

		&
		\subfigure[\small{Instances cost in Fuxi, IPA and IPA+RAA}]
		{\label{fig:e2e-bd1-eg3-cost}\includegraphics[height=6.5cm,width=5.5cm]{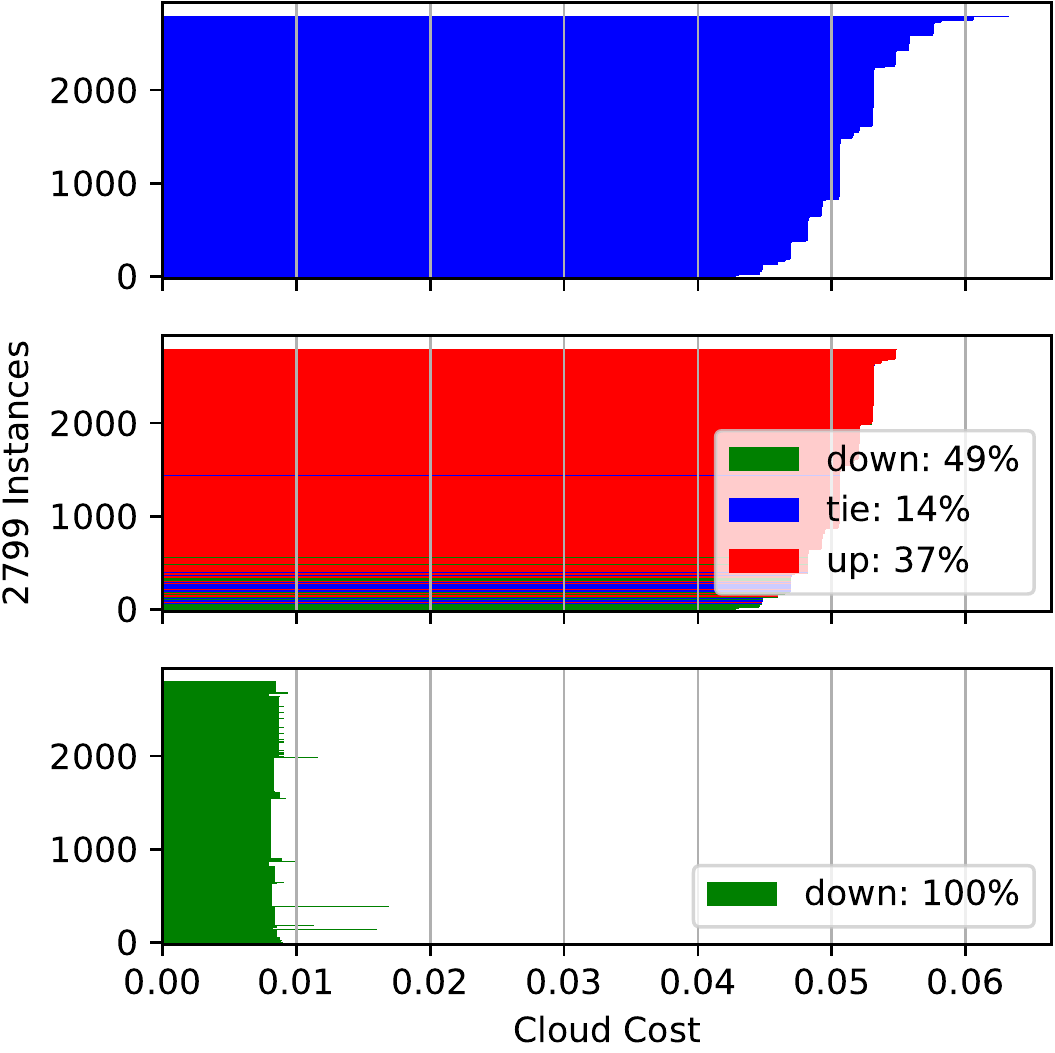}}

		&				
		\subfigure[\small{Instance cardinality, and resource plans}]
		{\label{fig:e2e-bd1-eg3-rp}\includegraphics[height=6.5cm,width=5.5cm]{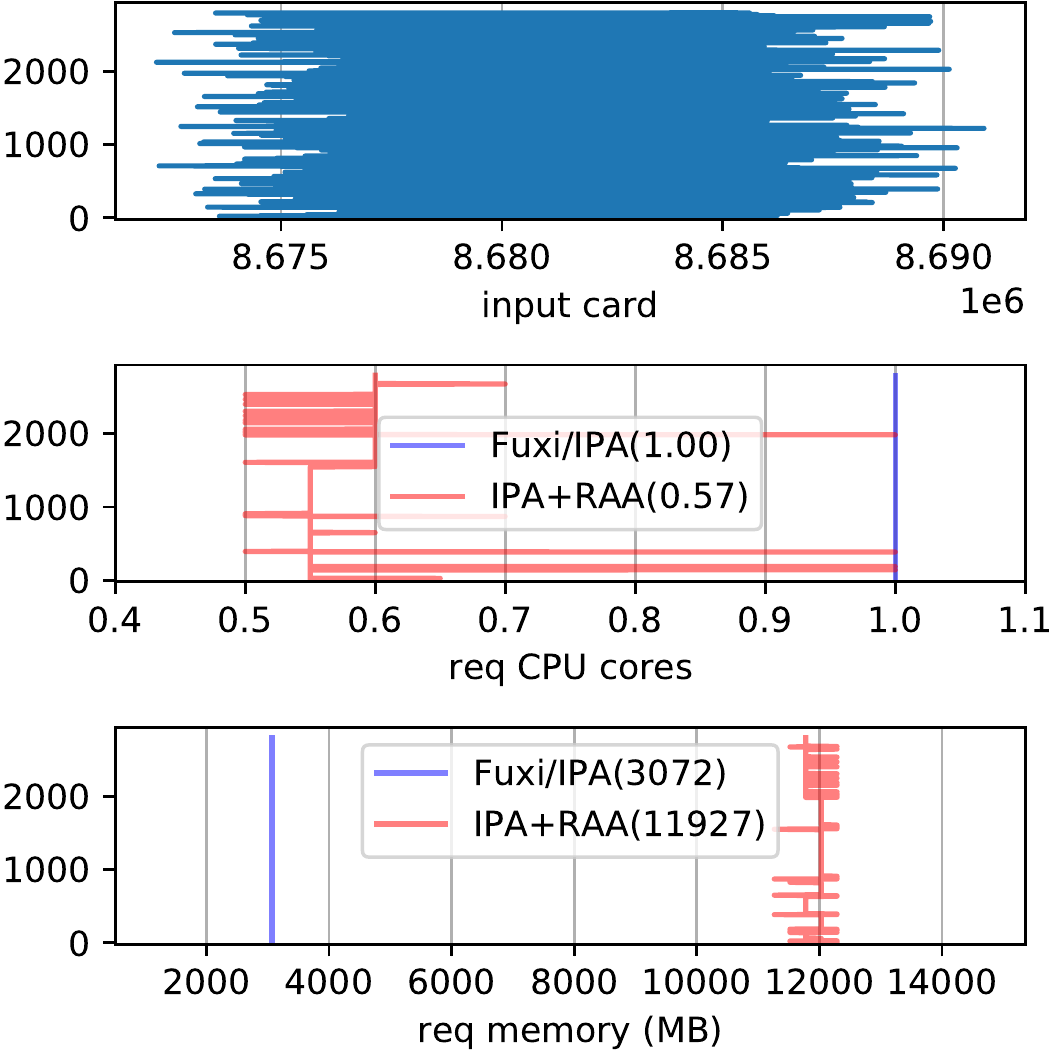}}
		
	\end{tabular}
	\caption{An example of instance latency spikes}
	\label{fig:e2e-breakdown-eg3}
\end{figure*}

To investigate the effects of IPA and RAA on the stages, we group all 1.9M stages into short-, median- and long-running categories based on the default stage latencies and show the statistics in Table~\ref{tab:e2e-breakdown-cost2} when assuming the model is accurate. 
Our IPA+RAA dominates the Fuxi policy in latency and cost on 93\%-97\% of the stages in each category, with a 34-66\% average latency reduction rate and 51-79\% average cost reduction rate.
Specifically, our approach reduces more latency and cost for long-running stages compared to the median- and short-running stages. 

Figure~\ref{fig:e2e-breakdown-eg1} shows the performances of the 479 instances (with similar input cardinality) in a long-running stage, regarding their latency, cost, cardinality, and resource plans.
Figure~\ref{fig:e2e-bd1-eg1-lat} and \ref{fig:e2e-bd1-eg1-cost} show the latency and cost of each instance by using Fuxi (row 1), IPA only (row 2), and IPA+RAA (row3). 
IPA reduces the stage latency by reducing the latency for 52\% instances and increasing the latency for the 34\% instances. 
IPA+RAA reduces more latency and cost after tuning the resource plan. In this example, IPA+RAA adjusts CPU cores and assigns x6 memory size for each instance and dominates the Fuxi policy in both latency and cost.

\subsection{Diagnoses of Unexpected Behaviors}

\begin{figure*}[t]
	\centering
    \includegraphics[width=.4\textwidth]{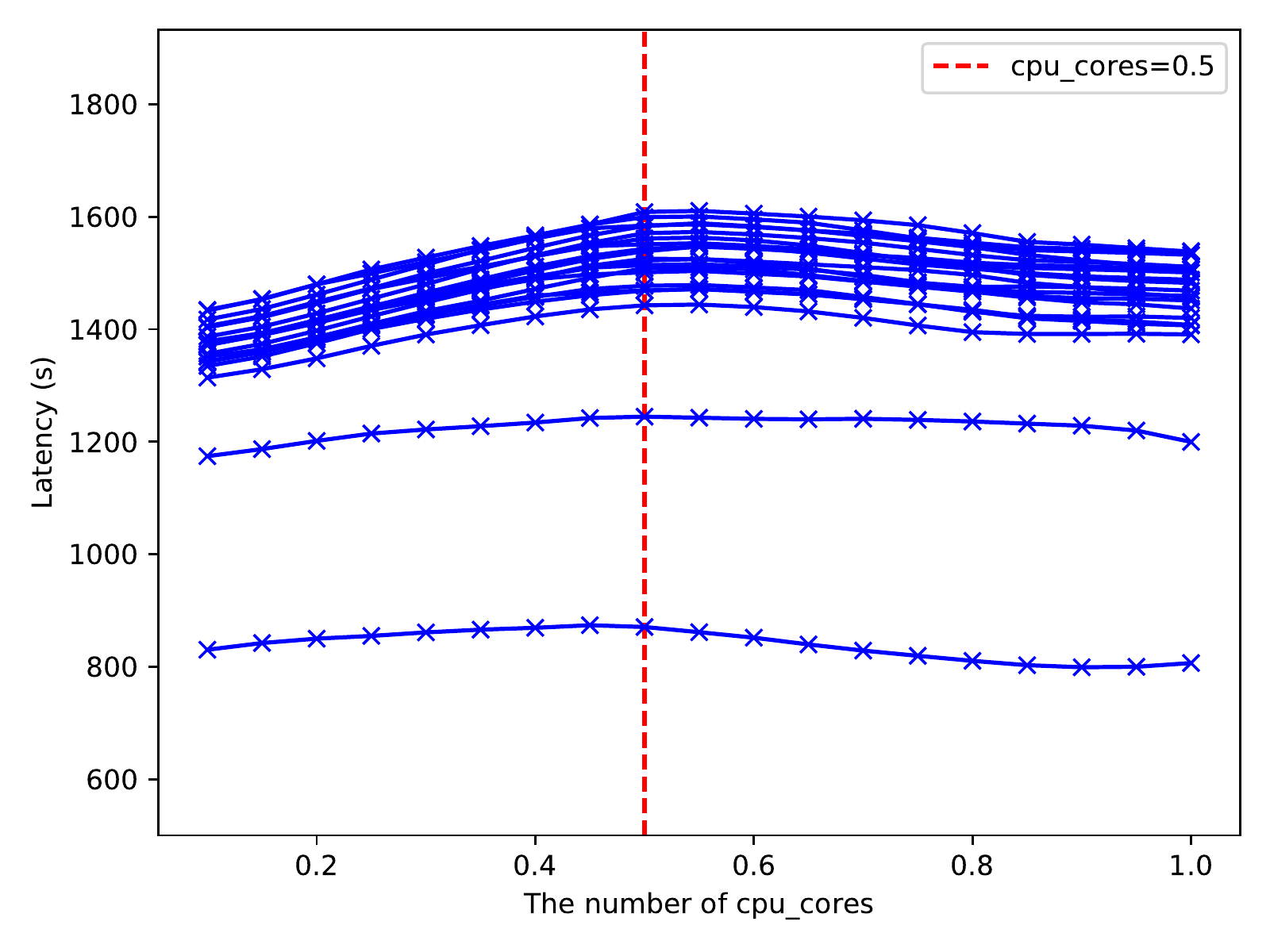}
    \caption{Latencies distribution over different number of cpu cores}
    \label{fig:uneven_diag_diff_cpu_cores}
\end{figure*}

\begin{figure*}[t]
	\centering
	\begin{tabular}{lcc}

		\subfigure[\small{}]
		{\label{fig:e2e-eg3-lat-all-clusters}\includegraphics[height=5.0cm,width=5.5cm]{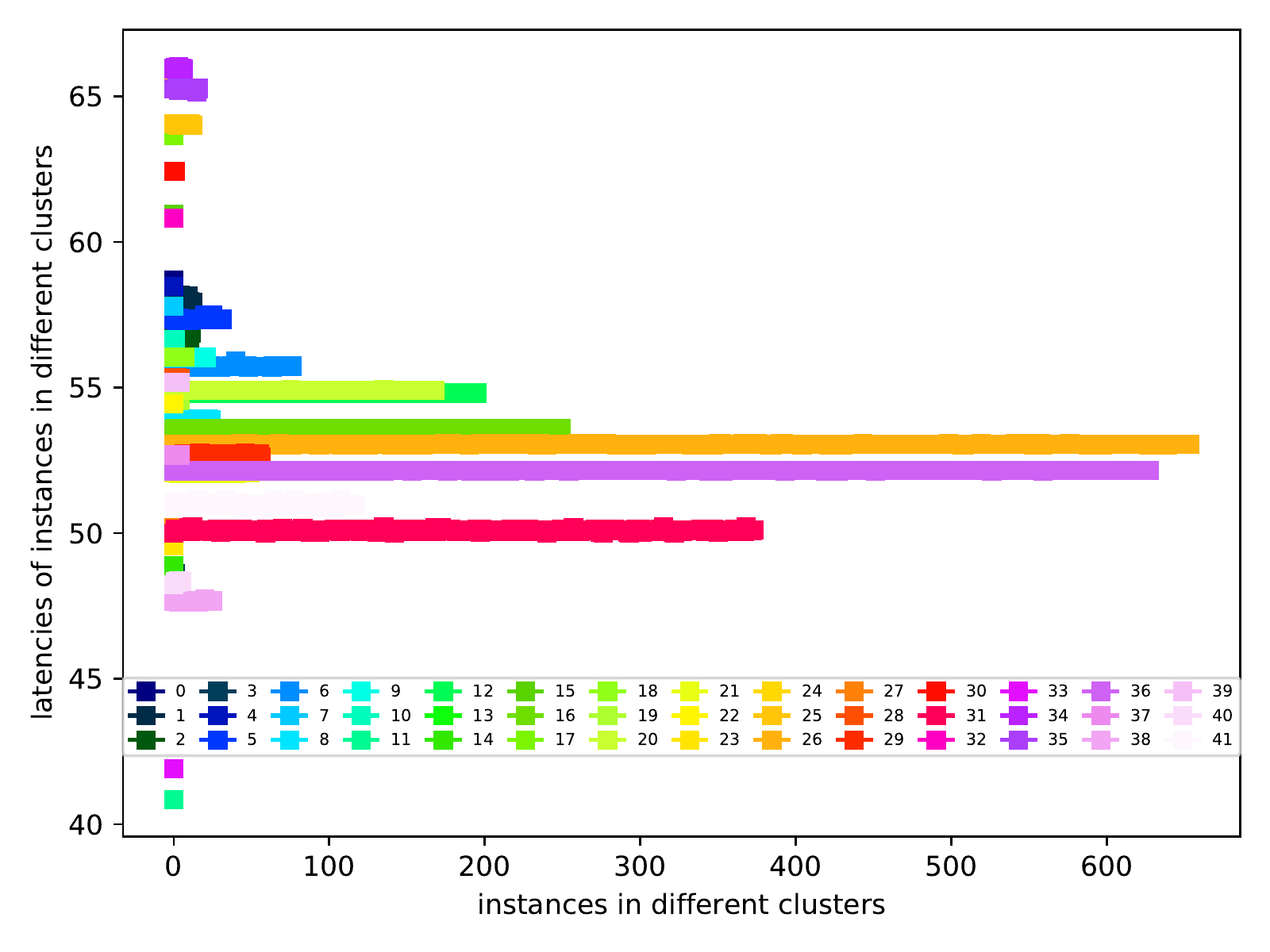}}

		&
		\subfigure[\small{}]
		{\label{fig:e2e-eg4-lat-all-clusters}\includegraphics[height=5.6cm,width=6.0cm]{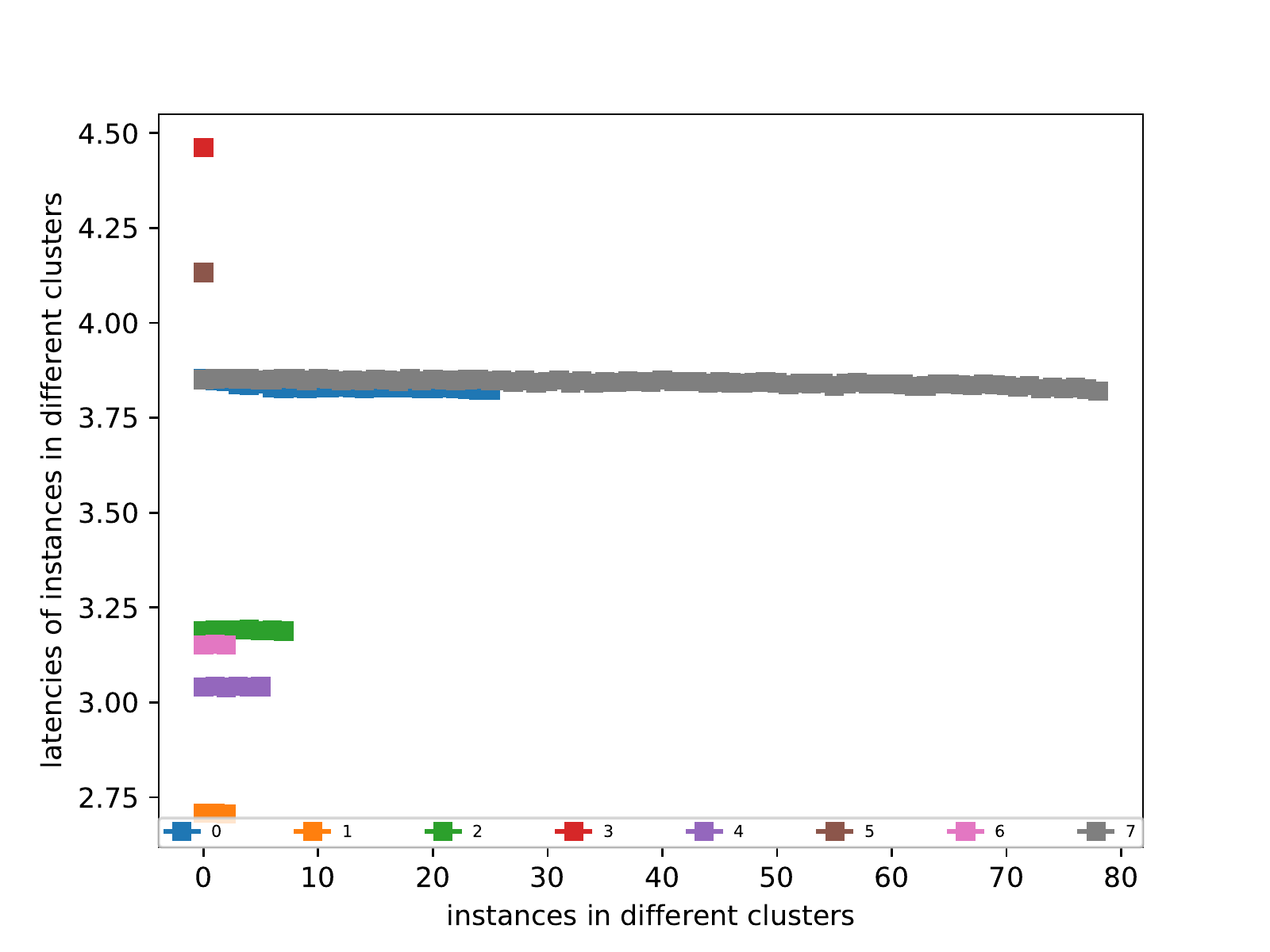}}

		&				
		\subfigure[\small{}]
		{\label{fig:e2e-eg5-lat-all-clusters}\includegraphics[height=5.0cm,width=5.5cm]{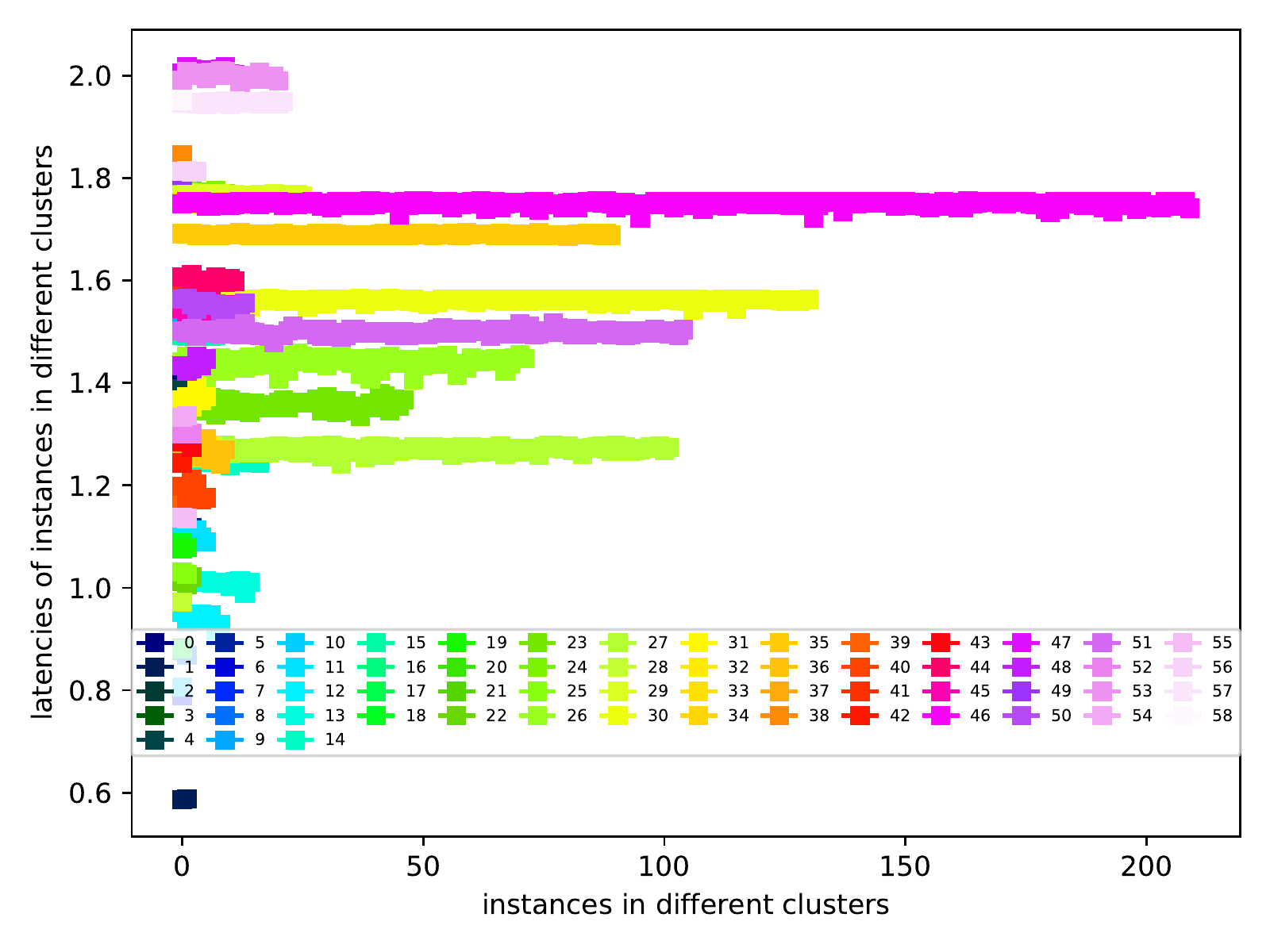}}
		
	\end{tabular}
	\caption{Examples of instance latencies in different instance clusters}
	\label{fig:e2e-breakdown-eg345}
\end{figure*}

Ideally, the IPA+RAA approach should increase the resources for long-running instances and decrease the resources for short-running instances in a stage to make the instances more even and, in the meantime, hit an instance-level \po point. 
While the IPA+RAA approach manages to reduce the stage latency and cost significantly, we observed the uneven behavior (Figure~\ref{fig:e2e-bd1-eg2-lat}) and spikes (Figure~\ref{fig:e2e-bd1-eg3-lat}) in the instance-level latency distribution after the optimization.

Now we enumerate the potential factors that cause the uneven behavior and spikes. 

First, the latency predictive model could have nonintuitive down-linear behaviors. E.g., 
Figure~\ref{fig:uneven_diag_diff_cpu_cores} shows the predicted latency over the CPU cores given a fixed memory size over several representative instances in a stage. Notice when the requested CPU is less than 0.5, the predicted latency decreases by using fewer cores, which is problematic.
The improper modeling results underestimate the latency of the short-running instances when the cores are small. So we would miss some \po points and get the unexpected behaviors.
It turns out that although our data set is quite large, it still does not have comprehensive coverage of the parameter space on the resource plan. E.g. among the 34M instances in workload A, we only observed 38 resource plans.

Second, we only consider the resource plan in a narrow range due to the limited coverage of resource plan space in the training data. 
The searching space of the resource plan is still in a narrow range even though being relaxed by assuming the choices of requested cores and memory are independent in their min-max ranges. 
However, when the resource plan is outside the searching space, the model is not guaranteed to function properly, preventing us from finding better solutions. Specifically, we cannot lower the cores anymore in the resource plan to enlarge the latencies of the short-running instances (and hence could reduce the cost) to make the distribution even. For example, the first 200 instances in Figure~\ref{fig:e2e-bd1-eg2-rp} have already reached the lowest cores in the searching space and cannot do things anymore.

Lastly, one may think our clustering could bring prediction errors when assuming the latencies of the instances in a cluster are the same, and hence contributes to the spikes.
However, according to our observations, the clustering works functionally well. 
E.g., Figure~\ref{fig:e2e-breakdown-eg345} shows the instance latency distributions of three stages, where instances from the same cluster are given the same color.
In each stage, the latencies of the instances from the same cluster are all quite close, verifying the good clustering performance.
In addition, it is worth mentioning that when we eliminate the factor of the clustering, the issue of spikes could degenerate into an uneven behavior, where instances of the spike latencies could belong to one cluster of fewer members.

Therefore, our follow-up work will include (1) augmenting the model with more offline sampling over the parameter space and (2) taking the model uncertainty (the expectation and variance of the latency prediction) into the concern.
}

\end{document}